\documentclass[tikz,12pt]{article}
\usepackage{mheckII}

\usepackage[T1]{fontenc}

\usepackage{tikz-feynman}

\usepackage{manfnt}
\usepackage{longtable}

\usepackage{fancybox}

\usepackage{multirow}

\usepackage{blkarray}

\usepackage{tikz} 
\usetikzlibrary{matrix}

\theoremstyle{theorem}
\newtheorem*{thm}{Theorem}

\newtheorem{fact}{Fact}

\newtheorem{corl}{Corollary}
\newtheorem{pro}{Proposition}

\newtheorem*{conj}{Conjecture}

\newtheorem{lem}{Lemma}
\newtheorem*{lem*}{Lemma}

\newtheorem{que}{Question}
\newtheorem*{hope}{Hope}

\newtheorem*{propo}{Proposal}
\newtheorem*{ans}{Answers}

\newtheorem*{prin}{Principle}
\newtheorem*{nec}{Necessary condition}

\newtheorem*{inverse}{Inverse Problem}

\newtheorem*{almost}{``Almost'' sufficient condition}

\newtheorem*{folk}{Folk-theorem}
\newtheorem*{scen}{Scenario}

\theoremstyle{definition}
\newtheorem{defn}{Definition}

\newtheorem{rem}{Remark}

\newtheorem*{sit}{Special situation}
\newtheorem*{warn}{Disclaimer}
\newtheorem{cave}{\begin{scriptsize}\textdbend\end{scriptsize} {Caveat}}
\newtheorem{exe}{Example}[section]

\def\Dsl{\,\raise.15ex\hbox{/}\mkern-13.5mu D}
\def\dsl{\,\raise.25ex\hbox{/}\mkern-10.5mu \partial}

\renewcommand{\arraystretch}{2.2}

\title{Direct and Inverse Problems\\ in Special Geometry 
}

\authors{Sergio Cecotti\footnote{e-mail: {\tt cecotti@sissa.it}, {\tt cecotti@bimsa.cn}}\vskip 9pt

\centerline{Yanqi Lake Beijing Institute of Mathematical Sciences and Applications (BIMSA)}
\centerline{Huaibei Town, Huairou District, Beijing 101408, China}
\centerline{and}
\centerline{Qiuzhen College, Tsinghua University, Beijing, China,}
\centerline{SISSA,
via Bonomea 265, Trieste, Italy}
}

\abstract{The \emph{inverse problem of special geometry} (Seiberg-Witten geometry of 4d $\cn=2$ SCFT)
asks for a recursive construction of all such geometries in rank $r$ by assembling together
known lower-rank ``strata''.  This leads to a program to understand/construct/classify  
all special geometries which looks surprising effective. After reviewing some advanced topics in special geometry, in this long note  we define the inverse problem and introduce the basic tools of the trade.
The program is essentially completed in rank 2, and we pave the way to proceed to higher ranks. A central role 
is played by various notions of geometric rigidity: in addition to the obvious one (triviality of the conformal manifold), Falting-Saito-Peters rigidity and Deligne-Simpson rigidity also enter in the story.}

\begin{document}
\maketitle

\tableofcontents

\newpage 

\section{Introduction: the Inverse Problem}

A first fundamental step in any program to understand/construct/classify all 4d $\cn=2$ SCFTs 
is to understand/construct/classify their \emph{special geometries} (a.k.a.
Seiberg-Witten geometries \cite{SW1,SW2,D1,D2}). Such a geometric program has been especially pursued
by Argyres and coworkers \cite{M1,M2,M3,M4,M5,M6,M7,M8,M9,M10,M11,M12,M13,M14,M15,M16}\footnote{\ For a recent review see \cite{Martone:2020hvy}.}; for other work in the same line see e.g.\! \cite{caorsi,C1,C2,Cecotti:2021ouq,Cecotti:2023mlc}.
This paper aims to introduce new powerful techniques and results in this direction.
\medskip   

A \emph{special geometry}\footnote{\ See section \ref{sec2} for precise definitions.} is a very nice mathematical
gadget, with multiple physical interpretations and deep relations with several other topics, which may be studied from 
diverse math perspectives. A special geometry may be seen as the (complexification)
of a classical integrable mechanical system \cite{D1,D2} or as a non-perturbative description of
a four-dimensional SUSY QFT \cite{SW1,SW2}: we have \emph{two} physical intuitions about them -- one classical and one quantum -- and the interplay between the two yields new unexpected insights. 
Special geometry in our SCFT sense looks very similar to the special geometry of
$\cn=2$ supergravity \cite{SG1,SG2,SG3,SG4,SG5,SG6,Cecotti:2015wqa}: formally the two are related by a simple ``Wick rotation'' in the signature of the fiber metric. The supergravity geometry describes the variation of Hodge structure (VHS \cite{VHS1,VHS2,VHS3,VHS4,griffithsMT,griffithsMT2,periods,BG,VHS5})
of a family of Calabi-Yau 3-folds (3-CY) \cite{SG3,SG4,SG5,SG6,Cecotti:2015wqa}, and all ideas and techniques introduced in the study
of moduli geometries of 3-CY can be equally well applied to the special geometries of $\cn=2$ SCFT.
Strangely enough, it seems that the previous SCFT literature did not leverage on the
deep extensive work done in the Calabi-Yau contest. The 3-CY methods will be instead central in our discussion.
Morally speaking, the program of the geometric classification of all 4d $\cn=2$ SCFT may be
seen as a simpler analogue of classifying all Calabi-Yau 3-folds, a basic problem with fundamental implications for quantum gravity \cite{SG6}.

\medskip

The main focus of this note is the \emph{inverse method} in special geometry. 
While the paper is somehow written in a review format, essentially all results in this direction are novel,
including the very notion of ``inverse problem''. In the rest of this introduction we outline what we mean
by ``inverse problem'': a more precise description of the method  is given in section \ref{s:invpro} in the proper 
geometric context.

\subparagraph{Special geometry.} In this paper a \emph{special geometry}  is a 
holomorphic integrable system\footnote{\ The integral systems of interest are actually \emph{algebraic} not just \emph{holomorphic}.}
associated to the Seiberg-Witten geometry of a 4d $\cn=2$ SCFT
with all its implied structures \cite{D1,D2}. Thus a special geometry is a (smooth) complex symplectic $\mathscr{X}$
variety together with a holomorphic fibration $\pi\colon \mathscr{X}\to \mathscr{C}$
over the Coulomb branch $\mathscr{C}$ with Lagrangian fibers
which generically are polarized Abelian varieties. 
The Coulomb branch $\mathscr{C}$ is the spectrum of the chiral ring $\mathscr{R}$,
hence an affine complex variety.
Most of our arguments work in the broader context of 4d $\cn=2$
QFT, but we focus on the superconformal case where the geometry carries an additional structure:
a $\C^\times$-action on $\mathscr{X}$ by automorphisms
 of all the relevant geometric structures.\footnote{\ Special geometries with this additional structure will be called $\C^\times$\emph{-isoinvariant.}}
The $\C^\times$-action reflects
 the existence of a conformal $U(1)_R$ symmetry which is spontaneously broken
 to a discrete subgroup everywhere in $\mathscr{C}$ but at a single closed point $0$ (the origin).
When the Coulomb branch is \emph{smooth} at $0$ (as a $\C$-affine variety)
 $\mathscr{C}\simeq\C^r$. We shall be concerned mainly, but not exclusively, with this case.

\subparagraph{``Stratification''.} The quantum field theoretic intuition points out that a special geometry $\mathscr{X}\to \mathscr{C}$ comes with a natural  ``stratification'' in
 lower-dimension special geometries \cite{M12,M13}. We write the word \emph{``stratification''}
 between quotes since the QFT intuition leads to a fairly new\footnote{\ \emph{Fairly new} $\equiv$ the author is not aware of any previous mention of this idea in the math literature.}
notion in symplectic geometry, which is \emph{not} a stratification in any usual sense, and lays well  
outside the horizon of classical Liouvillian physical intuition. 
In math terms the ``stratification'' is a generalization of the Marsden-Weinstein-Meyer symplectic quotient \cite{sympp,sasaki} which applies to
singular loci with suitable physical/geometric properties where the usual symplectic quotient becomes meaningless. The ``stratification'' of $\mathscr{X}$ arises from an ordinary stratification of the Coulomb branch $\mathscr{C}$ \cite{M12,M13}.   

We outline the physical heuristics beyond ``stratification''.
At the generic point of $\mathscr{C}$ the IR physics, while not necessarily boring, is relatively simple,
 but there are special loci $\mathscr{S}_a\subset\mathscr{C}$ where something ``more interesting'' happens. At a special point  $u\in\mathscr{C}$
 three new phenomena may happen:
 \begin{itemize}
\item[(1)] additional degrees of freedom become massless;
\item[(2)] the unbroken subgroup $R_u\subset U(1)_R$ of R-symmetry enhances;
\item[(3)] we have both extra light degrees of freedom and R-symmetry enhancement. 
\end{itemize}
In an interacting theory the special locus
 $\cd\subset \mathscr{C}$ where we have additional massless degrees of freedom (called the \emph{discriminant})
 is a complex analytic subspace of pure codimension 1. 
We identify $\cd$ with a simple effective divisor and write $\cd=\sum_a\cd_a$ for its decomposition into irreducible components. 
 In the SCFT case the $\cd_a$'s are preserved by the $\C^\times$-action.
 
We focus on a point $u\in\mathscr{C}\setminus\cd$ which is at distance $\epsilon$ from a generic point $u_0\in\cd_a$ of the $a$-th discriminant component.
The BPS states that are massless at $u_0$ are parametrically light at $u$;
since $\cd_a$ has codimension-$1$ in $\mathscr{C}$
these states are charged under a rank-$1$ gauge group. Thus, asymptotically as $u\to u_0$  the light degrees of freedom
are described by 
an effective IR $\cn=2$ theory which 
 looks like a system composed of $(r-1)$ IR-free vector-multiplets and a rank-$1$ model which can be either IR-free or
 an interacting SCFT in its own right.
This suggests that the special geometry in an infinitesimal neighborhood of $\cd_a$
should look as a rank-$1$ special geometry describing the light particles
 fibered over the local rank-$(r-1)$ geometry which describes the heavy states. 
 While the actual situation may be a bit more intricate,\footnote{\ See \cite{Cecotti:2023mlc}.} 
 this physical intuition yields a first \emph{rough} cartoon of the geometry in the vicinity of $u_0$.
 Making the idea a little more precise, one defines a rank-$(r-1)$ special geometry
 $\mathscr{X}_a\to\cd_a$ over the codimension-$1$ locus $\cd_a\subset \mathscr{C}$.
 The local geometry in an infinitesimal neighborhood of $\cd_a$ is fibered over this
 rank-$(r-1)$ special geometry.
 Iterating the construction, one gets a rank-$(r-s)$ special geometry
 $\mathscr{X}_{a_s}\to\mathscr{C}_{a_s}$ over each subvariety in a finite set 
 $\{\mathscr{C}_{a_s}\colon a_s\in A_s\}$ of
 codimension-$s$ subvarieties of $\mathscr{C}$ for all $0\leq s\leq r$.
 By the ``stratification'' of the special geometry $\mathscr{X}\to\mathscr{C}$ we mean the full collection
 of special geometries 
 \begin{equation}
 \big\{\mathscr{X}_{a_s}\to\mathscr{C}_{a_s}\colon a_s\in A_s,\ 0\leq s\leq r\big\}.
 \end{equation}
 The open sets $\mathring{\mathscr{C}}_{a_s}\equiv \mathscr{C}_{a_s}\setminus (\cup_{a_{s-1}\in A_{s-1}}
 \mathscr{C}_{a_{s-1}})\cap \mathscr{C}_{a_s}$ form a stratification of the Coulomb branch in the usual sense.
 The closures of the codimension-$(s+1)$ strata, $\mathscr{C}_{a_{s+1}}\subset \mathscr{C}$, are the irreducible
 components of the discriminants of the special geometries $\mathscr{X}_{a_s}\to\mathscr{C}_{a_s}$ over the (closed) codimension-$s$ strata.
 In particular the $\mathscr{C}_{a_1}$'s coincide with the discriminant component $\cd_{a_1}$ of $\mathscr{X}$. The geometry along a codimension-1 stratum $\mathring{\cd_a}$ is determined by the singular fiber $\mathscr{X}_{u_\star} $ at a generic point $u_\star\in \cd_a$. The most important invariant of a divisor $\cd_a$
 is the conjugacy class $[\varrho_a]$ of the local monodromy $\varrho_a$ around it; the type of the fiber $\mathscr{X}_{u_\star}$ is determined by $[\varrho_a]$ together with some subtler invariants \cite{oguiso1,oguiso2,oguiso3,jap1,jap2,sawon}.
The local monodromy classes $[\varrho_a]$ follows a 
 Kodaira-like $ADE$ classification \cite{koda1,koda2,koda3,bhm}. Thus to each component $\cd_a$
 we can associate the (unique)  rank-1 \emph{asymptotic} special geometry over the (small) disk $\Delta$
 with the same monodromy class.   
Roughly speaking, the effective IR field theory which governs the degrees of freedom which are light along $\cd_a$ is described by this rank-1 geometry defined by the local monodromy $\gamma_a$. The detailed geometry/physics is a bit trickier \cite{Cecotti:2023mlc} and depends on the particular point $u_\star$, not just on the discriminant component $\cd_a$, but the most important aspects are already captured by $[\gamma_a]$. 
 
The Coulomb branch $\mathscr{C}$ carries a second stratification by unbroken $U(1)_R$ symmetry.
At a generic point in $\mathscr{C}$ the unbroken subgroup is $\Z_k$
with $k\in\{1,2,3,4,6\}$ \cite{Cecotti:2021ouq}.\footnote{\ If $k\not\in\{1,2\}$ the model is isotrivial
and the geometry and physics
can be determined by group-theoretical methods \cite{Cecotti:2021ouq,Cecotti:2021ouq2}.}
The codimension-1 loci in $\mathscr{C}$ where the unbroken R-symmetry enhances to a bigger group
 form a divisor $E=\sum_f E_f$. The unbroken R-symmetry at the generic point of $E_f$
is $\Z_{\ell_f}$ where $\ell_f$ is a proper multiple of $k$.
On multiple intersections $\cap_jE_{f_j}$ we have further enhancement to
 $\Z_{\text{lcm}(\ell_{f_j})}$, etc. 
Some $D_a$ may coincide with some $E_f$ (case (3)). 
 
 By the \emph{stratification data} of a $\C^\times$-isoinvariant special geometry we mean
\be\label{sData}
 \{\cs_a,[\varrho_a],\ell_a\}_{a\in A} 
\ee
 where $\{\cs_a\}_{a\in A}$ is the finite set of ($\C^\times$-invariant) special divisors in $\mathscr{C}$ where something interesting happens,
 and $[\varrho_a]$ (resp.\! $\ell_a$) is the conjugacy class of the local monodromy around $\cs_a$ (resp.\! the order of the unbroken R-symmetry along $\cs_a$). 
 
 
\subparagraph{Inverse problem.} The very fact that all special geometry has a ``stratification'' whose non-trivial strata 
 are special geometries of strictly smaller dimension, suggests an inverse engineering strategy to construct
\emph{all} special geometries which, in particular, could lead to a complete classification of all 4d $\cn=2$ SCFT.  
 The idea is that, once we have constructed all $\C^\times$-isoinvariant special
 geometries up to rank $r$, we can construct the rank $r+1$ ones  
by ``gluing together'' a finite set of rank-$r$ geometries together with a rank-$(r+1)$ open stratum $\mathring{\mathscr{C}}=\mathscr{C}\setminus \mathscr{S}$
of ``boring'' generic vacua whose geometry is well understood. In addition to the stratification data \eqref{sData}, 
the problem then requires to specify the
generic vacuum which is encoded in two data: the action  of the R-symmetry $U(1)_R$ on $\mathring{\mathscr{C}}$
and the generic Higgs branch. For smooth Coulomb branches the first datum is more conveniently replaced by the $r$-tuple of the Coulomb dimensions $\{\Delta_i\}$.
For simplicity we assume that the generic Higgs branch is trivial,
but this is not an essential assumption for our analysis.

This inverse engineering program is what we call \emph{``the inverse problem of special geometry''.}
It requires to answer two main questions:

\begin{que} Characterize the sets of the form $\{(\Delta_i)_{i=1}^r,(\cs_a,[\varrho_a],\ell_a)_{a=1}^s\}$
where $\{\Delta_i\}$ are positive rational numbers, $\cs_a\in \C^r$ are homogeneous irreducible divisors, $[\varrho_a]$ Kodaira types, and $\ell_a$ positive integers,
which are the stratification data of indecomposable,\footnote{\ A geometry is indecomposable if it is not the product of lower-rank geometries, i.e. if the corresponding SCFT
is does not split in two non-interacting sectors. The notion of ``indecomposable'' should not be confused with the notion of ``irreducible''.} rank-$r$, \emph{global} special geometries.
\end{que}

\begin{que} When $\{(\Delta_i)_{i=1}^r,(\cs_a,[\varrho_a],\ell_a)_{a=1}^s\}$ is an allowed stratification datum,
construct all global special geometries that have this ``stratification''.
\end{que}


The answer to \textbf{Question 1} may be given at different levels.
There are necessary conditions on $\{(\Delta_i)_{i=1}^r,(\cs_a,[\varrho_a],\ell_a)_{a=1}^s\}$ for the existence of
a global special geometry and also sufficient conditions. 
Necessary conditions are physical properties which \emph{all}
4d $\cn=2$ SCFT should enjoy, and hence are statements of direct physical interest. Sufficient conditions produce
examples which may lead to new insights. It would be desirable to have a precise
necessary and sufficient condition. In this paper we shall come ``close''
to this goal but some work is still needed to get the complete story.
The trickier part of \textbf{Question 1}  is the classification of the 
special divisors $\sum_a\cs_a$ which may arise in a special geometry. 
This aspect is the main topic of section \ref{s:r3}. When the base $\mathscr{C}$
of the (algebraic) Lagrangian fibration is compact, i.e.\! $\mathbb{P}^r$,
the possible discriminants are well-understood \cite{disC}, but in our set-up $\mathscr{C}$
is affine, not projective, and \emph{a priori} the discriminant is much less constrained. 

\begin{ans} In this long note we shall give essentially complete
 answers to the two {\bf Questions} in rank-2.
For higher rank we give some necessary conditions, discuss general aspects and properties, provide
examples, and pave the way for a detailed study (which should be doable at least
for  low-ranks such as $3$ or $4$). 
\end{ans}

\subparagraph{Rigidity.} Our basic tool to answer the two \textbf{Questions} is \emph{rigidity} of various geometric structures implied by the definition of special geometry.
Several \emph{inequivalent} notions of rigidity are relevant in the present context: the ``weaker one'' (Faltings-Peter rigidity \cite{arak,falt,saito,peters1,peters2,peters3}) applies
to general (irreducible, non-isotrivial) special geometries, while the stronger Deligne-Simpson rigidity \cite{simpson,simpson?,simpson??} refers to the underlying monodromy representation and yields a typically sufficient condition. In addition we have (non)-rigidity in the special geometric sense,
see section \ref{s:rigidity}.

\subparagraph{Classification scheme.} The arguments in this paper suggest the
classifying scheme:
{\setlength{\arrayrulewidth}{0.4mm}
\renewcommand{\arraystretch}{1.8}$$
 \begin{tabular}{|c|p{11cm}|}\hline
 I & isotrivial geometries\\\hline
 II & quasi-isotrivial geometries not of class I\\\hline
 III & reducible indecomposable geometries not of classes I, II\\\hline
 IV & non-rigid geometries not of classes I, II, III\\\hline
 V & $\mu$-rigid geometries not of classes I-IV\\\hline
VI & others\\\hline
 \end{tabular}
 $$
 } Class I can be explored by elementary means for all rank $r$ \cite{Cecotti:2021ouq2}. 
 Classes II and III are severely
 constrained, ``few'' such geometries exist, and a systematic analysis may be possible
 also in higher rank. These geometries look rather ``bizarre'', but there is
 strong evidence that a small number of them do exist.
In view of the \textbf{Folk-theorem} (\S.\,\ref{s:Fthm}) class
IV is also essentially charted for all ranks. 
 For classes V and VI one falls back to the rank by rank strategy \emph{via} the \emph{inverse method}.
 One tries to construct them by induction on the rank by ``gluing together'' known lower rank ``strata''.
 Class VI is the less amenable to analysis and \emph{a priori} the largest one, but
 not a single example of class VI geometry is known to the author, so there is hope
 that this hardest class \emph{is actually empty} or at least very sparse.
 
 Each one of the classes I-V is naturally decomposed into subclasses which are conveniently studied separately.
 
 \subparagraph{Organization.}  The paper is organized as follows. Section 2 contains a review of the basic geometric structures with new results and complements perhaps known to the experts but
 without a good reference in print. In section 3 we state the inverse problem and explain its relation to the classification program. In section 4 we discuss the families of Abelian varieties which arise from a $\C^\times$-isoinvariant special geometry. In section 5 we discuss the deformation theory of special geometry, the conformal manifold, and the \textbf{Folk-theorem}. To illustrate the main issues here we also  discuss in detail the case of rank 2. In section 6 we consider the rigidity of implied geometric structures such as 
 the Abelian family and the monodromy representation. Here we state the Deligne-Simpson problem.
 In section 7 we discuss how rigidity can be used to understand/construct/classify $\C^\times$-isoinvariant special geometries. In section 8 the program is carried on in rank 2 in detail. Here several tricky points are analyzed such as reducible representations. Finally in section 9 we lay the groundwork for an extension of the program to higher rank, focusing on the problem of understanding the class of allowed discriminant divisors. Some technicality is deferred to the appendices.

 \section{Special Geometry: review with complements}
 \label{sec2}
 
 In this section we review special geometry to fix notation and terminology.
 We also present some complements which, while
 probably known to experts, are not spelled out in the literature (to the best of our knowledge).
 Some results are novel, at least in the present generality.  
 
 \subsection{Defining structures}
 
 We first review the geometric structures implied by special geometry.
 We start with a rough sketch of the physical foundations of the topic.
We omit all details and gloss over all subtle points since our purpose here
  is merely to \emph{motivate}
 the precise  definitions given in the next subsection (\S.\,\ref{yyy56z}) which we shall take as the geometric ``axioms''
for the rest of the paper.  
 \medskip
 
 An Euclidean 4d $\cn=2$ SQFT admits a topological twisting \emph{\'a la}
 Witten \cite{TFT1,TFT2} by modifying the couplings of the various fields to the holonomy connection of
 the 4-dimensional spacetime manifold $M$. The twisted model is a Topological Field Theory (TFT)
 invariant under all oriented diffeomorphisms $M\to M$ (see \cite{marino} for an extensive review).
By construction the twisting does not change the functional measure when
 $M$ is locally flat. 
 
\subsubsection{Chiral ring, Coulomb branch, Coulomb dimensions} The algebra of local operator of the twisted TFT is a subquotient $\mathscr{R}$ of the algebra  $\mathfrak{A}$ of local operators of the parent $\cn=2$ SQFT  
 called the \emph{chiral ring.} The correlation functions of operators in $\mathscr{R}$
  are independent of the choice of their representatives in $\mathfrak{A}$.
$\mathscr{R}$ is a commutative, associative $\C$-algebra with unit, which 
in a sensible SQFT\footnote{\ A non finitely-generated chiral ring will violate the entropy bound any good QFT should obey.} must be a finitely-generated integral domain.
Its spectrum $\mathscr{C}= \mathsf{Spec}\,\mathscr{R}$ is then an
 \emph{affine variety} over $\C$
 called the \textit{Coulomb branch}. Its complex dimension $r=\dim_\C\mathscr{C}$ is the \emph{rank}.
 The basic observation is that when the Riemannian spacetime $M$ is flat  the correlation functions of operators in $\mathscr{R}$  computed in the SQFT coincide with the ones computed in the corresponding twisted
 TFT.\footnote{\ This equality was the starting point of $tt^*$ \cite{Cecotti:1991me}.}
 
 By replacing $\mathscr{R}$ with its normalization \cite{Eisenbud}, if necessary, we assume
 $\mathscr{C}$ to be \emph{normal.} In particular, the Coulomb branch is regular
(as an affine variety\footnote{\ We stress that some \emph{other} relevant geometric structures on $\mathscr{C}$, such as the special K\"ahler metric, are typically non-regular along a divisor
in the geometries we are interested in. This \emph{caveat} applies to the question of \emph{smoothness} in the next paragraph.}) in codimension 1. 
\medskip
 
 We are mainly concerned with the situation where the original $\cn=2$ SQFT was a SCFT.
 In this special case the $\C$-algebra $\mathscr{R}$ is graded by the scaling dimension $\Delta$, that is:  \be\label{localgraded}
 \mathscr{R}= \bigoplus_{\Delta\in \frac{1}{m}\mathbb{Z}_{\geq0}}\mathscr{R}_\Delta,\qquad \mathscr{R}_0=\C,\qquad
 \mathscr{R}_{\Delta}\cdot \mathscr{R}_{\Delta^\prime}\subseteq \mathscr{R}_{\Delta+\Delta^\prime},
 \ee
 for some \emph{minimal} positive integer $m$. 
 
 \begin{defn} In this paper\footnote{\ Unless explicitly stated otherwise\,!!} by a \emph{chiral ring $\mathscr{R}$ of rank $r$}
 we shall always mean a normal affine $\C$-algebra of dimension $r$ with a grading as in  \eqref{localgraded}. By its \emph{Coulomb branch} $\mathscr{C}$ we mean the normal affine $r$-variety $\mathsf{Spec}\,\mathscr{R}$ (or its underlying complex $r$-fold).   
 \end{defn}
 
By definition the Coulomb branch $\mathscr{C}$ is an affine quasi-cone with an algebraic $\C^\times$-action. The maximal ideal $\mathscr{R}_+\equiv \oplus_{
\Delta>0}\mathscr{R}_\Delta$
 defines a closed point $0\in\mathscr{C}$, the \emph{tip} of the quasi-cone, which physically represents the unique vacuum where the superconformal symmetry is not spontaneously broken. $\{0\}$ is the only \emph{closed} $\C^\times$-orbit, hence it is contained in the closure of all $\C^\times$-orbits.

 \subparagraph{Smooth vs.\! non-smooth Coulomb branches.}
 
The Coulomb branch $\mathscr{C}$ is \emph{smooth} at $0$
 (hence everywhere) iff the chiral ring is \emph{regular} at the tip i.e.
 \be
 \dim (\mathscr{R}_+/\mathscr{R}_+^{\mspace{2mu}2})=\dim \mathscr{C}\equiv r,
 \ee
 in which case $\mathscr{R}=\C[u_1,\cdots\!, u_r]$
 where $u_i$ has dimension $\Delta_i\in \tfrac{1}{m}\mathbb{Z}_{>0}$. 
 
 In the physical applications one usually assumes the affine quasi-cone $\mathscr{C}$ to be smooth  because the \emph{typical} examples of Coulomb branches which arise in physics are smooth affine varieties. However, as it will be clear in the following, to construct a \emph{deep theory}
(i.e.\! to really understand the smooth case) 
one needs to allow for non-smooth Coulomb branches as well.
The situation leads to two questions: (1)
Which Coulomb branch singularities are allowed?
(2)
What is the physical meaning of a \emph{singular} Coulomb branch?
A plausible physical mechanism for a singularity is a ``discrete gauging''
which leads to Coulomb branches of the form $\C^r/G$ with $G$ a finite group.
In this paper we shall assume the Coulomb branch to be smooth,
except when explicitly stated otherwise.

\subparagraph{Coulomb dimensions.} 
When $\mathscr{C}$ is smooth, hence a copy of $\C^r\simeq\mathsf{Spec}\,\C[u_1,\dots,u_r]$
with $u_i\in\mathscr{R}_{\Delta_i}$, the $r$-tuple of rational numbers
 $\{\Delta_1,\dots,\Delta_r\}$ is called the \textit{Coulomb dimensions}.
 We shall deduce various restrictions on the allowed Coulomb dimensions below:
only finitely many $r$-tuples of Coulomb dimensions are allowed for a given $r$ \cite{caorsi}.
 
 \subparagraph{The projective Coulomb branch $\mathscr{P}$.}
 In the SCFT case it is natural to replace the ordinary Coulomb branch $\mathscr{C}\equiv\mathsf{Spec}\,\mathscr{R}$ by its projective counterpart 
 \begin{equation}
 \mathscr{P}=\mathsf{Proj}\,\mathscr{R}.
 \end{equation}
When the Coulomb branch is smooth, $\mathscr{P}$ is a weighted projective space \cite{dolgaWP}
\be
\mathscr{P}\equiv\mathbb{P}(d_1,\cdots\!,d_r),\qquad d_i=\frac{m\, \Delta_i}{\gcd\{m\, \Delta_j\}}.
\ee
In particular $\mathscr{P}$ is simply-connected $\pi_1(\mathscr{P})=1$.
$\pi_1(\mathscr{P})$ is at most finite even when $\mathscr{C}$ is not smooth: indeed it is a quotient of the finite group $F$.
In the smooth case $\mathbb{P}(d_1,\cdots\!,d_r)$ is isomorphic to a unique \emph{well-formed}
weighted projective space $\mathbb{P}(q_1,\cdots,q_r)$ such that $\gcd(q_1,\dots,\widehat{q_j},\dots,q_r)=1$ for all $j$ \cite{dolgaWP}.
%

\subsubsection{The Abelian fibration}
The arguments in this subsection are heuristics to motivate the definitions in \S.\,\ref{yyy56z} below.
\medskip

 By their very definition in all TFT the RG flow is trivial. Hence all topological correlation functions can be exactly
 computed using the low-energy effective theory \cite{WWW1,WWW2} which, roughly speaking, is a supersymmetric $\sigma$-model
 with target space the set of SUSY vacua. We say ``roughly speaking'' because the vacuum space is typically singular, and one should be careful with the precise definition of the low-energy theory.
For us the main point is that there are several distinct ``IR effective theories'' and 
we may use any one of them to compute TFT quantities. Independence of the
TFT amplitudes from the choice of effective description
 implies compatibility conditions between these effective theories
which we realize as additional geometric structures on $\mathscr{C}$.

Consider, for instance, the 4-manifold $M\equiv \R^3\times S^1$ where the circle has
 radius $R$. The local TFT observables are equal to their SQFT siblings (since $M$ is flat)
 and independent of $R$ by topological invariance.
 As $R\to\infty$ they may be computed using the 4d IR effective $\cn=2$ SQFT \cite{WWW1,WWW2},
while as $R\to0$ the very same local TFT observables may be computed using the 3d IR effective SQFT which now is a 3d $\cn=4$ SUSY theory.\footnote{\ \begin{tiny}\textdbend\end{tiny} The physical low-energy effective theory depends in an intricate way from $R$ \cite{gaiotto2}. Since the topological theory is independent of $R$, we are free to fix $R$ in any convenient way. We use the asymptotic limit $R\to\infty$ where things simplify dramatically.}
We take for granted that the 3d Coulomb branch $\mathscr{X}$ satisfies the same ``regularity'' conditions as its 4d counterpart $\mathscr{C}$, and hence that $\mathscr{X}$ is a smooth\footnote{\ Just as for the Coulomb branch, the condition that $\mathscr{X}$ is smooth may be relaxed.} 
scheme of finite type over $\C$ (in facts over $\mathscr{C}$). This follows from the physical heuristics sketched below,
as well as from the usual special geometry folklore in the physical literature.\footnote{\  Of course, if $\mathscr{X}$ is not an algebraic scheme (but still smooth), essentially everything we say will remain true in the \emph{analytic category.}}
On SUSY grounds \cite{Cecotti:2015wqa}, the 3d Coulomb branch $\mathscr{X}$
is a holomorphic symplectic manifold of quaternionic dimension $r$ (complex dimension $2r$):
indeed on the \emph{real} tangent bundle $T_\R\mspace{-1mu}\mathscr{X}$ we have a $\mathbb{H}$-action induced by the
3d $\cn=4$ R-symmetry which endows $\mathscr{X}$ with a torsionless $Sp(r,\C)$-structure. 
We write $\Omega$ for its $(2,0)$ holomorphic symplectic form.
The compatibility of the TFT computation from the 4d and 3d viewpoints
requires the algebra isomorphism
\be\label{basiciso}
\mathscr{R}\equiv\Gamma(\mathscr{C},\co_\mathscr{C})\simeq \Gamma(\mathscr{X},\co_\mathscr{X}),
\ee
which produces the canonical scheme-theoretic fibration
\be\label{basicpi}
\pi\colon\mathscr{X}\to \mathsf{Spec}\,\Gamma(\mathscr{X},\co_\mathscr{X})\simeq \mathscr{C}
\ee
over the Coulomb branch $\mathscr{C}$. The compatibility of the actions of the 3d and 4d R-symmetries
implies that
the fibers of $\pi$ are Lagrangian submanifolds, i.e. $\Omega|_{\mathscr{X}_u}=0$
for all $u\in\mathscr{C}$. 
The map $\pi$ has an obvious physical meaning in terms of the RG flow from $R=0$ to $R=\infty$.
We claim that $\mathscr{X}$ is a commutative group-scheme over $\mathscr{C}$.
In particular the fiber $\mathscr{X}_u$ over the \emph{generic} point $u$ is a commutative algebraic group.
Then $\pi$ has a section $s$ given by the zero of the group, $s(u)=0_u\in\mathscr{X}_u$,
 and $\mathscr{C}$
can be identified with the submanifold $s(\mathscr{C})\subset\mathscr{X}$ which should also be Lagrangian by
comparison of R-symmetry actions in 3d and 4d. The claim can be easily checked
in Lagrangian SCFTs and hence in full generality at generic points in $\mathscr{C}$:
this suffices to establish that $\mathscr{X}$ is a (smooth) birational model of a commutative group-scheme over
$\mathscr{C}$. In the usual physical situation where the electromagnetic charges obey
a Dirac quantization, the general fiber is actually a polarized Abelian variety.
These physical requirements lead to the ``conventional''
special geometries which are the object of these notes.
More exotic possibilities will be explored elsewhere. 


\subsection{Basic definitions}\label{yyy56z} 
We codify the above informal discussion in a definition.

\begin{defn}
A \emph{rank-$r$ special geometry} is
a holomorphic fibration $\pi$ with section $s$
\be
\xymatrix{\mathscr{X}\ar@<0.6ex>[rr]^{\pi}&& \mathscr{C},\ar@<0.6ex>[ll]^{s}}\qquad \pi\mspace{2mu}s=\mathrm{Id}_\mathscr{C},
\ee
over a complex (normal) affine variety $\mathscr{C}$ of dimension $r$, such that:
\begin{itemize}
\item[(a)] the total space $\mathscr{X}$ has a holomorphic symplectic structure $\Omega$,
and both the fibers $\mathscr{X}_u$ of $\pi$ and the section $s$ are Lagrangian submanifolds:
$\Omega|_{\mathscr{X}_u}=\Omega|_s=0$;
\item[(b)]
the \emph{generic} fiber $\mathscr{X}_u$ ($u\in\mathscr{C}$) is smooth and a polarized Abelian variety of dimension $r$ with zero element $s(u)$.
\end{itemize}
The locus $\cd\subset \mathscr{C}$ of points $u$ with non-smooth fiber $\mathscr{X}_u$
is called the \emph{discriminant.}
\end{defn}

Although it is not part of the definition, we shall assume the polarization of the fiber to be
\emph{principal} for definiteness. Physically this is the most natural situation.

\medskip

With this definition, a special geometry is (in particular) a birational model of an \emph{Abelian scheme over a
(quasi-)projective variety}, that is, of an Abelian variety defined over the function field $\C(\mathscr{C})$ (e.g. in rank-1 $\mathscr{X}$ is the \emph{N\'eron model} \cite{neron}
of an elliptic curve over the field of rational functions on a curve $\overline{\mathscr{C}}\simeq\mathbb{P}^1$).
$\mathscr{X}$ is also an \emph{Algebraic Integrable System}
where the Hamiltonians in involution are the holomorphic functions on $\mathscr{C}$,
i.e.\! the chiral ring\footnote{\ More precisely its Frech\'et completion.} $\mathscr{R}$,
while the conjugate angle variables are the doubly periodic coordinates along the fibers.
Two \emph{weaker} notions are often used:
\begin{itemize}
\item a \emph{non UV-complete} special geometry is a fibration $\mathscr{X}\to\mathscr{C}$
which satisfies all the properties of a special geometry except that $\mathscr{C}$
is not quasi-projective. A non UV-complete geometry may or may not have a UV completion, i.e.\! an analytic extension over some affine base. Non UV-complete special geometries describe low-energy effective 
$\cn=2$ QFTs which are not UV-complete, and may or may not have a completion;
\item a \emph{special K\"ahler geometry}
satisfies all the local conditions predicted by the $\cn=2$ superspace
approach to supersymmetry: existence of a local pre-potential $\cf$, existence of a special K\"ahler metric on the Coulomb branch, \emph{etc.}, but whose underlying global structures (such as the VHS\footnote{\ VHS=variation of Hodge structure.}) are not defined over $\mathbb{Q}$, or some other number field $\mathbb{F}$, but over a local field such as $\R$ or $\C$.
\end{itemize}

We stress that the results in this paper refer only to special geometries in the \emph{strict sense.}
In particular for us it will be absolutely essential that the Coulomb branch $\mathscr{C}$ is quasi-projective. Most of our statements are \emph{just false} 
when this condition does not hold.


\begin{defn} A special geometry is called:
\begin{itemize}
\item \emph{decomposable} if it is the product of two lower-rank special geometries. Decomposable geometries
 describe QFTs
which decompose in two non-interacting sub-systems;
\item \emph{free} if the fibration is trivial i.e.\! $\mathscr{X}=A\times \mathscr{C}$ for a fixed Abelian variety $A$.
Free geometries describe free $\cn=2$ QFT;
\item \emph{isotrivial} if all \emph{smooth} fibers are isomorphic to a fixed
polarized Abelian variety $A$. All $\cn\geq 3$ QFTs have isotrivial geometries,
but there are also genuine $\cn=2$ interacting SCFT with isotrivial geometries. The structure
 of isotrivial geometries is discussed in \S.\,\ref{s:isotrivial} below;
 \item\emph{quasi-isotrivial} if all smooth fibers are \emph{isogeneous} to a product
 $A\times Y_u$ where the first factor, $A$, is a fixed polarized Abelian variety while $Y_u$
 depends on $u\in\mathscr{C}$;
\item \emph{unitary} if it is complete as a Hamiltonian system, that is,
iff for all Hamiltonian $h\in\mathscr{R}$
the corresponding Hamiltonian vector field $v(h)$
\be
\iota_{v(h)}\Omega\overset{\rm def}{=}dh
\ee
is \emph{complete,} i.e.\! iff the solution of the (holomorphic) Hamilton equations
\be
\frac{dx^i}{dt}= v(h)^i
\ee
 exists for all times $t\in \C$ and initial condition.
The special geometries of unitary QFTs are unitary in this geometric sense. It is amusing that quantum unitarity of
the $\cn=2$ QFT amounts to classical ``unitarity'' of the underlying Hamiltonian system.
\end{itemize}
\end{defn}
 
 \begin{cave} In \S.\,\ref{ddddcom} we introduce the notion of \emph{reducible} special geometry.
 While decomposable implies reducible, there exist indecomposable geometries which are not irreducible.
 However their structure is very constrained, and their number rather small.
 \end{cave}

\begin{fact}
The discriminant $\cd$ is either empty or of pure codimension $1$. When $\cd=\varnothing$ the geometry is free. Otherwise there are finitely-many elements $p_a\in\mathscr{R}$
with $(p_a)$ prime, such that $\cd=\{p_1 p_2\cdots p_s=0\}\subset\mathscr{C}$.
\end{fact}
\begin{proof}[Argument] The first assertion is e.g.\! \textbf{Proposition 3.1(2)} in \cite{oguiso1}. The second one follows from the well-known fact that all families of smooth projective varieties parametrized by $\C$ are trivial (cf. \S.\,\textbf{13.7} of \cite{periods}). The third statement is well known.
\end{proof}

The above definitions refer to the geometry of a general 4d $\cn=2$ SQFT. 
The geometries of $\cn=2$ SCFTs have an extra geometric structure.

\begin{defn}\label{xxxyy67} A \emph{$\C^\times$-isoinvariant special geometry} is a special geometry with a \emph{complete} holomorphic vector field $\ce\in \Gamma(\mathscr{X},T\mspace{-2mu}\mathscr{X})$, called the \emph{Euler vector},
such that for all $z\in\C$ the exponential map $\exp(z\,\ce)\colon\mathscr{X}\to\mathscr{X}$ is a complex automorphism preserving all special geometric structures: i.e.\! $\exp(z\,\ce)$ preserves set-wise the section $s$ and maps fibers into fibers, while\footnote{\ Here and below $\mathscr{L}_v$ stands for the Lie derivative in the direction of the vector field $v$.} 
\be\label{dimOmega}
\mathscr{L}_\ce\mspace{2mu}\Omega=\Omega,\quad\text{that is,}\quad
\exp(z\,\ce)^{\mspace{-2mu}*}\mspace{1mu}\Omega= e^{z}\,\Omega.
\ee 
\end{defn}
The Euler vector $\ce$ induces a grading of the chiral ring $\Gamma(\mathscr{X},\co_\mathscr{X})\simeq \mathscr{R}=\oplus_\Delta\mathscr{R}_\Delta$ by
\be
h\in\mathscr{R}_\Delta \quad\Leftrightarrow\quad\mathscr{L}_\ce\mspace{2mu} h=\Delta\,h.
\ee
A Hamiltonian $h\in\mathscr{R}_\Delta$ is said to have \emph{dimension $\Delta$}.
The identity has dimension zero. We shall see momentarily that in a unitary special geometry
the identity $1$ is the only element of $\mathscr{R}$ with dimension zero, while all other Hamiltonians have $\Delta\geq1$.
We shall show below that the $\Delta$'s are rational numbers and that only a short list of rational numbers may appear at a given rank $r$.
Then $\mathscr{R}$ is a finitely generated local graded $\C$-algebra of the expected form 
\eqref{localgraded}. When the Coulomb branch is smooth $\mathscr{R}=\C[u_1,\cdots, u_r]$ and the coordinates have definite dimension 
$\mathscr{L}_\ce\, u_i=\Delta_i\, u_i$.
In particular $\exp(z\,\ce)$ induces a complex automorphism of the Coulomb branch,
$\exp(z\mspace{2mu}\boldsymbol{\Delta})\colon \mathscr{C}\to \mathscr{C}$,
such that
\be
\pi\exp(z\,\ce)=\exp(z\mspace{2mu}\boldsymbol{\Delta})\,\pi,
\ee
which for smooth $\mathscr{C}$ is simply $u_i\mapsto e^{z \Delta_i} u_i$. 
It follows from \textbf{Definition \ref{xxxyy67}} that the isoclass 
of the fiber, as a polarized Abelian variety, is invariant along an orbit of
the group of automorphisms $\exp(z\mspace{2mu}\boldsymbol{\Delta})$
\be
\text{for all }z\in\C,\ u\in\mathscr{C}\setminus\cd\colon\quad \mathscr{X}_{\exp(z\mspace{2mu}\boldsymbol{\Delta})u}\simeq \mathscr{X}_u\ \ 
\text{as polarized Abelian varieties.}
\ee
The kernel of group homomorphism 
$\exp(\cdot \,\ce)\colon\C\to \mathsf{Aut}(\mathscr{X})$ is $2\pi i m\Z$ where $m$ is the integer in eq.\eqref{localgraded}.
Indeed the induced map $\exp(2\pi i m\mspace{2mu}\boldsymbol{\Delta})$ is the identity on $\mathscr{C}$, so $\exp(2\pi i m\mspace{2mu}\ce)$ fixes the zero element $0_u$ of each Abelian fiber $\mathscr{X}_u$, while it acts as the identity on its Lie algebra
because of the isomorphism $\Omega\colon T_{0_u}\mspace{-2mu}\mathscr{X}_u\simeq T^*_u\mspace{2mu}\mathscr{C}$, so it is the identity in each fiber.
The automorphism group acting effectively on $\mathscr{X}$ is $\C^\times\simeq \C/2\pi i m\Z$.
The prime divisors $p_a$ in the discriminant $\cd$ are irreducible homogeneous elements of $\mathscr{R}$, i.e.
$\mathscr{L}_\ce\, p_a= \Delta(p_a)\, p_a$ for some $\Delta(p_a)\in \tfrac{1}{m}\mathbb{N}$.

\begin{rem} In the physical literature it is more common to require the existence of a
Seiberg-Witten (SW) differential $\lambda$ such that $\Omega=d\lambda$. $\lambda$ is simply the image of the vector $\ce$
under the natural isomorphism $\Omega\colon T\mspace{-1mu}\mathscr{X}\to T^*\mspace{-3mu}\mathscr{X}$, and the two data $\lambda$, $\ce$ are equivalent in the $\C^\times$-isoinvariant situation. It is more convenient to work with $\ce$ since it makes manifest that $\lambda$ corresponds to an automorphisms ($\equiv$ a symmetry) of the special geometry which is fully determined by its symplectic structure.
\end{rem}
%

\subsubsection{Geometric version of the unitarity bound}
\label{s:uni}

\begin{fact} {\bf (1)} In a unitary $\C^\times$-isoinvariant special geometry all homogeneous elements
of $\mathscr{R}$ except the identity, have dimension $\Delta\geq1$. {\bf (2)} In a unitary, indecomposable, non-free $\C^\times$-isoinvariant geometry the inequality is strict $\Delta>1$. 
\end{fact}

\begin{proof}[Argument] {\bf (1)} Let $h$ be a non-constant Hamiltonian of dimension $\Delta$. $dh$ is not zero
and $\mathscr{L}_\ce dh= d\mathscr{L}_\ce h=\Delta\, dh$.
The corresponding Hamiltonian vector $v(h)$ is non-zero and satisfies
\be
\Delta\, i_{v(h)}\Omega=\Delta\,dh= \mathscr{L}_\ce dh=
\mathscr{L}_\ce(i_{v(h)}\Omega)=i_{\mathscr{L}_\ce v(h)}\Omega+ i_{v(h)}\mathscr{L}_\ce\Omega= i_{\mathscr{L}_\ce v(h)}\Omega+i_{v(h)}\Omega
\ee
or $\mathscr{L}_\ce\mspace{2mu} v(h)=(\Delta-1)v(h)$. Therefore
\be
\exp(z\mspace{2mu}\ce)_\ast\mspace{2mu} v(h)= z^{\Delta-1} v(h).
\ee
If $\Delta<1$ the Hamiltonian vector field
 $v(h)$ diverges at the tip of the quasi-cone and so cannot be complete. {\bf (2)}
 follows from the following result whose proof is sketched in appendix \ref{app:zero}.
 \begin{lem}\label{le:zero} The connected component $(\mathscr{X}_0^\text{\rm reg})^0$
 of the smooth locus $\mathscr{X}_0^\text{\rm reg}\subset\mathscr{X}$ of the fiber over the origin $0$
 has the structure
 \be\label{zerostruc}
 (\mathscr{X}_0^\text{\rm reg})^0= A_0\times \C^{r-\ell},\quad\ell\equiv\dim A_0=\dim\mathscr{R}_{\Delta=1},
 \ee
 where $A_0$ is a polarized Abelian variety,
 while the smooth fibers have the form $\mathscr{X}_u=A_0\times\mathscr{Y}_u$ with
 $\mathscr{Y}_u$ a polarized Abelian variety of dimension $r-\ell$. 
 \end{lem}
 Clearly the geometry is the product of the free geometry $(A_0\times \C^\ell)\to\C^\ell$
 with a special geometry $\mathscr{Y}\to \C^{r-\ell}$ with fiber $\mathscr{Y}_u$ and no
 dimension 1 coordinate.
 \end{proof}
 
Eq.\eqref{zerostruc} says that the fiber over the origin has \emph{additive} (as contrasted to \emph{multiplicative}) degeneration, thus in a non-free $\C^\times$-isoinvariant geometry $\mathscr{X}_0$ is always an unstable degeneration of an Abelian variety. 
 
 \begin{rem}\label{kiuqwe44} The structure of the fiber over the origin, eq.\eqref{zerostruc} is one of the many properties of the special geometry which, taken together, guarantee that the corresponding $\cn=2$ QFT is free of Landau poles. For instance, in rank-1 it rules out singular fibers of Kodaira type $I_n$, which
 would correspond to $\cn=2$ SQCD a formal QFT affected by Landau poles.  
 \end{rem}

 \subsubsection{Implied structures: metric, SW periods, prepotential}
 
 Let $L$ be the fibers' polarization seen as a integral $(1,1)$  class which
 restricts to a K\"ahler class in each fiber. On the complement $\mathring{\mathscr{C}}=\mathscr{C}\setminus\cd$ of the discriminant we have the function
 \be
 K\colon u\mapsto i\int_{\mathscr{X}_u} L^{r-1}\lambda\wedge \lambda^*=
 i\sum_i\left(\int_{A_i}\lambda\int_{B^i}\bar\lambda-\int_{A_i}\bar\lambda\int_{B^i}\lambda\right)
 \ee
 where $\{A_i,B^j\}$ is a symplectic basis of 1-cycles on the fibers $\mathscr{X}_u$ defined \emph{locally in some $U\subset \mathring{\mathscr{C}}$} with respect to the polarization $L$
 and $\lambda\equiv \iota_\ce\Omega$.
 The function $K$ is smooth and pluri-subharmonic in $\mathscr{C}\setminus \cd$, so $K$ is the potential of a K\"ahler metric with K\"ahler form
 \be
 \omega=i\int_{\mathscr{X}_u}\! L^{r-1}\wedge\Omega\wedge\bar \Omega.
 \ee
 Except in the free case (when $\cd=\varnothing$) this ``special K\"ahler'' metric is not complete since it has singularities along the
 discriminant $\cd$ which are at finite distance in the special metric.
 
  We define the SW period $\Pi=(a^D_i,a^j)^t$ to be the \emph{local} functions in $U$
 \be
 a^D_i(u)=\int_{B^i}\iota_\ce\mspace{1mu} \Omega\big|_{\mathscr{X}_u},\qquad  a^j(u)=\int_{A_i}\iota_\ce\mspace{1mu} \Omega\big|_{\mathscr{X}_u},\quad u\in U,\ i,j=1,\dots,r
 \ee
 so that 
 \be
 K= i\, \Pi^\dagger\Omega \Pi,\quad\text{where}\quad \Omega=\left(\begin{smallmatrix} 0 & \boldsymbol{1}_r\\
 -\boldsymbol{1}_r & 0\end{smallmatrix}\right).
 \ee 
 The periods are locally-defined holomorphic functions on $\mathscr{C}$.
Locally on $\pi^{-1}(\mathring{\mathscr{C}})$ we may find Darboux coordinates such that
$\Omega=da^i\wedge dz_i$ where $dz_i$ are holomorphic differentials on the fiber
normalized as
\be
\int_{A_i}dz_j={\delta^i}_j,\qquad \int_{B^i}dz_j=\tau_{ij}.
\ee 
Then $\iota_\ce\mspace{1mu} \Omega=a^i dz_i$ and
\be
a^D_i=\int_{B^i}\iota_\ce\mspace{1mu} \Omega\big|_{\mathscr{X}_u}=\tau_{ij}\,a^j.
\ee 
In particular the $A$-periods $(a^1,\dots, a^r)$ are local complex coordinates in $\mathring{\mathscr{C}}$. Clearly
 \be
 \mathscr{L}_\ce\,\Pi=\Pi
 \ee
 so that on local functions on $\mathring{\mathscr{C}}$ we have $\mathscr{\ce}=a^i\partial_{a^i}$ and $\mathscr{L}_\ce\,\tau_{ij}=0$ gives $a^i\partial_{a^i}\tau_{jl}=0$ so that
 \be\label{hhhhleg}
 da^D_i=d(\tau_{ij}\,a^j)= \tau_{ij}\,da^j\quad\Rightarrow\quad da_i^D\wedge da^i=0.
 \ee
 The last equation says that, locally in $U\subset\mathring{\mathscr{C}}$ there exists a holomorphic function
 $\cf(a^j)$, called the prepotential, such that 
 \be
 \mathscr{L}_\ce \cf=2\,\cf\quad \text{and}\quad a^D_i=\frac{\partial\cf}{\partial a^i},\quad \tau_{ij}=\frac{\partial^2\cf}{\partial a^i\partial a^j}.
 \ee
 
 \begin{rem} The special coordinates yield an affine structure to $\mathring{\mathscr{C}}$. This is the holomorphic version of the integral affine structure on the base of a Liouville integrable system \cite{Arnold3}.  
 \end{rem}

\subsection{Unbroken R-symmetry and the ``characteristic dimension''}

In the smooth case it is convenient to see
the $r$-tuple of rational numbers\footnote{\ We assume the $\Delta_i$ to be ordered in a non-decreasing sequence, i.e.\! $\Delta_{i+1}\geq \Delta_i$.} $\{\Delta_1,\cdots, \Delta_r\}$
as a point in the projective space over $\mathbb{Q}$ and write
\be
(\Delta_1,\Delta_2,\dots,\Delta_r)=\lambda(d_1,d_2,\dots, d_r)
\ee
where $(d_1,d_2,\dots,d_r)$ are the unique positive integers with $\gcd(d_1,\dots\,d_r)=1$
which represent them same point in $\mathbb{P}^{\mspace{1mu}r-1}(\mathbb{Q})$. 
We shall refer to $d_i$ as the \emph{degree} of $u_i$, while $\Delta_i\equiv \lambda d_i$
is its \emph{dimension}. 
$\lambda$ is
a positive rational number which we write in minimal terms as 
\be
\lambda=\frac{n}{m}\qquad\text{with }
\gcd(m,n)=1.
\ee
$\lambda$ is a basic invariant of the $\C^\times$-isoinvariant special geometry.
It is well-defined also for non-smooth Coulomb branches: if $\{\phi_1,\dots,\phi_\ell\}$ is
a set of generators of $\mathscr{R}$, one sets
\be 
\big(\Delta(\phi_1),\dots,\Delta(\phi_\ell)\big)=\lambda\,\big(d(\phi_1),\dots,d(\phi_\ell)\big)
\ee
 with the
$d(\phi_i)$ integral and coprime.
By construction
the automorphism $\exp(2\pi i\mspace{1mu}\ce/\lambda)$ acts trivially on $\mathscr{R}$,
so it fixes all points of $\mathscr{C}$,
while it multiplies $\Omega$ by a $n$-th root of unity
\be
\exp(2\pi i\,\ce/\lambda)^*\mspace{2mu}\Omega = e^{2\pi i m/n}\Omega,
\ee
which depends only on $m\bmod n$. The root $\zeta\equiv e^{2\pi im/n}$, or equivalently
the pair $(n,m\bmod n)$ with $\gcd(m,n)=1$, was called the
\emph{characteristic dimension} in \cite{Cecotti:2021ouq} and written $\varkappa=n/\langle m\rangle_n$ 
where  $\langle m \rangle_n$ is the unique positive integer $\leq n$ congruent to $m$ mod $n$. 
The physical meanings of $m$ and $n$ are obvious:
\begin{itemize}
\item[(1)] the group which acts effectively on the geometry is
$\C^\times=\exp[\C/(z\sim z+2\pi i m)]$;
\item[(2)] at the \emph{generic} point in $\mathscr{C}$
the $\C^\times$ symmetry is spontaneously broken to $\Z_n$;
\item[(3)] at \emph{any} point $u\in\mathscr{C}$ the $U(1)_R$
symmetry is broken to a group which contains $\Z_n$. Hence, a part
for the tip $0$, to a discrete group of the form $\Z_{nk}$.
\end{itemize}
One shows \cite{Cecotti:2021ouq} that $n$ can take only five values
\be\label{whatn}
n\in\{1,2,3,4,6\},
\ee
hence we have only 8 possible $\varkappa$
\be\label{whatkappa}
\varkappa\in \big\{1, \tfrac{6}{5}, \tfrac{4}{3}, \tfrac{3}{2}, 2, 3, 4, 6\big\}
\ee 
in one-to-one correspondence with the  conjugacy classes of torsion elements of $SL(2,\Z)$,
or equivalently with the Kodaira elliptic fibers with semisimple monodromy \cite{koda1,koda2}. The possible values of $\varkappa$ are just the allowed dimensions of the Coulomb dimension $\Delta$ in a $r=1$ special geometry: they are in 1-to-1 correspondence with the
semisimple Kodaira fibers, see table \ref{table}.
The basic result of \cite{Cecotti:2021ouq} was that the Abelian fiber over the generic point $u\in\mathscr{C}$
admits complex multiplication by the cyclotomic field $\mathbb{Q}(\zeta)\equiv \mathbb{Q}(e^{2\pi i/\varkappa})$ and hence its isoclass is uniquely determined unless $\mathbb{Q}(\zeta)\equiv\mathbb{Q}$, i.e.\! unless $\varkappa\in\{1,2\}$. 

\begin{fact}
When $\varkappa\not\in\{1,2\}$ all smooth fibers are isomorphic as polarized Abelian varieties, and the full geometry is isotrivial, hence easily constructed and studied by simple group theory {\rm\cite{Cecotti:2021ouq}}. The statement holds independently of the assumptions that $\mathscr{C}$ is smooth and that the polarization is principal. 
\end{fact}

\begin{table}
\caption{\label{table} Semi-simple Kodaira fibers}\vskip9pt

{\begin{footnotesize}\renewcommand{\arraystretch}{1.5}
\begin{tabular}{p{2.2cm}p{1.25cm}p{1.5cm}p{1.35cm}p{1.35cm}p{1.35cm}p{1.35cm}p{1.4cm}p{1cm}}
\hline\noalign{\smallskip}
Kodaira fiber & $\mathbf{I}_0$ & $\mathbf{I}_0^*$ & $\mathbf{II}$ & $\mathbf{II}^*$ & $\mathbf{III}$ & $\mathbf{III}^*$ & $\mathbf{IV}$ & $\mathbf{IV}^*$\\
\noalign{\smallskip}\hline\noalign{\smallskip}
$\varkappa$ & $1$ & $2$ & $\tfrac{6}{5}$ & $6$ & $\tfrac{4}{3}$ & $4$ & $\tfrac{3}{2}$ & $3$\\
monodromy & $\left[\begin{smallmatrix}1 & 0\\
0 & 1\end{smallmatrix}\right]$ & $\left[\begin{smallmatrix}-1 & 0\\
0 & -1\end{smallmatrix}\right]$ & $\left[\begin{smallmatrix}1 & 1\\
-1 & 0\end{smallmatrix}\right]$ &$\left[\begin{smallmatrix}0 & -1\\
1 & 1\end{smallmatrix}\right]$ & $\left[\begin{smallmatrix}0 & 1\\
-1 & 0\end{smallmatrix}\right]$ & $\left[\begin{smallmatrix}0 & -1\\
1 & 0\end{smallmatrix}\right]$ & $\left[\begin{smallmatrix}0 & 1\\
-1 & -1\end{smallmatrix}\right]$ &$\left[\begin{smallmatrix}-1 & -1\\
1 & 0\end{smallmatrix}\right]$\\
Euler no. & $0$ & $6$ & $2$ & $10$ & $3$ & $9$ & $4$ & $8$\\ 
Lie algebra & $-$ & $\mathfrak{so}(8)$ & $-$ & $E_8$ & $\mathfrak{su}(2)$ & $E_7$ & $\mathfrak{su}(3)$ & $E_6$\\
\noalign{\smallskip}\hline\noalign{\smallskip}
\end{tabular}\end{footnotesize}}
\end{table}

For later reference we sketch a proof of these statements. Preliminary we recall
a few well know facts about polarized Abelian varieties.\footnote{\ For the facts about complex Abelian varieties used here and below see \cite{complexAbelian} expecially chapter 13.} Let $A$ be a polarized
Abelian variety and $\mathsf{Aut}(A)$ the group of automorphisms $A\to A$
which preserve the polarization. $\mathsf{Aut}(A)$ is a \emph{finite} group with two
natural representations. First we have the $\mathbb{Q}$-defined \emph{rational representation}
$\chi$ acting on $H_1(A,\Z)\otimes_\Z\mathbb{Q}$. For principal polarizations $\chi$ is
a group embedding $\chi\colon \mathsf{Aut}(A)\to Sp(2r,\Z)$; for a general polarization see
e.g.\! chapter 2 of \cite{cayley}.
Then we have the \emph{analytic representation} $\sigma\colon \mathsf{Aut}(A)\to GL(r,\C)$ which acts on the holomorphic tangent space
at $0$, $T_0A\equiv\mathsf{Lie}(A)$. We have the isomorphism of complex representations
\be
\chi\otimes_\mathbb{Q}\mspace{-2mu}\C \simeq \sigma\oplus \bar\sigma.
\ee 
Since $\chi$ preserves the lattice $H_1(A,\Z)\subset H_1(A,\mathbb{Q})$, this implies that for all $\eta\in\mathsf{Aut}(A)$
the characteristic polynomial $\det[z-\chi(\eta)]$ is a product of cyclotomic polynomials
whose roots are the eigenvalues of $\sigma(\eta)$ and their conjugates.
The automorphism 
\be
\xi\overset{\rm def}{=}\exp(2\pi i \mspace{2mu}\ce/\lambda)\colon \mathscr{X}\to\mathscr{X}
\ee 
generates a $\Z_n$ group of automorphisms of $\mathscr{X}$ which fixes point-wise the Coulomb branch $\mathscr{C}$ and hence acts trivially on its holomorphic cotangent bundle
$T^*\mathscr{C}$. 
$\xi$ then restricts to an automorphism $\xi_u\colon\mathscr{X}_u\to\mathscr{X}_u$ of each fiber (smooth or not).

Let $u\in\mathscr{C}$ be a generic point, and $\mathscr{X}_u$ its smooth Abelian fiber. The symplectic structure gives a linear isomorphism
\be
\Omega\colon T_{s(u)}\mspace{1mu}\mathscr{X}_u\to T_u^*\mspace{1mu}\mathscr{C}. 
\ee 
The automorphism $\xi$ acts as the identity
on $T_u^*\mspace{1mu}\mathscr{C}$, hence by eq.\eqref{dimOmega} as multiplication by $\zeta^{-1}$
 on $T_{s(u)}\mathscr{X}_u$, that is,
\be
\sigma(\xi_u)=\zeta^{-1}\cdot\mathbf{1}\in GL(r,\C)\quad\Rightarrow\quad
\det[z-\chi(\xi_u)]=\big((z-\zeta^{-1})(z-\zeta)\big)^r\in\Z[z]
\ee
which says that the root of unity $\zeta$ is either $\pm1$ or
an imaginary quadratic integer. This proves eqs.\eqref{whatn},\eqref{whatkappa}.
The final step is 
to recall\footnote{\ See \textbf{Corollary 13.3.5} in \cite{complexAbelian}.} the fact that when multiplication by $\zeta= i$ or $e^{2\pi i/3}$
is an automorphism of the Abelian variety $A$ we have\footnote{\ Through out this paper $E_\tau$ stands for the elliptic curve of period $\tau$.}
\be
A=\overbrace{E_\zeta\times E_\zeta\times\cdots\times E_\zeta}^{r\ \text{factors}}
\ee  
and the geometry is (in particular) isotrivial, hence well understood \cite{Cecotti:2021ouq,Cecotti:2021ouq2}.
It remains to study the two cases $\varkappa=1$ and $\varkappa=2$
where the geometry may or may not be isotrivial. 

\subsection{The index of an irreducible homogeneous divisor}
 
 By an irreducible homogeneous divisor in $\mathscr{C}$ we mean an irreducible codimension-1
 subvariety preserved by the $\C^\times$-action (the closure of the $\C^\times$-orbits it contains). 
An irreducible homogeneous divisor has the form $\{h=0\}\subset\mathscr{C}$ for some irreducible quasi-homogeneous\footnote{\ In the sequel we simplify `quasi-homogenous' in `homogeneous'.} 
Hamiltonian $h\in\mathscr{R}$ of degree $d$, i.e.
\be
\mathscr{L}_{\ce/\lambda} h= d\, h.
\ee
E.g.\! for a smooth geometry $h=h(u_i)$ is a polynomial  in the $u_i$'s of degree $d$, i.e.
\be
h(\lambda^{d_i} u_i)=\lambda^d\mspace{2mu} h(u_i)\quad \forall\; \lambda\in \C\ \text{and some }d\in\mathbb{N}.
\ee
Consider the Hilbert series $P_h(t)$ of the algebra $\mathscr{R}/(h)$ (graded by degree not dimension!).
E.g.\! for $\C[u_1,\dots, u_r]/(h)$
\be
P_h(t)=\frac{(1-t^d)}{\prod_i(1-t^{d_i})}.
\ee 
The \emph{enhancement index} $\nu(h)$ of the divisor $h$ is the largest positive integer such that
\be
P_h(t)\in \C[[t^{\nu(h)}]].
\ee
The physical meaning of $\nu(h)$ is:
\be
\nu(h)=\left[\left(\begin{smallmatrix}\textbf{unbroken subgroup of $\C^\times$ at the}\\
\textbf{generic point
in the $h=0$ locus}\end{smallmatrix}\right)\;:\; \left(\begin{smallmatrix}\textbf{unbroken subgroup of $\C^\times$}\\
\textbf{at the generic point
in $\mathscr{C}$}\end{smallmatrix}\right)\right]
\ee
i.e.\! $\nu(h)$ measures the enhancement of the unbroken R-symmetry along the divisor $h$.
A divisor $h$ is \emph{enhanced} (resp.\! \emph{non-enhanced}) iff $\nu(h)>1$
(resp.\! $\nu(h)=1$).

\begin{rem} The authors of \cite{M13} distinguish the divisors in \emph{unknotted} vs.\! \emph{knotted} ones.
Often, \emph{but not always,} this coincides with the distinction enhanced vs.\! non-enhanced.

\begin{defn} An irreducible homogeneous divisor $(h)$ is \emph{knotted}
iff the fundamental group of its complement $\pi_1(\C^r\setminus (h))$ is non-Abelian.
\end{defn}
\end{rem}
Smooth divisors are automatically unknotted.

\begin{exe}
For $r=2$ the irreducible homogeneous divisors have one of the two
forms:
\be
u_i\ \ \text{for }i=1,2\quad\text{or}\quad u_1^{d_2}- z\,u_2^{d_1}\quad \text{with }z\in\C^\times.
\ee
For divisors of the second kind $\nu=1$, while $\nu(u_i)=d_1d_2/d_i$.
The divisors of the second kind are knotted unless $d_1=1$,
while $u_i$ is never knotted.
\end{exe}

\noindent
In a general rank-$r$ smooth geometry, the irreducible divisors of degree $d_i\in \{d_1,\dots, d_r\}$
have $\nu=\gcd\{d_1,\dots,\widehat{d_i},\dots, d_r\}$
while divisors with $d\not\in \{d_1,\dots, d_r\}$ have $\nu=1$.
%
%
\subsection{Coulomb dimensions I: CM-types along submanifolds}
Suppose we have an irreducible reduced total intersection $H=\{h_1=\cdots=h_\ell=0\}$
with an enhancement index $\nu(H)\equiv \nu>1$ which is \emph{not} contained in the discriminant $\cd$.
The automorphism $\eta\equiv\exp(2\pi i\, \ce/(\nu\lambda))$
fixes all points of $H$, hence its acts as the identity on its cotangent bundle $T^*\mspace{-2mu}H\subset T^*\mspace{1mu}\mathscr{C}|_H$ while it acts by polarized automorphisms on the fibers $\mathscr{X}_u$
for $u\in H$.
We write $N^\ast\mspace{-1mu} H\to H$ for the rank-$\ell$ conormal bundle to $H$; it is spanned 
by the differentials $dh_a$ ($a=1,\dots, \ell$). Let $u\in H$ be a generic point; by assumption the fiber $\mathscr{X}_u$ is smooth. One has
\be
T_u^*\mspace{1mu}\mathscr{C}= T^*_uH\oplus N^\ast \mspace{-1mu}H_u
\ee
$\eta_u$ acts on the first summand as the identity
and on the second as multiplication by the
$\ell\times\ell$ matrix
\be
\mathrm{diag}\big\{\!\exp(2\pi i\mspace{2mu} m\mspace{1mu} \Delta(h_k))/(n\nu))\big\}.
\ee
Using the isomorphims $\Omega\colon T_{s(u)}\mathscr{X}_u\to T_u^*\mspace{1mu}\mathscr{C}$
we get the eigenvalues of its analytic representation
$\sigma(\eta_u)$ are
\be\label{listofroots}
(\overbrace{e^{-2\pi i m/(n\nu)},\cdots, e^{-2\pi i m/(n\nu)}}^{(r-\ell)\ \text{times}},e^{-2\pi i m(1-\Delta(h_1))/(n\nu)},\cdots,e^{-2\pi i m(1-\Delta(h_\ell)/(n\nu)})
\ee
This $r$-tuple of roots of unity, together with their conjugates, are the roots (with multiplicities)
of a product of cyclotomic polynomials. 
Enforcing these conditions on all total intersections of elements of $\mathscr{R}$
yield a large web of constraints on the possible $n$, $m$ and $\nu$
and in turn on the Coulomb dimensions $\Delta_i$. 
We stress that nothing depends on the assumption that $\mathscr{C}$ is smooth (except for the $\Delta_i$'s).

We give some examples.

\subsubsection{Example: dimension formulae along regular axes \cite{caorsi}} \label{jjuutreq}

For simplicity from now on we assume $\mathscr{C}$ smooth.
Suppose that the $i$-th coordinate axis $H_i=\{u_j=0\ \text{for }j\neq i\}$
is not contained in the discriminant: we say that $H_i$ is a \emph{regular} axis. Then $\nu=d_i$ and the list of roots in \eqref{listofroots} becomes
\be\label{rtupleorroots}
\big\{\exp[2\pi i(\Delta_1-1)/\Delta_i], 
\exp[2\pi i(\Delta_2-1)/\Delta_i],\cdots,\exp[2\pi i(\Delta_r-1)/\Delta_i]\big\}
\ee
Note that \emph{a priori} once we know this \emph{ordered} list of roots of 1,
 $1/\Delta_i$ is determined only mod 1, but since in a unitary
geometry
$0<1/\Delta_i \leq 1$ knowing  the $i$-th root in the list
suffices to determine $\Delta_i$ uniquely. Instead $\Delta_j$ for $j\neq i$
is determined by \eqref{rtupleorroots} only up to  addition of an integral multiple of $\Delta_i$.
The ambiguity\footnote{\ The role of this ambiguity was overlooked in \cite{caorsi}.} is however small since only a finite list of dimensions are allowed for any given $r$
(see \S.\,\ref{s:straDim}).
Moreover often there are several regular axes $H_j\not\subset \cd$ and 
 we get one such formula for the dimensions for each one of these axes,
 getting an overdetermined system which gets rid of all ambiguities.

When the $i$-th coordinate axis $H_i\not\subset\cd$
the list of $r$ $(n\,d_i)$-th roots of unity in eq.\eqref{rtupleorroots}
is a subset of the $2r$ roots of a product of cyclotomic polynomials
with the property that the union of this subset with its conjugate
yields the full set of all $2r$ roots. This entails that
\be
\exp(-2\pi i/\Delta_i)
\ee
is a root in a cyclotomic field of degree at most $2r$. That is,
the dimension $\Delta_i$ of a regular axis has the form
\be\label{whichdim}
\Delta_i= \frac{n_i}{m_i}\ \ 
\text{where }\phi(n_i)\leq 2r\ \text{and }\gcd(n_i,m_i)=1\ \ 1\leq m_i<n_i
\ee
so that there are only finitely many allowed dimensions at a given rank $r$ \cite{caorsi}.
We shall show in \S.\,\ref{s:straDim} that this conclusion remains valid even when 
$H_i\subset \cd$; in facts in that case we get the stronger condition $\phi(n_i)\leq 2(r-1)$.

\subsubsection{New dimensions and CM-type}\label{newdinwansio} The above results lead to the

\begin{defn} A rational number $\Delta$ is a \emph{new dimension in rank $r$}
if it is an allowed Coulomb dimension in rank $r$ but not in rank $(r-1)$. The set of new dimensions in rank-$r$ is
\be\label{Xidef}
\Xi(r)=\begin{cases}\Big\{\frac{n}{m}\colon  
\phi(n)= 2r,\ \gcd(n,m)=1,\ 1\leq m<n\Big\} &\text{for }r\geq2\\
\{1,6/5,4/3,3/2,2,3,4,6\} & \text{for }r=1
\end{cases}
\ee
see \cite{caorsi} for the cardinality $|\Xi(r)|$ and other interesting properties of the set $\Xi(r)$.
\end{defn}
We shall see in \S.\,\ref{s:straDim} that an axis $H_i$ with $\Delta_i$ a new dimension in rank $r$ is automatically regular.
In facts, we have more:
\begin{fact} If the axis $H_i$ has a new dimension $\Delta_i$, the fiber $\mathscr{X}_u$
at the generic point $u\in H_i$ has complex multiplication by the CM field
$\mathbb{Q}(e^{2\pi i/\Delta_i})$ of degree $2r$ and CM-type\footnote{\ For the notion of CM-type see e.g. chapter 13 of \cite{complexAbelian}.} given by eq.\eqref{rtupleorroots}. 
\end{fact}
It is well-known that the number of isomorphism classes of polarized Abelian varieties
of given CM-type is equal to the class number $h$ of the corresponding
number field $\mathbb{Q}(e^{2\pi i/\Delta_i})$. In particular for ranks $r\leq 9$
all new dimensions correspond to fields with  $h=1$ and 
the $r$-tuple of dimensions $\{\Delta_1,\dots,\Delta_r\}$ uniquely determines
the fiber $\mathscr{X}_u$ along a new dimension axis. 

\subsubsection{Sample results}\label{s:sample}

We list a few simple consequences of the above \textbf{Facts} for rank-2 that we shall use later in the examples.
The geometries are assumed to be smooth.
Here we use that an axis with a new dimension is automatically regular (see \S.\,\ref{s:straDim}).
\begin{itemize}
\item[(1)] A rank-2 special geometry having one of its dimensions in the list 
$\{12,12/5,12/7,12/11\}$ is automatically isotrivial;
\item[(2)] if both axes are regular of new dimension then the dimensions are
$\{8,12\}$ or $\{\tfrac{8}{7},\tfrac{10}{7}\}$. A geometry with dimensions $\{8/7,10/7\}$ cannot be isotrivial; 
\item[(3)] the rank-2 geometries where both axes are regular  must have dimensions
 $\{5/4,3/2\}$, $\{8/7,10/7\}$, $\{8,12\}$ or $\{5/3,4\}$.
The first two correspond to the Argyres-Douglas theories of types $A_5$ and $A_4$, respectively,
and the third one is the $G_8$ isotrivial geometry constructed in \cite{Cecotti:2021ouq};
\item[(4)] a rank-2 geometry with a new dimension of the type $5/k$ or $10/l$ is the Jacobian
fibration of a fibration in generically smooth genus-2 curves. 
\end{itemize} 
\begin{proof} (1) the dimension formulae yields pairs of dimensions with $\varkappa\not\in\{1,2\}$.
(2)(3) if both axes are regular we can compute $\{\Delta_1,\Delta_2\}$ in two ways; consistency
of the two computations yield the list. Second statement in (2): the fibers over the two axes
 have multiplication by $\mathbb{Q}(e^{2\pi i/8})$ and respectively $\mathbb{Q}(e^{2\pi i/10})$,
 so are not isomorphic. (4) the fiber over the regular axis with the new dimension has muliplication by
 $\mathbb{Q}(e^{2\pi i/5})$ hence is simple and not the product of two elliptic curves (cf. \textbf{Corollary 13.3.6.} of \cite{complexAbelian}). 
\end{proof}




\subsection{SW differential on enhanced axes}\label{s:iiiiii4uu5}

Suppose that the $i$-th axis $H_i$ is regular. 
The R-symmetry which is unbroken along $H_i$ is generated by $\xi_i\equiv\exp(2\pi i\,\ce/\Delta_i)$. $\xi_i$ acts by an automorphism of the generic (smooth)
fiber $\mathscr{X}_u$ with analytic representation 
\begin{equation}
\sigma(\xi_i)=\mathrm{diag}\big\{\!\exp[2\pi i(\Delta_1-1)/\Delta_i], 
\exp[2\pi i(\Delta_2-1)/\Delta_i],\dots,\exp[2\pi i(\Delta_r-1)/\Delta_i]\big\}
\end{equation}
and rational representation $\chi(\xi_i)\in Sp(2r,\Z)$ acting on $H_1(\mathscr{X}_u,\Z)$,
and 
\be
\chi(\xi_i)\simeq \sigma(\xi_i)\oplus \overline{\sigma(\xi_i)}
\ee
 as complex representations.
The polarization gives an isomorphism
\be
\mathsf{Lie}(\mathscr{X}_u)\simeq H^{0,1}(\mathscr{X}_u)\simeq \overline{H^0(\mathscr{X}_u, \Omega^1_{\mathscr{X}_u})}
\ee
Let $\lambda\equiv \iota_\ce\Omega$ be the SW differential
and $\lambda_u\equiv \lambda|_{\mathscr{X}_u}$ be its restriction to the fiber.
We have
\be
\xi^*_i\lambda_u= e^{2\pi i/\Delta_i}\,\lambda_u.
\ee
Let $\{e_a\}$ be a symplectic basis of $H_1(\mathscr{X}_u,\Z)$. The SW periods transform as
\be
\langle e_a, \lambda_u\rangle = \xi_i \cdot\langle e_a,\lambda_u\rangle=
\langle \chi(\xi_i)_{ab}\,e_b\cdot \alpha, \xi_i^*\lambda_u\rangle = 
e^{2\pi i/\Delta_i}\,\chi(\xi_i)_{ab}\, \langle e_b,\lambda_u\rangle
 \ee
 so, setting 
 {\renewcommand{\arraystretch}{1.1}
\be
 \Pi(u)_a\equiv \langle e_a,\lambda_u\rangle=\begin{pmatrix}a_i^D\\ a^j\end{pmatrix}
\ee
 }
\be
 \chi(\xi_i)_{ab}\,\Pi(u)_b= e^{-2\pi i/\Delta_i}\,\Pi(u)_a
\ee
i.e.\! the SW periods along $H_i$ are eigenvectors of the rational representation of $\xi_i$
of eigenvalue $\exp(-2\pi i/\Delta_i)$ (up to overall scale the period is independent of $u\in H_i$,
since the axis minus the origin is a $\C^\times$-orbit).
%

\subsection{$\varkappa=1$ vs.\! $\varkappa=2$ special geometries}\label{kk87777}

The characteristic dimension $\varkappa$ is a basic invariant of a $\C^\times$-isoinvariant special geometry. When $\varkappa$ is not $1$ or $2$ the geometry is isotrivial,
 all smooth fibers are the product of $r$ copies of $E_{e^{2\pi i/\varkappa}}$, and the possible
 geometries are
essentially known globally.\footnote{\ We stress that a geometry may be isotrivial even when $\varkappa\in\{1,2\}$.} We remain with two cases: $\varkappa=1$ and $\varkappa=2$.
What is the difference between these two (potentially) non-isotrivial situations?

The automorphism $\xi\equiv\exp(2\pi i\,\ce/\lambda)$ which fixes point-wise the Coulomb branch $\mathscr{C}$ acts on the SW differential $\lambda\equiv \iota_\ce\mspace{1mu}\Omega$ as 
\be
\xi^*\lambda=\begin{cases}\lambda & \varkappa=1\\
-\lambda &\varkappa =2
\end{cases}
\ee
and in facts 
\be
\xi^*\big|_{H^1(\mathscr{X}_u,\C)}= (-1)^{\varkappa+1}.
\ee
This means that when $\varkappa=2$ the SW periods $(a_i^D,a^j)^t$ are well-defined only
up to overall sign. 
In Lagrangian SCFTs with $\varkappa=2$ the sign flip is part of the gauge group (it is the longest word in the gauge Weyl group), so
the two signs are obviously physically equivalent and indistinguishable.
This property extends to special geometries with $\varkappa=2$ which do not describe Lagrangian SCFTs. 

We can state the difference between $\varkappa=1$ and $2$ in another (related) way.
When $\varkappa=1$ there is one element $h$ of the Coulomb branch function field\footnote{\ $\C(\mathscr{C})$ is the quotient field of the integral domain $\mathscr{R}$.} $\C(\mathscr{C})$
of dimension $1$, while there is no such object when $\varkappa=2$.
When $\varkappa=2$ we have instead an element $h^2\in\C(\mathscr{C})$ of dimension $2$. 

This difference between the two situations will have fundamental consequences
for the global aspects of a non-isotrivial special geometry.

\subsection{Discriminant complement and monodromy representation}\label{ddddcom}

As before $\cd\subset\mathscr{C}$ is the discriminant i.e.\! the divisor 
of points with singular fibers; we write
$\sum_i \cd_i$ for its decomposition into irreducible
components. $\cd_i=(p_i)$ for irreducible quasi-homogeneous polynomials $p_i$
of degree $g_i$ and index $\nu_i$. 
We write $\mathring{\mathscr{C}}\equiv \mathscr{C}\setminus \cd$
for the complement of the discriminant, i.e.\! the open domain of ``good'' points with smooth 
Abelian fiber. 
We have a holomorphic family of polarized Abelian varieties
\be
\mathring{\pi}\colon \mathring{\mathscr{X}}\equiv\mathscr{X}|_{\mathring{\mathscr{C}}}\to \mathring{\mathscr{C}},
\ee 
parametrized by the complement $\mathring{\mathscr{C}}$. 
We stress that $\mathring{\mathscr{C}}$ is quasi-projective over $\C$. 
Then we may apply the standard methods and results of
 the theory of Variations of Hodge Structure (VHS) over
a quasi-projective base \cite{VHS1,VHS2,VHS3,VHS4,VHS5,griffithsMT2,griffithsMT,periods,BG}.

All fibers of $\mathring{\pi}$ are diffeomorphic. This entails that over $\mathring{\mathscr{C}}$
we have a flat Gauss-Manin local system with fiber $H_1(\mathscr{X}_u,\Z)$ \cite{VHS2,periods}.
Going along a non-trivial closed curve $\ell\subset\mathring{\mathscr{C}}$
the fiber of the local system gets back to itself up to a rotation by an element $\gamma(\ell)$ 
of its polarized automorphism group: for principal polarizations this is the Siegel modular group $Sp(2r,\Z)$ (see \cite{cayley} for the general case). This yields
a group homomorphism
\be\label{kiiqwert}
\varrho\colon \pi_1(\mathring{\mathscr{C}})\to Sp(2r,\Z)
\ee
called the \emph{monodromy representation} \cite{VHS1,VHS2,VHS3,VHS4,VHS5,griffithsMT2,griffithsMT,periods,BG,monopeters},
a crucial invariant of the special geometry
which controls many physical properties of the associated SCFT.
The image $\Gamma\subset Sp(2r,\Z)$ of $\varrho$ is the \emph{monodromy group}.

\subsubsection{Topology of the complement}
A basic ingredient of special geometry is then the topology of the complement of the discriminant
$\mathring{\mathscr{C}}$, in particular its fundamental group $\pi_1(\mathring{\mathscr{C}})$.
While the \emph{Abelianization} of $\pi_1(\mathring{\mathscr{C}})$,
i.e.\! $H_1(\mathring{\mathscr{C}},\Z)$, is a pretty simple group \cite{dimca}
\be
H_1(\mathring{\mathscr{C}},\Z)\simeq \Z^k,\quad \text{where $k$ is the number of components of $\cd$,}
\ee
the group $\pi_1(\mathring{\mathscr{C}})$ is typically rather complicated. Roughly speaking, its
derived group
\be
D\pi_1(\mathring{\mathscr{C}})\overset{\rm def}{=}\mathrm{ker}\Big(\pi_1(\mathring{\mathscr{C}})\to
H_1(\mathring{\mathscr{C}},\Z)\Big)
\ee
 is a measure of the interactions in the corresponding SCFT.
A simple example will illustrate the issue.

\begin{exe}\label{n4su3}
Consider the special geometry of 4d $\cn=4$ SYM with $G=SU(3)$
(seen as a special instance of $\cn=2$ theory). The Coulomb dimensions $\{\Delta_1,\Delta_2\}$
 are the 
degrees of $\mathsf{Weyl}(SU(3))$, $\{2,3\}$, while physically we expect the monodromy group
to be the $SU(3)$ Weyl group $\Gamma=\mathsf{Weyl}(SU(3))\simeq \mathfrak{S}_3$.
Let us see how this comes about topologically. The discriminant has a unique irreducible component
corresponding to the locus in $\mathscr{C}$ where a $SU(2)$ gauge symmetry is restored
and the corresponding non-Abelian d.o.f. get light.
Thus $\cd=(u_1^3-u_2^2)$. One has \cite{oka}
\be
\pi_1(\mathring{\mathscr{C}})\equiv\pi_1(\C^2\setminus\{u_1^3-u_2^2=0\})=\langle \alpha, \beta\colon \alpha^3=\beta^2\rangle
\ee
which is the Artin braid group $\mathcal{B}_3$ of type $A_2$.
Since the geometry is isotrivial with dimensions $\{2,3\}$, the monodromy representation $\varrho$ factors through the finite quotient $\mathsf{Weyl}(SU(3))\simeq \mathcal{B}_3/\mathcal{P}_3$ producing the physically expected results. ($\mathcal{P}_3$ is the pure braid group in 3 strands). A similar story holds for $\cn=4$ with any gauge group.   
\end{exe}  

We list a few useful facts about the topology of hypersurface complements \cite{dimca}. 
When the divisor $\cd$ is \emph{normal crossing} with $k$ irreducible components the 
 $\pi_1$ of the complement is an Abelian group,
generated by the $k$ lassos circling the irreducible components. By the Zariski conjecture (proved in \cite{severi}) the same holds for all  nodal planar curves.
In the general case 
we have an explicit description of $\pi_1(\mathring{\mathscr{C}})$ given by the Zariski-van Kampen theorem
\cite{dimca,ZvK1,ZvK2}
that we shall review in \S.\,\ref{s:zariskithm}.

\subsubsection{Subtleties when $\varkappa=2$} \label{kiuuqwa}
We assume our special geometry not to be isotrivial. At a generic point
$\ast\in\mathscr{C}$ the R-symmetry is totally broken for $\varkappa=1$
while we have
a residual unbroken $\Z_2$ symmetry when $\varkappa=2$.
We define the SW periods $\Pi\equiv (a^D_i,a^j)^t$ in a small neighborhood of $\ast$,
and analytically continue them along a non-trivial loop $\ell\subset\mathring{\mathscr{C}}$
  based at $\ast$. One would expect that the analytically continued periods $\Pi^\prime$ satisfy
  \be
 \Pi^\prime = \varrho(\ell)\,\Pi.
  \ee
However when $\varkappa=2$
 the periods $\Pi$ are well defined only up to sign (cf.\! \S.\ref{kk87777})
and
 only the quotient representation
 \be
 \check{\varrho}\colon \pi_1(\mathring{\mathscr{C}})\to PSp(2r,\Z)\equiv Sp(2r,\Z)/\{\pm1\}
 \ee
has a well-defined physical meaning.

\subsubsection{Algebraic invariants of $\Gamma$: the Hodge ring}

We have a vector bundle $\mathscr{V}$ over the complement $\mathring{\mathscr{C}}$
with fiber 
\be
H_1(\mathscr{X}_u,\C)\equiv H_1(\mathscr{X}_u,\Z)\otimes_\Z\C
\ee
equipped with the flat Gauss-Manin connection. We fix a reference point $u_\star$ and write $V$ for the fiber $\mathscr{V}_{u_\star}$.
$V$ carries the representation $\varrho$ of $\Gamma$ (defined over $\mathbb{Q}$)
while\footnote{\ The isomorphism is induced by the polarization.} $V\simeq V^\vee$.
This induces representations (over $\mathbb{Q}$) of $\Gamma$ on all tensor spaces
$V^{\otimes k}$.
Since $\mathring{\mathscr{C}}$ is quasi-projective, one has:
\begin{fact}[See \cite{monopeters}]
The monodromy group $\Gamma$ is finitely generated and semisimple, i.e.\! all its finite-dimensional representations
split in direct sums of irreducible ones.
\end{fact} 
 The subtle point about the previous statement is that we must work with representations defined
 over $\mathbb{Q}$ instead of the algebraically closed field $\C$.
 However there are some simple situations.
 \begin{fact}\label{uni2}
 Let $V=V_0\oplus V_1$ where $V_0$ is the largest $\C$-subspace where $\Gamma$ acts trivially. Then:
 \begin{itemize}
\item[{\bf(1)}] The chiral ring contains precisely $\dim V_0/2$
 elements with $\Delta=1$. 
 \item[{\bf(2)}] The special geometry is a product with one factor a free geometry of rank $\dim V_0/2$.
 \end{itemize}
 \end{fact}
\begin{proof}[Argument] Let $\{\gamma_I\}\equiv \{B^i,A_j\}$ ($I=1,\dots,2r$, $i,j=1,\dots,r$) 
be \emph{multi-valued} sections of the local system with section $H_1(\mathscr{X}_u,\Z)$
which form a symplectic basis in each fiber. The SW periods
 {\renewcommand{\arraystretch}{1.1}\be
\Pi_I(u)\overset{\rm def}{=}\oint_{\gamma_I} \iota_\ce\mspace{1mu} \Omega\big|_{\mathscr{X}_u}, \qquad
\Pi_I(u)\equiv\begin{pmatrix}a^D_i(u)\\ a^j(u)\end{pmatrix}
\ee}
\hskip-4pt are a set of $2r$ \emph{multi-valued} holomorphic sections of
$\mathscr{V}^\vee\simeq\mathscr{V}\to \mathring{\mathscr{C}}$ which span 
the fibers and satisfy
\be
\mathscr{L}_\ce\, \Pi_I= \Pi_I.
\ee 
Only half the periods are functionally independent:
indeed
$da_i^D(u)\wedge da_i(u)=0$ (cf.\! eq.\eqref{hhhhleg}) and
$a_i^D(u)=\tau_{ij}(u)\,a^j(u)$
where $\tau_{ij}(u)$ is the (multivalued) period matrix of $\mathscr{X}_u$ in the given basis $\{B^j,A_i\}$,
 Let ${P_I}^J$ be the projector $V\to V_0$ written in the basis $\{\gamma_I\}$.
By construction  ${P_I}^J a_J$ are now \emph{global} holomorphic functions on
$\mathring{\mathscr{C}}$. It is easy to check that they are regular along
 the discriminant,\footnote{\ By Hartogs theorem it suffice to check regularity in codimension 1.} and hence are defined and regular everywhere on $\mathscr{C}$. By definition, these holomorphic global functions belong to $\mathscr{R}$ and have $\Delta=1$. It is easy to see that exactly half of them are functionally independent.
 \end{proof}
 
We can generalize the story. Suppose that the symmetric $k$-th power $V^{\odot k}$
 contains a $\Gamma$-invariant tensor $T^{I_1\cdots I_k}$ (automatically defined over $\mathbb{Q}$); the function
 \be
 T^{I_1\cdots I_k}\Pi_{I_1}\Pi_{I_2}\cdots \Pi_{I_k},
 \ee
(if not zero) is an element of $\mathscr{R}$ with $\Delta=k$. The invariants of $\Gamma$ form a
subalgebra of the chiral ring $\mathscr{R}$ called the \emph{Hodge ring} $\mathscr{H}$ \cite{griffithsMT,griffithsMT2},
a fundamental invariant of special geometry. For instance, for a $\cn=4$ special geometry
$\mathscr{H}=\mathscr{R}$. More generally, for an isotrivial geometry $\mathscr{R}$
is a finitely-generated $\mathscr{H}$-module. The ring $\mathscr{H}$ is non-trivial for
all quasi-isotrivial geometries as well as for the ones having a non-generic Mumford-Tate group.

More generally, if we have an invariant tensor $T^{I_1\dots I_k;J_1\dots J_l}\in V^{\odot k}\otimes V^{\wedge l}$ we can form the global $l$-form of scaling dimension $k+l$
\be
T^{I_1\dots I_k;J_1\dots J_l}\,\Pi_{I_1}\cdots\Pi_{I_k}\,d\Pi_{J_1}\wedge\cdots\wedge d\Pi_{J_l}
\ee
When $T^{J_1J_2}\equiv\Omega^{J_1J_2}$ is the polarization, the corresponding invariant 2-form is trivial, $\Omega^{IJ} d\Pi_I\wedge d\Pi_J=0$, in facts we have the stronger condition $\Omega^{IJ}\Pi_I\wedge d\Pi_J=0$.

\subsubsection{\begin{scriptsize}\textdbend\end{scriptsize} The reducible case: subtleties} \label{s:ugly}
We consider now the case of a monodromy representation $V$ which decomposes over $\mathbb{Q}$ as 
$V_1\oplus V_2$ where the summands are non-trivial and non-isomorphic.

We consider two different scenarios. In the fist one the image
of $\Gamma$ in $GL(V_s,\mathbb{Q})$ is infinite for $s=1,2$.
Working in a sufficiently small domain $U\subset\mathring{\mathscr{C}}$,
we  define the SW periods and use half of them as special coordinates in $U$.
Taking $\mathbb{Q}$-linear combinations, we may split the periods in 
two sets, $\Pi^{(1)}_I$ and $\Pi^{(2)}_J$ taking value respectively in $V_1$ and $V_2$.
 When we go around a non-trivial path in $\mathring{\mathscr{C}}$
and return back to $U$, each $\mathbb{Q}$-period gets transformed
into linear combinations of $\mathbb{Q}$-periods from the same set. The two sets never mix,
and the splitting of the local special coordinates in two sets $a_{(1)}^i$ and $a_{(2)}^j$
is globally defined. 
The cotangent bundle of $\mathscr{C}$ splits into the subbundle spanned by the
differentials $da_{(1)}^i$ and the one spanned by the $da_{(2)}^j$'s.
 Indeed,
 up to isogeny, the fibers $\mathscr{X}_u$
split into a product $A^{(1)}_{u}\times A^{(2)}_{u}$ of Abelian varieties with $H_1(A_{u}^{(s)},\mathbb{Q})\simeq V_s$. Going to a finite cover we may assume the fiber is a product.
We have
\be\label{treqw12}
T_0A^{(1)}_u\oplus T_0A^{(2)}_u\simeq T_u^*\mathscr{C}.
\ee
The decomposition \eqref{treqw12} defines two complementary integrable distributions
whose leaves have local equations in $U$ of the form 
\be
a_{(s)}^i=\mathrm{const.}\quad s=1,2,
\ee
 so that (restricting $U$ if necessary) \emph{locally}
the Coulomb branch can be written as $U= U_{(1)}\times U_{(2)}$
with $a_{(s)}^i$ local coordinates on the factor $U_{(s)}$ ($s=1,2$).
The local pre-potential is a sum
\be
\cf(a_{(1)}^i,a_{(2)}^j)=\cf_{(1)}(a_{(1)}^i)+\cf_{(2)}(a_{(2)}^j)
\ee
where the period matrix of $A_{u}^{(s)}$ is given by $\partial^2 \cf_{(s)}(a_{(s)})\,\partial a^i_{(s)}\partial a^j_{(s)}$,
and the analytic continuation of $\cf(a_{(1)}^i,a_{(2)}^j)$ (as a \emph{multivalued}
holomorphic function in $\mathring{\mathscr{C}}$) is the sum of the analytic continuations of the
two multivalued functions in the \textsc{rhs}.

Thus \emph{locally} in $U$ we have a product of two special geometries
with Coulomb branches  $\mathscr{C}^{(s)}$ and fibers $A^{(s)}$ ($s=1,2$).
This local product structure is preserved by the monodromy so it is a kind of ``global''
feature.

At first sight it may seem that we have got a decomposable geometry. \textit{Is this really so?}
Not quite. A decomposable geometry has some peculiar properties which are not implied by the above discussion of \emph{reducible} monodromy.

\subparagraph{The $r=2$ situation.}
For definiteness we discuss the case $r=2$ which already presents all essential phenomena. All decomposable $r=2$ geometries are isotrivial
and their discriminant is contained in the union of the two axes $H_1\cup H_2$ (in particular, it is  normal crossing). Neither properties
follow from reducibility of the monodromy representation. Indeed the geometry cannot be isotrivial
and the discriminant cannot be normal crossing
when $\Gamma$ is infinite (as we are assuming). 
The point is that while the generic fiber is indeed
the product of two elliptic curves $E_{\tau_{(1)}}\times E_{\tau_{(2)}}$, the periods of these curves need not to be constant: we  have just $\tau_{(s)}=\partial^2\cf_{(s)}(a_{(s)})/\partial^2 a_{(s)}$ for $s=1,2$. The only conclusion that we may draw is that the monodromy representation factors through
\be
SL(2,\mathbb{Q})\times SL(2,\mathbb{Q})\hookrightarrow Sp(4,\mathbb{Q})
\ee 
while (by assumption) its image is infinite in both factors, in fact the $\mathbb{Q}$-Zariski closure of the image in each factor is the full  $SL(2,\mathbb{Q})$ modulo finite groups.
In particular there are no invariant symmetric elements in the tensor algebra $T^{\bullet,\bullet}V_s$ of each
summand $V_s$,
so while the splitting of the periods in the two sets $\Pi^{(1)}$ and $\Pi^{(2)}$ is
globally defined, we cannot use each set of special coordinates \emph{separately}
 to define global functions on $\mathscr{C}$
i.e.\! elements in $\mathscr{R}$ which may play the role of coordinates $u_i$ in the factor Coulomb branches of a decomposable geometry. Therefore we expect that the Coulomb coordinates $u_1,u_2$
are non-trivial functions of both $a_{(1)}$ and $a_{(2)}$. However, in order for the functions $u_i=u(a_{(1)},a_{(2)})$, initially defined locally in $U$, to have a \emph{univalued} analytic continuation everywhere in
$\mathring{\mathscr{C}}$, the monodromy transformations of the two summands should be correlated,
and given the Zariski density of $\Gamma$ the only possibility which comes to mind is a further factorization
\be
\pi_1(\mathring{\mathscr{C}})\to SL(2,\mathbb{Q})\xrightarrow{\rm\ diag\ } SL(2,\mathbb{Q})\times SL(2,\mathbb{Q}) \hookrightarrow Sp(4,\mathbb{Q}).
\ee
If this was the case, we would have an element of $\mathscr{R}$ with $\Delta=2$
given by $\epsilon^{ij} \Pi^{(1)}_i \Pi^{(2)}_j$.

In the second scenario (for $r=2$) the image of $\Gamma$ in the first $SL(2,\mathbb{Q})$ factor
is infinite but is finite in the second one. (If both images are finite the geometry is isotrivial, see \S.\,\ref{s:isotrivial}).
When pulled back to a finite cover $\mathscr{C}^\sharp$ of $\mathring{\mathscr{C}}$
 the family takes the form
 \be
 E\times\mathscr{Y}\to  \mathscr{C}^\sharp
 \ee
 with $E$ a constant elliptic curve and $\mathscr{Y}\to  \mathscr{C}^\sharp$ a non-isotrivial
 family of elliptic curves. This is what we call a \emph{quasi-isotrivial} situation. If such
 a special geometry exists, it cannot be decomposable. However it looks pretty hard to cook up
 such a quasi-isotrivial family in a way which is consistent with all the other required structures
 to produce a \emph{regular} $\C^\times$-isoinvariant special geometry. A quasi-isotrivial geometry is a highly over-constrained
 problem which may have a solution only under very specific conditions: peculiar Coulomb dimensions,
 specific form of the discriminant, etc. \emph{A priori} no one will bet a cent on their existence.
 However, wait for \S\S.\,\ref{reddd},\,\ref{s:knott},\,\ref{puzzle:s} and see.

This discussion motivates the following

\begin{defn} An indecomposable geometry is \emph{reducible} iff it is not isotrivial
and its underlying monodromy representation is reducible. 
\end{defn}

Later we shall show that

\begin{fact}\label{ugly} An indecomposable but reducible geometry is a very rare, bizzarre and ugly animal.
They exist only for some special ``troublesome'' dimension $r$-tuples $\{\Delta_1,\dots,\Delta_r\}$.
In rank $2$ they are either quasi-isotrivial or have a dimension $2$ Coulomb coordinate.
\end{fact}
%

 \subsection{Singular fibers along the discriminant}\label{s:sing}

\begin{table}
\caption{\label{table2} Non-semisimple Kodaira fibers. Here $n\geq1$}\vskip9pt
\centering
{\begin{small}\renewcommand{\arraystretch}{1.5}
\begin{tabular}{p{3cm}p{1.5cm}p{1.8cm}}
\hline\noalign{\smallskip}
Kodaira fiber & $\mathbf{I}_n$ & $\mathbf{I}_n^*$ \\
\noalign{\smallskip}\hline\noalign{\smallskip}
monodromy & $\left[\begin{smallmatrix}1 & n\\
0 & 1\end{smallmatrix}\right]$ & $\left[\begin{smallmatrix}-1 & -n\\
0 & -1\end{smallmatrix}\right]$\\
Euler number & $n$ & $n+6$ \\ 
Lie algebra $\mathfrak{g}$ & $\mathfrak{su}(n)$ & $\mathfrak{so}(2n+8)$ \\
\noalign{\smallskip}\hline\noalign{\smallskip}
\end{tabular}\end{small}}
\end{table}

In this technical subsection we quickly sketch the structure of singular fibers as
described in the math literature \cite{oguiso1,oguiso2,oguiso3,jap1,jap2,sawon}
and reviewed in \cite{Cecotti:2023mlc}.
We stress that our assumption that there is a section $s\colon\mathscr{C}\to\mathscr{X}$ implies that only simple singular fibers may be present, so the phenomena associated with more general kinds of fibers are not present in this restricted setting (for ``special geometries'' with multiple fibers see \cite{Cecotti:2023mlc}).

In rank $1$ the classification of the singular elliptic fibers is given by Kodaira \cite{koda1,koda2,koda3,bhm};
the local monodromies of
 \emph{simple} exceptional elliptic fibers
with semisimple monodromy are listed in table \ref{table}\footnote{\ In the following we omit type $I_0$ which is non-singular.} and the ones with non-semisimple
monodromy in table \ref{table2}. 
In rank-$r$ the possible conjugacy classes of the local monodromy around a discriminant component $\cd_a$ are again classified in Kodaira types but new phenomena appear. 
The local monodromy is not sufficient to determine the singular fiber in general \cite{oguiso1,oguiso2,oguiso3,jap1,jap2,sawon}.

\subsubsection{Generalities}

We stress that the following analysis applies under \emph{genericity}
assumptions. For the highly non-generic situations in \textbf{Fact \ref{ugly}}
one needs to introduce the appropriate modifications.
\medskip

We consider a generic point $u$ in the $i$-th discriminant component.
The fiber $\mathscr{X}_u$ is a union of components $F_{u,a}$; the
normalization of each component $F_{u,a}^\text{nor}$ is a fibration over an Abelian variety $A_u$ of dimension
$(r-1)$ (the Albanese variety of the fiber)
which is the same for all components. For \emph{simple} fibers 
(i.e.\! fibers of multiplicity 1) the fiber of $F_{u,a}^\text{nor}\to A_u$ is
a copy of $\mathbb{P}^1$.  
  However it is not always true
that $\mathscr{X}_u\to A_u$ is globally a fibration: two components
$F_{u,a}$, $F_{u,b}$
are glued together by identifying a section $\simeq A_u$ of the first component with a section
$\simeq A_u$
of the second one, but the identification does not necessarily carry the zero to the zero.
When the fiber is not simply-connected this may result in an obstruction to the fiber being
a global fibration over $A_u$.
One may classify the fibers type in terms of the \emph{characteristic cycle}:
i.e.\! a maximal connected sequence of fibers.
The possible (simple) characteristic cycles are classified by the Kodaira types
plus an extra type $I_\infty$ \cite{oguiso1,oguiso2,oguiso3}. 
If the type is semi-simple ($I_0^*$, $II$, $II^*$, $III$, $III^*$, $IV$, and $IV^*$)
we have a fibration over $A_u$ with fiber the Kodaira singularity.
In these cases by some blow-up/down and a finite cover we can
get a fibration with no singular fiber (as we see from the period map) \cite{oguiso2,oguiso3}.
These same operations allow to reduce the case $I^*_b$ to the case
$I_{2b}$.  The case $I_b$ (including $b=\infty$) is the most subtle one,
 but (at least when the
polarization is principal) it is described in great detail in \cite{oguiso3}. The rough
picture, is as follows: each irreducible component $F_{u,a}$
is a fibration over $A_u$ with fiber a $\mathbb{P}^1$.
Two components $F_{u,a}$, $F_{u,b}$ cross along finitely many sections.
In facts their intersection matrix $F_{u,a}\# F_{u,b}$ (suitably defined)
 is equal to the 
intersection matrix $\Theta_a\cdot \Theta_b$ of the components a Kodaira fiber $K$.
If $K$ has type $I_b$, the type of the characteristic cycle
can be $I_{m b}$ where $m$ is a positive integer or
$\infty$ (in this case the characteristic cycle warps $m$ times around the fiber).
When the fiber is simply-connected, $K$ coincides with the type of the characteristic 
cycle.

We stress that, while the type $K$ is locally constant as we move the general point $u$
in the discriminant locus, the type of the characteristic cycle may change abruptly:
as shown in the examples of \cite{oguiso3}, points whose
fibers are of type $I_\infty$ or type $I_k$ ($k$ finite) may be both dense in the discriminant.\footnote{\ See \textbf{Proposition 5.3} of \cite{oguiso3}.} The physical
meaning of the characteristic cycle was discussed in \cite{Cecotti:2023mlc}; we refer the interested reader to that paper.

The type $K$ controls the local monodromy $\varrho_a$ around the $a$-th divisor component.
It is clear from the above discussion that $\varrho_a$ is conjugate in $Sp(2r,\mathbb{Q})$
to the direct sum of the Kodaira monodromy of type $K$, $\varrho_K\in SL(2,\Z)$, and the 
$(2r-2)\times (2r-2)$ identity matrix. In particular, its characteristic polynomial
has the form
\be
\det[z-\varrho_a]=(z-1)^{2(r-1)} P_K(z)
\ee       
where $P_K(z)$ is the characteristic polynomial of $\varrho_K$. Notice that the conjugacy class
in $Sp(2r,\mathbb{Q})$ distinguishes types $II$, $III$, $IV$ from the ones with  the
same characteristic polynomial $II^*$, $III^*$, $IV^*$ (respectively). A priori we would have an ambiguity
for types $I^*_b$ and $I_b$. However this is precisely the case where the detailed analysis of \cite{oguiso2,oguiso3}
applies, and the relation of the monodromy class to the fiber geometry is established:

\begin{pro}[See \cite{oguiso3}] Assume principal polarization.
The monodromy $\varrho_a$ around a semi-stable divisor $\cd_a$
(characteristic cycle of type $I_b$) is conjugate over in $Sp(2r,\Z)$ to
{\renewcommand{\arraystretch}{1.1}\be\label{kkii7612}
\begin{pmatrix} 
1 & \ell \\
0 & 1
\end{pmatrix}\oplus\boldsymbol{1}_{r-1}
\ee}
where $\ell$ is the number of components of the fiber. As a consequence if $\varrho_a$ is a local monodromy around
a divisor of type $I^*_b$, $\varrho_a^2$ is conjugate over $Sp(2r,\Z)$ to the matrix \eqref{kkii7612}
with $\ell=2b$.  
\end{pro}
The geometry also depends on the characteristic cycle, see \cite{oguiso1,oguiso2,oguiso3}
 or \cite{Cecotti:2023mlc} for details. 

\subsubsection{More details on the stable reduction}\label{s:stable}

This subsection may be skipped in a first reading.
Having constructed a putative special geometry, we have to check that it is a genuine one, that is,
analytically \emph{regular.} One important step is to check its behavior along the discriminant which is the locus where things may go wrong. This is done by comparing the near discriminant geometry with the model local geometry which describes the singular fiber type at the generic point in the discriminant. For the present purposes
the important invariant of the singular fiber is its local monodromy class. As always we assume
principal polarization. The following local models are borrowed from refs.\!\cite{oguiso1,oguiso2,oguiso3}.

\subparagraph{A. Transformation of the fiber type.}
Just as in the Kodaira situation (rank-1), locally in 
a neighborhood of $y\in\cd_i$ 
we can always reduce to the semi-stable
case by a sequence of blow up/down of components of the fiber
 together with
a (branched) cyclic base change, that is, by making $\tilde z\to \tilde z^{\mspace{1mu}m}=z$, where $z$
is the local coordinate in $\mathscr{C}$ such that $\cd_i$ is given by $z=0$ and
$m$ is the order of the semi-simple part of the monodromy around $\cd_i$:
{\renewcommand{\arraystretch}{1.4}\be
\begin{tabular}{c@{\hskip30pt}c@{\hskip30pt}c}\hline\hline
model fiber & $m$ & semi-simple reduction\\\hline
$I^*_n$, $n\geq1$ & $2$ & $I_{2n}$\\
$II$, $II^*$ & $6$ & $E(e^{2\pi i/3})$\\
$III$, $III^*$ & $4$ & $E(i)$\\
$IV$, $IV^*$ & $3$ & $E(e^{2\pi i/3})$\\\hline\hline
\end{tabular}
\ee}
\hskip-4.5pt where $E(\tau)$ stands for the elliptic curve of period $\tau$,
and fiber type $E(\tau)$ means that the fiber over the point $\tilde y$
covering $y$ has the form
\be 
E(\tau)\times A
\ee 
for an Abelian variety $A$ of dimension $r-1$.

\subparagraph{B. Transformation of the symplectic form.}
$y$ is a generic point in a  component $\cd_i$ of the discriminant, and
$U\ni y$ a sufficiently small neighborhood. Let $\tilde U\to U$ be the branched cover
of the (local) semi-stable reduction. We write $V=\pi^{-1}(U)$ and $\tilde V=\tilde\pi^{-1}(\tilde U)$,
where tilded symbols refer to the semi-stable reduction.
We can find local coordinates $(x_i,y^i)$ in $V$ (resp.\! $(\tilde x_i, \tilde y^i)$ in $\tilde V$)
such that
\begin{itemize}
\item $x_i$, $\tilde x_i$ are coordinates in the bases $U$ and $\tilde U$;
\item the local equation of $\cd_i$ is $x_r=\tilde x_r=0$.
\end{itemize}
There are three cases \cite{oguiso2,oguiso3}:

\textbf{(I) Type $I_k^*$, $k\geq1$.}
In this case
\be
x_r=\tilde x_r^2,\quad y^r=\frac{\tilde y^r}{\tilde x_r},\qquad  x_i=\tilde x_i\quad y^i=\tilde y^i\ \ \text{for } i=1,\dots,r-1
\ee
\be
\sum_i dx_i\wedge dy^i =\sum_{k=1}^{r-1} d\tilde x_i\wedge d\tilde y^i+2\, d\tilde x_r\wedge d\tilde y^r
\ee

\textbf{(II) Types $I_0^*$, $II^*$, $III^*$ and $IV^*$.} One has
 \be\label{type*}
x_r=\tilde x_r^m,\quad y^r=\frac{\tilde y^r}{\tilde x_r^{m-1}},\qquad  x_i=\tilde x_i\quad y^i=\tilde y^i\ \ \text{for } i=1,\dots,r-1
\ee
\be
\sum_i dx_i\wedge dy^i =\sum_{k=1}^{r-1} d\tilde x_i\wedge d\tilde y^i+m\, d\tilde x_r\wedge d\tilde y^r
\ee
where $m=2,6,4,3$, respectively.

\textbf{(III) Types $II$, $III$ and $IV$.} One has
 \be
x_r=\tilde x_r^m,\quad y^r=\frac{\tilde y^r}{\tilde x_r},\qquad  x_i=\tilde x_i\quad y^i=\tilde y^i\ \ \text{for } i=1,\dots,r-1
\ee
\be
\sum_i dx_i\wedge dy^i =\sum_{k=1}^{r-1} d\tilde x_i\wedge d\tilde y^i+m\, \tilde x^{m-2}\,
d\tilde x_r\wedge d\tilde y^r
\ee
where $m=6,4,3$, respectively.

The monodromy in the cover geometry as $\tilde x_r\to e^{2\pi i}\tilde x_r$ is equal to the monodromy
in the original space as $x_r\to e^{2\pi i m}x_r$ i.e.\! to the $m$-th power of the original monodromy
which is trivial for all semisimple types and unipotent of type $I_{2n}$ in the $I^*_n$ case. 

\subparagraph{C. Local pre-potential in the semi-stable case.}

We consider a (sufficiently small) neighborhood $U$ of a generic point $x$ on a semi-stable divisor $\cd_i$.
Locally we can always reduce to this case.

\begin{pro}[\!\cite{oguiso3}] In $U$ we can find special coordinates $(a^1,\dots, a^r)$ such that
\begin{itemize}
\item the local equation of $\cd_i$ is $a^r=0$;
\item the prepotential has the form
\be\label{juqweraa}
\cf(a^1,\cdots, a^r)=\tilde \cf(a^1,\cdots, a^r)+\frac{\ell}{4\pi i}\,(a^r)^2\,\log a^r
\ee
where $\tilde\cf(a^1,\cdots, a^r)$ is holomorphic in $U$ and $\ell$ is a positive integer;
\item the period matrix is
\be
\boldsymbol{\tau}_{ij}=\frac{\partial^2\cf}{\partial a_i \partial a_j}.
\ee 
\item the local monodromy around the semi-stable discriminant with pre-potential 
\eqref{juqweraa} is unipotent of type $I_\ell$.
\end{itemize}
\end{pro}
Indeed under the transformation $a_r\to e^{2\pi i}a_r$ the periods will transform as
\be
a^i\to a^i,\qquad a^D_j\equiv \frac{\partial\cf}{\partial a^j}\to a^D_j\ \ j\neq r,\qquad
a^D_r \equiv \frac{\partial\cf}{\partial a^r}\to a^D_r+\ell a^r.
\ee
%

\subsection{Period map and structure theorem}
We saw above that the restriction of the fibration $\pi$ to the complement $\mathring{\mathscr{C}}$ is in particular a family $\mathring{\mathscr{X}}$ of polarized Abelian varieties parametrized by the quasi-projective space $\mathring{\mathscr{C}}$.   
The moduli space parametrizing isomorphism classes of
principally polarized Abelian varieties of dimension $r$
is the \emph{Siegel variety} 
\begin{equation}
\mathsf{S}_r\equiv Sp(2r,\Z)\backslash Sp(2r,\R)/U(r).
\end{equation}
We then have a \textit{holomorphic period map} $\check{p}$
\be
\check{p}\colon \mathring{\mathscr{C}}\to \mathsf{S}_r
\ee
which sends a point $u\in  \mathring{\mathscr{C}}$ into the
isoclass $[\mathscr{X}_u]\in \mathsf{S}_r$ of its fiber.  The geometry is isotrivial iff $\check{p}$ is the constant map.

We stress that
the family $\mathring{\mathscr{X}}$ of Abelian varieties over $\mathring{\mathscr{C}}$
is not fully determined by the period map.\footnote{\ This is well-known in the case
of families of elliptic curves over a curve \cite{koda1,koda2,koda3,miranda}. There are several families of elliptic curves over a curve
with the same period map: they are related by the so-called
\emph{quadratic transform}  \cite{miranda,MW}.} Since $\{\pm1\}\subset Sp(2r,\Z)$ acts trivially on the
symmetric space $Sp(2r,\R)/U(r)$,
the period map $\check{p}$ only determines the quotient representation
\be
\check{\varrho}\colon\pi_1(\mathring{\mathscr{C}})\to PSp(2r,\Z)\equiv Sp(2r,\Z)/\{\pm1\},
\ee
while the family of Abelian varieties over $\mathring{\mathscr{C}}$ depends on the actual monodromy  i.e.\! on the specific lift $\varrho$ of $\check{\varrho}$
\be
\begin{gathered}
\xymatrix{\pi_1(\mathring{\mathscr{C}})\ar[rr]^\varrho\ar@/_1.5pc/[rrd]_(0.35){\check{\varrho}} && Sp(2r,\Z)\ar[d]^{\text{can}}\\
&& Sp(2r,\Z)/\{\pm1\}}
\end{gathered}
\ee
 which is a \emph{finer} invariant of the family $\mathring{\mathscr{X}}$ than $\check{\varrho}$.
There may be at most 
finitely many special geometries with the same period map \cite{ffffin}. When $\varkappa=2$
only $\check{\varrho}$ is a \emph{a priori} well-defined (cf.\! \S.\,\ref{kiuuqwa}), and one has to be careful with for the proper
lift $\varrho$ which yields the correct family of Abelian varieties arising from  the particular special geometry.

\subsubsection{The Abelian family over $\mathring{\mathscr{P}}$}

The holomorphic map $\exp(z\,\ce)$
acts by automorphisms, hence 
\be
\mathscr{X}_{\exp(z\,\ce)u}\simeq \mathscr{X}_u\quad\text{for all }z\in\C/2\pi i m\,\Z,
\ee
 i.e.\! the fibers over all points in a $\C^\times$-orbit are isomorphic (as polarized Abelian varieties).
 Then the map $\check{p}$ factors through a holomorphic period map
 \be\label{kkkkkaaa12}
 p\colon \mathring{\mathscr{C}}/\C^\times \simeq \mathsf{Proj}\,\mathscr{R}\setminus Y \to \mathsf{S}_r
 \ee
 where $Y\subset \mathsf{Proj}\,\mathscr{R}\equiv\mathscr{P}$ is the hypersurface
$\cd/\C^\times$. 
When $\mathscr{C}$ is smooth
 \be
 \mathsf{Proj}\,\mathscr{R}=\mathbb{P}(d_1,\cdots, d_r)\simeq \mathbb{P}(q_1,\dots,q_r)\quad \text{(well-formed)},
 \ee 
 and $Y$ is the weighted projective hypersurface with quasi-cone the reduced discriminant $\cd$.
%
%
%

However there is a subtlety. $\mathbb{P}(d_1,\cdots, d_r)$ is not
projectively\footnote{\ The underlying complex space may be nevertheless smooth (this happens when the associated well-formed weighted projective space is $\mathbb{P}^{r-1}\equiv\mathbb{P}(1,\dots,1)$ \cite{dolgaWP}). See \cite{sasaki} for a discussion. } smooth on points which corresponds to $\C^\times$-orbits with enhanced R-symmetry,
since the quotient of $U(1)_R$ which acts effectively on $\mathscr{C}$
does not act freely there. We are mainly interested
in singularities in codimension one:
over an enhanced divisor of index $\nu$
we have $\Z_\nu$ quotient singularities.

It is convenient to replace the divisor $Y$ by the special locus divisor $\mathscr{S}\subset \mathscr{P}$
whose irreducible components $\cs_i$ are either projective images $\pi(\cd_i)$ of the discriminant
components or divisors of $\mathscr{P}$ which are the projective images
of R-enhanced irreducible divisors. 
Note that a special divisor $\cs_i$ can belong to both the discriminant and the R-enhancement locus.
 
We consider the projective monodromy representation
\be
\mu\colon \pi_1(\mathscr{P}\setminus\mathscr{S})\to \begin{cases} \textit{Sp}(2r,\Z) & \varkappa=1\\
\textit{PSp}(2r,\Z) &\varkappa=2
\end{cases}
\ee
which we call the $\mu$-monodromy to distinguish it from the $\varrho$-monodromy representation of
$\pi_1(\mathscr{C}\setminus\cd)$. The two monodromies are related and determine each other.

\begin{fact} For $\varkappa=1$ (resp.\! $\varkappa=2$) the group $\mu(\pi_1(\mathscr{P}\setminus\mathscr{S}))$ is equal to the Coulomb branch monodromy group $\Gamma$
(resp.\! the image $\check{\Gamma}$ of $\Gamma\subset \mathit{Sp}(2r,\Z)$
in the quotient group $\mathit{Sp}(2r,\Z)/\{\pm1\}$).
\end{fact}

\begin{proof}[Argument] Let $\ast\in \mathscr{C}\setminus \pi^{-1}(\mathscr{S})$ be a generic
point in the Coulomb branch. A closed path $\ell\in\mathring{\mathscr{C}}\equiv\mathscr{C}\setminus\cd$,
based at $\ast$, may be continuously deformed to lay in $\mathscr{C}\setminus \pi^{-1}(\mathscr{S})$.
Then $\ell^\sharp\equiv\pi(\ell)$ is a closed path in $\mathscr{P}\setminus\mathscr{S}$ based at $\pi(\ast)$.
On the other hand if $\gamma(t)$ is a closed path in  $\mathscr{P}\setminus\mathscr{S}$ based at $\pi(\ast)$
we can lift it into a path $\tilde\gamma(t)$ starting at $\ast$ and ending at some point $\tilde\gamma(1)$
in the same $\C^\times$-orbit as $\ast$. Hence 
\be\label{jjuuuuull}
\ast=\exp(a\,\ce)\tilde\gamma(1),\qquad \text{for some $a\in\C/2\pi i m\Z$}
\ee
and we can define a closed loop $\gamma^\flat$ in  $\mathscr{C}\setminus \pi^{-1}(\mathscr{S})$
based at $\ast$
\be
\gamma^\flat(t)= \begin{cases}\tilde\gamma(2t) & 0\leq t\leq 1/2\\
\exp[(2t-1)a\,\ce]\tilde\gamma(1) & 1/2\leq t\leq 1
\end{cases}
\ee
Note that the homotopy class of $\gamma^\flat$ is independent of the lift.
When $\varkappa=1$ the two maps
\be
\xymatrix{\pi_1(\mathscr{C}\setminus \pi^{-1}(\mathscr{S})) \ar@<0.5ex>[rr]^(0.55)\sharp && \ar@<0.5ex>[ll]^(0.45)\flat \pi_1(\mathscr{P}\setminus \mathscr{S})}
\ee
are clearly inverse of each other and $\mu(\ell)=\varrho(\ell^\flat)$. 
However when $\varkappa=2$
the base point $\ast$ is fixed by an automorphism of the form  $\exp(2\pi i m\,\ce/2)$ ($m$ odd)
and $a$ in eq.\eqref{jjuuuuull} is ambiguous by the shift $a\mapsto a+i\pi m$.
Hence $\mu(\ell)=\pm\varrho(\ell^\flat)$. 
\end{proof}

\begin{exe} We return to \textbf{Example \ref{n4su3}}: $\cn=4$ with $G=SU(3)$.
As complex manifolds, one has $\mathbb{P}(2,3)\simeq \mathbb{P}^1$ via the map
$(u_1,u_2)\mapsto (u_1^3 :u_2^2)$ \cite{dolgaWP}. The only component of the discriminant is mapped to the point $(1:1)\in\mathbb{P}^1$. The sphere less one point is simply-connected, so all monodromy groups are trivial, while we expect it to be the non-Abelian (but solvable)
group $\mathsf{Weyl}(SU(3))$. The local monodromy $\mu_1$ around the discriminant point  $1\equiv (1:1)$
is physically expected to generate $\mathsf{Weyl}(SU(2))\simeq \Z_2$,
while the full monodromy group is $\mathsf{Weyl}(SU(3))$ (cf.\! \textbf{Example \ref{n4su3}}).
Taking into account the monodromy around the two R-enhanced divisors
we get the correct answer:  the local $\mu$-monodromy around $0$ is of order $\nu_0=2$
and the local monodromy around $\infty$ is of order $\nu_\infty=3$;
hence the $\mu$-monodromy group has the presentation
\be
\Gamma=\big\{\mu_0,\mu_1,\mu_\infty\colon\mu_0^2=\mu_1^2=\mu_\infty^3=\mu_1\mu_2\mu_3=1\big\}
\ee
which is the standard Coxeter presentation of $\mathsf{Weyl}(SU(3))$.
\end{exe}

In conclusion: the family of polarized Abelian $r$-varieties $\mathring{\mathscr{X}}\to\mathring{\mathscr{C}}$
defines a family of polarized Abelian $r$-varieties 
\be
\mathscr{A}\to \mathring{\mathscr{P}}\equiv\mathscr{P}\setminus\mathscr{S}
\ee
which is our main object of study in this paper.
\medskip
 
A basic property of $\Gamma$ is that it contains finite-index normal subgroups
which are \emph{neat} (so, in particular, torsionless) \cite{monopeters}. We write $\Upsilon\triangleleft \Gamma$ for such a subgroup and $\mathsf{G}=\Gamma/\Upsilon$ for the quotient \emph{finite} group. By general theory we have
a finite holomorphic cover $f\colon\mathscr{Q}\to\mathscr{P}$, branched over $\mathscr{S}$, with deck group $\mathsf{G}$ such that the pulled back family 
\be
f^*\mspace{-2mu}\mathscr{A}\to \mathring{\mathscr{Q}}\equiv \mathscr{Q}\setminus f^{-1}(\mathscr{S})
\ee
 has neat monodromy group
$\Upsilon$. For many purposes it is technically  more convenient to consider the $\mathsf{G}$-twisted period map
\be
P\colon \mathring{\mathscr{Q}}\to \Upsilon\backslash Sp(2r,\R)/U(r),
\ee
which contains the same information but is simpler since we got rid of the finite groups.

 \subsubsection{Local monodromies at special divisors}\label{s:localmod} 
 The special divisor $\cs_i$ may be of three kinds:
\begin{itemize}
\item[(1)] projective discriminant non-enhanced: the local monodromies 
satisfy $[\mu_i]=[\rho_i]$;
\item[(2)] R-enhanced non-discriminant. In this case the period map $p$
may be extended holomorphically on the divisor, as it is evident from eq.\eqref{kkkkkaaa12}.
Let $(u,z)$ be local coordinates along a regular enhanced divisor with local homogeneous 
equation $z=0$. If its index is $\nu_i$, we have that the points $p_k\equiv(u, e^{2\pi i k/\nu_i}\epsilon)$
 are all in the same $\C^\times$ orbit, indeed $p_{k+1}=\exp[2\pi i \,\ce/(\Delta(z)\nu_i)]p_k$.
A curve in $\mathring{\mathscr{C}}$ from $p_k$ to $p_{k+1}$ will project to a closed loop $\ell$ in $\mathring{\mathscr{P}}$. Going around $\ell$ for $\nu_i$ times we return to the original point $p_0\equiv p_{\nu_i}$ in $\mathring{\mathscr{C}}$; since the locus $z=0$ is not in the determinant, there is no $\varrho$-monodromy.
 Then the $\mu$-monodromy along the loop $\ell\subset\mathring{\mathscr{P}}$ satisfies
\be
\mu_i^{\nu_i}=1\in \begin{cases}\mathit{Sp}(2r,\Z) & \varkappa=1\\
\mathit{PSp}(2r,\Z) & \varkappa=2.
\end{cases}
\ee
One can try to be more precise. By continuity as $\epsilon\to0$, $\mu_i$ is (up to conjugacy)
the same as the rational representation of the automorphism $\exp[2\pi i \,\ce/(\Delta(z)\nu_i)]$
of the generic fiber over the R-enhanced discriminant. Hence $\mu_i$ is a $\nu_i$-root of $1\in PSp(2r,\Z)$ which lifts in $Sp(2r,\Z)$ to a matrix with eigenvalues
\be
\big\{\!\exp[\pm 2\pi i(1-\Delta_1)/(\Delta(z)\nu_i)],\dots,\exp[\pm 2\pi i(1-\Delta_r)/(\Delta(z)\nu_i)]\big\}
\ee
\item[(3)] R-enhanced and discriminant. The same argument as before yields
\be
\mu_i^{\nu_i}=\varrho_i\in \begin{cases}\mathit{Sp}(2r,\Z) & \varkappa=1\\
\mathit{PSp}(2r,\Z) & \varkappa=2.
\end{cases}
\ee
To find the proper root of $\varrho_i$ (more precisely: to determine the minimal polynomial of $\mu_i$) one compares with the case in which the geometry is locally the product of a special geometry with
fiber an Abelian variety of dimension $(r-1)$ -- the Albanese variety (cf.\! \S.\,\ref{s:sing}) -- on which $\exp[2\pi i\, \ce/(\Delta(z)\,\nu_i)]$
acts by polarized automorphisms times an elliptic fibration over a small disk with a
Kodaira central fiber of the appropriate type.
Up to conjugacy in $Sp(2r,\mathbb{Q})$
(and overall sign for $\varkappa=2$) $\mu_i$ is the block diagonal
 matrix having a $(2r-2)\times (2r-2)$ block of eigenvalues 
\be
\big\{\exp[\pm 2\pi i(1-\Delta_1)/(\Delta(z)\nu_i)],\dots,\widehat{\exp[\pm 2\pi i(1-\Delta(z))/(\Delta(z)\nu_i)]},\dots\big\}
\ee
and a $2\times 2$ block which is the monodromy obtained from the Kodaira monodromy of the fiber
type in $\mathscr{X}$ by the
base change $z\mapsto z^{\nu_i}=w$ (see e.g.\! \textbf{table 5.2} of \cite{MW}).
\end{itemize}

\subsubsection{The structure theorem}\label{s:structurethm}
 The structure theorem of VHS gives the period map in terms of the neat monodromy group
$\Upsilon$ \cite{VHS4,griffithsMT,griffithsMT2,periods}. By definition $\Upsilon$ is a discrete subgroup of the algebraic
group $Sp(2r,\mathbb{Q})$ defined over $\mathbb{Q}$. The $\mathbb{Q}$-Zariski closure $\overline{\Upsilon}^{\,\mathbb{Q}}$ of $\Upsilon$ is the smallest $\mathbb{Q}$-algebraic
group containing $\Upsilon$. $\Upsilon$ and $\overline{\Upsilon}^\mathbb{Q}$ are
algebraically indistinguishable, in the sense that no algebraic invariant can discriminate between them. $\overline{\Upsilon}^\mathbb{Q}$ is also called the algebraic monodromy group.
Let $\overline{\Upsilon}^\mathbb{Q}(\R)\subset Sp(2r,\R)$ be the \emph{real Lie group}
of $\R$-valued points of the $\mathbb{Q}$-algebraic group  $\overline{\Upsilon}^\mathbb{Q}$:
it is a semi-simple Lie group with finitely many connected components.\footnote{\ To show the last statement one uses the fact that the Coulomb branch is - by definition - quasi-projective as a complex manifold hence Liouvillic for the plurisubharmonic functions.}
We write $H$ for a real Lie group such that $\overline{\Upsilon}^\mathbb{Q}(\R)\times H$ is a maximal subgroup of $Sp(2r,\R)$ and
write $K$, $L$ for maximal compact subgroups of the real Lie groups $\overline{\Gamma}^\mathbb{Q}(\R)$ and $H$, respectively. The structure theorem says that the 
period map factorizes through
\be
\mathring{\mathscr{B}}\to \Big(\Upsilon\backslash \overline{\Gamma}^{\,\mathbb{Q}}(\R)/K\Big)\times H/L \to Sp(2r,\Z)\backslash Sp(2r,\R)/U(r)\equiv \mathscr{S}_r,
\ee
and is the constant map in the factor $H/L$. 
\begin{corl} In special geometry
\begin{itemize}
\item[\rm(1)] if $\Gamma$ is finite $p$ is the constant map;
\item[\rm(2)] if $\Gamma$ is Abelian $p$ is the constant map;
\item[\rm(3)] if $\Gamma=\Gamma_1\times \Gamma_2$
the period map decomposes into two distinct period maps
$p_1$ and $p_2$
factoring, respectively, through $\Gamma_1\backslash \overline{\Gamma_1}^\mathbb{Q}(\R)/K_1$ and $\Gamma_2\backslash \overline{\Gamma_2}^\mathbb{Q}(\R)/K_2$.
\end{itemize}
\end{corl}

A special geometry where $p$ is the constant map is called \emph{isotrivial}.
All its smooth fibers are isomorphic as polarized Abelian varieties.
In the situation (3) if one of the two maps is constant we say that the special geometry
is \emph{quasi-isotrivial}.

\begin{corl}
The  $\R$-algebraic monodromy group $M\equiv\overline{\Upsilon}^{\,\mathbb{Q}}(\R)$ of a non-isotrivial special geometry is isogeneous to
\be
M_1\times M_2\times \cdots\times M_r
\ee
where $M_i$ are \textbf{simple, non-compact} real Lie subgroup of
$Sp(2r,\R)$ of \emph{odd} Mumford-Tate type and $\boldsymbol{2r}=\oplus_i V_i$
for a real symplectic representation $V_i$ of $M_i$.
\end{corl}

For the list of allowed pairs $(M_i,V_i)$ see e.g.  the table in \cite{milneShi} or \cite{griffithsMT}.
For a given rank $r$ there are finitely many allowed groups $M$, and we may classify 
the special geometries group by group. The simpler $\R$-algebraic monodromy $M$
the simpler the corresponding classification program.
The case $M=\{1\}$ is specially easy.

This has a consequence already noted in \cite{M13}:
\begin{corl}
{\bf(1)} If the discriminant $\cf$ is a simple normal crossing divisor, the geometry is a product of
$r$ rank-1 geometries; {\bf(2)} let $\mathscr{R}=\C[u_1,\cdots,u_r, v_1,\cdots, v_s]$
and $\cd=(f_1)+(f_2)$ where $f_1\in \C[u_1,\cdots,u_r]$, $f_2\in \C[v_1,\cdots,v_s]$.
Then the special geometry is the product of two special geometries with
chiral rings  $\mathscr{R}_1\in \C[u_1,\cdots,u_r]$, $\mathscr{R}_2\in \C[v_1,\cdots,v_s]$.
\end{corl}

\begin{exe} The discriminant of an \emph{indecomposable} rank-2 geometry
 contains \emph{at least}:
\begin{itemize}
\item[(1)] one knotted component;
\item[(2)] \textbf{\emph{or}} two
unknotted components crossing non-transversely;
\item[(3)] \textbf{\emph{or}} three unknotted components with
pairwise transverse crossings at the origin.
\end{itemize} 
The geometries of the SCFT numbered
 $\#20$, $\#21$, $\#60$ in the tables of \cite{Martone} have a minimal discriminant of type (1);
 those numbered $\#1$, $\#7$, $\#11$, $\#14$, $\#16$, $\#18$, $\#51$-$\#53$, $\#61$-$\#65$, $\#68$
 have a minimal discriminant of type (2); the ones numbered
 $\#24$, $\#27$, $\#30$ a minimal discriminant of type (3). All other SCFT
 in the tables of \cite{Martone} have non-minimal discriminants with more components.
 Model $\#32$ is not indecomposable.
\end{exe}
The case of higher $r$ will be sketched in section \ref{s:r3}.

\subsubsection{$\Gamma$ finite: isotrivial special geometries}
\label{s:isotrivial}

 The
isotrivial special geometries have several (equivalent)
 characterizations:
 \begin{fact} For a $\C^\times$-isoinvariant special geometry
 the following are equivalent:
 \begin{itemize}
 \item[\rm(1)] $\Gamma$ is finite or, equivalently, $M=\{1\}$;
 \item[\rm(2)] all smooth fibers of the special geometry $\mathscr{X}\to\mathscr{C}$
 are isomorphic as polarized Abelian varieties to a fixed variety $A$;
 \item[\rm(3)] there is a \emph{finite} cover $f\colon\check{\mathscr{C}}\to \mathscr{C}$, branched only over the discriminant, such that the pulled-back fibration is trivial, i.e. $f^*\mathscr{X}= A\times \check{\mathscr{C}}$ for a fixed polarized Abelian variety $A$. In particular the special geometry
 $\pi\colon\mathscr{X}\to\mathscr{C}$ is a global quotient  {\rm\cite{VHS4,griffithsMT,griffithsMT2,periods}}
 \be\label{juu7651z}
 \begin{gathered}
 \xymatrix{A\times \check{\mathscr{C}}\ar@{->>}[d]\ar[rr]^{p_2} && \check{\mathscr{C}}\ar@{->>}[d]\\
 \mathscr{X}\equiv (A\times \check{\mathscr{C}})/\Gamma \ar[rr]^\pi&& \mathscr{C}=\check{\mathscr{C}}/\Gamma}
 \end{gathered}
 \ee
   \item[\rm(4)] all local monodromies around the discriminant components $\cd_i$ have finite order. 
   \end{itemize}
 \end{fact}
$(1)\Leftrightarrow(2)\Leftrightarrow(3)$ is very well-know in VHS
 and valid for \emph{all} holomorphic
 fibration in polarized Abelian varieties over a quasi-projective base. We summarize the arguments:
 (1) $\Leftrightarrow$ (2) is the structure theorem together with the Godemen criterion \cite{morris}.
 (1) $+$ (2) $\Rightarrow$ (3) follows from the obvious Galois cover of the complement
 $\check{\mathring{\mathscr{C}}}= \widetilde{\mathring{\mathscr{C}}}/\mathrm{ker}\,\varrho$
 with finite deck group $\Gamma$, together with Riemann's existence theorem \cite{EGA}
 which extends it to a finite \emph{branched} cover $\check{\mathscr{C}}\to\mathscr{C}$.
 (3) $\Rightarrow$ (1) follows from the fact that $\Gamma$ is a subgroup of the finite group of
 polarized automorphisms of $A$. 
 
 The tricky (and important for non-perturbative physical considerations) part of the statement is the equivalence of (4) with one (hence all) of the first 3 conditions. It is the only point where one uses that  $\mathscr{X}\to\mathscr{C}$ is a $\C^\times$-isoinvariant special geometry and not merely a  fibration in polarized Abelian varieties.  (1) $\Rightarrow$ (4) is trivial, and one needs to prove
 (4) $\Rightarrow$ (1). There is a \emph{finite} cover $\mathscr{B}\to\mathscr{C}$, branched only on the discriminant,
 such that the pull-back of the Abelian family to $\mathscr{B}$ has trivial local monodromies.  
The pull-back of the special K\"ahler metric extends to a regular metric over
the inverse image of the discriminant components of types $I_0^*$, $II^*$, $III^*$, $IV^*$.
If we have only local monodromies of these types, the pulled back special K\"ahler metric is conic and regular,
hence it is the flat metric, which implies that the period $\tau_{ij}$ is constant. Otherwise
we have a covering $\mathscr{B}\to\mathscr{A}$, branched only on discriminant components of types $II$, $II$ and $IV$, such that the singular special K\"ahler metric on $\mathscr{B}$ is the pull back of a regular one in $\mathscr{A}$ which then must be flat. In all cases the period matrix $\tau_{ij}$ is locally constant.

 Many $\cn=2$ SCFT have isotrivial geometries:
 \textit{(i)} all $\cn\geq3$ models
 (regarded as special $\cn=2$ SCFT); \textit{(ii)}
all models with $\varkappa\not\in\{1,2\}$; \textit{(iii)}
all rank-1 models. In principle
the isotrivial special geometries can be constructed and classified
 in all ranks $r$ \cite{Cecotti:2021ouq2}. More details will be given elsewhere.

\subsection{``Stratification'' of the special geometry}
The Coulomb branch $\mathscr{C}$ is naturally stratified \cite{M12,M13}. The closure of a codimension-$s$
stratum is a connected component of a multiple intersection of 
discriminant components $\cd_{1_1}\cap\cdots\cap \cd_{i_s}$ of codimension $s$.
To this stratification there corresponds a ``stratification'' of the special geometry
in the sense that the special geometry $\mathscr{X}\to \mathscr{C}$
induces a Lagrangian fibration
$\mathscr{X}_{a_s}\to\mathscr{C}_{a_s}$ over each codimension-$s$ stratum $\mathscr{C}_{a_s}$ whose generic fibers are polarized Abelian varieties
of dimension $\dim \mathscr{C}_{a_s}$.

Let $\cd_{a_1}$ be an irreducible component of the discriminant and $u\in \cd_{a_1}$ a general point.
To $u$ we may associate a polarized Abelian variety $A_u$,
 namely the Albanese variety of the normalization of a connected
component of the fiber $\mathscr{X}_u$. $A_u$ is well-defined
since it does not depend on the chosen component \cite{oguiso1,oguiso2,sawon}. We can glue the various $A_u$'s
along $\cd_{a_1}$ to get a fibration $\mathscr{X}_{a_1}\to \cd_{a_1}\equiv\mathscr{C}_{a_1}$ in polarized Abelian varieties which we claim is a special geometry
in its own right. The geometry $\mathscr{X}_{a_1}\to \cd_{a_1}\equiv\mathscr{C}_{a_1}$ may be seen 
as a symplectic reduction of $\mathscr{X}\to \mathscr{C}$
in a sense which is more general than the usual Marsden-Weinstein-Meyer (MWM) one \cite{sympp,sasaki}
which applies when the usual regularity assumptions of the MWM theorem \emph{do not} hold, but the Lagrangian fibration is holomorphic.
 
 \subsubsection{A generalization of Marsden-Weinstein-Meyer symplectic reduction}
We recall how it works the usual MWM symplectic reduction of a phase space $X$ with symplectic form $\Omega$:
first one considers a level subset $\iota\colon h^{-1}(c)\hookrightarrow X$
of a Hamiltonian function $h\colon X\to \C$, and then one constructs the space $Y$
which parametrizes the orbits
of the flow generated by the Hamiltonian vector $v(h)$ dual to $dh$. In general $Y$ may be quite weird.
In the lucky case when the orbit space $Y$ happens to be a nice manifold, it inherits a symplectic structure $\omega$ from the one of $X$
such that
\be\label{uuuyqrt}
\iota^*\Omega=\pi^*\omega
\ee 
where $\pi\colon h^{-1}(c)\to Y$ is the canonical projection sending a point to the orbit it belongs. In this regular case $Y$ is a phase space in its own right, called the \emph{symplectic reduction} of $X$ along the hypersurface $h=c$. The typical application is when $h$ is the momentum map
of a $U(1)$ action on $X$.
 Unfortunately, most often $Y$ is
\emph{not} a nice manifold, indeed it can be a very ugly space, and the symplectic reduction does not make sense unless the necessary regularity conditions are met.
\medskip

To get the special geometric ``stratification'' one may try to apply the MWM strategy to the zero level set of the prime homogeneous Hamiltonian $p_a\in\mathscr{R}$ such that $\cd_a=(p_a)$. However now the level subset
$p_a^{-1}(0)\subset \mathscr{X}$ is non-smooth by definition. The divisor $p_a^{-1}(0)$
contains several non-reduced irreducible components intersecting each other in a Kodaira pattern. 
The zero-section $s$ is contained in one component of multiplicity 1
of  $p_a^{-1}(0)$.
We focus on the normalization of this zero component, 
and consider the space parametrizing its orbits under the flow generated by the vector $v(p_a)$. 
Away from special points in $\cd_a$, this orbit space $\mathscr{X}_a$ is smooth. In facts 
the normalization of the zero component of the fiber over
the generic point $u\in \cd_a$ is a fibration over the Albanese polarized Abelian variety $A_u$ 
and the vector $v(p_a)$ is tangent to the fiber \cite{oguiso1},
so we may identify $A_u$ with the fiber at $u\in\cd_a$
of the orbit space $\mathscr{X}_a\to \cd_a$.
In other words $\mathscr{X}_a\to \cd_a$ is the fibration with generic fiber $A_u$. 
Note that if we start with the
normalization of any other irreducible component we
would get the same orbit space $\mathscr{X}_a$, which then is an `intrinsic' property of the divisor $\cd_a$.
It is easy to check tha $\mathscr{X}_a$ carries a natural symplectic form such that \eqref{uuuyqrt} holds.\footnote{\ As mentioned before, in the stable case (monodromy type $I_m$) there is some additional subtlety related to the type of the characteristic cycle. This does not affect the conclusion, just requires some more care in the precise definitions of the various entities.} 
$\mathscr{X}_a$ may (in facts, \emph{should}) have singularities
over a codimension-$1$ locus in $\cd_a$. If these singularities admit a crepant resolution,
the resolved $X_a$ is a \emph{bona fide} special geometry of rank $(r-1)$ and
we can consider recursively its own ``stratification''. We assume that this can be done.

\subsection{``Stratification'' vs. Coulomb dimensions}\label{s:straDim}

``Stratification'' gives new information on the allowed Coulomb dimensions.
The formulae of \cite{caorsi} reviewed in \S.\,\ref{jjuutreq} assumed the existence of
a \emph{regular} coordinate axis $H_i$ not contained in the
discriminant. While this is often the case, there are many geometries where 
no such ``good'' axis is present, and we need a more general analysis.

We may assume without loss that $\Delta_j>1$ for all $j$.
Fix an axis $H_i$ and let $\ell$ be the smallest integer such that $H_i$ is contained
in a codimension-$\ell$ (open) stratum; clearly $1\leq \ell\leq r$.
Let $u\in H_i$ be a generic point.
Let $\mathscr{X}_{u,0}$ be the unique component  of the fiber which contains the smooth point $s(u)$. 
The normalization of $\mathscr{X}_{u,0}$ is
a fibration over an Albanese variety $A_{u}$ of dimension $\ell$. 
$T_{s(u)} \mathscr{X}_{u,a}$ is isomorphic to $T_u^*\mathscr{C}$.
The automorphism $\xi_u=\exp(2\pi i\, \ce/\Delta_i)$ fixes $u$ and acts
linearly on
$T_{s(u)}\mathscr{X}_{u,a}$ with eigenvalues
\be\label{jytqert}
\exp(2\pi i (\Delta_j-1)/\Delta_i),\quad j=1,2,\cdots, r
\ee
Since $\xi_u$ is an automorphism of the fiber, it preserves its Albanese variety $A_u$
and acts on it as a polarized automorphism for some \emph{polarization}.
Hence $\ell$ out of the $r$ eigenvalues \eqref{jytqert} are the eigenvalues of
the analytic representation $\sigma$ of a polarized automorphism of the Abelian variety
$A_u$ of dimension $\ell$.
Moreover $\xi_u$ acts trivially on $\C\, du_i|_u \subset T_u\mathscr{C}$
which corresponds to a vector field tangent to the Albanese variety. Hence
$\exp(-2\pi i/\Delta_i)$ is one of the eigenvalues of $\sigma$.
Writing $\Delta_i=m/n$ with $\gcd(m,n)=1$ and $n<m$, we get
$\phi(m)\leq 2\ell$.

\begin{fact}Let $H_i$ be an axis of dimension $\Delta_i$ which is new in rank $\ell$. The axis $H_i$ is not contained 
in a stratum of $\mathscr{C}$ of dimension $<\ell$.
In particular, if $\Delta_i$ is a new dimension
$H_i$ does not belong to the discriminant. 
If $\Delta_i$ is new in rank $r$, the full $r$-tuple $\{\Delta_1,\dots,\Delta_r\}$
is given by the formulae of {\rm\S.\,\ref{jjuutreq}}. If $\Delta_i$ is new in rank $\ell$,
at least a subset of $\ell$ dimensions are given by those formulae
with $r$ replaced by $\ell$.
\end{fact}

In particular, when the rank is 2 and $\Delta_i$ is new  rank-2
the corresponding axis $H_i$ does not belong to the discriminant, and the dimension formulae of
\S.\,\ref{jjuutreq}
can be used giving a small set of possibilities for the second dimension.

 \section{Classification: the inverse problem}\label{s:invpro}
 
 One possible strategy to construct/classify all $\cn=2$ SCFTs is to
 construct first all possible $\C^\times$-isoinvariant
 special geometries and then work out the (finitely many)
 SCFT associated with each one of them. This program was
 initiated and carried on to a certain extend by Argyres and coworkers
(see \cite{M1,M2,M3,M4,M5,M6,M7,M8,M9,M10,M11,M12,M13} and references therein). The first step in such a program requires
 to construct all $\C^\times$-isoinvariant special geometries, 
 or at least to chart their ``geography/zoology'', by following either
 one of two natural strategies:
 \begin{itemize}
 \item[(A)] rank after rank;
 \item[(B)] or general subclass after general subclass
  in increasing order of ``complexity''.
 \end{itemize}
 Subprogram (A) is completed in rank 1 \cite{M3,M4,M5,C1,C2} and in good shape in rank 2 \cite{M14,M15,M16},
 while subprogram (B) was initiated in \cite{Cecotti:2021ouq} with the introduction of a first
 numerical invariant, $\kappa$, which partially characterize
 the complexity  subclasses. The subclasses with $\kappa\not\in\{1,2\}$ can be worked out explicitly for arbitrary $r$ \cite{Cecotti:2021ouq2} (at least when $\mathscr{C}$ is smooth). The Lagrangian SCFT are also classified for arbitrary $r$ \cite{Tach}, but not all their global special
 geometries are explicitly known. In our view the best course is to follow a ```mixed'' strategy where we describe
 the simplest subclasses for arbitrary ranks and then discuss the residual ``hard'' geometries
 rank by rank. We propose the following coarse classification of
 the indecomposable, non-free, unitary geometries:
 {\setlength{\arrayrulewidth}{0.4mm}
\renewcommand{\arraystretch}{2}$$
 \begin{tabular}{|c|p{11cm}|}\hline
 I & isotrivial geometries\\\hline
 II & quasi-isotrivial geometries not of class I\\\hline
 III & reducible indecomposable geometries not of classes I, II\\\hline
 IV & non-rigid geometries not of classes I, II, III\\\hline
 V & $\mu$-rigid geometries not of classes I-IV\\\hline
VI & others\\\hline
 \end{tabular}
 $$
 } Classes I and IV are essentially charted for all ranks. Classes II and III are very
 constrained, and a systematic analysis may be doable. In these cases one aims to produce strong necessary conditions which rule out such geometries except under ``exceptional'' circumstances such as peculiar values of the dimensions $\{\Delta_1,\dots,\Delta_r\}$ and of the local monodromies determining their ``zoology''. 
 For classes V and VI one falls back to the rank by rank strategy.
 Class VI is the less amenable to analysis and \emph{a priori} the largest one, but 
 
 \begin{hope}[Supported by some experimental evidence] Class VI is hard but empty.
 \end{hope}
 
 Each class is further divided into subclasses.

 \subsection{The inverse problem}

 For the last two classes of geometries at the moment there is no better plan
 than studying one rank at the time. 
In this direction a
 natural idea, already put forward in \cite{M12,M13}, is to exploit the
``stratification'' of the special geometry in lower rank ones,
and try to construct recursively
all possible geometries in rank $r$
by ``gluing'' together local rank-$(r-1)$ special geometries around codimension-$1$ loci $\cd_a$ and a ``boring'' rank-$r$ open stratum of generic vacua in the complement
$\C^r\setminus\cup_a\cd_a$ in ``all consistent ways''.
 This strategy leads to what we call the \textbf{Inverse problem} which now we  state in concrete terms: 
%
%
 \begin{inverse}
 Given the following \emph{data:}
 \begin{itemize}
  \item[\textit{(i)}] the Coulomb dimensions $\{\Delta_1,\dots,\Delta_r\}$,
 \item[\textit{(ii)}] the discriminant $\cd$,
 \item[\text{(iii)}] the conjugacy class of the local monodromy at each irreducible component $\cd_a$  of $\cd$,
 which we label by the corresponding Kodaira type $C_a$;
 \end{itemize}
 construct the \emph{indecomposable} special geometries of rank $r$ with these data
 or conclude that no such geometry exists. 
 \end{inverse}
 We can restrict the set of allowed data. First of all we may consider only $r$-tuple $\{\Delta_1,\dots,\Delta_r\}\equiv\lambda\{d_1,\dots,d_r\}$ with $\Delta_i$'s in the finite set $\cup_{s=1}^r \Xi(s)$ of allowed dimensions
 in rank $r$,
 since otherwise the answer is \emph{negative} for obvious reasons. 
 We may also assume that $1/\lambda$
 is integral ($\varkappa=1$) or half-integral ($\varkappa=2$) since otherwise the geometry is automatically isotrivial and hence can be constructed/classified by group-theoretical means.
 If some of the $\Delta_i$ is equal $2$ the geometry is \emph{non-rigid} (see \S.\,\ref{s:rigidity}). A \textbf{folk-theorem} states that a non-rigid
 special geometry describes a Lagrangian SCFT which is more conveniently studied by different (physical) means.
 There are many other restrictions in the $r$-tuple from the
 dimension formulae and their generalizations.
We can limit ourselves to admissible $r$-tuples $\{\Delta_1,\dots,\Delta_r\}$
 with $\Delta_i\neq 1,2$ and $2\lambda\in\mathbb{N}$ without essential loss.  
 We know that when all Kodaira types $C_a$ are semi-simple the geometry is isotrivial,
 therefore we may assume that at least one monodromy class $C_a$ is $I_n$ or $I_n^*$ with $n>0$.
This, in particular, requires the discriminant $\cd$ not to be normal crossing; more generally
we have the condition 
that $\pi_1(\mathscr{C}\setminus \cd)$ has an infinite quotient $\Gamma$ which is semisimple
with a symplectic action on $\mathbb{Q}^{2r}$. When focusing on the ``harder'' classes IV, VI we may assume the action to be irreducible.

\medskip
 
 In physical terms the \textbf{Inverse problem} asks to find a common UV completion to the several
 low energy effective theories describing locally the degrees of freedom which are light
  along each discriminant component, \emph{if this (geometric) global UV completion exists}.
%
 
 \begin{fact}
There are at most finitely many geometries for a given datum $(\{\Delta_i\},\cd,\{C_a\})$.   
 \end{fact}  
 
The general expectation is that the special geometry associated to an inverse datum, \emph{when it exists,} 
  is actually unique.\footnote{\ In this statement is important that $C_a$ is defined as a Kodaira type, i.e.\! as a conjugacy class in $SL(2,\Z)$; if $C_a$ is seen as a conjugacy class in a larger group, then multiple geometries do exist.} The idea is that the existence of a global special
 geometry is already a very overdetermined problem, which can be solved positively  only for a tiny set
 of data $(\{\Delta_i\},\cd,\{C_a\})$. It looks extremely unlikely that we have more than one solution for
 the same datum. The finiteness in the \textbf{Fact} follows from the fact that there are finitely many
 fibrations in Abelian $r$-varieties over $\mathscr{C}\setminus\cd$ for fixed $\cd$.
 For almost all choices of $\cd$, the monodromy of \emph{all} these finitely-many fibrations would not
 be compatible with \emph{any} Kodaira-type/dimension data, so that
 \textit{only a ``short list'' of discriminants are allowed in a given rank $r$ and for a given
 discriminant $\cd$ only a ``small set'' of Kodaira-type/dimension datum is consistent}. 
 We shall discuss in \S.\,\ref{s:r3} the conditions on $\cd$. They turn out to be very restrictive.
 
For each consistent set of ``stratification'' data there exist only finitely many Abelian fibrations,
but this is just one of the many geometric structures required for a special geometry.
To be a special geometry the total space of an Abelian family over $\mathscr{C}\setminus\cd$ 
should have a symplectic structure which extends over the ``bad'' fibers
such that all fibers are Lagrangian. This is quite a severe restriction to impose on the few surviving
fibrations, and for most consistent-looking data no one will pass the test. The existence
of more than one geometry for the same datum looks totally
unlikely.
\medskip
 
\noindent The \textbf{inverse problem} may be studied with different aims.
 In decreasing order of ambition:
 \begin{itemize}
 \item[(1)] to construct all $\C^\times$-isoinvariant geometries recursively in the rank $r$;
 \item[(2)] to discover interesting properties which hold for all such geometries;
 \item[(3)] to get necessary conditions for the existence of a geometry;
 \item[(4)] to construct 
 examples with prescribed/desired properties.
 \end{itemize}
 
 \subsection{Taking ``stratification'' seriously}\label{s:finD}
 
 At first sight the weak aspect of the \textbf{Inverse problem} strategy is that we need
 to specify 
\emph{a priori} a discriminant $\cd\subset\mathscr{C}$ (or, equivalently, a special divisor
 $\mathscr{S}\subset\mathscr{P}$). While the datum of a dimension $r$-tuple $\{\Delta_1,\dots,\Delta_r\}$
 is to be chosen from a known \emph{finite} set, and the Kodaira classes $C_i$ from a known discrete set which becomes
 finite when we consider conjugacy classes in $Sp(2r,\R)$ instead of $Sp(2r,\Z)$,
 the discriminant $\cd=(P)$ \emph{a priori} is just a square-free homogeneous element $P\in \mathscr{R}$ which may have (and indeed has) several components, with arbitrarily bad singularities, intersecting each other is all ugly non-transverse ways. At face value, {describing the
 set of meaningful initial data} for the problem looks harder than
 finding needles in a haystack.
 
 Luckily it is not so. Consider (the normalization of) an irreducible component $\cd_a$ of the discriminant.
 If we take ``stratification'' seriously we conclude that it is the Coulomb branch
 of a $\C^\times$-isoinvariant special geometry on its own right. Thus $\cd_a$
 cannot be an arbitrary hypersurface in $\mathscr{C}$, it should satisfy the conditions
of special geometry. $\cd_a$ cannot be, say, a quasi-cone over a smooth Calabi-Yau.
The affine variety $\cd_a$ needs not to be smooth, so its normalization may not be just a copy of $\C^{r-1}$
(although this is rather common in examples), but at the very least
it should be a quasi-cone over a Fano variety or (according to physical intuition)
satisfy the even more restrictive condition of admitting
a polynomial parametrization of the form $u_i=u_i(v_j)$.

Going on, consider the divisor
\be
\cd^\prime_a\equiv \cd_a\cap (\cd\setminus \cd_a)\subset\cd_a
\ee
 (or rather the respective normalizations). The pair $(\cd_a,\cd_a^\prime)$
 should be in the list of allowed pairs $(\textsf{Coulomb branch},\textsf{Discriminant})$ in rank $(r-1)$,
 a condition that we call \emph{hereditarity}.
 This sets a strong condition on $\cd$. 
 
 Then there are two other sets of constraints on $\cd$.
 The first one comes from topology: the fundamental group of the complement
 $\pi_1(\mathscr{C}\setminus\cd)$ should have the right properties
 (not to be finite, or solvable, or a product, etc.). Thus $\cd$
 cannot be normal crossing, nor $\cd_1+\cd_2$ with all components of $\cd_1$ crossing transversely the ones of $\cd_2$, etc.
 
 However the most strong condition on $\cd$ comes from the hero of the present paper:
 \emph{rigidity} which will be the central topic of the rest of the note.
 
We introduce a coarse equivalence relation between discriminants
 or rather special divisors (``commensurability'').
 $\mathscr{S}\sim \mathscr{S}^\prime$ iff their complements $\mathscr{P}\setminus\mathscr{S}$
 and $\mathscr{P}\setminus\mathscr{S}^\prime$ admit a common finite unbranched cover,
 so that their fundamental groups have an isomorphic finite-index normal subgroup.

\begin{hope}[Well supported] In rank-$r$ there are only finitely many equivalence
classes of  (deformations classes of) allowed special divisor.
\end{hope} 
   
  The validity of \textbf{Hope} is a tautological consequence of the (physically expected) fact that there are only finitely many special geometries at a given rank $r$. But what we really hope for is that the allowed divisor classes
   can be listed \emph{a priori} without working out the classification of the geometries themselves.
   In the rest of the paper the hope will be illustrated by several examples in various situations.
 
 \section{Abelian families arising from a special geometry}\label{s:raising}
 
 We adopt a viewpoint (and language) which is more familiar in the context of
 Calabi-Yau moduli spaces and mirror symmetry. As stated in the introduction, the two
 problems are \emph{formally} similar.
 \medskip
 
 Suppose we are given an $\C^\times$-isoinvariant Abelian family $\mathring{\mathscr{X}}\to \mathring{\mathscr{C}}$,
or equivalently a family 
\be\label{profamily}
\mathscr{A}\to \mathring{\mathscr{P}}
\ee
 where 
$\mathring{\mathscr{P}}\equiv\mathscr{P}\setminus\mathscr{S}$
with $\mathscr{S}$ the special divisor in $\mathscr{P}$.
The Abelian family $\mathscr{A}$ defines a monodromy representation 
$\mu\colon \pi_1(\mathring{\mathscr{P}})\to Sp(2r,\Z)$ and
also a flat Gauss-Manin (GM) differential equation of rank $2r$
of the form
\begin{equation}\label{iuytqz}
\nabla\mspace{1mu}\Phi \equiv d\mspace{1mu}\Phi-\Phi A=0
\end{equation}
 where the connection $A$ is locally a $2r\times 2r$ matrix
of holomorphic differentials and $\Phi$ is the fundamental solution
(a $2r\times 2r$ matrix normalized to be the identity at a reference point $\ast\in\mathring{\mathscr{P}}$). The differential equation has only \emph{regular singularities} along $\mathscr{S}$\cite{deligne}. When $\dim\mathscr{P}=1$ it reduces to a classical Picard-Fuchs equation, and we use the terms Gauss-Manin connection and Picard-Fuchs system interchangeably. For simplicity we consider the smooth case
where $\mathscr{P}=\mathbb{P}(q_1,\dots,q_r)$ and we work in the affine patch
$z_r\neq0$. In this patch ($\simeq \C^{r-1}$) the connection can be written as
\be
\nabla=d+\sum_i \frac{A_i(u)}{p_i(u)}
\ee
where the  $p_i(u)$ are the equations of the irreducible components of $\mathscr{S}$ and the $A_i(u)$ are $2r\times 2r$ matrices of holomorphic differentials. The connection should be flat
\be
\nabla^2= dA+A^2=0,\qquad A\overset{\rm def}{=}\sum_i \frac{A_i(u)}{p_i(u)}
\ee

 The Riemann-Hilbert correspondence states
that there is a one-to-one equivalence between the GM PDEs and
the monodromy representation up to the obvious equivalences. In our case the monodromy takes values in $Sp(2r,\Z)$.
Since we are mainly interested in irreducible non-isotrivial geometries,
we may assume that the monodromy group $\Gamma$ acts irreducibly and contains
non-trivial unipotent elements.\footnote{\ If $\Gamma$ does not contain non-trivial unipotent elements it is automatically finite.} Since $\pi_1(\mathscr{P})$ is trivial,
the monodromy is generated by the images of loops encircling the
irreducible components $\cs_i$ of the special divisor $\mathscr{S}=\sum_i \cs_i$.
The monodromy group acts on the solution as
\begin{equation}
\mu([\ell])\cdot \Phi= \Phi\bigg|_{\text{\begin{tiny}$\begin{smallmatrix}\text{analytically}\\ \text{continued\phantom{m}}\\ \text{along $\ell$\phantom{mml}}\end{smallmatrix}$\end{tiny}}}
\end{equation}
where $\ell$ is a loop based at a reference point $\ast\in\mathring{\mathscr{P}}$
which represents a given class in $\pi_1(\mathring{\mathscr{P}},\ast)$.
The Zariski closure over $\C$ of the group $\Gamma\subset Sp(2r,\Z)$
is the \emph{differential Galois group of the PDE \eqref{iuytqz}.}
We use the symbol $\Pi$ for a column $2r$-vector given by a linear combination of
the columns of $\Phi$ (with complex coefficients). The symplectic $2r\times 2r$ matrix will be written $\Omega$ (so $\gamma\in\Gamma$ $\Rightarrow$ $\gamma^t\Omega\gamma=\Omega$). In a basis where  the monodromy elements are integral symplectic matrices,
the entries of $\Pi$ are the periods of some holomorphic differential $\eta$
with respect to a symplectic basis $\{A^i,B_j\}$ of 1-cycles on the fibers
{\renewcommand{\arraystretch}{1.4}\be
\Pi =\begin{pmatrix}\int_{B^i}\eta\\
\int_{A_i}\eta\end{pmatrix}
\ee}
\hskip-4pt The period vector $\Pi$ is multivalued on $\mathring{\mathscr{P}}$. We see $\Pi$ as
a map defined on the universal cover $\widetilde{\mathring{\mathscr{P}}}$
of $\mathring{\mathscr{P}}$
\be
\Pi\colon \widetilde{\mathring{\mathscr{P}}}\to \mathbb{P}^{2r-1}
\ee
equivariant with respect to the deck group $\pi_1(\mathring{\mathscr{P}})$
\be
\sigma^\ast \Pi= \mu(\sigma)\Pi,\qquad \sigma\in \pi_1(\mathring{\mathscr{P}}).
\ee

\begin{fact} The family \eqref{profamily} of Abelian varieties arises from a $\C^\times$-isoinvariant special geometry, iff its GM equation has a preferred period
$\Pi$ such that
its image in $\mathbb{P}^{2r-1}$ is a Legendre submanifold of the canonical contact structure of
 $\mathbb{P}^{2r-1}$, that is,
\be\label{Legendre}
\Pi^t \,\Omega\, d\Pi=0
\ee 
while the $(1,1)$-form $i\, d\Pi^\dagger\,\Omega\,d\Pi$ is positive definite.
\end{fact} 
\begin{proof}(Adapted from \cite{BG}) Write $\Pi=(a^D_i,a^j)^t$ ($i,j=1,\dots,r$). For a Legendre submanifold in generic position,
locally there exists a holomorphic function $\cf(a^j)$, homogeneous of degree $2$,
such that
\be
a_i^D= \frac{\partial \cf(a)}{\partial a^i}.
\ee
On the other hand $a^i=\int_{A_i}\eta$ implies $\eta=\sum_i a^i \omega_i$
where $\{\omega_i\}$ are holomorphic differential normalized as $\int_{A_i}\omega_j={\delta^i}_j$.
Then the period matrix of the Abelian variety is 
\be
\tau_{ij}=\int_{B^i}\omega_j =\frac{\partial}{\partial a^j}\!\int_{B^i}\eta=\frac{\partial}{\partial a^j} a_i^D= \frac{\partial^2\cf(a)}{\partial a^i\,\partial a^j}.
\ee
In particular $\mathrm{Im}\,\tau_{ij}>0$ by the second Riemann relation for the Abelian fibers.
\emph{A posteriori} the differential $\eta$ defined by the Legendre condition is identified (when it
exists) with the Seiberg-Witten differential $\lambda$. 
In other words, the Legendre-manifold condition is equivalent to the existence of a Seiberg-Witten differential on the pull-back of the Abelian family from $\mathring{\mathscr{P}}$ to $\mathring{\mathscr{C}}$.
When such differential exists, it pulls back to a  Lagrangian fibration $\mathscr{X}$
over $\mathscr{C}$ by taking as local pre-potential $\cf$ the one defining (locally) the Legendre submanifold.
The Legendre condition is also called the ``integrability'' condition.
\end{proof}

\subparagraph{Comparison with the Calabi-Yau case.}
A family of Calabi-Yau 3-folds with $h^{2,1}=r-1$ is also described by a GM connection with a period
$\Pi$ satisfying the Legendre condition \cite{BG}. The difference is that in the CY case the
matrix $\mathrm{Im}\,\tau_{ij}$ has signature $(r-1,1)$ instead of being positive-definite.
A part for this crucial sign condition, all equations look formally identical, so we can borrow methods developed in the context of Calabi-Yau periods and mirror symmetry.

\subsection{The Legendre condition}
\emph{A priori} the Legendre condition may look a severe constraint on the family of
Abelian varieties. However often it is not hard to satisfy. Consider the 1-form $\Pi^t \,\Omega\, d\Pi$
for a generic period. If it extends over the special divisor $\mathscr{S}$, it is a global
holomorphic 1-form on $\mathscr{P}$ hence zero since the latter has $h^{1,0}=0$.
So the question boils down to the existence of a period $\Pi$ which is ``sufficiently regular''
 along the special divisors $\cs_a$. Since the local monodromy around $\cs_a$ corresponds (by construction) to the IR physics of the degrees of freedom which purportedly becomes light there,
 ``sufficiently regular'' means that the period $\Pi$ approaches along $\cs_a$ the ones
 describing
the local physics of the light degrees of freedom, i.e.\! that $\Pi$ agrees near $\cs_a$
 with one particular local period, the ``good'' one. The ``good'' local period $\Pi$ exists
 because we choose the exponents of the monodromy (encoded in the datum $C_a$
 of the inverse problem) precisely to guarantee that the local  geometry follows the appropriate model (cf.\! \S.\,\ref{s:stable}).
 
 Seen in this way, the Legendre condition is a requirement on the connection formula for the PDE:
  the ``good'' local period near $\cs_a$, when analytically continued to a neighborhood of $\cs_b$, should agree with the ``good'' local solution there for all $b$. 
  In the language of the \textbf{inverse problem}:
  the several local geometries along the strata should ``glue well'' together.
  Alternatively one has to show that
  the holomorphic one-form $\Pi^t\Omega d\Pi$ vanishes in a neighborhood of $\cs_a$
  and not just asymptotically when we approach it. This point of view 
clarifies that the special period $\Pi$ (when it exists) is unique.

In the classical case of rank-2
we may be very precise. Here we use the analogy with 
mirror symmetry for one-modulus Calabi-Yau's \cite{draco,AZthm}.

\subsection{The classical case $r=2$} \label{k87xxx}
In the case $r=2$ the normalization of $\mathscr{P}$ is just
the Riemann sphere $\mathbb{P}^1$ for all pairs of Coulomb dimensions $\{\Delta_1,\Delta_2\}$,
while $\mathring{\mathscr{P}}\simeq \mathbb{P}^1\setminus\{z_1,\dots,z_k\}$
where $\{z_i\}$ are $k$ distinct points on the sphere.
The Gauss-Manin PDE reduces to a Fuchsian ODE. Being primarily interested in irreducible geometries,
we assume the monodromy to be irreducible. Then, by a ``choice of gauge''
the ODE can always
be written in the form \cite{beau,haraoka}
\begin{equation}\label{kuyqwert888}
\frac{d}{d z}\Pi=\Pi A,\qquad
A\equiv\sum_i \frac{A_\ell}{z-z_\ell}
\end{equation}
for constant $4\times 4$ matrices $A_\ell$ satisfying the non-resonant condition.\footnote{\ That is: if two eigenvalues of $A_\ell$ are equal mod 1, then they must be equal.}
The conjugacy class in $GL(4,\C)$ of the local monodromy $\mu_\ell$ at $z_\ell$ is
\begin{equation}\label{uuuu76cz}
\big[\mu_\ell\big]=\big[\exp(2\pi i A_\ell)\big].
\end{equation}
The equation \eqref{uuuu76cz} can be recast
in the form of the 4-th order Fuchsian linear equation for a scalar function $y$:
\begin{equation}\label{ODe4}
L y=0,\quad\text{where }L=\frac{d^4}{dz^4}+a_3(z)\,\frac{d^3}{dz^3}+
a_2(z)\frac{d^2}{dz^2}+a_1(z) \frac{d}{dz}+ a_0(z)
\end{equation} 
whose coefficients are rational functions on $\mathbb{P}^1$
determined by the points $z_i$ and the monodromy representation of
the original equation. 
In \cite{AZthm} (see also \cite{draco}) it is shown that the coefficients $a_i(z)$
of a 4-th order ODE \eqref{ODe4}
with differential Galois group
contained in $Sp(4,\C)$ satisfy the differential relation\footnote{\ The inverse statement holds only modulo finite groups (equivalently modulo finite covers).}
\begin{equation}\label{thecondition}
4 a_2 a_3 -8 a_1- a_3^3+8 a^\prime_2-6 a_3a_3^\prime-4 a^{\prime\prime}_3=0.
\end{equation}
Let $y_1$, $y_2$, $y_3$, $y_4$ be four linearly independent solutions to 
the ODE \eqref{ODe4}. Out of these 4 functions we may produce 6 independent 
$2\times 2$ Wronskians
{\renewcommand{\arraystretch}{1.2}\begin{equation}
W_{ab}\equiv\left|\begin{matrix}y_a & y_b\\ y_a^\prime & y^\prime_b\end{matrix}\right|,\quad 1\leq a<b\leq 4.
\end{equation}}
In a general 4-th order Fuchsian ODE the Wronskians satisfy a linear ODE of order 6.
The Almkvist-Zudilin theorem \cite{AZthm} states that when the coefficients obey the condition \eqref{thecondition} the Wronskians $W_{ab}$ solve
an ODE of just order 5. In other words, a linear combination of the
$W_{ab}$ vanishes, or equivalently there exists a constant, non-degenerate, antisymmetric,
$4\times 4$ matrix $M_{ab}$ (unique up to overall normalization) such that
\be\label{rrrt54312}
M_{ab}\, y_a\, y_b^\prime=0.
\ee
We can find a constant matrix $S\in GL(4,\C)$ such that $S^t M S=\Omega$, with
$\Omega$ the standard symplectic
 $4\times 4$ matrix. Define
 \begin{equation}\label{juytqq}
 \Pi_a\overset{\rm def}{=} {(S^{-1})_a}^b y_b 
 \end{equation}
Eq.\eqref{rrrt54312} becomes
 \be\label{ju77654231i}
 (\Pi^t \mspace{2mu}\Omega \mspace{2mu}\partial_z \Pi) dz=0,
 \ee
 which is the Legendre condition \eqref{Legendre}. We see that the Legendre condition is an automatic consequence of having the ``right''
 algebraic monodromy/differential Galois group. More precisely here there are \emph{two}
 ingredients: the ``right'' differential Galois group, and the choice of the ``good'' solution
 for the
 4 periods $\Pi_a$.
 Since all solutions are linear combinations of any one basis of solutions, 
 the ``good'' ones are special linear combinations and the matrix ${S_a}^b$ specifies 
 the one we need to keep, cf.\! eq.\!  \eqref{juytqq}. Since $\Pi$ are the periods of the so-called Seiberg-Witten
 differential, this shows that the monodromy representation also fixes which linear
 combination of the differentials on the fiber is the Seiberg-Witten differential.
 Note that we have still some freedom: making $S\to t\,Sg$ with $g\in Sp(4,\C)$
 and $t\in\C^\times$ we get a new $\Pi$ which still satisfy eq.\eqref{ju77654231i}. 
 The overall scale freedom just reflects the fact that we consider the image of
 the periods in projective space. The freedom by a symplectic rotation is more
 interesting: for the moment $\Pi$ are periods defined over some complex sympletic basis of
 $H_1(\mathscr{A}_u,\Z)\otimes_\Z\C$. Using $g\in Sp(4,\C)$ we transform this basis
 into an integral symplectic basis of 1-cycles (which exists since, by assumption, our ODE
 arises from an actual family of Abelian surfaces) characterized by the fact that the monodromy elements becomes integral matrices in $Sp(4,\Z)$ when written in terms of such a basis. Then $\Pi\in \mathbb{P}^{2r-1}$
becomes unique modulo the action of the Siegel modular group $Sp(2r,\Z)$.
 Finally the condition $\mathrm{Im}\,\tau_{ij}>0$ holds since the ODE arises from
 a family of Abelian varieties.
 We conclude
 
 \begin{fact}\label{AZth} In the classical case $r=2$ all \emph{non-isotrivial, irreducible} families of Abelian surfaces
 $\mathscr{A}\to\mathbb{P}^1\setminus\{z_1,\dots,z_s\}$,
 \emph{having a monodromy representation satisfying the appropriate conditions,}
  admit an essentially unique
 SW differential $\lambda$ making them into $\C^\times$-isoinvariant rank-2 special geometries. 
 \end{fact}
 
In the ``appropriate conditions'' we include the properties which insures that the special geometry is regular at the special points $z_i$, i.e.\! along the discriminant components and the enhanced divisors. To insure these regularity conditions requires some easy extra checks (see below).
\medskip  
 
 The message is that all that matters is the monodromy representation $\mu$.
 If the monodromy representation has the right properties, the integrability conditions will follow.
 This is the same statement as in the Calabi-Yau case, but of course, the monodromy representations
and other important aspects are rather different in the two cases.
From an abstract viewpoint we may state that special geometry is a topic in Representation Theory (RT), and the main source of difficulty is that 
the representation is defined over $\Z$ (or another ring of integers) which makes
the analysis harder and subtler than plain RT over (say) $\C$.

 The corollary is that to classify all irreducible rank-2 geometries (which are not isotrivial)
  it suffices to classify all non-isotrivial families of Abelian surfaces with the appropriate
  properties. 
  We expect the higher rank case to work morally in the same way.   
  
  \subsubsection{Recovering the special coordinates ($r=2$)}\label{s:reovvering}
  In terms of the solution $\Pi(z)$ to the Picard-Fuch equation the
  SW periods on $\mathscr{C}\simeq\C^2$ read
  \be\label{iuuyyyqw}
  \Pi_\text{SW}(u_1,u_2) =\begin{cases}h(u_1,u_2)\, \Pi(u_1^{d_2}/u_2^{d_1}) & \varkappa=1\\
  \sqrt{h(u_1,u_2)^2}\; \Pi(u_1^{d_2}/u_2^{d_1}) & \varkappa=2
  \end{cases}
  \ee
  where $h(u_1,u_2)$ (resp.\! $h(u_1,u_2)^2$) is an homogeneous
  element of the field $\C(\mathscr{C})$ of dimension 1 (resp.\! 2). Note that in the
  $\varkappa=2$ case the periods are defined up to sign.
  
In order for a solution $\Pi$ of the Picard-Fuchs equation to represent a regular special geometry
we need in particular that the SW periods $\Pi_\text{SW}$
along the special divisors $\cs_i$ (whose leading behavior we can read from the local solutions
defined by the conjugacy class $C_i$) agree asymptotically with the
local models of the singular fibers of type $C_i$ which we reviewed in \S.\,\ref{s:stable}.
Later we shall use eq.\eqref{iuuyyyqw} to check the regularity of the candidate geometries along the special loci.

\subsection{``Stratification'' vs.\! Picard-Fuchs systems}\label{s:PFstrat}

Let us summarize the logic of our work. The starting point is that specifying a
special geometry is the same as giving a Gauss-Manin connection $\nabla$ on the projective
Coulomb branch $\mathscr{P}$ with the appropriate properties (monodromy in $Sp(2r,\Z)$,
a ``good'' solution with the Legendre property, local regularity, and so on). The idea is to construct the global
connection $\nabla$ from the local data around the special divisors $\cs_a=(p_a)$.
The obvious local datum is the conjugacy class $C_a$ of the local monodromy along
$\cs_a$. The class of the local monodromy is
just the class of the exponential of the residue matrix $A_a(u)|_{\cs_a}$ of the connection $\nabla=d+\sum_a A_a(u)\, d\log p_a(u)$: 
\be
C_a=\big[\exp(-2\pi i\, A_a(u)|_{\cs_a})\big].
\ee
It is a basic fact that the integrability condition $\nabla^2=0$ implies that the class $C_a$
is independent of the point along the special divisor $\cs_a$ \cite{haraoka}.
The monodromy representation (which fully determines the geometry)
is generated by the local monodromies $\mu_a$ around $\cs_a$ whose conjugacy classes
is a datum of the inverse problem.

``Stratification'' yields another strong constraint on $\nabla$. The restriction of
the rank $2r$ Gauss-Manin connection on a special divisor (or rather to a special divisor which
originates as a irreducible component of the discriminant) should reproduce
the rank $(2r-2)$ Gauss-Manin connection of the special geometry along the strata.
Since we are working recursively on $r$, this GM connection is supposed to be know when
running the inverse problem program (but it may be reducible or trivial, corresponding to
an isotrivial geometry along the discriminant component).

The technical details of the procedure to restrict the GM PDE on a special divisor are described, for instance,
in \S.12.2 of the book \cite{haraoka}. For simplicity we consider the case when 
$\cs_a$ is a non-enhanced discriminant (we can always reduce to this case by
a change of base and the elimination of apparent singularities).
In this case the local monodromy $\mu_a$ leaves invariant a subspace (defined over $\mathbb{Q}$) $V_\mathbb{Q}\subset \mathbb{Q}^{2r}$ of dimension $2r-2$. The solutions taking value
in $V_\mathbb{Q}$ have trivial monodromy around the divisor, so can be taken to be
holomorphic. The connection restricted to $\cs_a$ and $V_\mathbb{Q}$
should reproduce one of the allowed GM connection in one less rank,
that we have already classified in the previous step of our recursive procedure. 
This fixes a $(2r-2)\times (2r-2)$ diagonal block of the residue matrix $A_a(u)$.
The complementary $2\times 2$ block should be the $(2\pi i)^{-1}$ times the log
of the matrix giving the known Kodaira monodromy around the discriminant
(see \S.\,\ref{s:sing}), up to conjugacy. The restriction relations for the GM connection following from ``stratification''
is a defining property of special geometry.

 \section{Rigidity I: the conformal manifold}\label{s:rigidity}

 To get necessary/sufficient conditions for the solvability of the inverse problem,
 and to construct explicit solutions in lower rank, 
 we use \emph{rigidity theorems.} Three different
 notions of  rigidity are relevant for our story:
 \begin{itemize}
 \item[(a)] rigidity as a $\C^\times$-isotrivial special geometry;
 \item[(b)] rigidity of the underlying family of Abelian varieties over the quasi-projective
 manifold $\mathring{\mathscr{P}}$ in the sense of Arakelov \cite{arak}, Faltings \cite{falt},
 Saito \cite{saito}, and Peters \cite{peters1,peters2,peters3};
 \item[(c)] rigidity of the underlying monodromy representation in the sense of Deligne \cite{Deligne?},
 Simpson \cite{simpson,simpson?,simpson??}, and Katz \cite{katz}.
 \end{itemize}
 A fourth rigidity notion, the VHS one in the sense of Griffiths \cite{VHS1,VHS2,VHS3,VHS4} and Schmidt \cite{schmidt},
 will not be discussed here since in the present setting it is subsumed in (b).

 \subsection{Rigidity as a $\C^\times$-isoinvariant special geometry}
 
 The first notion is rigidity of the $\C^\times$-isoinvariant special geometry in its full glory.

\begin{defn}
 A \emph{family} of $\C^\times$-isoinvariant special geometries over the Coulomb branch
 $\mathscr{C}$ parametrized by the complex manifold $\mathscr{B}$ is a holomorphic fibration
 \be\label{iiii8x34}
 \varpi\colon\mathscr{Z}\to \mathscr{C}\times \mathscr{B},\qquad \text{$\mathscr{B}$ connected}
\ee 
 such that for all $b\in\mathscr{B}$ the restricted fibration
\be
  \varpi|_{\varpi^{-1}(\mathscr{C}\times \{b\})}\colon \mathscr{Z}|_{\varpi^{-1}(\mathscr{C}\times \{b\})}\to \mathscr{C}
\ee
 is a $\C^\times$-isoinvariant special geometry.
 \end{defn}
 We may assume the manifold $\mathscr{B}$ to be quasi-projective with no loss.
 Then, as always when working with families \cite{moduli}, we may (or rather should) introduce a moduli functor.
 
 \begin{defn} The \emph{(fine) conformal manifold} $\mathscr{M}$ of a $\C^\times$-isoinvariant
 special geometry is the scheme
 which represents the moduli functor. Typically $\mathscr{M}$
  does not exist as a \emph{fine} moduli space and one must content himself with a \emph{coarse}
  conformal manifold. This happens (for instance) when points in the conformal manifold are fixed by
  finite subgroups of the $S$-duality group. Alternatively on considers a fine extended conformal manifold
  which parametrizes suitable pairs $(\textsf{SCFT},\textsf{duality frames}/\!\!\sim)$.\footnote{\ A typical example ($SU(2)$ with $N_f=4$) works
as follows: the coarse conformal manifold is $H/SL(2,\Z)$ while taking $\sim$ to be the equivalence relation that two frames are identified if they differ by a rotation in $\Gamma(2)\subset SL(2,\Z)$ leads to
  a fine conformal manifold for the pair, namely $H/\Gamma(2)$.}
 \end{defn}

Roughly speaking, $\mathscr{M}$ parametrizes continuous deformations of the geometry
 which keep it $\C^\times$-isoinvariant and special. The dimensions
 $\{\Delta_1,\dots,\Delta_r\}$, being rational, are constant along the conformal manifold $\mathscr{M}$, as is the monodromy representation. 
 \begin{defn}  A $\C^\times$-isoinvariant special geometry is \emph{rigid}
 if its conformal manifold $\mathscr{M}$ is zero dimensional (i.e.\! a point). 
 \end{defn}

 \begin{fact} One has $\dim \mathscr{M}=\dim\mathscr{R}_{\Delta=2}$.
 \end{fact}
 \begin{proof} Using the periods $\Pi(u,m)$ ($u\in\mathscr{C}$, $m\in\mathscr{M}$) we construct the intrinsic linear map\footnote{\ Here $m^a$ are local coordinates around the point $m\in\mathscr{M}$.}
 \be
\phi\colon T_m\mathscr{M}\to \mathscr{R}_{\Delta=2},\qquad
\phi\colon\partial_{m^a}\mapsto \frac{1}{2} \Pi(u,m)^t\mspace{1mu}\Omega\mspace{2mu} \partial_{m^a}\Pi(u,m)
 \ee
 \emph{A priori} the function in the \textsc{rhs} is only defined in $\mathring{\mathscr{C}}$, but it is easy to check\footnote{\ Use the classification of the possible singular fibers along the discriminant components and check the statement on a case by case basis.} that it extends everywhere on $\mathscr{C}$ and hence is an element of the chiral ring $\mathscr{R}$ of scaling dimension $\Delta=2$. 
 We have to show that the map $\phi$ is an isomorphism. Consider a small coordinate patch $U\subset\mathring{\mathscr{C}}$ and let
 $\cf(a,m)_U$ be the local pre-potential in $U$ as a function of the conformal parameters $m^a$.
 Locally in $U$ one has
 \be
 \frac{1}{2}\, \Pi(u,m)^t\mspace{1mu}\Omega\mspace{2mu} \partial_{m^a}\Pi(u,m)=
  \frac{1}{2}\,a^i\mspace{1mu}\partial_{m^a}\partial_{a^i}\cf_U= \partial_{m^a}\mspace{-5mu}\left(\frac{1}{2} a^i\partial_{a^i}\cf_U\right)=
  \partial_{m^a}\cf_U,
 \ee
 which shows that $\phi$ is injective since $\partial_{m^a}\cf_U=0$ in the open set $U$
 implies that the infinitesimal deformation is trivial by uniqueness of the analytic continuation. On the other hand, consider an open cover $\mathring{\mathscr{C}}=\cup_\alpha U_\alpha$ and
  locally modify
 the pre-potential $\cf_\alpha$ in  each $U_\alpha$ as
  \be
  \cf_\alpha\to \cf_\alpha+\epsilon\, h +O(\epsilon^2) \quad\text{for }h\in\mathscr{R}_{\Delta=2}.
  \ee
  It is easy to check that in the intersection $U_\alpha\cap U_\beta$
   the deformed local special geometries on $U_\alpha$ and $U_\beta$
  agree  to the first
  order in $\epsilon$. Hence the infinitesimal deformation makes sense at the global level on $\mathscr{C}$
  and defines a non-trivial infinitesimally deformed special geometry which corresponds to an element of the tangent space $T_m\mathscr{M}$.
  Hence $\phi$ is also surjective. \end{proof}
  
  Recall our standing assumption that $\Delta_i>1$ for all $i$. In this case:
  
  \begin{fact} When the Coulomb branch $\mathscr{C}$ is \emph{smooth}  the dimension of the
  conformal manifold $\mathscr{M}$ is simply the multiplicity of $2$ as a Coulomb dimension.
  In particular when $2$ is not a Coulomb dimension
the special geometry is \emph{rigid}.
 \end{fact}
 
 \begin{cave} When the Coulomb branch is \emph{non-}smooth we may have $\dim\mathscr{M}>\dim\mathscr{C}\equiv r$. 
 \end{cave}

 \subsection{Non-rigid geometries: the Folk-theorem}\label{s:Fthm}
 
 Everybody believes that the following is true:
 
 \begin{folk} A non-rigid, smooth, indecomposable $\C^\times$-isoinvariant special geometry with $\dim_\C\mathscr{M}=s$
 is the Seiberg-Witten geometry of
 a coupled $\cn=2$ SYM $+$ ``matter'' model, where the gauge group $G$ has $s$ simple factors, and the  ``matter'' is a SCFT --  which may be non-Lagrangian -- such that the YM beta-functions vanish for all  factor gauge groups.\footnote{\ This is equivalent to requiring that the central charge of the $G_i$ current algebra of the ``matter'' SCFT has the correct value $h(G_i)^\vee$ for all simple factor $G_i$ of $G$. } The matter sector decouples into
 a set of SCFTs (possibly free) whose special geometries are rigid.  
 \end{folk}
  
  The main relevance for us of the \textbf{Folk-theorem} is that it says that all non-rigid special geometries can be constructed by ``gauging'' global symmetries of the rigid ones; therefore for the classification purpose we may (and do!)
  assume the geometry under consideration is rigid. This restriction will simplify our discussion.
In this subsection we spend some words about the \textbf{Folk-theorem} for the benefit of the skeptical reader, sketching a proof of it
  in some special situations, and giving some applications to the construction of Seiberg-Witten ``curves''
  including ones that were not known before. 
  \medskip

  The \textbf{Folk-theorem} is known to be true for \emph{isotrivial} geometries with a smooth
  Coulomb branch \cite{caorsi}. In this subsection we present some general considerations which apply also to the non-isotrivial case.
  
  \begin{cave} The \textbf{Folk-theorem} is \emph{false} in the non-smooth case.
  \end{cave}

  \subsubsection{Folk-theorem vs.\! classical special geometries}

 The precise statement of the \textbf{Folk-theorem} is that (under the assumptions) the coarse conformal manifold is a quasi-projective space $\mathscr{M}=\overline{\mathscr{M}}\setminus\sum_i W_i$
 while along \emph{some} divisor at infinity $W_{i_0}$ the gauge coupling gets parametrically \emph{small} 
  almost everywhere in $\mathscr{C}$. This means that (in an appropriate $Sp(2r,\Z)$ frame)
$r_G\leq r$ eigenvalues\footnote{\ The eigenvalues of $\mathrm{Im}\,\tau_{ij}(u)$ may be identified with $4\pi/g(u)_i^2$ where $g(u)_i$ are the effective gauge couplings in vacuum $u$.} of the matrix $\mathrm{Im}\,\tau_{ij}(u)$ go to infinity
for ``most'' points $u\in\mathscr{C}$.
 Physically  we expect that  the SCFT dynamics in this low-couplings regime
 has a weakly-coupled Lagrangian description,\footnote{\ In this description the ``matter'' is treated as a ``black box'': only the gauge degrees of freedom are described by weakly coupled fields.} and moreover that the 
 semiclassical treatment of the gauge sector gets (asymptotically) exact  near $W_i$.
 Going from the dynamics of the SCFT to its special geometry,
 one is led to think that the special geometry should also approach a semiclassical geometry.
 This is morally right, but the actual story is subtler. We start our discussion of the \textbf{Folk-theorem}
 by briefly reviewing  the geometry of a SCFT with a weakly-coupled Lagrangian description (hence with a Yang-Mills subsector).
 \medskip

 The special geometry of an $\cn=2$ model containing a Yang-Mills sector
  with the  gauge group $G=G_1\times\cdots\times G_s$ ($G_a$ simple Lie groups)
  has a Coulomb branch of the form\footnote{\ The sector
  called ``matter'' may contain further Yang-Mills sectors.}
  \be\label{cccciiib}
  \mathscr{C}_{G_1}\times  \mathscr{C}_{G_2}\times\cdots   \mathscr{C}_{G_s} \times\mathscr{C}_\text{\,``matter''}
  \ee where\footnote{\ As always $\mathsf{Weyl}(G_a)$ is the Weyl group of the Lie group $G_a$.}
  \be
  \mathscr{C}_{G_a}= \C^{r_{G_a}}/\mathsf{Weyl}(G_a)= \mathsf{Spec}\,\mathscr{R}_{G_a}\simeq \C^{r_{G_a}}
  \ee
 with $r_{G_a}$ the rank of the simple Lie group 
 $G_a$ while  $\mathscr{R}_{G_a}\simeq\C[x_1,\dots,x_{r_{G_a}}]^{\text{Weyl}(G_a)}$ is the graded polynomial ring in
 $r_{G_a}$ variables $u_i$ whose degrees are the degrees of $G_a$.
 
 \begin{defn}\label{iu77cccv} A \textit{classical special geometry with gauge group $G=G_1\times\cdots\times G_s$}
 is a $\C^\times$-isoinvariant special geometry over a Coulomb branch of the form
 \eqref{cccciiib} whose general fiber is an Abelian variety containing a \emph{fixed} Abelian subvariety $A_G$ of
 rank $r_G\equiv \sum_a r_{G_a}$ whose automorphism group $\mathsf{Aut}(A_G)$
 contains $\mathsf{Weyl}(G)\equiv\prod_a \mathsf{Weyl}(G_a)$. 
 \end{defn} 
 
 
 In particular when $r=r_G$, that is, when ``matter'' is a system of free hypermultiplets,
  a classical geometry is automatically isotrivial. In this case the
 Abelian varieties $A_G$ are as described in \cite{car1,car2} (see also \cite{caorsi}). In rank-2 all classical geometries are isotrivial, since the matter has at most rank-1 and hence is automatically isotrivial. The following statement is very well known:\footnote{\ A mathematical proof goes as follows.
 $\cn=4$ SUSY implies that the scalar metric is flat \cite{Cecotti:2015wqa}, hence the special geometry is isotrivial for the reasons explained in \S.\,\ref{s:isotrivial}.}
 \begin{fact} When our $\cn=2$ SCFT is actually $\cn=4$ invariant, the non-renormalization theorems of
 \emph{extended} supersymmetry imply that the special geometry is classical.
 \end{fact}
 
When $r_G\leq r$ eigenvalues of $\mathrm{Im}\,\tau_{ij}$ go to infinity the Abelian generic fiber  of the classical geometry $\mathscr{X}$ degenerates.
$\mathscr{X}_u$ ($u$ generic) remains a commutative algebraic group (hence still smooth)
 but is no longer compact: $\mathscr{X}_u$ is a fibration over an Abelian variety of dimension $(r-r_G)$
 with fiber the algebraic torus $(\C^\times)^{r_G}$ (again the Chevalley-Barsotti theorem\footnote{\ A
 more general and deeper treatment will be presented elsewhere.}). In particular in the (much simpler)
 case where $r=r_G$, the classical
 special geometry reduces in this weak-coupling limit  to the classical geometry of the Coulomb branch of the  3d $\cn=4$ model with the same 
 gauge group and matter content. However in 3d QFT quantum effects modify the geometry of the
Coulomb branch, so that the \emph{quantum} 3d $\cn=4$ Coulomb branch is not the classical one:
see \cite{3dN4} and references therein for a rigorous mathematical definition of the 3d Coulomb branch.
It is reasonable to expect that on the divisor $W_i\subset\overline{\mathscr{M}}$ at infinity in the conformal manifold 
the special geometry reduces (in this simpler situation) to the quantum 3d one.
Since the ``proof'' of the \textbf{Folk-theorem} asserts that ``almost everywhere'' on the Coulomb branch
we get the semiclassical situation, we deduce two geometric properties of the 3d $\cn=4$ Coulomb branch: 
\begin{fact}
{\bf(1)} the 3d $\cn=4$ Coulomb branch of a gauge theory coupled to hypermultplets
is a Lagrangian fibration over the 4d Coulomb branch of the corresponding model with
fibers the dual of the maximal torus in the gauge group $G$, i.e.\! the fiber $\simeq(\C^\times)^{r_G}$.
{\bf(2)} the special geometry on the divisor $W_i\subset\overline{\mathscr{M}}$ is birational to the classical 3d $\cn=4$ Coulomb branch. 
\end{fact}   
\noindent This statement is in fact a math theorem, see eq.(3.17), \textbf{Proposition 5.19}, and
\textbf{Corollary 5.21} of \cite{3dN4}. This \textbf{Fact} prompts the

\begin{defn}\label{uuuu453} By a \emph{weakly-coupled geometry} we mean a \emph{generalized} special geometry
$\pi\colon\mathscr{W}\to\mathscr{C}$ with Lagrangian fibers and section whose general fiber $\mathscr{W}_u$ 
 is a commutative algebraic group
 which is an extension of an Abelian variety $A_u$ (possibly of dimension 0)
 \be
 1\to T^\vee\to \mathscr{W}_u\to A_u\to 0
 \ee
  by the dual $T^\vee\simeq (\C^\times)^{r_G}$ of the maximal torus
 of a (compact, semisimple) Lie group $G$ which is the product of $s$ simple factors. 
 The set of degrees of the group $G$ are contained, with the appropriate multiplicities,
 in the $r$-tuple $\{\Delta_i\}$ of Coulomb dimensions.   
\end{defn}

\begin{rem}
In \textbf{Remark \ref{kiuqwe44}} we noticed that the degeneration of the fiber as we approach the origin in $\mathscr{C}$ is always additive. On the contrary, the degeneration at infinity in $\mathscr{M}$ is always multiplicative.
\end{rem}

 \begin{rem} The theory of generalized special geometries, their number-theoretical flavor,
 and their applications will be discussed elsewhere \cite{toappear}.
 \end{rem}

 In view of these remarks, we restate the \textbf{Folk-theorem} as follows
 \begin{folk} A smooth indecomposable $\C^\times$-isoinvariant special geometry where 2 has multiplicity $s$ as a Coulomb dimension has a (coarse) conformal manifold $\mathscr{M}$ which is a quasi-projective variety i.e.\!
 can be compactified in a projective variety $\overline{\mathscr{M}}$ while $\mathscr{M}=\overline{\mathscr{M}}\setminus W$ for a simple divisor $W=\sum_i W_i$ at infinity. There is (at least one) component $W_i$ of $W$ such that as we go close to it
the special geometry asymptotically approaches a classical
 \emph{(weakly-coupled)} geometry for the appropriate Lie group $G$.
 \end{folk}
 
\begin{rem} The \textbf{Folk-theorem} implicitly contains
the idea that the conformal manifold $\mathscr{M}$ has a \emph{canonical} compactification.
In the known examples this is true. For instance, for class $\cs$ theories \cite{gaiotto,classS} (which form
a sizable portion of all $\cn=2$ SCFTs) the conformal manifold
is the moduli space $\cm_{g,n}$ of genus $g$ curves with $n$ punctures which has the canonical
Deligne-Mumford compactification \cite{DMum,harris}.
\end{rem}

\begin{rem} Typically the divisor $W$ has several irreducible components at which
we have \emph{inequivalent} weakly-coupled geometries. When this happens we say
that there is \emph{duality of the Argyres-Seiberg type} \cite{AS} between the corresponding weakly-coupled QFTs. The basic example \cite{AS} is $SU(3)$ SYM with 6 fundamentals whose conformal manifold has a second divisor at infinity where it describes $SU(2)$
coupled to ``matter'' consisting of one fundamental hyper plus one Minahan-Nemeshanski interacting SCFT of type $E_6$ \cite{MN1,MN2}.
\end{rem}
 
 \begin{warn} Proving (or disproving) the \textbf{Folk-theorem} is \emph{not} a purpose of this paper.
 We limit ourselves to sketch a proof in the special case of rank-$2$
 where the necessary mathematical analysis has already been done in the
 math literature for different purposes. As a preparation we consider a slightly more general
 situation which is of independent interest.
 \end{warn}
 
 \subsubsection{A special class of dimension $r$-tuples $\{\Delta_i\}$}
 
 We focus on the following:
 \begin{sit} A non-isotrivial $\C^\times$-isoinvariant special geometry with $r$ even, having
 a $r$-tuple of dimensions $\{\Delta_1,\dots,\Delta_r\}$
which contains $2$ with multiplicity $1$ and in addition a dimension $\Delta_{i_1}$ which is new in rank $r/2$, while the axis spanned by $u_{i_1}$ (called the ``good'' axis henceforth)
 is \emph{not} contained in the discriminant.
\end{sit}
\begin{rem}
This is not a too ``special'' situation. In particular all $r=2$ models with precisely one dimension 2 have this form
with the following \textbf{\emph{qualification}} (see \S.\,\ref{s:falting}): as we move in the conformal manifold $\mathscr{M}$
the discriminant divisor $\cd\subset\C^2$ will also move. In particular there is a point at infinity
in $\mathscr{M}$ where the ``good'' axis ends up being inside the discriminant: our arguments here apply only as long as we keep away from that particular corner of the conformal manifold.
\end{rem}

\begin{rem} The \textbf{Special situation} covers all Lagrangian $\cn=2$ SCFT with a simple gauge group $G$
 of rank $r$ which has a ``good''\footnote{\ For the groups in the list \eqref{kkkkbbbb7} the Coxeter number $h$ is always a ``good'' degree except for $E_6$ which has the single ``good'' degree $9$.
 $F_4$ has 2 ``good'' degrees $8$ and $12$, while $E_8$ has 3 ``good'' degrees $20$, $24$ and $30$.} degree $d$ such that\footnote{\ Throughout this paper $\phi\colon \mathbb{N}\to\mathbb{N}$ is
 Euler's totient (multiplicative) function.} $\phi(d)=r$. Their list is
 \be\label{kkkkbbbb7}
 SU(p)\ \text{($p$ prime),}\quad SO(2^{k}+1),\quad Sp(2^{k+1}),\quad G_2,\quad F_4,\quad E_6,\quad E_8
 \ee
 which contains all simple Lie groups of rank-2.
\end{rem}

 In the \textbf{Special situation} $\mathscr{M}$ is a non-compact complex curve which we may write as
 a projective curve $\overline{\mathscr{M}}$ minus a few points (``cusps at infinity''). In the known examples
 $\mathscr{M}$ is a (non-compact) modular curve \cite{modulllar}: one expects this to be true in general
 when $2$ has multiplicity 1. When $r=2$ one proves that $\mathscr{M}$ is a non-compact modular curve of genus zero (see below).

We fix a point $u$ on the ``good'' axis with $u\neq0$.
By restricting \eqref{iiii8x34} to $u$ we get a one-parameter family of Abelian varieties over the 
conformal manifold\footnote{\ Here we are oversimplifying since $\mathscr{M}$ is typically coarse. Implicitly we add some extra structure to the geometry to make some finite cover of $\mathscr{M}$ into a fine conformal manifold. This is implicitly done all the time in the physical literature about dualities of $\cn=2$ SCFTs. }
\be\label{uuuuuu12qa}
\mathscr{Z}_u\to \mathscr{M}
\ee
whose fibers are principally polarized Abelian varieties of dimension $r$
which
are left invariant by the (polarized) automorphism
$\xi\equiv\exp(2\pi i\,\ce/\Delta_{i_1})$ which acts on the fibers via the analytic representation
\be\label{ttttt873}
\sigma(\xi)=\mathrm{diag}(e^{2\pi i(1-\Delta_1)/\Delta_{i_1}},\dots,e^{2\pi i(1-\Delta_r)/\Delta_{i_1}}).
\ee
Since $2$ is a dimension, the family \eqref{uuuuuu12qa} is non-isotrivial while both conjugate roots
\be\label{conjugr}
\exp(2\pi i/\Delta_{i_1})\quad\text{and}\quad \exp(-2\pi i/\Delta_{i_1})
\ee
are eigenvalues of $\sigma(\xi)$ and hence both must have multiplicity at least 2 as eigenvalues of the rational representation $\chi(\xi)$. 
By assumption the algebraic number $\exp(2\pi i/\Delta_{i_1})$ has degree $r$,
and hence
\be
\det[z-\chi(\xi)]=\Phi_m(z)^2
\ee
where $\Phi_m(z)$ is a cyclotomic polynomial with $\phi(m)=r$.
Then $\det[z-\sigma(\xi)]=\Phi_m(z)$ and the $r$ phases in eq.\eqref{ttttt873} are precisely the primitive $m$-th roots of 1. 

The classification of all \textbf{Special situation} geometries is then reduced to the construction of all non-isotrivial one-parameter families of
principally polarized
Abelian $r$-folds with a (polarization preserving) $\Z_m$ automorphism $\xi$ (with $\phi(m)=r$)
having the analytic representation in eq.\eqref{ttttt873}.
We \textbf{\emph{Claim}} that a sufficient (but not necessary!) condition for the validity of the \textbf{Folk-theorem} in the \textbf{Special situation} is that all one-parameter families of Abelian varietes with the stated properties have an automorphism group containing a non-trivial Weyl group of the proper kind. Next we show that this is the case in rank 2.

 \subsubsection{Folk-theorem: proof in rank 2}\label{Folk2}
 
 We return to the \textbf{Folk-theorem} in rank 2. We focus on the case with dimensions  $\{2,\Delta\}$
 where $\Delta$ is a rank-1 dimension $\neq1,2$. Since the discriminant locus $\cd_q\subset\mathscr{C}$ will move with the point $q\in\mathscr{M}$, at generic points
 in $\mathscr{M}$ it stays away from the ``good'' axis of dimension $\Delta$
and we are exactly in the \textbf{Special situation}. 
 In particular $m=3,4,6$, so either $\Delta=m$ or $\Delta=m/(m-1)$,
 and the fibers $X_u$ along the axis of dimension $\Delta$
 have $\Z_m\subset \mathsf{Aut}(X_u)$.
 
 \begin{pro}[Cf.\! \textbf{Theorem 4.8} of \cite{fujiki}\!\!
 \footnote{\ \label{fffnote}A slightly weaker result is in \textbf{Theorem 13.4.5} of \cite{complexAbelian}. For our purposes we need the strong version of the theorem.}] \label{pro3}
 The one-dimensional non-isotrivial families of
 principally polarized Abelian surfaces $A_\tau$ with a non-trivial automorphism group $\mathfrak{A}$
are:
{\rm$$
\begin{tabular}{cccccc}\hline\hline
$\begin{smallmatrix}\textbf{Abelian}\\
\textbf{surface $A_\tau$}\end{smallmatrix}$ & $\mathfrak{A}$ & $\begin{smallmatrix}\textbf{normal}\\
\textbf{subgroup}\end{smallmatrix}$ & $\tau_{ij}$ & $\begin{smallmatrix}\textbf{analytic}\\
\textbf{representation}\end{smallmatrix}$ & $\begin{smallmatrix}\textbf{Jacobian of}\\
\textbf{$g=2$ curve}\end{smallmatrix}$\\\hline
$E_\tau\times E_\tau$ & $D_8\equiv\mathsf{Weyl}(C_2)$ & $\Z_4$ & $\left(\begin{smallmatrix}\tau & 0\\
0 &\tau\end{smallmatrix}\right)$ & $\left\langle \left(\begin{smallmatrix}0 & 1\\
-1 &0\end{smallmatrix}\right),\left(\begin{smallmatrix}0 & 1\\
1 &0\end{smallmatrix}\right)\right\rangle$ & \\
$(E_\tau\times E_\tau)/\Z_2^2$ & $D_8\equiv\mathsf{Weyl}(C_2)$ & $\Z_4$ & $\left(\begin{smallmatrix}\tau & 1/2\\
1/2 &\tau\end{smallmatrix}\right)$ & $\left\langle \left(\begin{smallmatrix}0 & 1\\
-1 &0\end{smallmatrix}\right),\left(\begin{smallmatrix}0 & 1\\
1 &0\end{smallmatrix}\right)\right\rangle$ & $y^2=x(x^4+\alpha x^2+1)$\\
$(E_\tau\times E_\tau)/\Z_3$ & $D_{12}\equiv\mathsf{Weyl}(G_2)$ & $\Z_6$ & $\left(\begin{smallmatrix}\tau & \tau/2\\
\tau/2 &\tau\end{smallmatrix}\right)$ & $\left\langle \left(\begin{smallmatrix}0 & -1\\
1 &1\end{smallmatrix}\right),\left(\begin{smallmatrix}0 & 1\\
1 &0\end{smallmatrix}\right)\right\rangle$ & $y^2=x^6+\alpha x^3+1$\\
\hline\hline
\end{tabular}
$$} 
\hskip-0.15cm where $E_\tau$ is the elliptic curve of period $\tau\in H\equiv SL(2,\R)/U(1)$, $\alpha\equiv\alpha(\tau)$ is a certain complex coordinate on the curve $\mathscr{M}=H/\Gamma$, $\Gamma$ a suitable congruence subgroup of $SL(2,\Z)$,
and $D_{2k}$ is the dihedral group with $2k$ elements. 
 \end{pro}
  Note that modulo isogeny we have $A_\tau\sim E_\tau\times E_\tau$ in all three cases (cf.\! footnote \ref{fffnote}).
  \medskip
 
 From the \textbf{Proposition} we see that the invariance of the one-parameter family
  under the cyclic automorphism 
 group $\langle\xi\rangle\simeq\Z_m$ implies its invariance under a full Weyl group
 which is $\mathsf{Weyl}(G_2)\supset\mathsf{Weyl}(A_2)$, $\mathsf{Weyl}(C_2)$,
 and $\mathsf{Weyl}(G_2)$ for respectively $m=3,4,6$. Note that the derivative
 of the gauge coupling $\tau_{ij}$ along the ``good'' axis with respect to the coordinate $\alpha$ on $\mathscr{M}$
 is\footnote{\ For the analogue statement in higher rank see \cite{car2}.}
 \be
 \frac{\partial\tau_{ij}}{\partial \alpha}\Big|_\text{``good''}\propto\left(\begin{smallmatrix}\textbf{invariant scalar product}\\
 \textbf{$C_{ij}$ on the weight lattice}\end{smallmatrix}\right)\!.
 \ee
 This is precisely the condition which we claimed above to be sufficient for
 the validity of the \textbf{Folk-theorem}. However we have still to explain why the  \textbf{\emph{Claim}}
 holds.
 
 \subparagraph{Justification of the \textbf{\emph{Claim}} for $r=2$.}
 The first observation is that -- for each choice of $m\in\{3,4,6\}$ and family $\{A_\tau\}_{\tau\in\cm}$ 
 with $\Z_m\subset\mathfrak{A}$ -- we have a one-parameter family of
 \emph{isotrivial} special geometries
 \be
 (A_\tau\times \C^2)\big/\mathsf{Weyl}(G)\to \C^2\big/\mathsf{Weyl}(G),\qquad \tau\in\cm
 \ee 
  where $G=A_2,C_2,G_2$ for respectively $m=3,4,6$. The Weyl symmetry acts
  on the factor $A_\tau$ via the embedding $\mathsf{Weyl}(G)\subset\mathfrak{A}$ and on the factor $\C^2$ via its
  degree 2 representation as a reflection group. In this way we get \emph{four} rank-2 special geometries
  which are \emph{classical} in the sense of \textbf{Definition \ref{iu77cccv}}.
  For $G=SU(3)$ or $G_2$ we have just \emph{one} classical geometry, but for $Sp(4)$
  we get \emph{two} inequivalent classical geometries with $A_\tau=(E_\tau\times E_\tau)/\Z_2^2$
  and $A_\tau=E_\tau\times E_\tau$ respectively.
  
  By the non-renormalization theorems of extended SUSY, the quantum special geometry 
  of $\cn=4$ SYM coincides with its semiclassical one. Hence for $G=SU(3)$
  or $G_2$ the only semiclassical geometry must be the $\cn=4$ one. In the case of
  $Sp(4)$, on the other hand, in addition to the $\cn=4$ geometry, we have a second ``exotic'' semiclassical geometry.
  We stress than the ``exotic'' semiclassical geometry has the peculiarity that 
  its fibers are not the Jacobians of genus 2 curves, i.e.\! there is no (smooth) Seiberg-Witten
  curve.
  
Roughly speaking the \textbf{Folk-theorem} asserts that in
  the  limit $\tau\to i\infty$ the geometry approaches a classical one;
  however the classical geometry itself is degenerating in the limit to a weakly-coupled one,
  so the actual statement is that as $\tau\to i\infty$ we get a generalized
  geometry as in \textbf{Definition \ref{uuuu453}}.
  Physically this is merely the statement that at extreme weak coupling the semiclassical approximation becomes exact. On the other hand \textbf{Proposition \ref{pro3}} states that in a sufficiently
narrow $\C^\times$-stable neighborhood $U_\tau\subset \mathscr{C}$ of the ``good'' axis the semiclassical description is
a good approximation for all $\tau$. The size of the neighborhood $U_\tau$ where the semiclassical approximation is reliable  depends on the coupling $\tau$. Heuristically the \textbf{\emph{Claim}} says that in the limit
$\tau\to i\infty$ the semiclassical domain
$U_\tau\subset \mathscr{C}$ will enlarge to cover ``almost all'' the Coulomb branch $\mathscr{C}$ so that semiclassical physics in the ``zero coupling'' limit $\tau\to i\infty$ gives the correct prediction for ``almost all'' observables  and almost all vacua $u\in\mathscr{C}$. We formalize this heuristic picture as follow. Fix a huge constant $\Lambda\ggg1$
and consider the $\C^\times$-stable
 open domain $\mathscr{W}_{\Lambda,\tau}\subset\mathscr{C}$ consisting of Coulomb vacua where
the gauge couplings are ``very small'' in the sense that all eigenvalues of the real matrix
$\mathrm{Im}\,\tau_{ij}$ are larger than the fixed huge quantity $\Lambda$.
The local physics at these vacua is essentially semiclassical, and we call (somehow abusively\footnote{\ The terminology may be justified as follows: the period map $p$ extends holomorphically \cite{borelX}
 to the Barley-Borel compactification $\overline{\mathsf{S}}_r$ of the Siegel variety $\mathsf{S}_r\equiv Sp(2r,\Z)\backslash Sp(2r,\R)/U(r)$ \cite{borel-li}. $\mathscr{W}_{\Lambda,\tau}$ gets mapped by $p$ in a small neighborhood of a point at infinity in $\overline{\mathsf{S}}_r$ which shrinks to zero as $\Lambda\to\infty$, i.e.\! for large $\tau$ the period map is ``essentially constant'' which is the trademark of a semiclassical geometry. However, this holds after replacing the period map by its Borel extension and should be taken with a pitch of salt.})
$\mathscr{W}_{\Lambda,\tau}$ the ``semiclassical domain''.

The \textbf{\emph{Claim}} states that as $\tau\to i\infty$ in $\mathscr{M}$ the \emph{size} of the semiclassical domain $\mathscr{W}_{\Lambda,\tau}\subset \mathscr{C}$
gets \emph{infinite} for all $\Lambda$ however large. 
That this is true follows
from Ahlfors' lemma applied to the period map $\widetilde{p}\colon\widetilde{\mathscr{Q}}\to Sp(2r,\R)/U(r)$ where $\mathscr{Q}$ is $\mathscr{P}$ with all special points omitted \emph{but} the one corresponding to the ``good'' axis. This makes sense since the ``good'' axis does not belong to the discriminant, so it has a local monodromy of finite order, and hence the covering period map $\widetilde{p}$ extends holomorphically over that special point.\footnote{\ See e.g.\! \textbf{Theorem 13.7.5} of \cite{periods}.} 
Since the geometry is not rigid, $\widetilde{\mathscr{Q}}$ is the universal cover of $\mathbb{P}^1$ less a number $\geq3$ of points
(see \S.\,\ref{s:falting} below),
hence $\widetilde{\mathscr{Q}}$ is the upper half-plane $H$ with its canonical
Poincar\'e distance function $d_H(\cdot,\cdot)$. We write $B_H(\ast,R)\subset H$ for the Poincar\'e
ball of radius $R$ centered in a point $\ast\in H$ which is (a lift of) the projection of the ``good axis''
on $\mathscr{P}$ and set $U(R)\equiv \varpi^{-1}(B_H(\ast,R))\subset \mathscr{C}$.
The \emph{size} of the semiclassical domain is the
largest radius $R_{\Lambda,\tau}$ such that
$U(R_{\Lambda,\tau})\subset \mathscr{W}_{\Lambda,\tau}$. 

The precise version of the \textbf{\emph{Claim}} is the statement that $R_{\Lambda,\tau}\to\infty$ as $\tau\to i\infty$ for all $\Lambda$.
For $\tau\gg \Lambda$ very large but finite $\tilde{p}(\ast)=\tau\, C_{ij}+O(1)$
and $\mathscr{W}_{\Lambda,\tau}$ contains the preimage $U_{\Lambda,\tau}$ of the open ball $B(\widetilde{p}(\ast),r_{\Lambda,\tau})_\text{Siegel}$ in Siegel space $Sp(2r,\R)/U(r)$
of radius 
\be
r_{\Lambda,\tau}\equiv\mathrm{inf}_{M_{ij}}d_S(M_{ij}, \tilde{p}(\ast)),
\ee
 where $d_S(\cdot,\cdot)$ is
the canonical distance in the Siegel space and the infimum is taken over the set of complex
symmetric matrices $M_{ij}$ such that at least one of the eigenvalues of $\mathrm{Im}\,M_{ij}$
is $\leq \Lambda$:
\be
\mathscr{W}_{\Lambda,\tau}\supset U_{\Lambda,\tau}\equiv \{x\in H\colon d_S(\widetilde{p}(x),\widetilde{p}(\ast)) < r_{\Lambda,\tau}\}
\ee
and $R_{\Lambda,\tau}$ is not smaller than the largest $R$ such that
$U(R)\subset U_{\Lambda,\tau}$.
By Ahlfors' lemma\footnote{\ See e.g.\! \textbf{Corollary 13.7.2} in \cite{periods}.} $U_{\Lambda,\tau}$ contains $U(r_{\Lambda,\tau})$, that is,
\be
R_{\Lambda,\tau} \geq r_{\Lambda,\tau}
\ee
 As $\tau\to i\infty$, $r_{\Lambda,\tau}\to\infty$ giving the \textbf{\emph{Claim}}.

\subsubsection{Argyres-Seiberg dualities} The above discussion refers only to limits in the conformal manifold which are
semiclassical of the special situation kind, i.e.\! which correspond to a weakly coupled gauge theory of rank 2. The boundary of the conformal manifold $\mathscr{M}$ typically contains other components where a different kind
of semiclassical geometry appears whose fixed Abelian variety $A_\tau$ (in the sense of \textbf{Definition \ref{iu77cccv}}) has rank $<2$. This happens on a component of the boundary of $\mathscr{M}$ where
the axis parametrized by the coordinate $v\equiv u_2$ with $\Delta_2=m$ lays in the discriminant.
In this case $\exp(2\pi i\,\ce/m)$ acts by automorphisms of the singular fiber, hence of its Albanese variety
which then has complex multiplication by $\mathbb{Q}(e^{2\pi i/m})$ and hence is rigid.
The one-dimensional stratum with fiber this Albanese curve is the special geometry of a
rank-1 SCFT with $\Delta=m$. Following \cite{AS} we can describe 
this Albanese stratum using the genus-2 curve  whose Jacobian yields the fiber over the $v$-axis.
For, say, $m=3$ this curve is 
\be
y^2=x^6+\alpha\, v\, x^3+v^2.
\ee
 We rescale $\alpha^{1/3}x= \tilde x/v$ $y=\tilde y/v$ 
getting 
\be
\tilde y^2= \frac{\tilde x^6}{\alpha^2 v^4}+\tilde x^3+ v^4
\ee
 which as $\alpha\to\infty$ reduces to the elliptic curve for the $E_6$ Minahan-Nemeshanski model.
 On the contrary, the divisor at infinity where the weakly coupled $SU(3)$ SYM appears is $\alpha^2\to 4$; in this limit the hyperelliptic curve degenerates to \cite{AS}
  \be
 y^2=(x^3+v)^2.
 \ee

\subsubsection{Examples and checks}

The above analysis gives predictions for the geometry of the rank-$2$ models
with one dimension equal 2 which we wish to check against the known models.
We stress that the same geometry may be shared by several inequivalent SCFTs.
We write $u$, $v$ for the Coulomb coordinate of dimension $2$ and $\Delta$, respectively.
The first consequence of the validity of the \textbf{Folk-theorem} in rank-2 is that
$\Delta$ is an integer $\Delta\equiv m=3,4,6$. We consider each case in turn.

\subparagraph{Dimensions $\{2,3\}$: generalities.} From the table
in \textbf{Proposition \ref{pro3}} we see that \emph{all}
geometries with these dimensions have a genus-2
SW curve which when restricted to the locus $u=0$ is isomorphic to the one written
in the last column of that table. Taking into account the scaling properties,
we get a SW curve of the form
\be
y^2= x^6+\alpha(\tau)\,x^3\,v+v^2+ z\, P(x,v;z)
\ee 
where $P(x,v;z)$ is a polynomial in the variables $x,v$, of degree at most 4 in $x$,
with coefficients in the ring of power series $\C[[z]]$ in the local
uniformizing parameter $z$ of the normalization of the projective Coulomb branch $\mathscr{P}$
at the special point associated with the ``good'' axis $u=0$,
that is, $z= u^3/v^2$. It satisfies the homogeneity condition
\be
P(t\, x,t^3 v;z)=t^4\,P(x,v;z),\quad t\in\C.
\ee
The SW differential $\lambda$ is uniquely determined by the analytic representation of $\xi$ to be
\be\label{pppppoqw12}
\lambda=v\,\frac{dx}{y}+ u\, \frac{x\,dx}{y}.
\ee

The simplest possibility is that $z\,P(x,v;z)$ is a polynomial of the form $u\,Q(x,v,u)$;
by its scaling properties should have the form
$u\,Q(x,v,u)=c_1 x^4 u+c_2 x^2 u^2+ c_3 x u v+c_4 u^3$. The coefficients $c_i$
are further restricted by the condition that the discriminant has precisely 3 irreducible components.
These predictions from the automorphism $\xi$ are consistent with the known SW curve
which yields the \emph{non-isotrivial} geometry common to $SU(3)$ with 6
fundamentals ($\#12$ of \cite{Martone}) and $SU(3)$ with one fundamental and one symmetric
(SCFT $\#38$ in \cite{Martone})  given by the equation \cite{M14}
\be\label{jjj777xxz}
y^2=x^6+\alpha(\tau)\, x^3\,(v-u x)+(v-u x)^2
\ee  
whose discriminant is proportional to
\be
v^6\,\big(27\, v^2-2(\alpha-\sqrt{\alpha^2-4})u^3\big)\big(27\, v^2-2(\alpha+\sqrt{\alpha^2-4})u^3\big),
\ee
corresponding to two $I_1$ component and one $I_6$ component along the enhanced divisor $v=0$.

\subparagraph{Dimensions $\{2,3\}$: the isotrivial geometry.}
In addition to \eqref{jjj777xxz} there must be (at least) another SW
 curve corresponding to $SU(3)$ with one adjoint.
In this case we know from \S.\,\ref{s:isotrivial} that the geometry is isotrivial (all smooth fibers are
the Jacobian of the genus 2 curve in the table) and that the discriminant is irreducible
(hence knotted). We also know that the monodromy around the discriminant is
an automorphism of the fixed Abelian fiber of order $2$ whose analytic representation
has eigenvalues $\{+1,-1\}$. Since the cotangent space to the fiber (the dual of its Lie algebra) 
is spanned by the two forms
\be
\frac{dx}{y},\qquad \frac{x\,dx}{y}
\ee
this requires $x$ to flip sign when we go around the unique component of the discriminant.
Modulo trivial redefinitions, there is precisely one equation which satisfies all these requirements
\be\label{yyyyyqw}
y^2=x^6+\alpha(\tau)\,x^3\,(v^2-u^3)^{1/2}+(v^2-u^3)
\ee
whose discriminant is a bona fide algebraic divisor in the Coulomb branch
\be
729 \big(\alpha(\tau)^2-4\big)^3(v^2-u^3)^5
\ee
with \emph{only one} irreducible component $v^2-u^3=0$. The ``odd-looking'' square-root in
\eqref{yyyyyqw} is actually required to compensate for the flip of sign of $x^3$ when we go around the
discriminant. As always with isotrivial geometries, the SW curve is not a good way to describe
the situation.

\subparagraph{Dimensions $\{2,6\}$: generalities.}
Again \textbf{Proposition \ref{pro3}} fixes  the fiber along the ``good'' axis
to be (up to isomorphism) the Jacobian of the same curve as in the $\{2,3\}$ case 
for \emph{all} special geometries with dimensions $\{2,3\}$ while
the SW differential is still given by eq.\eqref{pppppoqw12}. The scaling argument yields
a SW curve of the form
\be
v\, y^2= x^6+\alpha(\tau)\, x^3\, v^2 + v^4 + z\, P(x,v;z)
\ee 
where now $z= u^3/v$. There should be an isotrivial geometry corresponding to
$\cn=4$ with gauge group $G_2$ and also a non-isotrivial one describing $G_2$ SYM
coupled to 4 fundamentals (model $\#49$). The last model has 3 smooth discriminant components \cite{Martone}.

For the non-isotrivial geometry analogy with the $\{2,3\}$ case (which is strictly related)
suggests making the following replacement in the $z=0$ curve
\be\label{repppre}
v\leadsto v- u\,x
\ee
getting the SW curve 
\be
v\, y^2= x^6+\alpha(\tau)\, x^3\, (v-u\,x)^2 + (v-u\,x)^4
\ee
whose (properly normalized!!) discriminant is (up to an overall constant depending on $\tau$)
\be
v^8 \big(v-\beta_+(\tau) u^3\big)\big(v-\beta_-(\tau) u^3\big)
\ee 
where $\beta_\pm(\tau)$ are known functions of $\alpha(\tau)$.
This expression says that we have two Kodaira fibers of Euler characteristic
$e=1$, i.e.\! of type $I_1$, while the third one
has $e=8$, i.e.\! $I_8$, $I^*_2$ or $II^*$.
This is consistent with $G_2$ with 4 fundamental which is known \cite{Martone} to have two
discriminants of type $I_1$ and one of type $I_2^*$.

\subparagraph{Dimensions $\{2,6\}$: the isotrivial geometry.}
The discriminant has two \emph{smooth} components meeting \emph{non-transversally}
\be
(v-2 u^3)(v+2 u^3)=0
\ee
Consider the curve
\be
v\, y^2= x^6+\alpha(\tau)\, x^3\,(v^2-4 u^6)+(v^2-4 u^6)^2
\ee
which is isotrivial with the appropriate fixed fiber and scaling properties for dimensions $\{2,6\}$. The discriminant
of the polynomial in the \textsc{rhs} is (up to an overall constant) is
\be
(v-2 u^3)^{10}(v+2 u^3)^{10}
\ee
with the correct components. 
As for all isotrivial geometries, the SW curve is not a good description.

\subsubsection{The subtle case: dimensions $\{2,4\}$}

\subparagraph{Generalities.}
There are several known SCFTs with dimensions $\{2,4\}$ we list them in table \ref{24} together with the relevant information on them taken from refs.\cite{M14,Martone}.
As always $\#$ refers to the number of the model in Martone's tables \cite{Martone}.

\begin{table}
$$
\begin{tabular}{c@{\hskip25pt}c@{\hskip25pt}cccc}\hline\hline
$\#$ & SCFT & types & isotrivial? & Jacobian? & curve in [cite]?\\\hline
$10$ & $Sp(4)$ w/ $6\cdot\mathbf{4}$ & $2\, I_1,I_2^*$ & no & yes & yes\\
$11$ & $Sp(4)$ w/ $4\cdot\mathbf{4}\oplus \mathbf{5}$ & $2\;I_0^*$ & yes & no & non Jacobian\\
$43$ & $Sp(4)$ w/ $3\cdot\mathbf{5}$ & $2\, I_1,I_2^*$ & no & yes & yes\\
$48$ & $Sp(4)$ w/ $2\cdot\mathbf{4}\oplus 2\cdot\mathbf{5}$ & $2\, I_2,I_1^*$ & no & ? & no\\
$56$ & $\cn=4$ $Sp(4)$ & $2\; I_0^*$ & yes & yes & no\\
$69$ & $Sp(4)$ w/ $\tfrac{1}{2}\cdot\mathbf{16}$ & $6\, I_1$ & no & ? & no\\\hline\hline
\end{tabular}
$$
\caption{\label{24}Known SCFTs with dimensions $\{2,4\}$ and their data from \cite{M14,Martone}. 
``Types'' refers to the Kodaira type of each irreducible component of the discriminant which
are all smooth curves which intersect non-transversely. When the number of components
is less than 3 the geometry is necessarily isotrivial (see \S.\,\ref{s:structurethm}).
Models $\#10$ and $\#43$ share the same SW curve. For model $\#11$
ref.\!\cite{M14} determines that the geometry is not a Jacobian family arising from a family of smooth
genus-2 curves.}
\end{table}

\medskip

From \textbf{Proposition \ref{pro3}} we know that there are \emph{two distinct classes} of $\{2,4\}$
geometries:
\begin{itemize}
\item[(A)] geometries whose fiber\footnote{\ At a general point in the conformal manifold $\mathscr{M}$.}
over the ``good'' axis $u=0$ is isomorphic to $(E\times E)/\Z_2^2$ i.e.\!  to the Jacobian of the hyperelliptic curve $y^2=x(x^4+\alpha x^2+1)$
\item[(B)] geometries whose fiber over the ``good'' axis $u=0$ is isomorphic to the non-Jacobian Abelian surface $E\times E$.
\end{itemize}

Class (A) geometries are expected to have a SW curve which is ``nice'' when the geometry is non-isotrivial. Class (B) geometries are not expected to have a SW curve neither nice nor ugly. 
The authors of ref.\!\cite{M14} conclude that the SCFT $\#11$ has a non-Jacobian geometry,
i.e.\! model $\#11$ is a first example of a $\{2,4\}$ SCFT of
class (B).

\subparagraph{The two isotrivial geometries.} We have \emph{two} isotrivial $\{2,4\}$ geometries one of class (A)
and one of class (B); they have the form
\be
\pi\colon (A\times \C^2)/\mathsf{Weyl}(C_2)\to \C^2/\mathsf{Weyl}(C_2)\simeq \C^2\equiv\mathscr{C}
\ee
where the fixed Abelian variety $A$ is, respectively, the Jacobian and non-Jacobian
 principally polarized
Abelian surface with automorphism group $\mathsf{Weyl}(C_2)$, cf.\! the table
attached to \textbf{Proposition \ref{pro3}}. This math result fits with the observation in \cite{M14} that there are
two inequivalent $\{2,4\}$ isotrivial geometries. Thus the SCFT $\#56$ has the class (A)
isotrivial geometry and the SCFT $\#11$ the class (B) isotrivial geometry. 

\subparagraph{Class (A) non-isotrivial.} These geometries are the Jacobian of a fibration in genus-$2$ curves which on the locus $u=0$ are isomorphic to the curve in the last column of 
the table attached to \textbf{Proposition \ref{pro3}} and the SW differential
in eq.\eqref{pppppoqw12}. Taking into account the proper scaling we have
\be
y^2=x(x^4+\alpha(\tau)\,x^2\,v+v^2)\quad \text{along }u=0
\ee 
which is promoted  by the replacement \eqref{repppre} to the SW curve
\be
y^2=x\big(x^4+\alpha(\tau)\,x^2\,(v-x\,u)+(v-x\,u)^2\big)
\ee
which is the curve given in \cite{M14} for models $\#10$ and $\#43$.

The geometry of the model $\#48$ looks ``of the same kind'' as the ones for the 
models $\#10$, $\#43$. Indeed the discriminant types suggest that these two geometries
share the same $\C$-VHS (of the Appell class). If this is the case, the two geometries
are related by a Number-Theoretical phenomenon that we shall discuss in \S.\,\ref{s:noother}.

The geometry of 
 $Sp(4)$ with $\tfrac{1}{2}\,\mathbf{16}$ is very subtle. It will be discussed in \S.\,\ref{reddd}.

\subsubsection{A few words about $r>2$}

To use the above methods in arbitrary rank $r$ we need the explicit
classification of polarized Abelian varieties of dimension $r$
 corresponding to \textbf{Proposition \ref{pro3}} for $r=2$.
 Without working out that classification we may only make some general consideration.

 In the \textbf{Special situation} we expect -- at least when the class number of the number field
 $\mathbb{Q}(e^{2\pi i/\Delta_{i_1}})$ is 1 -- that the fiber along the ``good'' axes
 has the form 
\be
A_\tau=\Big(\overbrace{E_\tau\times E_\tau\times \cdots\times E_\tau}^{r\ \text{copies}}\Big)\Big/\left(\textsf{finite group}\right)
\ee
where $E_\tau$ is an elliptic curve of arbitrary period $\tau$,
and
\be
\mathsf{Aut}(A_\tau)\supseteq \mathsf{Weyl}(G)
\ee
for a simple Lie group $G$ of rank $r$. Now everything boils down to the classification
of the allowed \textsf{(finite groups)}. Just as in the case of dimensions $\{2,4\}$ we have
two possible groups leading to the geometry classes (A) and (B), in higher rank we expect
a short list of finite group for each dimension $r$-tuple consistent with the \textbf{Special situation} condition. This will lead to a finite list of special geometry classes
for each dimension $r$-tuple.
Situations more general than the \textbf{Special} one can be studied along similar lines.

 \section{Rigidity II}

 \subsection{Rigidity of the underlying Abelian family}\label{s:falting}
 
A different notion of rigidity refers to the underlying Abelian fibration 
 stripped of all other special-geometric structures 
 (cf.\! Faltings \cite{falt}, Saito \cite{saito}, and Peters \cite{peters1,peters2,peters3}).
 In this context we ask if we can continuously deform the holomorphic Abelian fibration
 $\mathscr{A}\to\mathring{\mathscr{P}}\equiv \mathscr{P}\setminus \mathscr{S}$
 while keeping fixed its basis (i.e.\! the special divisor $\mathscr{S}\subset\mathscr{P}$).
Here the crucial fact is that the base $\mathring{\mathscr{P}}$ of the family is quasi-projective.
 
It is known that all non-isotrivial  families of Abelian $r$-varieties
 over a quasi-projective base $B$ are rigid
 if one of the following situations occur\footnote{\ I thank Chris Peters for clarifications on this issue.} (cf.\! \textbf{Example 2.4(4)} of \cite{peters3} and references therein):
 \begin{itemize}
\item $r\leq 7$;
 \item the monodromy is irreducible and $r$ is prime;
\item the base $B$ is non-compact, the underlying variation of Hodge structure
 is irreducible and some local monodromy
operator at the boundary has infinite order.
 \end{itemize}
We have seen in \S.\,\ref{s:isotrivial} that a non-isotrivial geometry has at least one local monodromy of infinite order,
while our main interest is in geometries with irreducible monodromy. Then 
 \begin{fact} The Abelian family $\mathscr{A}\to \mathring{\mathscr{P}}$
arising from a non-isotrivial, irreducible $\C^\times$-isoinvariant special geometry is rigid.
 \end{fact}
 Note that this holds independently of the assumption that $\mathscr{C}$ is smooth.
 In practice we shall work only in lower rank, so rigidity in the non-isotrivial case is guaranteed
 by the condition that the rank is less than 8, \emph{even if the special geometry is reducible.} 
 
\begin{corl} The deformations of a non-isotrivial, irreducible, $\C^\times$-isoinvariant special geometry
arise from isomonodromic deformations of the discriminant divisor $\cd$. The same holds in general for non-isotrivial geometries of rank less than 8.
\end{corl}

In particular when $\mathring{\mathscr{P}}$ is rigid, the non-isotrivial special geometry is also rigid.
Conversely:
\begin{corl}\label{kjjjasqw} If $2$ is a Coulomb dimension and $\mathring{\mathscr{P}}$ is rigid then
the special geometry must be isotrivial.
\end{corl}

 \subsection{The case of rank 2}
 In rank-$2$ all irreducible components of the special divisor $\mathscr{S}$ are points in $\mathscr{P}\simeq\mathbb{P}^1$. Let $s$ be the degree of $\mathscr{S}$.
The moduli of allowed divisors is then a subspace of $\cm_{0,s}$, the moduli space of 
curves of genus zero with $s$ punctures. Then 
 \begin{corl}\label{uuuuyyyy}
For a rank 2 non-isotrivial special geometry
\be\label{yyy5123}
\dim \mathscr{M}\leq s-3.
\ee
In particular the special geometry of a rank-2 SCFT with one (resp.\! two) Coulomb dimensions equal 2
 either has an isotrivial geometry or at least 4 (resp.\! 5) special points in $\mathbb{P}^1$.
An irreducible rank-2 SCFT has at least 3 special points. 
 \end{corl}
 
Next we argue that the inequality is generically saturated in $r=2$. Indeed, suppose
we have a special geometry with $s$ special points. Its $\mu$-monodromy be generated by the $s$ local monodromies $\mu_i\in Sp(4,\Z)$ with $\mu_1\cdots\mu_s=1$. Now move the special points a little bit to new positions nearby. By the Riemann-Hilbert correspondence there is a Picard-Fuchs connection with
the monodromies $\mu_i\in Sp(4,\Z)$ around these points. By the Almkvist-Zudilin theorem this defines
a special geometry (the positivity of $\mathrm{Im}\,\tau_{ij}$ will be not destroyed by small deformation, nor it will spoil regularity).
The deformed geometry is isomorphic to the original one if we can undo the deformation by a
automorphism of $\mathbb{P}^1\simeq\mathscr{P}$. Thus -- if we are allowed to move the points freely -- the number of proper deformations
is at least $s-3$. Since it is at most $s-3$, we get that it is precisely $s-3$. 
More generally, the number of deformations is at least the number of points which we may move freely
less 3. Usually we have no constraint on the positions, and the inequality is saturated.
The only mechanism that may freeze some degrees of freedom of the special discriminant
$\mathscr{S}=\sum_i z_i$ is the condition that the family should preserve some extra symmetry.
This can only happen under exceptional circumstances. In this case fixing the positions
of a suitable subset of $\dim\mathscr{M}+3$ points of $\mathscr{S}$ determines the other
$s-3-\dim\mathscr{M}$.

\begin{corl} The conformal manifold of geometries with dimensions $\{2,\Delta\}$, $\Delta\neq2$
is (modulo finite covers) the moduli space of the sphere less 4 points, i.e.\! the modular curve $H/\Gamma_0(2)$, while the moduli space
of geometries with dimensions $\{2,2\}$ is the moduli of the sphere less 5 points. 
\end{corl}

\subsubsection{Comparison with known models: non-rigid geometries}
Let us compare the last statement with the list of known rank-2 SCFTs (table 1 of \cite{Martone}): the models with one $\Delta_i=2$
are the numbers $\#10$, $\#\underline{11}$, $\#12$, $\#38$, $\#43$, $\#49$, $\#\underline{56}$, $\#\underline{60}$, $\#\underline{68}$, $\#69$ where we have underlined the isotrivial ones. All isotrivial geometries have 3 special points,\footnote{\ Cf. \textbf{Corollary \ref{kjjjasqw}}.} and all the non-isotrivial ones are known to have 4 special points
except 
 the last one which looks to have 7 special points, namely, 6 divisors of type $I_1$ and the enhanced axis of dimension 4. The models with $\Delta_1=\Delta_2=2$ are the numbers $\#13$, $\#31$ and $
 \#32$. The last one is reducible and so has only 2 special points, the other two have 5 special points.
 See e.g. \cite{Cecotti:2022bya} where it is shown explicitly that the conformal manifold of the $\{2,2\}$ geometries is $\cm_{0,5}$ (the moduli space of the 5-punctured sphere)
 in agreement with Gaiotto's class-$\cs$ picture \cite{gaiotto}.  
 Thus, with the single exception of
model $\#69$, all other saturate the bound \eqref{yyy5123}. 
Note that in the non-rigid isotrivial special geometries (which correspond to $\cn=4$ SCFTs)
 the continuous deformations of the geometry
leave fixed the discriminant $\cd$. This is not possible in the non-isotrivial case -- that is, when at least along one component $\cd_i$ the light degrees of freedom are IR-free -- because
of Falting-Saito-Peters rigidity.

\subsubsection{A scenario for the mysterious $\#69$}

$Sp(4)$ SYM coupled to $\tfrac{1}{2}\,\textbf{16}$ has several peculiar properties which single it out from all other SCFT. It has no known construction in string theory so we have no clue on its geometry from physics. It has dimensions $\{2,4\}$,
six discriminant components of the form $u_2-z_a\, u_1^2=0$, $z_a\in\C$ ($a=1,\dots,6$) and one regular enhanced axis $u_1=0$ \cite{Martone}. The conformal manifold is one-dimensional,
so we have to freeze 3 out of the 7 special points. A scenario which comes to mind is the following one. Consider
the $\Z_2$-quotient
\be
s\colon(u_1,u_2)\to (v_1,v_2)=(u_1,u_2^2)
\ee
and suppose that on the base we have an Abelian family $\mathscr{B}\to \mathbb{P}(2,8)\setminus \mathscr{S}$ with the 4 special points
$v_1=0$, $v_2-\alpha_k^2\, v_1^4=0$ ($k=1,2,3$). We may fix $\alpha_1$, $\alpha_2$ using the residual automorphisms of $\mathbb{P}^1$, remaining with one essential parameter $\alpha_3$. The pull-back of this family via $s$ would be an Abelian family over $\mathscr{P}$ with 7 special points $u_1=0$ and $u_2\pm\alpha_k\, u_1^2=0$ ($k=1,2,3$),
which depend on the single parameter $\alpha_3$. 
One should validate (or reject\,!) this scenario by constructing the Abelian family on the quotient.
We shall leave the hard analysis to future work. A possibility to be explored is that the local monodromy $\gamma$ of $\mathscr{B}$ around
the $v_2$ axis is such that $\gamma^2$ is the automorphism of the fiber of $\mathscr{A}$
on the $u_2$ axis (described in \textbf{Proposition \ref{pro3}}) with analytic representation
$\left(\begin{smallmatrix}1 & 0\\ 0 &-1\end{smallmatrix}\right)\in\mathsf{Weyl}(Sp(4))$.   
Note that the family $\mathscr{B}$ cannot describe a special geometry on its own right since $\{2,8\}$ is not an allowed dimension pair (it existence is also ruled out by the \textbf{Folk-theorem}).

%

\subsubsection{Comparison with known models: rigid geometries}
One expects that a rank-$2$ \emph{rigid} special geometry has \emph{exactly} 3 special points.
In the table of known SCFTs in \cite{Martone} there is only one apparent
 counterexample to this statement, namely model $\#22$ with dimensions $\{4,6\}$. However in \S.\,4.2
 of that paper it is said:
 \begin{quote}\begin{footnotesize}
 Unfortunately there is another possible solution which is compatible with everything that we know, i.e.\! a single knotted stratum with $\mathfrak{T}_{u^3+v^2} = [II, \varnothing]$. [...] we pick the former but for not better reason than symmetries with the other unknotted stratum, also of $I_n$ type. It is important to clarify, that these two choices {\it are indistinguishable as far as the analysis we are performing here goes} [...] we expect that a more careful analysis of the global structure of the CB could very likely distinguish the two cases.
\end{footnotesize} \end{quote}
Our guess is that the second choice is the correct one, thus eliminating the counterexample to equality.
However another geometric scenario does exist, a covering scenario similar to the one suggested for
$\#69$. Consider the $\Z_2$ cover
\be
f\colon(u_1,u_2)\to (v_1,v_2)= (u_1^2,u_2)
\ee
 over a plane with coordinates of dimensions $\{8,6\}$. There is a known geometry with these dimensions, the one of the SCFT $\#2$ of \cite{Martone}, which has 
 a discriminant component $v_1=0$ of type $I^*_6$ and a divisor $v_1^3-v_2^4=0$ of type $I_1$.
 The pull-back of the Abelian family $\mathscr{A}\to \mathbb{P}(8,6)$ to an Abelian family
 $f^*\!\mathscr{A}\to \mathbb{P}(4,6)$ will have a discriminant component $u_1=0$
 with monodromy the square of the monodromy at $v_1=0$, i.e.\! Kodaira type $I_{12}$.
 In addition we have the two discriminant components $u_1^3\pm u_2^2=0$ both of type $I_1$. 
 This is precisely the ``stratification'' proposed in \cite{Martone} for the $\#22$ model.
 For these precise monodromies, the would-be $\{4,6\}$ geometry passes the text of local regularity of the periods around the axis of dimension 4.
 
 Thus the present analysis of the global structure of the Coulomb branch cannot decisively support one possibility over the other. It is even possible that both geometries exist: this will explain why both pass all possible checks geometrical as well as physical.

\subsubsection{Covering and primitive geometries}

The discussion following \textbf{Corollary \ref{uuuuyyyy}} shows that whenever
$s-3>\dim\mathscr{M}$ there must be relations between the positions of the special points
and we suggested above that this \emph{very rare} phenomenon may arise when there
is a cyclic cover $f\colon\mathscr{P}\to\mathbb{Q}$ -- branched at most over the points
$(1,0)$ and $(0,1)$ -- and the special divisor is invariant under the deck group. On $\mathring{\mathscr{Q}}\setminus f(\mathscr{S})$ we have
an Abelian family $\mathscr{B}\to\mathring{\mathscr{Q}}$ whose pull-back to $\mathring{\mathscr{P}}$ is the Abelian family defining the original special geometry. 
 
 \begin{defn} A special geometry is \emph{primitive} iff its underlying Abelian family
 $\mathscr{A}\to\mathring{\mathscr{P}}$ cannot be written as the pull-back of a family on a non-trivial quotient of $\mathring{\mathscr{P}}$.
 \end{defn} 
 
 \begin{fact} An indecomposable non-isotrivial rank-2 geometry with $s-3=\dim\mathscr{M}$ is primitive.
 \end{fact}
 Indeed the quotient geometry has $\dim\mathscr{M}+3-t\geq1$ special points $\neq(1:0),(0:1)$
 where $1\leq t\leq 2$ is the number of non-trivial monodromies around the points $(1:0)$,$(0:1)$ in the quotient.
 The number of special points of the covering geometry is $s=d(\dim\mathscr{M}+3-t)+\tilde t$, where $d$ is the degree of the cyclic cover and $\tilde t$ the the number of non-trivial monodromies around the points $(1:0)$,$(0:1)$ in the quotient. For a non-isotrivial geometry $\tilde t\geq 2-\dim \mathscr{M}$.
When $t=\tilde t$ we have
 \be
 0=s-3-\dim\mathscr{M}=(d-1)(\dim\mathscr{M}+3-t)
 \ee
 which can be satisfied only if $d=1$, i.e.\! if the geometry is primitive. Otherwise we have $\dim\mathscr{M}=1,2$ and
 \begin{equation}
 \begin{aligned}
&\dim\mathscr{M}=1 &\quad &4=s=d(4-t)+\tilde t\geq 2d +1\\
&\dim\mathscr{M}=2 &\quad &5=s=d(5-t)+\tilde t\geq 3d 
\end{aligned}
 \end{equation}
 both of which can be satisfied only for $d=1$.
 Thus for a rank-2 indecomposable non-isotrivial geometry $s-3=\dim\mathscr{M}$ if and only if
 the geometry is primitive. Most geometries are primitive.
 We shall focus on primitive geometries. The non-primitive ones (if they exist)
 can be studied by the same methods going to the quotient.

 \subsection{Finiteness theorems}
 These  rigidity properties, together with a bound on the degree,
give Falting-Peters finiteness theorems (cf. \textbf{Theorem 4.2} of \cite{peters1}):
 
 \begin{thm}[Falting, Peters] For fixed $\mathscr{S}\subset\mathscr{P}$ 
 there are only finitely many non-isotrivial Abelian families
 which can arise from an irreducible non-isotrivial special geometry with special divisor $\mathscr{S}$.
 \end{thm} 
 
 \begin{rem} Since all simple Abelian surfaces are Jacobians of genus 2 curves,
 in rank 2 the above essentially reduces to Arakelov rigidity \cite{arak}.
 \end{rem}

 This finiteness result parallel another deep result from the VHS side: 
 \begin{thm}[Deligne finiteness theorem \cite{delignefin}] Fix a connected quasi-projective manifold $\mathring{\mathscr{P}}$ and an integer $N$. There are at most finitely many conjugacy classes of rational representations of $\pi_1(\mathring{\mathscr{P}})$ of dimension $N$ giving local systems that occur as a direct factor of a polarized variation of Hodge structure on $\mathring{\mathscr{P}}$ (of any weight).
\end{thm}
 
To get the expected result that only finitely many special geometries exist for each rank $r$
it remains to show that only finitely many divisors $\mathscr{S}$ are possible
for \emph{rigid} irreducible geometries. 
In rank-$2$ we have the stronger statement that there is a unique $\mathscr{S}$ modulo automorphisms of $\mathscr{P}\simeq\mathbb{P}^1$ namely $\mathscr{S}$ should consists of three distinct points. In the general case we already said some words on this issue in \S.\,\ref{s:finD} and
we shall return to the question in sect.\,\ref{kkki9223b}.

 \subsection{Rigidity in the Deligne-Simpson sense}
 
 For the third notion of rigidity we strip the special geometry of all its
 structures but for the $\mu$-monodromy representation\footnote{\ For $\varkappa=2$ we need to choose a suitable lift of the monodromy from $PSp(2r,\Z)$ to $Sp(2r,\Z)$.}
 \be
 \mu\colon\pi_1(\mathscr{P}\setminus\mathscr{S})\to Sp(2r,\Z)\subset GL(2r,\C),
 \ee
 and we forget that the monodromy was defined over
 $\mathbb{Q}$ (and in fact integral), seeing it as a mere
  representation in $GL(2r,\C)$. In particular we consider two
  representations to be the same if they are conjugated in $GL(2r,\C)$.
  
  Then $\mu$ is a complex representation of dimension $2r$,
  which in the interesting situations 
  is irreducible, with the property that
  the local monodromy 
  $\mu_i$ around the $i$-th component of the divisor $\mathscr{S}$ 
  belongs to a prescribed quasi-unipotent
  conjugacy class $C_i\subset GL(2r,\C)$.

 \begin{defn} An irreducible representation $\mu\colon 
 \pi_1(\mathring{\mathscr{P}})\to GL(2r,\C)$
 with $\mu_i\in C_i$ is called \emph{rigid}
 if all nearby irreducible representations with local monodromies $\mu_i^\prime\in C_i$
 are conjugated in $GL(2r,\C)$ to it. It is called \emph{absolutely rigid}
 if all such irreducible representations are conjugate. 
 \end{defn} 
 
 \noindent When $\dim\mathring{\mathscr{P}}=1$ rigidity implies absolute rigidity.
 The Simpson celebrated conjecture \cite{simpson} states\footnote{\ A somehow related statement is the
Bombieri-Dwork conjectures about $G$-functions \cite{andre1,andre2}. The Laplace transform of $G$-functions are the BPS brane amplitudes in the asymmetric limit \cite{Cecotti:2014wea}.}
 
 \begin{conj}[Simpson]
 Every rigid irreducible representation of the fundamental group of a smooth complex algebraic variety
 $\mathring{\mathscr{P}}=\mathscr{P}\setminus\mathscr{S}$, with quasi-unipotent monodromies
 around the components of $\mathscr{S}$ is \emph{motivic}
 i.e.\! arises from a VHS.
 \end{conj}
 
 The \textbf{Conjecture} is a theorem when $\dim\mathscr{P}=1$; this follows
 from the work of many distinguished mathematicians culminating in the explicit
 construction of the rigid monodromies by Katz \cite{katz} later reinterpreted in the representation theoretical language by Crawley-Boevey\footnote{\ For a textbook account see \cite{haraoka}.} \cite{CW1,CW2,CW3,CW4,CW5}. In ref.\cite{Cecotti:2022uep} it was
 shown how these results can be rephrased in terms of Seiberg duality for SUSY gauge theories.
 
 The \textbf{Conjecture} is proven in general in the category of $\C$-VHS,
 that is, it is proven that a rigid irreducible representation (with the appropriate $C_i$'s)
 yields a \emph{weak special geometry} (see \textbf{Lemma 4.5} of \cite{simpson}).
The aspect of the \textbf{Conjecture} which is still open
is the number-theoretic one.  
 When the monodromy is rigid, the representation $\mu$ automatically  takes value
 in $GL(2r,\mathbb{F})$ for some number field $\mathbb{F}$,
 that is, it corresponds to a direct factor of a $\mathbb{Q}$-variation of Hodge structure
 (\textbf{Theorem 5} of \cite{simpson}). The \textbf{Conjecture}
claims that this representation is also \emph{integral},
 i.e.\! $\mu$ takes values in $GL(2r,\co_{\mathbb{F}})$ where $\co_{\mathbb{F}}$ is the ring of algebraic integers in
 $\mathbb{F}$. Equivalently: the traces $\mathrm{tr}\,\gamma$ ($\gamma\in \Gamma$)
 are rational integers. 
 
 \subsubsection{Properties of rigid monodromy representations}
 Let us explain in elementary terms a special (and simpler) instance of
 the `surprising' result of \cite{simpson}, 
 that is, why  an absolutely rigid, irreducible, complex, monodromy representation $\mu$ 
 with local conjugacy classes
 $C_i\subset GL(2r,\C)$ such that:
 \begin{itemize}
 \item[(1)] for all $i$ the elements of $C_i$ are quasi-unipotent with at most one non-trivial Jordan block of size 2;
 \item[(2)] for all $i$: $C_i\cap Sp(2r,\Z)\neq\varnothing$;
 \item[(3)] the elements of at least one $C_i$ have infinite order i.e.\! some power of the elements of
 $C_i$ is unipotent with a single Jordan 2-block 
 \end{itemize}
  is automatically contained in a real form of $Sp(2r,\C)$, so yields a $\R$-VHS
  which, under some ``mild'' regularity condition (cf.\! \S.\,\ref{s:raising}) describes a \emph{weak} special geometry. 
 By assumption the $C_i$'s are invariant under $g\mapsto (g^t)^{-1}$ and
 $g\mapsto g^*$. Hence by rigidity
 \begin{equation}
\mu=M^{-1}(\mu^t)^{-1}M,\qquad
 \mu = R^{-1}\mu^* R
 \end{equation} 
for some invertible matrices $M$, $R$ with
\be\label{schhhrur}
(M^t)^{-1}M\,\mu=\mu\,(M^t)^{-1}M,\qquad RR^*\,\mu=\mu\, RR^*.
\ee 
 These equations imply
 \begin{equation}\label{ju1222z}
 \mu^t M\mu =M\quad\text{and}\quad \mu^\dagger H\mu = H\quad\text{where }\  H\equiv M^*R.
 \end{equation}
 Since the representation is irreducible, the invariant bilinear form
 $M$ must be either symmetric or antisymmetric while, being invertible, $M$ is non-degenerate.
 A non-degenerate symmetric $M$ is not consistent\footnote{\ A matrix $\mu$ with 
 $\mathrm{rank}(\mu-1)=1$ which satisfies the first eq.\eqref{ju1222z} for $M=M^t$ must have the form $\mu=1+v\otimes v^t M$
 with $v^tMv=-2$, while the condition $(\mu-1)^2=0$ requires $v^t M v=0$.} with condition (3) and should be ruled
 out. Changing basis, we may assume that $M$ is the standard (real) symplectic matrix $\Omega$.
 This shows that the rigid monodromy group $\Gamma\subset Sp(2r,\C)$. The second eq.\eqref{ju1222z} then implies that $\Gamma$
 is in a real form of $Sp(2r,\C)$.
More specifically: by Schur's lemma and eq.\eqref{schhhrur} $RR^*$ and $H^\dagger H^{-1}$
 are multiples of the identity, so $H^\dagger H^{-1}=\text{(phase)}$. We can set $RR^*=1$ as a choice of normalization,
 fixing $R$ modulo an overall phase. Now
 \be
 (\Omega R)^\dagger =H^\dagger =\text{(phase)}\,H= \text{(phase)}\,\Omega R
 \ee 
 and we may fix the phase of $R$ so that $H$ is anti-Hermitian i.e.\!
 $R^\dagger\Omega=\Omega R$ or
 \be
 R^t \mspace{1.5mu}\Omega\, R \equiv(R^\dagger)^{-1}\Omega\, R=\Omega,
 \ee
that is, $R\in Sp(2r,\C)$. Since $Sp(2r,\C)$ is simply connected and $R^*=R^{-1}$,
we can find $L\in\mathfrak{sp}(2r,\C)$ with $L^*=-L$ such that $R=\exp(L)$.
Let $S\equiv\exp(L/2)\in Sp(2r,\C)$; one has $S^2=R$ and $SS^*=1$, hence
 \be
(S\mu S^{-1})^*= S^* \mu^* (S^*)^{-1}= S^{-1}\mu^* S= S^{-1}(S^{2}\mu^* S^{-2})S= S\mu S^{-1}
 \ee
 that is, after conjugating with $S$, the monodromy $\mu$ takes
 value in $Sp(2r,\R)$.

 \subsubsection{Obstruction to be rational: the Brauer group}

 By rigidity, 
the traces $\mathrm{tr}(\gamma)$ ($\gamma\in\Gamma$)
are complex numbers invariant under all automorphisms of $\C$,
so are algebraic numbers. Since the Galois group $\mathsf{Gal}(\overline{\mathbb{Q}}/\mathbb{Q})$
fixes the traces, they take value in $\mathbb{Q}$. Then there exists a finite extension $\mathbb{F}$
of $\mathbb{Q}$ such that (after a suitable overall conjugation) $\Gamma\subset Sp(2r,\mathbb{F})$.
The number field $\mathbb{F}$ is either totally real or a purely imaginary quadratic extension of
a totally real field \cite{simpson}. As already mentioned, the Deligne-Simpson conjecture requires $\Gamma$ to be
integral:\footnote{\ Recall that the symbol $\co_\mathbb{F}$ stands for the ring of integers in the number field $\mathbb{F}$.}  $\Gamma\subset Sp(2r,\co_\mathbb{F})$. In particular the traces are ordinary integers.
However to be the monodromy of an actual family of Abelian $r$-folds $\mu$
should be integral and \emph{rational} i.e.\! $\Gamma\subset Sp(2r,\co_\mathbb{Q})\equiv
Sp(2r,\Z)$.

The Brauer group of the number field $\mathbb{F}$ gives a potential obstruction to the
rationality of the monodromy representation.
 Let $\sigma\in \mathsf{Gal}(\mathbb{F}/\mathbb{Q})$; 
by rigidity we have
\be\label{8776ggs}
\gamma_i^\sigma= M_\sigma^{-1} \gamma_i \mspace{1mu}M_\sigma,\quad \forall\; i 
\ee 
for some matrix $M_\sigma\in GL(2r,\mathbb{F})$ which is unique up to multiplication by an element of $\mathbb{F}^\times$ by irreducibility. Now
\be
\begin{split}
M_{\sigma\tau}^{-1}\gamma_iM_{\sigma\tau}&= \gamma_i^{\sigma\tau}=(\gamma_i^\tau)^\sigma=
(M_\tau^{-1}\gamma_i M_\tau)^\sigma=\\
&= (M_\tau^\sigma)^{-1} \gamma_i^\sigma (M_\tau)^\sigma=
(M_\tau^\sigma)^{-1}M_\sigma^{-1} \gamma_i M_\sigma (M_\tau)^\sigma,
\end{split}
\ee
that is,
\be
C_{\sigma\tau}\equiv M_{\sigma\tau}(M_\sigma (M_\tau)^\sigma)^{-1}\in\mathbb{F}^\times,\qquad \sigma,\tau\in \mathsf{Gal}(\mathbb{F}/\mathbb{Q}).
\ee
is a 2-cocycle on the Galois group with values in $\mathbb{F}^\times$.
When the cohomology class $[C_{\sigma\tau}]$ in the Brauer group $H^2(\mathbb{F},\mathbb{F}^\times)$
vanishes, we can chose $M_\sigma$ so that $M_{\sigma\tau}=M_\sigma (M_\tau)^\sigma$,
and $M_\sigma$ is a Galois 1-cocycle with values in $GL(2r,\mathbb{F})$.
The non-Abelian extension of Hilbert's \textbf{Theorem 90} states \cite{serre0}
\be
H^1(\mathsf{Gal}(\mathbb{F}/\mathbb{Q}), GL(2r,\mathbb{F}))=\{1\},
\ee
that is, there exists a matrix $N\in GL(2r,\mathbb{F})$
such that 
\be
M_\sigma= N^{-1}\,N^\sigma\qquad \forall \; \sigma\in \mathsf{Gal}(\mathbb{K}/\mathbb{Q}).
\ee
Substituting this equation in \eqref{8776ggs} we get
\be
(N\gamma_i N^{-1})^\sigma\equiv N^\sigma \gamma_i^\sigma (N^\sigma)^{-1}= N\gamma_i N^{-1}\quad
\forall\; i,\ \ \forall\; \sigma,
\ee
that is, the matrices $\eta_i=N\gamma_i N^{-1}$ yield a conjugate representation
defined over $\mathbb{Q}$ which is integral by the Simpson conjecture.

However the Brauer class $[C_{\sigma\tau}]$ is \emph{not} zero in general.
Therefore to conclude that a particular rigid irreducible representation satisfying (1)(2)(3)
yields a viable monodromy group for special geometry we have to check that its Brauer class is zero.
Otherwise we can only conclude that the monodromy is a direct summand  in some higher degree
VHS \cite{simpson}, and hence not directly relevant for special geometry.

\subsection{Use of rigidity to construct special geometries}

It may happen that a special geometry has a rigid monodromy representation $\mu$. 
\begin{defn}
A special geometry is \emph{$\mu$-rigid} if it is rigid (i.e.\! its conformal manifold $\mathscr{M}$ is a point), and its monodromy representation
$\mu\colon \pi_1(\mathring{\mathscr{P}})\to GL(2r,\C)$ is either
irreducible and rigid (as a monodromy representation) or it decomposes in the direct sum
of rigid representations.
\end{defn}

By the Riemann-Hilbert correspondence being $\mu$-rigid
 is equivalent to saying that the associated Gauss-Manin (GM) connection (Picard-Fuchs ODE for $r=2$)
is rigid (modulo gauge equivalence). In \S.\,\ref{s:raising} we suggested that ``under some mild conditions''  an irreducible Gauss-Manin connection over $\mathring{\mathscr{P}}$ with the
appropriate local monodromies will yield a special geometry; in $r=2$ this follows\footnote{\ Subject to the checks of $\mathrm{Im}\,\tau_{ij}>0$ and local regularity.} from the 
Almkvist-Zudilin theorem \cite{AZthm}  and we expect that a similar story holds for all $r$. 
The basic idea is that one can construct the GM connection of the $\mu$-rigid special geometries by exploiting their rigidity.
This is a method pioneered by Riemann in his 1851 lecture notes on the hypergeometric ODE:
he was able to construct integral representations for the solutions and compute the connection coefficients
using rigidity of the monodromy. Since Riemann's times a quantity of important results have been obtained using this method. 

Our interest in the class of $\mu$-rigid special geometries stems from the fact that
all rigid special geometries tend to be $\mu$-rigid.

\begin{conj} Most, and possibly \emph{all,} rigid $\C^\times$-isoinvariant special geometries
are $\mu$-rigid. 
\end{conj}

\begin{proof}[Evidence] In rank 1 the statement is trivially true. The seventy or so known rank $2$ special geometries \cite{Martone}
have rigid monodromy representation, including the non-rigid (i.e.\! Lagrangian) ones.
For all but (at most) two indecomposable rank-$2$ geometries the monodromy representation is in addition \emph{irreducible}. As emphasized on page 360 of \cite{haraoka}, rigidity  becomes a milder and milder requirement on the monodromy as we increase the rank $r$. The point is that the fundamental group of the complement $\pi_1(\mathscr{P}\setminus\mathscr{S})$ is finitely-presented for all geometries
in any $r$. But while the number of generators grows slowly with $r$ the number of independent relations grows quite rapidly, producing a typically overdetermined
set of algebraic equations for the moduli space of monodromy representations 
which then is typically empty for most choices of the $\{C_a\}$'s. Since already for $r=2$
all (or at least the vast majority) of the special geometries are $\mu$-rigid, this suggests
that this pattern would hold for all $r$.        
\end{proof}

Regardless for the validity of the \textbf{Conjecture}, the $\mu$-rigid geometries form a
big chunk of the finite set of all rigid special geometries of given $r$. Then in our classification scheme (cf.\! \S.\,\ref{s:invpro}) we introduce class V of special geometries which are $\mu$-rigid and do not belong to simpler classes
I-IV which can be addressed with simpler techniques.
Then we proceed to study, and hopefully 
solve, the inverse problem for this important class of geometries by exploiting the wonders of rigid monodromies  (with the hope that the \textbf{Conjecture} is true and class VI is small or empty).
We shall pursue this program
for the rest of the paper.
\medskip  

The crucial question is when a rigid monodromy
leads to a non-isotrivial special geometry. A necessary condition is that the local monodromies
are quasi-unipotent, real symplectic, with at most one non-trivial Jordan block of dimension 2,
and at least one local monodromy non-semisimple. The Deligne-Simpson
conjecture (theorem for $r=2$) then says that the monodromy arises from
 a VHS. In particular the monodromy group is a subgroup of $Sp(2r,\co_\mathbb{F})$
where $\co_\mathbb{F}$ is the ring of algebraic integers in some number field $\mathbb{F}$ (a finite extension of $\mathbb{Q}$). Given the rigid monodromy representation $\mu$
we can write
the corresponding rigid GM connection $\nabla$ on $\mathring{\mathscr{P}}$
 and solve the flat-section equation $\nabla\Phi=0$.
To declare that we have constructed a $\C^\times$-isoinvariant
special geometry in the sense of
\textbf{Definition \ref{xxxyy67}} 
 it remains only to check:
\begin{itemize}
\item[\it (i)] the existence of a particular solution $\Pi$ of the flat-section equation
which satisfies the Legendre condition.
This is automatically true for $r=2$ (cf.\! \ref{s:raising}) and plausibly also for general $r$ under some
(yet to be fully understood) ``mild'' conditions;
\item[\it (ii)] the Brauer class of $\mu$ is trivial so that
$\Gamma\subset Sp(2r,\Z)$;
\item[\it (iii)] $\mathrm{Im}\,\tau_{ij}>0$, i.e.\! the VHS implied by the 
Deligne-Simpson conjecture has weight 1;
\item[\it (iv)] regularity of the geometry, that is, that the SW periods pulled back on 
$\mathscr{C}$ (as in \S.\,\ref{s:reovvering} for $r=2$) have the behavior along
the irreducible components of the discriminant dictated by the local models
of the corresponding singular fiber types, cf.\! \S.\,\ref{s:sing}, and no singularity along
the enhanced divisors, where they should satisfy the conditions in \S.\,\ref{s:iiiiii4uu5}.
\end{itemize}

When the monodromy representation is an appropriate one, we expect that
all these conditions are automatically satisfied, without imposing any new one.
Since it is still not fully obvious what we shall mean by ``appropriate'' monodromy representation,  when in doubt
we shall perform the checks \emph{ex post}.

%
%

\section{The inverse problem vs. rigidity}

We now return to our inverse problem in rank $r$ which we state in more precise terms:

\begin{inverse}
Suppose we are given the following set of data:
\begin{itemize}
\item[\bf (a)] the Coulomb dimensions $\{\Delta_1,\dots,\Delta_r\}$
from which we construct the (normal projective) manifold
$\mathscr{P}\equiv \mathbb{P}(d_1,d_2,\dots,d_r)$,
and identify its cyclic quotient orbifold divisors;\footnote{\ Here we mean ``orbifolds'' in the projective sense, the underlying manifold/variety may be smooth along the divisor (this happens when the quotient is by a reflection group \cite{dolgaWP}). For the distinction between the two notions we refer to the discussion in \cite{sasaki}.}
\item[\bf (b)] a finite list $\{p_1,\cdots, p_\ell\}$ of irreducible
quasi-homogeneous polynomials in $\C[u_1,\cdots,u_r]$
($\deg u_i=d_i$) as candidate irreducible components of the discriminant. The $\ell$-tuple of polynomials is
severely restricted by the conditions briefly discussed in {\rm\S.\,\ref{s:finD}}; we shall return to them in {\rm\S.\,\ref{kkki9223b}};
\item[\bf (c)] a conjugacy class $C_a\subset Sp(2r,\Z)$ of Kodaira type
per each component $(p_a)$ of the discriminant;
\end{itemize}
and ask:
\begin{itemize} 
\item[\bf(1)]there exists an indecomposable special geometry
with these data?
\item[\bf(2)] how many are they? What are the degrees of the allowed
polarizations?
\item[\bf(3)] What can we say about the geometry (when it exists)?
\end{itemize}
\end{inverse}

We consider only data $(\{\Delta_i\},\{p_a\},C_a)$ for which the answer to \textbf{(1)} is not ``no''
for trivial reasons. So the $\{\Delta_i\}$ is an allowed $r$-tuple of dimensions,
 the discriminant satisfies the various necessary conditions, and in particular
 the group $\pi_1(\mathscr{C}\setminus \sum_a(p_a))$ is not solvable, finite, or a non-trivial product,
 and when an axis (or a coordinate plane) is not part of the discriminant
 the corresponding restrictions on the dimensions are satisfied,
 etc. Moreover we do not loose nothing in assuming that the data correspond to a 
 geometry of  a non-isotrivial rigid geometry, that is,
 \begin{itemize}\it
 \item[\bf (d)] neither 1 nor 2 belong to the list $\{\Delta_i\}$
 \item[\bf (e)] at least one local monodromy is not semisimple.
 \end{itemize} 

The first step to give a positive answer to {\bf(1)}  is to show that there exists
 a monodromy representation of $\pi_1(\mathring{\mathscr{P}})$ of dimension 
$2r$ which has the given local monodromies $\mu_a\in C_a$. Again, we may
set a part the reducible indecomposable
case (which carries additional constraints) for a separate discussion, and focus on the typical case of
\emph{irreducible} monodromy representations.
The irreducible monodromies
with the appropriate local classes $\{C_a\}$
 are parametrized by a moduli space $\boldsymbol{\cm}(\{\Delta_i\},\{p_a\},C_a)$. 
The monodromy is rigid (resp. properly rigid) iff  $\boldsymbol{\cm}(\{\Delta_i\},\{p_a\},C_a)$
is non empty of dimension zero (resp. a point).
Then a necessary condition for a positive answer to {\bf(1)} is: 

\begin{nec} The moduli space of irreducible representations of
$\pi_1(\mathring{\mathscr{P}})$ with the fixed local monodromy classes 
should be non-empty $\boldsymbol{\cm}(\{\Delta_i\},\{p_a\},\{C_a\})\neq\varnothing$.
\end{nec} 

For ``most'' data $(\{\Delta_i\},\{p_a\},\{C_a\})$ all representations with the required local monodromies turn out to be reducible, i.e.\! the moduli space $\boldsymbol{\cm}(\{\Delta_i\},\{p_a\},\{C_a\})$ is empty. Thus the necessary condition is already very stringent. Indeed, for large $r$ the set of representations with
given $\{C_a\}$ is empty even before imposing irreducibility.

Our main point is that there is also an ``almost sufficient'' condition:
\begin{almost}
$\dim\boldsymbol{\cm}(\{\Delta_i\},\{p_a\},\{C_a\})=0$.
\end{almost}
This is ``almost'' sufficient because\footnote{\ Assuming
the validity of the Deligne-Simpson conjecture (which is
a theorem in many special cases).} when it holds 
we have still to make three checks before we can declare
that we have a \emph{bona fide} special geometry:
\begin{itemize}
\item[\it(i)] the vanishing of
the Brauer class $[C_{\sigma\rho\tau}]\in H^2(\mathbb{F},\mathbb{F}^\times)$;
\item[\it(ii)]  the regularity
of the solution (which is expected to take care of itself most of the
time);
\item[\it(iii)] the positivity of $\mathrm{Im}\,\tau_{ij}$ (also expected to be typically true).
\end{itemize}
``Experimentally'' the criterion looks pretty
good as a ``sufficient'' condition in the sense that the checks
 have the tendency to go well. By definition the special geometries satisfying the ``almost sufficient'' condition are of class V.

\section{The inverse problem in rank-2}

In this section we run the inverse problem program in rank-2 in more or less full detail. We suppose the geometry to be indecomposable (thus $\mathscr{R}_1=0$).
\medskip

When the geometry is primitive and irreducible the situation is pretty clear.
Let $s=\dim_\C\mathscr{R}_2$ be the dimension of the conformal manifold $\mathscr{M}$.
The complement $\mathring{\mathscr{P}}$ is
the Riemann sphere $\mathbb{P}^1$ minus exactly $s+3$ special points $p_a$.
Then
\be
\pi_1(\mathring{\mathscr{P}})= \big\langle\, l_1,\dots, l_{s+3} \colon \ l_1\,l_2\cdots l_{s+3}=1\big\rangle.
\ee
When $\varkappa=1$ we are interested in monodromy representations of $\pi_1(\mathring{\mathscr{S}})$ taking value in the arithmetic group
$Sp(4,\Z)$, but let start by considering the underlying complex representations valued in $GL(4,\C)$.
An irreducible representation of $\pi_1(\mathring{\mathscr{P}})$ consists of $s+3$ 
complex $4\times 4$ matrices $\mu_a\equiv\mu(\ell_a)$ ($a=1,\dots,s+3)$
such that 
$\mu_1\cdots\mu_{s+3}=1$ which have no non-trivial
 common invariant subspace in $\C^4$. In addition, the conjugacy classes of the
 $s+3$ local monodromies $\mu_a$ are prescribed. The space of all such complex representations is
 \be
 \big\{\mu_1,\dots,\mu_{s+3}\in GL(4,\C)\colon \mu_1\cdots\mu_{s+3}=1,\ \mu_a\in C_a\ a=1,\dots,s+3,\ \text{irreducible}\big\}
 \ee
Two representations are considered to be the same (as {complex representations}) iff they are related by an overall
 conjugation $\m_a^\prime= U \m_a \,U^{-1}$ with $U\in GL(4,\C)$.
 
When $\varkappa=2$ the $\mu$-monodromy representation is well defined in $PSp(4,\Z)$. Suppose we lift the local monodromies $\mu_a$ to elements of $Sp(4,\Z)$.
 Now the product $\mu_1\cdots\mu_{s+3}$, which is the identity in $PSp(4,\Z)$,
 may lift to either $+1$ or $-1$ in $Sp(4,\Z)$. In the second case we may add a dummy special point
$p_{s+4}$ with local monodromy $\mu_{s+4}=-1$ and reduce back to the standard case
\be\label{dummy}
\mu_1\cdots\mu_{s+3}\,\mu_{s+4}=1.
\ee
Alternatively we may flip the sign $\mu_{s+3}\to -\mu_{s+3}$ to get back to the usual story.

The case of main interest is the rigid one, $s=0$. In this case $\mathring{\mathscr{P}}$
is the sphere less three points, which we call $0$, $1$ and $\infty$,
while $\pi_1(\mathring{\mathscr{P}})\simeq F_2$. The three local monodromies
will be written as $\mu_0$, $\mu_1$, $\mu_\infty$, satisfying $\mu_0\mu_1\mu_\infty=\pm1$
(the sign is always $+$ when $\varkappa=1$).

\subsection{Reducible geometries in rank 2}\label{reddd}

Indecomposable but reducible special geometries are rare and \emph{quite subtle.}
\medskip

Suppose our $\mu$-monodromy representation $M$
 is reducible, $M=V\oplus W$, while it contains non-trivial unipotents: then at least one summand,
 say $V$, should be irreducible of dimension 2.\footnote{\ \textit{Proof.} Suppose $L\subset M$ is an invariant subspace of dimension 1. Then $$L^\perp\overset{\rm def}{=}\{v\in M\colon v^t\Omega\, l=0\ \ \forall\, l\in L\}\supset L$$ is also invariant of dimension 3. Since the monodromy is totally reducible (Deligne theorem \cite{delII}) $L^\perp =V\oplus L$ for $V$ an invariant subspace of dimension 2. We claim that $V$ is irreducible. Indeed, otherwise, it would be the direct sum of two invariant spaces of dimension 1, and the monodromy representation would be the direct sum of 4 one-dimensional ones, hence Abelian, which is a contradiction since the monodromy group contains non-trivial unipotents.   }
 Since $\Omega\colon M\xrightarrow{{}_{\boldsymbol{\sim}}} M^\vee$ we have three possibilities:
 \begin{itemize}
 \item[1.] $V^\vee\simeq W$ and $\Omega\colon V\xrightarrow{{}_{\boldsymbol{\sim}}} W^\vee$;
 \item[2.] $V^\vee\not\simeq W$ and $\Omega\colon V\xrightarrow{{}_{\boldsymbol{\sim}}} V^\vee$, $\Omega\colon W\xrightarrow{{}_{\boldsymbol{\sim}}} W^\vee$;
 \item[3.] $V^\vee\simeq V\simeq W$ and $\Omega\colon V\xrightarrow{{}_{\boldsymbol{\sim}}} V^\vee$.
 \end{itemize}
 In the first case the monodromy is contained in $GL(2,\mathbb{Q})\subset Sp(4,\mathbb{Q})$,
 in the second one is a subgroup of $SL(2,\mathbb{Q})\times SL(2,\mathbb{Q})$, and in the third situation
 in the diagonal $SL(2,\mathbb{Q})$.
 Case 3. is a special instance of both 1. and 2.
 
 \subsubsection{Case 1}
 In case 1.\! (so also in 3.) we have a symmetric Hodge 2-tensor,
 hence an element of the chiral ring of dimension 2, i.e.\! the special geometry
 \emph{cannot be rigid.} According to the \textbf{Folk-theorem}, the geometry should describe a Lagrangian SCFT whose list is known.
We shall see later in this section that all geometries describing rank-2 SCFTs with $s+3=\dim\mathscr{M}$
 are irreducible. It remains the \emph{mysterious}
 model $Sp(4)$ with $\tfrac{1}{2}\mathbf{16}$ ($\#69$ of \cite{Martone})
whose geometry requires further analysis.   
 
%
%
%
 
 \subsubsection{Case 2}
 In case 2. the monodromy is contained in $SL(2,\mathbb{Q})\times SL(2,\mathbb{Q})$
 so, going to a finite cover if necessary, we have a fibration with fiber the product of two elliptic
 curves $E_1\times E_2$, that is, the fibered product of two elliptic fibrations
over $\mathbb{P}^1$ less some points which we may replace by the smooth
surfaces given by their respective 
 Kodaira-N\'eron models\footnote{\ See \cite{MW}.} $\pi_a\colon\ce_a\to\mathbb{P}^1$ ($a=1,2$)
 \be\label{kkk88888x}
 \begin{gathered}
 \xymatrix{& \ce_1\times_{\mathbb{P}^1} \ce_2\ar[dl]\ar[dr]\ar[dd]^\varpi\\
 \ce_1\ar[dr]_{\pi_1}&& \ce_2\ar[dl]^{\pi_2}\\
 & \mathbb{P}_1}
 \end{gathered}
 \ee

 If $\ce_a$ are non-trivial fibrations, the monodromy group has Zariski-dense images in each factor of $SL(2,\mathbb{Q})\times SL(2,\mathbb{Q})$, so either we are back to case 3.\! or
 there is no operator of dimension 2, and the geometry is rigid, so we have
exactly 3 special points on $\mathbb{P}^1$. Since $1,2$ are not Coulomb dimensions, $d_1,d_2\neq1$,
 and hence the axes in $\mathscr{C}$ gives two special points $(1:0)$ and $(0:1)$.
We remain with only another special point, around which -- generically -- only
 the monodromy in one of the two $SL(2,\mathbb{Q})$ factors is non-trivial.
 Therefore the other elliptic surface is isotrivial. An isotrivial family is not rigid
 but in the present case it may be rigidified by the extension on a special point
 if the monodromy around that point is constrained to be $\not=\pm1$
 say because the special point is regular enhanced or with a non-trivial semisimple
 monodromy. Then either:
 \begin{itemize}
 \item[\textit{(i)}] we are not in the generic situation, i.e.\! we are in case 3 which is a particular case of 1;
 \item[{\it(ii)}] the geometry is quasi-isotrivial.
 \end{itemize}

\subsubsection{Scenario {\it(ii)}: quasi-isotrivial geometries}\label{SS:quasi-isotrivial}
In this case in the diagram \eqref{kkk88888x} the second rational elliptic surface
$\ce_2$ is isotrivial while the first one $\ce_1$ is non-isotrivial. 
 
The non-rigid geometries are known not to be quasi-isotrivial except for the still \emph{mysterious} one which deserves a separate analysis. Here we assume the quasi-isotrivial geometry to be rigid. Both dimensions $\Delta_1,\Delta_2$ should be old and not 1 or 2  and $\varkappa\in\{1,2\}$.
 Thus the order of $1/\Delta_1$ and $1/\Delta_2$ should be in $\{3,4,6\}$ and not both divisible by 4 or 3. Then one dimension should be either 4 or $4/3$, but the second possibility is ruled out by
 the values of the $d_i$. We remain with $\{\Delta_1,\Delta_2\}=\{4,3\}$ or $\{4,6\}$
 depending on $\varkappa=1$ or $2$.

 We have 3 special points in $\mathbb{P}^1$.
 Two special points are $p_0\equiv(0:1)$ and $p_\infty\equiv(1:0)$ 
 corresponding to the axes and then we have a third `knotted' special point
 with monodromy of the form $\alpha_1\oplus \boldsymbol{1}_2$.
The formulae for local regularity along the axes, yield \emph{one} exponent of the local $\mu$-monodromy at each of the points $p_0$, $p_\infty$ 
regardless of the Kodaira type of the axes as discriminant components.
For the present two allowed dimension pairs the statement is simply that $1/\Delta_i$ is an exponent for
the local monodromy around the special point corresponding to an axis of dimension $\Delta_i=3,4,6$.

 Since a non-isotrivial family should have at least 3 non-trivial local holonomies,
 $\alpha_1$ is a local monodromy for $\ce_1$, while $\ce_2$ has only two special fibers
 over $(1:0)$ and $(0:1)$ and their local $\mu$-monodromy, $\beta$ and $\beta^{-1}$,
 are the inverse of  
 each others so $\ce_2$ is an isotrivial rational elliptic surface
 with a singular fibers' configuration of the forms $\{F,F^*\}$ where $F=II,III,$ or $IV$.
The local monodromies
of $\ce_1$ are $\alpha_0$, $\alpha_1$ and $\alpha_\infty$ with
$e(\alpha_0)+e(\alpha_1)+e(\alpha_\infty)=12$ where $e(\alpha)$ stands for the Euler
characteristics of a Kodaira fiber with local monodromy $[\alpha]\in SL(2,\Z)$.
The $\varrho$-monodromy
\be
(\alpha_0\oplus\beta)^{d_1}=\begin{cases}\epsilon_0(\boldsymbol{1}\oplus K_0)\\
\epsilon_0(K_0\oplus \boldsymbol{1}),\end{cases}
\quad
(\alpha_\infty\oplus\beta^{-1})^{d_2}=\begin{cases}\epsilon_\infty(\boldsymbol{1}\oplus K_\infty)\\
\epsilon_\infty(K_\infty\oplus \boldsymbol{1})
\end{cases}
\ee
where $K_0$, $K_\infty$ are the local Kodaira monodromies along the axes in $\mathscr{C}$
and $\epsilon_0,\epsilon_\infty$ signs from the lift ($\epsilon_0=\epsilon_\infty=+1$ when $\varkappa=1$). A simple very preliminary ``regularity'' argument
leads to the following

\begin{fact}\label{uuuu63} All rigid non-isotrivial quasi-isotrivial rank-2 geometries -- {\rm IF THEY EXIST }--
should be one of the following:\footnote{\ Not to confuse the reader, we recall that we are assuming that the Coulomb branch $\mathscr{C}$ is smooth.}
\begin{itemize}
\item[\rm(1)] Coulomb dimensions $\{3,4\}$, the axis of dimension 3 is a discriminant
of type $I^*_0$ and the axis of dimension 
4 is a discriminant of type $IV^*$. In addition we have a knotted discriminant of type $I_m$ ($m>0$);
\item[\rm(2)] Coulomb dimensions $\{3,4\}$, the axis of dimension 3 is a discriminant of type $III^*$, the axis of dimension 4 is a discriminant of type $I_n$ ($n>0$), and the knotted discriminant has type $I_m$
($m>0$);
\item[\rm(3)]
Coulomb dimensions $\{4,6\}$, the axis of dimension 4 is a discriminant of type
$II^*$, the axis of dimension 6 is a discriminant of Kodaira type $I_n$ with $n>0$, and we
have a knotted discriminant of type $I_m$ ($m>0$);
\item[\rm(4)] Coulomb dimensions $\{4,6\}$, the axis of dimension 4 is a discriminant of  type $II^*$,
 the axis of dimension $6$ is a discriminant of type $I_0^*$, and we have a knotted discriminant of type $I_m$ ($m>0$).
\end{itemize}
\end{fact}
All statements are shown in appendix \ref{A:quasi} except for the claim about the knotted
discriminant that follows from considerations about the hypergeometric Picard-Fuchs ODE (see below).
The 4 geometries are very similar. For future reference we list the eigenvalues of the $\mu$-monodromy at the points $(1:0)$ and $(0:1)$ (cf.\! appendix \ref{A:quasi}):
\be
\begin{tabular}{c|cccc}\hline\hline
 & (1) & (2) & (3) & (4)\\\hline
 $(1:0)$ & $\{e^{\pm2\pi i/3}, e^{\pm2\pi i/6}\}$ & $\{e^{\pm2\pi i/3}, e^{\pm2\pi i/4}\}$ &
  $\{e^{\pm2\pi i/4}, e^{\pm2\pi i/6}\}$ & $\{e^{\pm2\pi i/4}, e^{\pm2\pi i/6}\}$\\
 $(0:1)$ &$\{e^{\pm2\pi i/4}, e^{\pm2\pi i/3}\}$ & $\{e^{\pm2\pi i/4}, e^{\pm2\pi i/2}\}$ &
  $\{e^{\pm2\pi i/6}, e^{\pm2\pi i/2}\}$ & $\{e^{\pm2\pi i/6}, e^{\pm2\pi i/3}\}$\\\hline\hline
 \end{tabular}
\ee  

\begin{rem} In this section we shall return again on the intriguing possibility of the existence of quasi-isotrivial geometries from different viewpoints. We shall get further restrictions on them. There is evidence from physics
\cite{Martone} that the geometry (3) indeed exists and represents a known SCFT.
\end{rem}

From now on we shall (mainly) focus on \emph{irreducible} (principal) rank-2 special geometries.

 \subsection{Review of the Deligne-Simpson problem}
 The Deligne-Simpson problem \cite{DS1,DS2,DS3,DS4,DS5} asks for the determination of the
irreducible degree-$n$ representations of the fundamental group of the sphere $\mathbb{P}^1$
less $\ell$ points $\{p_1,\dots,p_\ell\}$ with prescribed conjugacy classes $C_a\subset GL(n,\C)$
of the local monodromies at the marked points. The moduli space of such representations is given by
\be
\boldsymbol{m}=\Big\{\mu_1,\dots,\mu_\ell\in GL(n,\C)\colon \mu_1\cdots\mu_\ell=1,\ 
\mu_a\in C_a\ \forall a,\ \begin{smallmatrix}\text{the $\mu_a$ have no non-trivial}\\
\text{common invariant subspace}\end{smallmatrix}\Big\}/\sim
\ee 
where $\sim$ stands for equivalence up to overall conjugacy in $GL(n,\C)$.
The representation is rigid iff $\boldsymbol{m}$ is non-empty and zero dimensional.

The Deligne-Simpson problem was solved by W. Crawley-Boevey by relating it to the representation theory of quivers,
see \cite{CW1,CW2,CW3,CW4,CW5,CW6}.\footnote{\ See \cite{haraoka} for a textbook account.} The theory is reviewed in the physical language of
Seiberg duality in \cite{Cecotti:2022uep}. To summarize the story, we start by describing
a convenient way to
encode the conjugacy classes $C_a\subset GL(n,\C)$ of a $\ell$-tuple $\{\mu_a\}_{a=1}^\ell$
of $n\times n$ matrices.
For each marked point $p_a$ we choose a polynomial $P_a(z)$ which is solved by the matrix $\mu_a$, i.e.\! such that $P_a(\mu_a)=0$. Note that the polynomial depends only on the conjugacy class of $\mu_a$.
While we can choose any polynomial with this property, for definiteness
we assume $P_a(z)$ to be \emph{minimal.}
Then we factorize $P(z)$ into linear factors $P_a(z)=\prod_{i=1}^{m_a}(z-\xi_{a,i})$
choosing an order between them.
Now we set
\be
E_{a,k}\overset{\rm def}{=} \prod_{i=1}^k (\mu_a-\xi_{a,i})\C^n,
\quad N_{a,k}\overset{\rm def}{=} \dim E_{a,k}\ \ \ a=1,\dots,\ell,\ \ k=0,\dots,m_a.
\ee
When the polynomial $P_a(z)$ is minimal the non-negative integers are strictly-decreasing
\be
N_{a,k}<N_{a,k-1},
\ee
 and we choose the order of the factors in $P_a(z)$
such that the integers $N_{a,k}$ take their smallest possible values (we call this the \emph{minimal form}). This choice is not necessary,
but convenient. Then the conjugacy class $C_a$ of $\mu_a$ in $GL(n,\C)$ is encoded in the two sequences
\be
0\equiv N_{a,m_a} < N_{a,m_a-1}<\cdots < N_{a,1} < N_{a,0}\equiv n, \quad\text{and}\quad \xi_{a,1},\dots,\xi_{a,m_a}
\ee
indeed from these data we recover the Jordan-block structure of $\mu_a$.
Henceforth we write $\{(N_{a,i},\xi_{a,i})\}$ for the datum $\{C_a\}$.
A conjugacy class $C_a$ consisting only the element $\lambda\,\boldsymbol{1}$
($\lambda\in\C^\times$) will be called ``dummy''. $m_a=1$ if and only if $C_a$ is dummy.
In our application the local monodromy belongs to a dummy class only if it corresponds to
the extra special point with local monodromy $-1$ arising from the lift of the
monodromy representation from $PSp(4,\Z)$ to $Sp(4,\Z)$ in rank-2 geometries with
$\varkappa=2$ (cf.\! eq.\eqref{dummy}).  

To the Deligne-Simpson datum $\{(N_{a,i},\xi_{a,i})\}$ one associates a star-shaped graph
with one branch for each class $C_a$ containing $m_a-1$
nodes (not counting the central one $\bigstar$) \cite{CW1,CW2,CW3,CW4,CW5,CW6}:
\be\label{uuutttttw}
\begin{gathered}
\xymatrix{& \bullet_{1,1} \ar@{-}@<-0.9ex>[ldd] & \cdots\ar@{-}[l] & \bullet_{1,m_1-1}\ar@{-}[l]\\
& \bullet_{2,1}\ar@{-}@<-0.6ex>[dl] & \cdots\ar@{-}[l] & \bullet_{2,m_2-3}\ar@{-}[l]&\bullet_{2,m_2-2}\ar@{-}[l] &\bullet_{2,m_2-1}\ar@{-}[l]\\
\bigstar &\bullet_{3,1}\ar@{-}[l] & \cdots\ar@{-}[l] & \bullet_{3,m_3-2}\ar@{-}[l] & \bullet_{3,m_3-1}\\
& \vdots & \vdots & \vdots & \vdots\\
& \bullet_{\ell,1}\ar@{-}@<0.9ex>[uul] & \cdots\ar@{-}[l] & \bullet_{\ell,m_\ell-1}\ar@{-}[l]
}
\end{gathered}
\ee
which we see as the Dynkin graph of a simply-laced Kac-Moody Lie algebra $\mathfrak{G}$.
We write $\alpha_\bigstar$ and $\alpha_{\bullet_{a,i}}$ for (respectively) the simple root associated to the node
$\bigstar$ and $\bullet_{a,i}$ while $\Delta^+(\mathfrak{G})$ stands for the set of positive roots of $\mathfrak{G}$.
Note that a ``dummy'' class yields a zero-lenght branch which has no effect on the definition of $\Delta^+(\mathfrak{G})$ but will enter in the definition
of the sub-sets $R_+$ and $\Sigma$ below.
For $\alpha$ a positive root
\be\label{hhhhhasss}
\alpha=k_\star\,\alpha_\bigstar+\sum_{a=1}^\ell\sum_{i=1}^{m_a-1} k_{a,i}\,\alpha_{\bullet_{a,i}}\in\Delta^+(\mathfrak{G}),\qquad k_\star, k_{a,i}\in\Z_{\geq0}
\ee
 we set
\be
d(\alpha)\overset{\rm def}{=} 1-\frac{1}{2}\langle\alpha,\alpha\rangle\equiv 1-k_\star^2+\sum_a\left(\sum_{i=1}^{m_a-1}k_{a,i}^2+\sum_{i=1}^{m_a-1}k_{a,i}k_{a,i-1}\right),
\ee
where $\langle-,-\rangle$ is the integral quadratic form on the root lattice of $\mathfrak{G}$
given by the Cartan matrix, and we use the convention $k_{a,0}\equiv k_\star$ for all $a$.
The integer $d(\alpha)$ is non-negative for all roots and zero iff the root $\alpha$ is real. We write $R_+\subset\Delta^+(\mathfrak{G})$ for the sub-set of positive roots
\eqref{hhhhhasss} such that the complex number
\be
\xi(\alpha)\overset{\rm def}{=} \prod_{a=1}^\ell\xi_{a,1}^{k_\star}\!\prod_{i=1}^{m_a-1} \big(\xi_{a,i}^{k_{a,i-1}}/\xi_{a,i}^{k_{a,i}}\big)
\ee
is equal 1, that is,
\be
R_+\overset{\rm def}{=}\big\{\alpha\in\Delta^+(\mathfrak{G})\colon \xi(\alpha)=1\big\}.
\ee
Note that adding the ``dummy'' class $\lambda\,\boldsymbol{1}$ has the effect
$\xi(\alpha)\leadsto \lambda^{k_\star}\, \xi(\alpha)$.
For convenience we rewrite $\xi(\alpha)$ in the form
\be
\xi(\alpha)=\lambda_\star^{k_\star}\prod_a\left(\prod_{i=1}^{m_a-1}\lambda_{a,i}^{k_{a,i}}\right)
\ee
and call $\lambda_\star=\prod_a\xi_{a,1}$ and $\lambda_{a,i}=\xi_{a,i+1}/\xi_{a,i}$ the \emph{fugacity} of the quiver node
$\bigstar$ and respectively $\bullet_{a,i}$. We specify the function $\xi(\cdot)$
by drawing the quiver with the values of the fugacities written over the respective nodes.

\begin{defn}
 $\Sigma\subset R_+$ is the subset of positive roots $\alpha$ such that for all nontrivial
decomposition $\alpha=\beta_1+\cdots+\beta_s$ with $\beta_j\in R_+$ we have 
\be\label{iiiiuu77cc}
d(\alpha)>\sum_i d(\beta_i).
\ee
\end{defn}

\begin{thm}[Crawley-Boevey \textbf{Theorems 1.1, 1.9} of \cite{CW3}]
There is an irreducible representation of degree $n$ with local monodromies in the
given conjugacy classes $\{(N_{a,i},\xi_{a,i})\subset GL(n,\C)\}$, iff 
\be\label{hhhhhasss2}
\boldsymbol{N}=n\,\alpha_\bigstar+\sum_a\sum_i N_{a,i}\alpha_{\bullet_{a,i}}\in\Sigma.
\ee
In this case $\dim\boldsymbol{m}=2\, d(\boldsymbol{N})$
and hence an irreducible representation is rigid iff $\boldsymbol{N}\in \Sigma$ is a \emph{real} root of the Kac-Moody algebra $\mathfrak{G}$.
\end{thm}

\begin{rem} When $\alpha\in\Sigma$ admits a non-trivial decomposition in $R_+$
we have both reducible and irreducible representations with the given local classes $C_a$.
When $\alpha$ is a real root in $R_+$ either all representations are irreducible or all reducible.
\end{rem}

\begin{rem} If the graph $\mathfrak{G}$ is the Dynkin graph of a finite-dimensional, affine, or
hyperbolic Kac-Moody algebra, the condition $d(\alpha)=0$ (resp.\! $d(\alpha)\geq1$)
implies that $\alpha$ is a real root (resp.\! imaginary root). 
\end{rem}

\subsection{Irreducible monodromies in $Sp(4,\R)$}

With the possible exceptions discussed in \S.\,\ref{reddd}, all smooth $\C^\times$-isoinvariant,
indecomposable, non-isotrivial, rank-2 geometries are irreducible
with 3, 4 or 5 special points for, respectively,
rigid geometries, geometries with one dimension 2, and geometries with two
dimensions 2. The associated Dynkin graphs are stars with 3, 4, or 5 branches
of the following kinds
\be\label{poooor}
\begin{tabular}{l@{\hskip25pt}l}\hline\hline
$\delta$ & Dynkin branch\\\hline
1& $\xymatrix{\cdots\ar@{-}[r]& \text{\ovalbox{1}}}$\\
2 & $\xymatrix{\cdots\ar@{-}[r]& \text{\ovalbox{2}}}$\\
3& $\xymatrix{\cdots\ar@{-}[r]& \text{\ovalbox{2}}\ar@{-}[r]&\text{\ovalbox{1}}}$\\
4 &\xymatrix{\cdots\ar@{-}[r]& \text{\ovalbox{3}}\ar@{-}[r]& \text{\ovalbox{2}}\ar@{-}[r]& \text{\ovalbox{1}}}
\\
\end{tabular}
\ee
where the circled number on the $(a,i)$-th node is $N_{a,i}$.
Note that a $\delta=1$ branch can correspond only to a non-enhanced divisor of type $I_n$ ($n>0$).

The ``virtual'' dimension of the moduli space $\boldsymbol{m}$ of irreducible representations
whose graph has $s$ branches of kinds $\{\delta_1,\dots,\delta_s\}$
is $2\,d(\delta_1,\dots,\delta_s)$ where
\be
d(\delta_1,\dots,\delta_s)= 2s+\sum_{\ell=1}^s\delta_s -15
\ee
and this quantity should be non-negative (zero in the rigid case). In the addition we have the other conditions in the \textbf{Theorem} which depend on the eigenvalues of the local monodromies.

\subsubsection{The conjugacy classes which may appear}
We list the conjugacy classes in $GL(4,\C)$ which may appear.
To the $a$-th special point in $\mathscr{P}\simeq\mathbb{P}^1$ there corresponds
 a branch of the quiver of one of the four kinds \eqref{poooor} with
 their fugacities.
 
In the case $\varkappa=2$ only the image of the monodromy
in $PSp(4,\Z)$ is non-ambiguous, so its conjugacy class
is determined up to overall sign.
However the integers $N_{a,i}$ are independent of the sign, as are the fugacities
$\lambda_{a,i}\equiv\xi_{a,i+1}/\xi_{a,i}$ of all nodes except the central one.
Thus the only ambiguity is the sign of $\lambda_\star$.
Therefore when $\varkappa=2$ we have to explore the two possible signs
of the central fugacity  $\lambda_\star$. 

%

\subparagraph{Discriminant component non-enhanced.} In this case $[\mu_i]=[\varrho_i]=[\sigma_i\oplus \boldsymbol{1}_2]$ with $\sigma_i\in SL(2,\Z)$ a Kodaira monodromy, so
\be\label{fffufug}
\begin{tabular}{l|lll}\hline\hline
type & $\delta$ & Dynkin branch & fugacities\\\hline
$I_n$ $n>0$ &1& $\xymatrix{\cdots\ar@{-}[r]& \text{\ovalbox{1}}}$ & 
\xymatrix{\cdots\ar@{-}[r]& 1}\\
$I_0^*$ &2 & $\xymatrix{\cdots\ar@{-}[r]& \text{\ovalbox{2}}}$ & 
\xymatrix{\cdots\ar@{-}[r]& -1}\\
$I_n^*$ $n>0$& 3& $\xymatrix{\cdots\ar@{-}[r]&\text{\ovalbox{2}} \ar@{-}[r]&\text{\ovalbox{1}}}$ &$\xymatrix{\cdots\ar@{-}[r]& -1\ar@{-}[r]&1}$\\
order $m>2$ & 3& $\xymatrix{\cdots\ar@{-}[r]& \text{\ovalbox{2}}\ar@{-}[r]&\text{\ovalbox{1}}}$
& $\xymatrix{\cdots\ar@{-}[r]& e^{2\pi i/m}\ar@{-}[r]& e^{-4\pi i/m}}$
\\
\end{tabular}
\ee

\subparagraph{Enhanced divisor non-discriminant.} Suppose that $u_1=0$ is not part of the discriminant.
The local coordinate around the corresponding special point in $\mathbb{P}^1$ is $z=u_1^{d_2}/u_2^{d_1}$. A path from $(\epsilon, u_2)$ to $(e^{2\pi i/d_2}\epsilon,u_2)$ covers a closed loop
encircling $u_1=0$. Hence the local $\mu$-monodromy $\mu$ satisfies 
\be\label{whyyyy}
\mu^{d_2}=\begin{cases}\boldsymbol{1} &\text{for }\varkappa=1\\
\pm\boldsymbol{1}&\text{for }\varkappa=2.
\end{cases}
\ee
 Since the central fiber in the covering family $\mathscr{X}$ is
a smooth Abelian fiber $A_{u_1=0}$, the monodromomy $\mu$ should be the rational representation of an element of the automorphism group of $A_{u_1=0}$. The element $\eta=\exp(2\pi i\,\ce/\Delta_2)\in\mathsf{Aut}(A_{u_1=0})$
satisfies
\be\label{uuujjjjbbb}
\sigma(\eta)^{d_2}=\sigma(\exp(2\pi i\,\ce/\lambda))= \exp(2\pi i/\varkappa)\cdot \boldsymbol{1}, 
\ee
so that, modulo the sign when $\varkappa=2$, a solution to \eqref{whyyyy} is $\mu=\chi(\eta)$
whose eigenvalues are 
\be\label{4phases}
\big\{\exp[\pm 2\pi i(1-\Delta_1)/\Delta_2],\exp[\pm 2\pi i(1-\Delta_2)/\Delta_2]\big\}.
\ee  
When $\Delta_2$ is a new dimension this is the only solution to \eqref{whyyyy} up to equivalence. In this case 
the 4 eigenvalues are distinct and we get a $\delta=4$ branch of $\mathfrak{G}$. Otherwise
one can find in principle other automorphisms of order $d_2$ which may play the role
of $\mu$. This seems unlikely, and we shall assume the equality $\mu=\chi(\eta)$ in general.
If the axis $u_2=0$ does not belong to the discriminant, the same formulae
apply with $1\leftrightarrow2$.

$\mu$ is semisimple and integral, so a part for the dummies $\pm\boldsymbol{1}$,
it may have
\begin{itemize}
\item 2 distinct eigenvalues both of multiplicity 2 of the form $(+1,-1)$
or $(e^{2\pi i/m}, e^{-2\pi i/m})$ ($m=3,4,6$) which lead to a $\delta=2$ branch,
\item
 3 distinct eigenvalues of multiplicites (2,1,1) of the form $(\pm1,\e^{2\pi i/m},e^{-2\pi i/m})$
($m=3,4,6$) which lead to a $\delta=3$ branch with fugacities
\be
\xymatrix{\cdots \ar@{-}[r]& e^{2\pi i/m}\ar@{-}[r]& e^{-4\pi i/m}}
\ee 
\item in all other cases we have a $\delta=4$ branch.
\end{itemize}

Eigenvalues $\pm1$ are expected only when we have a regular axis with $\Delta=2$
which is rather uncommon. Hence (possibly with a few exceptions) we expect that
only $\delta=2$ and $\delta=4$ branches appear at an enhanced non-discriminant
orbifold point. $\delta=2$ means that the four
phases \eqref{4phases} are equal in pairs. This  requires
\be\label{cffffeq*}
\frac{1-\Delta_1}{\Delta_2}=\pm\frac{1}{\Delta_2}\bmod 1
\ee
i.e.\! either \textit{(i)} $\Delta_1=m\Delta_2$ which is possible only for $\Delta_2=2$ and $m>1$,
or \textit{(ii)} $\Delta_1=2+m\Delta_2$ where $\Delta_2$ is an old dimension.
In all other cases we have a $\delta=4$ branch.
\begin{rem} The argument also shows that $\phi(d_i)\leq 4$ for a rank-2 theory.
\end{rem}

\subparagraph{Enhanced divisor in the discriminant.} 
We suppose that $u_1=0$ is both enhanced and part of the discriminant.
Arguing as around eq.\eqref{uuujjjjbbb}, we conclude that the monodromy $\mu$ around the corresponding special point $z=0$ in
$\mathbb{P}^1$ is related to
 the monodromy $\varrho$ around
the divisor $(u_1)\subset\mathscr{C}$ by
\be
\mu^{d_2}=\begin{cases}\varrho &\varkappa=1\\
\pm\varrho &\varkappa=2.
\end{cases}
\ee
We need to determine $\mu$ up to conjugacy in $GL(4,\C)$.
As argued in \S.\,\ref{s:sing}, $\mu$ induces an action on the homology of the Albanese curve of the singular fiber which must be
 the rational representation of
an automorphism, which (as in the previous paragraph) has
eigenvalues $\exp(\pm 2\pi i/\Delta_2)$ (up to the usual sign ambiguity when $\varkappa=2$).
This is consistent since $\Delta_2$ is necessarily a rank-1 dimension. 
The other two eigenvalues are determined by $d_2$ and the Kodaira type of the singular fiber.
If $\varrho$ is non-semisimple, so should be $\mu$, and hence these two eigenvalues must be
equal, hence $\pm1$ since the determinant of $\mu$ is $1$. In particular

\begin{fact} Assume $\varkappa=1$.
{\bf(1)} if $d_2$ is even the divisor $u_1=0$ cannot be of type $I^*_n$
with $n>0$. {\bf(2)} the only allowed semisimple types are:
\begin{description}
\item[$d_2=2$:] $I_0^*$, $IV$, $IV^*$
\item[$d_2=3$:] $I_0^*$, $III$, $III^*$
\item[$d_2=4$:] $IV$, $IV^*$
\item[$d_2=6$:] $I_0^*$
\end{description}
When $\varkappa=2$ no a priori restriction on the non-semisimple fibers,
while for the semisimple fibers  we have the same table as above
with $IV$, $IV^*$ replaced by $IV$, $IV^*$, $II$ and $II^*$.
\end{fact}
No counterexample to these statements appears in the tables of \cite{Martone}. The statement
eliminates an ambiguity in the analysis of model $\#37$ in that reference. 

When the fiber type is additive, the other two eigenvalues besides $\exp(\pm 2\pi i/\Delta_2)$ are related to the Kodaira type of the $\varrho$-monodromy by the base change formula (modulo the sign ambiguity when $\varkappa=2$) which may be read e.g.\! in table 5.2 of \cite{MW}. Since $d_2>1$ for an enhanced
discriminant, when the fiber is semisimple the two eigenvalues are distinct.

In conclusion: an enhanced discriminant $u_1=0$ produces always a $\delta=4$ branch of the Dynkin graph except when $\Delta_2=2$ which produces a $\delta=3$ graph. 

\subsection{Rank-2 special geometries with rigid $\mu$-monodromy} 

The classification/construction of non-isotrivial irreducible special geometries with rigid monodromy is
relatively straightforward in rank-2. This would produce the full list of non-isotrivial irreducible geometries
if the $\mu$-rigidity conjecture is true, and a large portion of them otherwise. The results below
will strongly support the conjecture.  

\subsubsection{Goursat rigid monodromies}

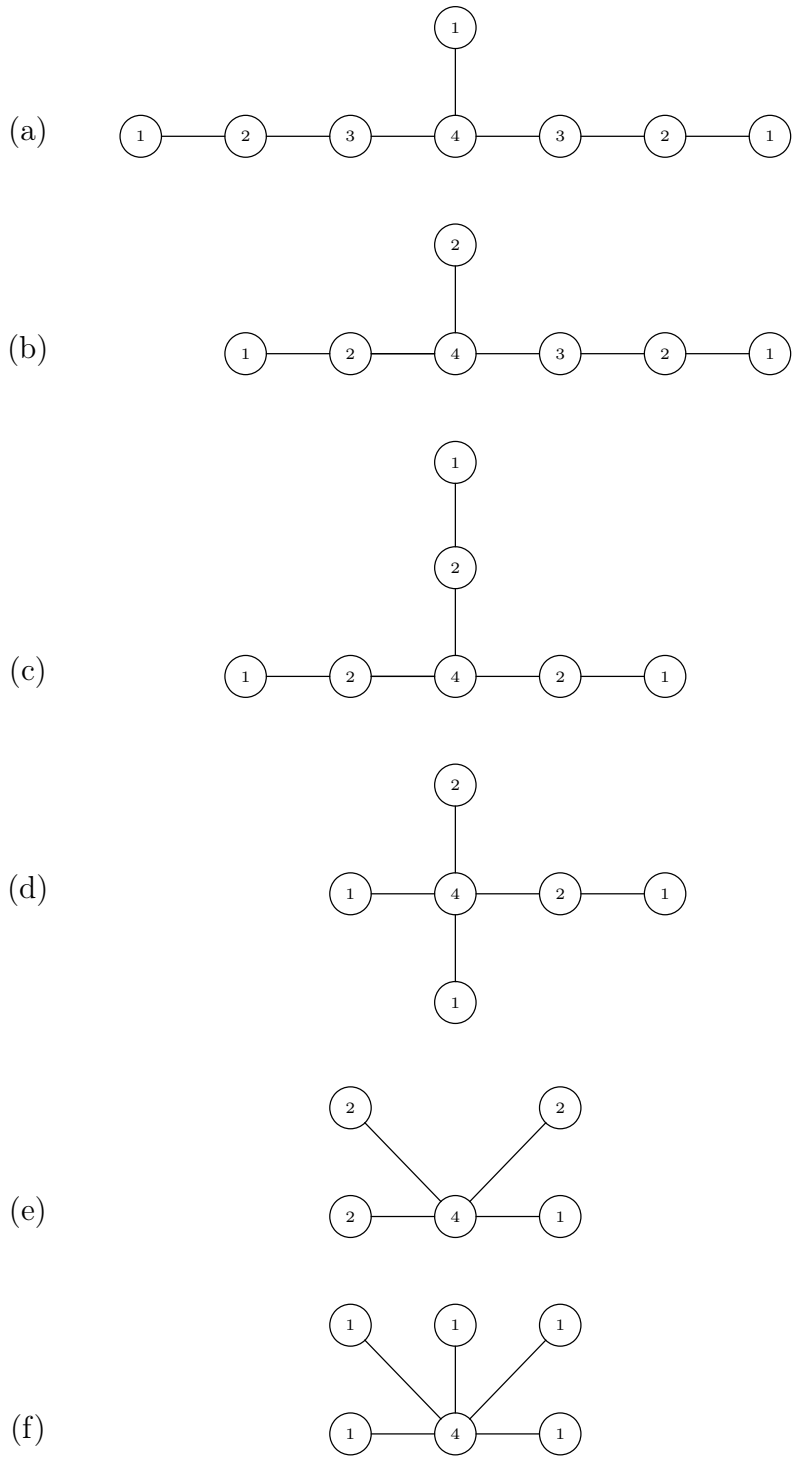
\begin{figure}
\begin{tiny}
$$
\xymatrix{&&&&*++[o][F-]{1}\ar@{-}[d]\\
\text{\normalsize(a)}&*++[o][F-]{1}\ar@{-}[r] & *++[o][F-]{2}\ar@{-}[r] & *++[o][F-]{3}\ar@{-}[r]& *++[o][F-]{4}\ar@{-}[r]
&*++[o][F-]{3}\ar@{-}[r]&*++[o][F-]{2}\ar@{-}[r]& *++[o][F-]{1}
\\
&&&&*++[o][F-]{2}\ar@{-}[d]\\
\text{\normalsize(b)}&&*++[o][F-]{1}\ar@{-}[r] & *++[o][F-]{2}\ar@{-}[r] \ar@{-}[r]& *++[o][F-]{4}\ar@{-}[r]
&*++[o][F-]{3}\ar@{-}[r]&*++[o][F-]{2}\ar@{-}[r]& *++[o][F-]{1}
\\
&&&&*++[o][F-]{1}\ar@{-}[d]\\
&&&&*++[o][F-]{2}\ar@{-}[d]\\
\text{\normalsize(c)}&&*++[o][F-]{1}\ar@{-}[r] & *++[o][F-]{2}\ar@{-}[r] \ar@{-}[r]& *++[o][F-]{4}\ar@{-}[r]
&*++[o][F-]{2}\ar@{-}[r]&*++[o][F-]{1}
\\
&&&&*++[o][F-]{2}\ar@{-}[d]\\
\text{\normalsize(d)}&&&*++[o][F-]{1}\ar@{-}[r] &  *++[o][F-]{4}\ar@{-}[r]
&*++[o][F-]{2}\ar@{-}[r]&*++[o][F-]{1}\\ 
&&&&*++[o][F-]{1}\ar@{-}[u]
\\
&&&*++[o][F-]{2}\ar@{-}[dr] && *++[o][F-]{2}\ar@{-}[dl] \\
\text{\normalsize(e)}&&&*++[o][F-]{2}\ar@{-}[r] &  *++[o][F-]{4}\ar@{-}[r]
&*++[o][F-]{1}
\\
&&&*++[o][F-]{1}\ar@{-}[dr] & *++[o][F-]{1}\ar@{-}[d]&*++[o][F-]{1}\ar@{-}[dl]\\
\text{\normalsize(f)}&&&*++[o][F-]{1}\ar@{-}[r] &  *++[o][F-]{4}\ar@{-}[r]
&*++[o][F-]{1}}
$$
\end{tiny}
\caption{\label{goursat}The six order 4 rigid ODEs: graphs and roots in minimal form. The first graph is the
order-4 hypergeometric, the 4-th one is the Appell equation (an ODE truncation of the 2-variable Appell
PDE), the fifth one is a particular instance of Simpson equation,
and the last one is the order-4 Pochhammer equation.}
\end{figure}

In order 4 there are six Fuchsian ODEs with rigid monodromy. They were already discussed by Goursat in 1886 \cite{Gou,Gou2}.
Writing all conjugacy classes in their minimal form,
the corresponding pairs \textsf{(Dynkin graph, real root)} are as in table \ref{goursat}. Some of these graphs
describe classically well-known Fuchsian OPEs:
\begin{itemize}
\item[(a)] is the order-4 hypergeometric ODE of Thomae;
\item[(d)] is a reduction to one-variable of the 2-variable
Appell hypergeometric PDE, aka as the Okubo equation $II^*$
\item[(f)] is the order-4 (Jordan-)Pochhammer ODE.
\end{itemize} 

\begin{rem} The graph corresponding to 3 unipotent monodromies and one with minimal polynomial of degree 4
\begin{tiny}\begin{equation*}
\begin{gathered}
\xymatrix{&*++[o][F-]{1}\ar@{-}[d]\\
*++[o][F-]{1}\ar@{-}[r]&*++[o][F-]{4}\ar@{-}[r]&*++[o][F-]{3}\ar@{-}[r]&*++[o][F-]{2}\ar@{-}[r]&*++[o][F-]{1}\\
&*++[o][F-]{1}\ar@{-}[u]}
\end{gathered}
\end{equation*}\end{tiny}
\hskip-4pt has $d(\alpha)=0$, so superficially seems to be a rigid solution to the Deligne-Simpson problem,
but it is not so since the Kac-Moody algebra is not hyperbolic and $\alpha$ is not a root.
We conclude that these conjugacy classes cannot appear.
\end{rem}

\subsubsection{Graph selection rules}\label{s:selection}
 We note that
$\xi(\alpha)$ remains invariant if we cancel one or more $\delta=1$ branches.
This fact eliminates some possibilities thus further restricting the allowed geometries.

\begin{fact}\label{selrule} In a non-isotrivial, irreducible, rank-$2$ geometry:
\begin{itemize}
\item[\bf (1)] graph {\rm(c)} is excluded;
\item[\bf (2)] graph {\rm(e)} is excluded;
\item[\bf (3)] graph {\rm(f)} is excluded  when $\varkappa=1$.
\end{itemize}
\end{fact}
\begin{proof}[Argument]
\textbf{Part (1).} At least one branch is non-semisimple,
so the fugacities are as follows
\be
\begin{gathered}
\xymatrix{&&(1)\ar@{-}[d]\\
&&(-1)\ar@{-}[d]\\
(\zeta^{-2})\ar@{-}[r]& (\zeta) \ar@{-}[r] &(\pm 1)\ar@{-}[r] & (\xi)\ar@{-}[r]& (\xi^{-2})}
\end{gathered}
\ee
where $\zeta$ and $\xi$ are either $-1$ (if the branch is non-semisimple)
or a root of $1$ of degree 2 (if it is semisimple).
As before $\pm$ corresponds to the two possible lifts in $Sp(4)$ when $\varkappa=2$.
If the sign is $+$ the simple root at the central node $\alpha_\star\in R_+$
and $\alpha-\alpha_\star$ is the minimal imaginary root of $\widehat{E}_6$ which also belongs to
$R_+$.
Since $d(\alpha)=0$ the representation is reducible. Suppose now that the sign is $-$: we have the decomposition
\be
\begin{smallmatrix}&& 0\\
&& 1\\
1&2&3&2&1\end{smallmatrix}\ +\ \begin{smallmatrix}&&1\\
&&1\\ 0&0&1&0&0\end{smallmatrix}
\ee
into roots of the $E_6$ and $A_3$ subgraphs which also belong to $R_+$ and again the representation is reducible.

\textbf{Part (2).} The reader may find in appendix \ref{defpprr} the proof of the following
\begin{lem}\label{ttthelemma}
A monodromy representation associated to the graph {\rm(e)},
whose local monodromies are in the allowed lists, is irreducible only when
$\varkappa=2$ and the three semisimple monodromies have eigenvalues
\be\label{kkki1234}
(1,1,-1,-1),\quad (i,i,-i,-i),\quad (\pm \zeta_3,\pm \zeta_3,\pm\zeta_3^{-1},\pm\zeta_3^{-1}),
\ee
where $\zeta_3$ is a primitive 3-rd root of 1.\footnote{\ This is the first example in \cite{koEx}. I thank V. Kostov for clarifications.} 
\end{lem}

If it exists, an irreducible geometry based on graph (e) should
have a conformal manifold of dimension 1 hence dimensions $\{2,\Delta\}$,
with $\Delta$ an \emph{old} dimension $\neq2$. To have $\varkappa=2$ 
we must have $\Delta\in\{6/5,4/3,4,6\}$. When $\Delta$ is an integer $=4,6$
there is only one enhanced divisor, and the spectra \eqref{kkki1234} are inconsistent.
When $\Delta=6/5$ we have $(d_1,d_2)=(5,3)$ and the eigenvalues \eqref{kkki1234} are inconsistent.
We remain with $\{\Delta_1,\Delta_2\}=\{2,4/3\}$ which has $(d_1,d_2)=(3,2)$.
But then the axis of dimension $\Delta_1$ has a monodromy $\varrho=\pm\mu_1^3=\pm\boldsymbol{1}$
and hence is not part of the discriminant, so 
\be
\exp(\pm 2\pi i/\Delta_1),\qquad \exp(\pm 2\pi i(1-\Delta_2)/\Delta_1),
\ee
should be all square roots of 1. This is not the case. 

\textbf{Part (3).} When $\varkappa=1$ there is no extra sign from the lift from $PSp(4,\Z)$ to
$Sp(4,\Z)$ while all local monodromies in graph (f) are unipotent, so that the group homomorphism 
\be
\xi(\alpha)\colon \bigoplus_{a\in\mathfrak{G}} \Z\,\alpha_a\to \C^\times
\ee is the trivial one,
and the only roots in $\Sigma_\xi$ are the ones in the fundamental domain (which are never real!)
and the simple roots (cf.\! \S.\,5 of \cite{CW1}). When $\varkappa=2$ we have the 
central node $\star$ may have fugacity $\zeta_\star=-1$, so that now $R_+$ is the set of positive roots of the form $2k\alpha_\star+\cdots$. It is easy to check that no decomposition of the root in (f) into elements of $R_+$ exists. Alternatively, since the ODE associated to graph (f) is the classical order-4 Pochhmmer equation, irreducibility for this fugacity follows from \textbf{Proposition 1.3} of \cite{pochh}.
\end{proof}

\begin{rem} The argument for part {\bf(2)} is greatly simplified if one
assumes the \textbf{Folk-theorem} (which we argued is true in rank-2).
A non-isotrivial geometry based on graph (e) should be the geometry of a
Lagrangian SCFT, which were classified \cite{Tach}, and no SCFT in that list
corresponds to graph (e). Indeed a Lagrangian SCFT has integral $\Delta_i$ which is incompatible with (e) on the nose.
\end{rem}

\begin{rem} According to the \textbf{Folk-theorem} graph (f) corresponds to dimensions $\{2,2\}$
which has automatically $\varkappa=2$ as predicted by rigidity and irreducibility.
\end{rem}

\begin{fact}\label{distinctr} In graph {\rm(a)} we assume the relation $\mu_1\mu_2\mu_3=\boldsymbol{1}$ ($\mu_1$ corresponds to the $\delta=1$ branch) by absorbing  the possible extra sign from the lift from $PSp(4,\Z)$ to $Sp(4,\Z)$ when $\varkappa=2$ in the definition of $\mu_3$. With this convention, the matrices $\mu_2$ and $\mu_3$ cannot have any common eigenvalue.
\end{fact}
\begin{proof}[Argument] \textbf{(1)} Suppose that they have the eigenvalue $\eta$ in common.
Since the monodromies are symplectic, also $\eta^{-1}$ is an eigenvalue and by a suitable ordering of the factors in their respective minimal polynomial we get the graph fugacities
\be
\begin{gathered}
\xymatrix{&&&1\ar@{-}[d]\\
\xi_1^{-2}\ar@{-}[r]&\xi_1\,\eta^{-1}\ar@{-}[r] &\eta^2 \ar@{-}[r] & 1\ar@{-}[r] & \eta^{-2}\ar@{-}[r] &\xi_2\,\eta\ar@{-}[r] &\xi_2^{-2}}
\end{gathered}
\ee
and the simple root $\alpha_\star$ and the imaginary root $\alpha-\alpha_\star$ belong to $\Sigma_+$.
\end{proof}

\subsection{From rigid monodromy to special geometry}

All non-isotrivial, irreducible, $\mu$-rigid rank-2 geometries are associated to one of four
Kac-Moody graphs (a), (b), (d), (f) according to the rule:
\begin{itemize}
\item Graph (a): rigid geometries (2 not a dimension) with one ``knotted''
discriminant of type $I_n$ for some $n>0$;
\item Graph (b): rigid geometries (2 not a dimension) with one ``knotted''
discriminant not of type $I_n$;
\item Graph (d) geometries with dimensions $\Delta_1=2$, $\Delta_2\neq2$;
\item Graph (f) geometries with $\{\Delta_1,\Delta_2\}=\{2,2\}$.
\end{itemize}
We shall say that a $\mu$-rigid special geometry is of class $\mathfrak{G}$
iff its monodromy corresponds to a real root of the Kac-Moody algebra $\mathfrak{G}$.

\subsubsection{The inverse algorithm}
The inverse algorithm for rank-$2$ works through the following steps:
\begin{itemize}
\item[1.] If the datum $(\{\Delta_i\},\{\varrho_a\})$ corresponds to an isotrivial,
a decomposable, a reducible, or a covering geometry, 
the special geometry can either be constructed 
by direct means or easily proven not to exist;
\item[2.] Otherwise (assuming the $\mu$-rigidity conjecture)
either the datum $(\{\Delta_i\},\{\varrho_a\})$ corresponds to a graph (a), (b), (c), or (d) or no geometry with the datum exists;
\item[3.] From the datum one reads the fugacities of the nodes
of the graph $\mathfrak{G}$. Sometimes there may be an ambiguity, but with only few alternatives. One checks whether one or more of these
alternative fugacities define an irreducible monodromy using the characterization of the root set $\Sigma(\xi)$. At this point we have a rigid $GL(4,\C)$ monodromy which is known to
be in $Sp(4,\C)$ in facts it can be conjugated to lay in $Sp(4,\R)$ or $Sp(4,\mathfrak{o}_\textbf{F})$;
\item[4.] We have to check that the Brauer class of the resulting monodromy is trivial, so
that the monodromy can be defined over $\mathbb{Q}$;
\item[5.] At this point we know that the monodromy is integral, rational, and symplectic for a unique
skew-symmetric form $\Omega$. However we need the symplectic matrix $\Omega$ to be principal, that is, with integral entries and determinant 1.
Thus we have to check that by a simultaneous transformation
\be\label{oooo12z}
\mu_a\leadsto U^{-1}\mu\, U,\qquad \Omega\leadsto U^t \Omega\, U,\qquad U\in GL(4,\R) 
\ee  
we can set all matrices $\mu_a$, $\Omega$ to be integral while $\det\Omega=1$.
\item[6.] \begin{scriptsize}\textdbend\end{scriptsize}
the matrix $U$ with the properties in 5. (if it exists) may be non-unique. 
Two such matrices $U$ and $U^\prime$
yield equivalent integral monodromies iff $U^\prime U^{-1}\in Sp(4,\Z)$, otherwise they produce \emph{inequivalent} integral monodromies (all equivalent in the $GL(4,\C)$-sense).\footnote{\ This subtle phenomenon is already present in rank-1 where geometries with dimensions $m$ and $m/(m-1)$ (with $m=3,4,6$) have monodromies conjugates over $GL(2,\R)$ but not over $SL(2,\Z)$. } Thus we get \emph{a list} of inequivalent integral monodromies associated
to each fugacity assignment. In particular, it is quite frequent to have distinct special geometries arising from the same graph and fugacities which differs because the types $I_n$ of the semistable discriminants
have different values of $n$ (see below for examples).
We stress that the list is finite (an typically quite short)
because of Deligne finiteness (or Faltings finiteness in view of 7.); 
\item[7.] For each (rigid) integral monodromy in the list of item 6.
 we get a candidate $\C^\times$-isoinvariant special geometry.
To write it explicitly, we use Riemann's observation that the solutions
to all Fuchsian ODE with rigid monodromy have a known explicit
integral representation (see \cite{katz}, for a textbook treatment \cite{haraoka}, for a review in 
physical language see \S.\,4 of \cite{Cecotti:2022uep}).
The Simpson-Deligne integrality implies that these integral representations are \emph{motivic}
i.e.\! the integrand is given in terms of algebraic functions. 
The periods of the SW differential
$(a^D_i,a^j)$ are linear combinations of these explicit transcendents, with
coefficients uniquely fixed by the Almkvist-Zudilin theorem (\S.\,\ref{k87xxx}). The SW periods of the various geometries associated to a single
set of fugacities are linear combinations of each other with coefficients dictated by the matrices $U$.
We stress that, up to equivalence, only finitely many such $U$ may exist by Faltings finiteness.
Having constructed the SW periods as multivalued functions of the Coulomb branch coordinates
we may check their regularity (which typically gives no problem). If this final text is passed, we declare to have constructed an explicit special geometry. If one wishes to write its equations as a variety,
one may read them from the algebraic functions in the integrand of the integral representation of the periods. 
\item[8.] 
the special geometry with the original datum $(\{\Delta_i\},\{\varrho_a\})$ is either an exceptional one in step 1, or one of the finite list
produced in step 7, or it does not exist.
\end{itemize}

\subparagraph{The classification program.} Since only finitely many fugacity assignments
are possible, we have only a finite list of cases to analyze along the lines of items 3.-7. above.
This will produce all non-isotrivial, irreducible, $\mu$-rigid
geometries thus completing the classification in rank-2 unless -- contrary to expectations\,! --
non-trivial non-$\mu$-rigid geometries exist.

\subparagraph{How to proceed in practice.} To work out items 3.-7. one has to construct explicitly the matrices 
$\mu_i$'s realizing the rigid monodromy of the Abelian family over $\mathring{\mathscr{P}}$. This has been an active field of research for almost two centuries, so a lot is known. Unfortunately, for some of the graphs the results in the literature 
refer to the simpler case of generic eigenvalues/fugacities whereas we are interested in a highly non-generic case. Thus for some graphs we cannot limit ourselves to search the answer in the existing math literature and 
novel computations are needed. Then we have to determine the list of relevant matrices $U$ in step 6,
and this gives us a ``Diophantine'' problem which may be not easy to solve. 
Anyhow a very large portion of the program can be carried on rather explicitly without going into detailed computations
as we are going to show.

\subsection{Checking regularity}

Having found a rigid monodromy representation with the prescribed properties we have to check that it
describes a regular special geometry. In principle the check splits in two parts:
\begin{itemize}
\item[(1)] local checks around the special divisors. One has to check that there exist a local solution to the ODE which are (multivalued) local special coordinates which behave as predicted by the local models
of the Lagrangian fibration near a singular fiber of the appropriate Kodaira type, cf.\! \S.\,\ref{s:sing} .
Since the leading behavior of the local solutions is determined by the exponents of the ODE, i.e.\!
the inverse data, these local check are easy and quick (see \S.\,\ref{s:regcheck} below);
\item[(2)] if the local checks are satisfied, in principle one has still to check that the ``good''
solution at one special divisor is the same one which is ``good'' at all other divisors.
This requires the use of the connection formulae which are explicitly known for rigid ODEs
\cite{haraoka} so the check may be time-consuming but has no fundamental difficulty.
\end{itemize}  
We shall not do checks of type (2) in this paper. We do not expect them to be a major issue:
we know from the Almkvist-Zudilin theorem \cite{AZthm} that there is a global solution which embeds $\mathring{\mathscr{P}}$
as a Legendre submanifold of $\mathbb{P}^3$. The ``good'' local solutions
should agree with this global ``good'' solution. 

\subsubsection{Local regularity along the axes}\label{s:regcheck}

We consider a non-isotrivial rank-2 geometry with dimensions $\{\Delta_1,\Delta_2\}=\lambda\{d_1,d_2\}$
where $\lambda$ is written in minimal terms is $k/\varkappa$ ($k\in\mathbb{N}$) and $\varkappa\in\{1,2\}$. 
When $\varkappa=1$ we have an element $h\in\C(\mathscr{C})$ with $\Delta(h)=1$
while for $\varkappa=2$ we have an element $h^2\in\C(\mathscr{C})$ with $\Delta(h^2)=2$.
The standard coordinate on the normalization $\mathbb{P}^1$ of $\mathscr{P}$ is $z=u_1^{d_2}/u_2^{d_1}$ \cite{dolgaWP}.

The SW periods have the form
\be\label{kkkkyyyy}
a^i=h\cdot f_i(z),\qquad a_j^D=h\cdot f_j^D(z)
\ee
where $(f_j^D(z),f_i(z))$ is the solution to the Picard-Fuchs equation selected by the Almkvist-Zudilin theorem.
In  \eqref{kkkkyyyy} $h$ is an abusive notation which stands for $h$ when $\varkappa=1$
and $\sqrt{h^2}$ when $\varkappa=2$. In the latter case the periods are defined only up to overall sign.

We specialize the above expression along the coordinate axes. They correspond to the two special points
$z=0$ and $z=\infty$ in $\mathscr{P}$. To get the behavior along the axes we use the local
solutions of the Picard-Fuchs equations around these special points \cite{haraoka}. If the monodromy is semisimple, the leading term has the form
$z^{\alpha_{\ell,i}}\big(1+O(z)\big)$
with $\alpha_{\ell,i}=\log\xi_{\ell,i}/2\pi i$ an exponent 
 of the local monodromy $\mu_\ell$ at these special points ($\xi_{\ell,i}$ are the eigenvalues of $\mu_\ell$ as before). If $\mu_\ell$ has non-trivial Jordan blocks the leading term contains logarithms of $z$; however for a suitable choice of the symplectic basis only the ``dual'' special coordinates $a^D_i$
 contain logarithms, while the special coordinates $a^j$ have the local form of a fractional power of $z$
 times a holomorphic function.
 
 We consider regularity along the first axis $(u_2=0)$ ($z=\infty$). The same discussion applies to the second axis
 with $1\leftrightarrow 2$.
 Two ($\C$-linear combination of the) periods, $a^{\|}$ and $a^{\perp}$ correspond to the
 special coordinates associated respectively to the first axis and its orthogonal bundle
 \cite{caorsi,M13}: these two particular periods do not contain logarithms. Comparing with the local models in \S.\,\ref{s:stable}
 we see that regularity along the first axis requires the existence of two exponents $\alpha_{\|}$ and $\alpha_{\perp}$ of
 $\mu_\infty$ such that
 \be\label{Cregg}
 \begin{aligned}
a^\| &=h\cdot f_{\|}(z)\sim h\cdot z^{\alpha_{\|}}= u_\ell^{1/\Delta_1}\\
a^\perp &=h\cdot f_{\perp}(z)\sim h\cdot z^{\alpha_{\perp}}=  u_2^{1/\Delta_K}\cdot u_1^{(1-\Delta_K)/\Delta_1}
\end{aligned}
 \ee   
 where $\Delta_K$ is the dimension of the rank-1 (weak) geometry with the same
 Kodaira type as the $\varrho$-monodromy $\varrho_1=\pm \mu_\infty^{d_1}$ around the first axis, i.e.
 \be
 \begin{array}{c|cccccccc}
\text{type}& I_n & I_n^* & II & II^* & III & III^* & IV & IV^*\\\hline
\Delta_K & 1 & 2 & 6/5 & 6 & 4/3 & 4 & 3/2 & 3
 \end{array}
 \ee 
 with the convention that an enhanced regular axis is of type $I_0$.

\subsubsection{Check of regularity: an example}\label{s:eeeexam}

We start by illustrating how regularity of the SW periods works in an non-isotrivial example where we have an enhanced
discriminant of semisimple type, we consider the geometry of the SCFT $\#3$ of \cite{Martone}
which has dimensions $\{4,10\}$, $\varkappa=2$, and a class-$\widehat{E}_7$
geometry where the axis of dimension 4 is
a component of the discriminant with Kodaira type $II^*$. The second axis is regular since
it has the new dimension 10. Given that this geometry describes an actual SCFT, we know in advance that the regularity conditions are satisfied: we just want to illustrate how they work.
The monodromies of the family over $\mathring{\mathscr{P}}$ are explicitly
\be
\mu_1=\left(\begin{smallmatrix}1 & -1 & 0 &0\\
2 & 0 & 1 & 1\\
0 & 0 &0 &-1\\
-1 & 0 & 0 &0\end{smallmatrix}\right),\quad
\mu_2=\left(\begin{smallmatrix}0 & 0 & 0 &1\\
1 & 0 & 0 & 1\\
0 & -1 &-1 &-1\\
0& 0 & 1 &0\end{smallmatrix}\right),\quad \mu_3=\left(\begin{smallmatrix}1 & 0 & 0 &0\\
0 & 1 & 0 & 1\\
0 & 0 &1 &0\\
0 & 0 & 0 &1\end{smallmatrix}\right),
\ee
which satisfy
\be\label{opppo1}
\begin{aligned}
&(\mu_1^2+\boldsymbol{1})(\mu_1^2-\mu_1+\boldsymbol{1})=\mu_2^4+\mu_2^3+\mu_2^2+\mu_2+\boldsymbol{1}=0,
&&\mu_1\mu_2\mu_3=-\boldsymbol{1},\\
&
\mu_a^t\,\Omega\,\mu_a=\Omega,\quad a=1,2,3 \quad\text{where}\quad\Omega\equiv\left(\begin{smallmatrix}0 & 0&1 &0\\
0& 0 &0 &1\\
-1&0 &0&0\\
0&-1&0&0\end{smallmatrix}\right)
\end{aligned}
\ee
while from the form of $\mu_3$ we see that there is a knotted discriminant of type $I_1$.
The $\varrho$-monodromy around the axis of dimension 4,
$\varrho=-\mu_1^2$, satisfies the minimal equation
\be
(\varrho-\boldsymbol{1})(\varrho^2-\varrho+\boldsymbol{1})=0
\ee
which is consistent with type $II^*$ (or $II$). Let us look at the special coordinates as multivalued functions in a neighborhood of
 the axes. Since $\varkappa=2$ they have the form
 \be
 a^i=\sqrt{h^2}\,f_i(z)
 \ee
where $h^2$ is an element of the quotient ring of dimension $2$, $z$ is the coordinate
on $\mathbb{P}^1$ and $f_i$ are solutions to the Picard-Fuchs ODE whose local expressions along the axes
scale with the appropriate local exponents. In the present case $h^2=u_2/u_1^2$ and $z=u_1^5/u_2^2$.

The local exponents along the $u_2$-axis ($u_1=0$) for the special coordinates are
$1/5$ and $2/5$ (as one sees using $\mu_2^5=\boldsymbol{1}$ and the dimension formulae along a regular axis) so that
\be
\begin{aligned}
a_\|&\sim \sqrt{h}\, z^{1/5} = \frac{u_2^{1/2}}{u_1}\cdot\frac{u_1}{u_2^{2/5}}= u_2^{1/10}\\
a_\perp&\sim \sqrt{h}\, z^{2/5} = \frac{u_2^{1/2}}{u_1}\cdot\frac{u_1^2}{u_2^{4/5}}= u_1\,u_2^{-3/10}
\end{aligned}\qquad\quad \text{as }u_1\to0
\ee
which are precisely the regular behaviors expected along a regular axis, see \cite{caorsi}\!\!\cite{M13}. More interesting are the periods along the $u_1$-axis which belong to a semisimple component of the discriminant.
The local exponents are read from the minimal polynomial of $\mu_1$,
 eq.\eqref{opppo1}: they are $1/4$ and $1/6$. Then as $u_2\to 0$
\be
\begin{aligned}
a_\|&\sim \sqrt{h}\, z^{1/4} = \frac{u_2^{1/2}}{u_1}\cdot\frac{u_1^{5/4}}{u_2^{1/2}}= u_1^{1/4}\\
a_\perp&\sim \sqrt{h}\, z^{1/6} = \frac{u_2^{1/2}}{u_1}\cdot\frac{u_1^{5/6}}{u_2^{1/3}}= u_2^{1/6}\,u_1^{-1/6}
\end{aligned}
\ee
which is the expected result for an axis of dimension 4 which is a discriminant
of type $II^*$, cf.\! eq.\eqref{Cregg}. Note that regularity rules out type $II$.

\subsection{Geometries with $\{\Delta_1,\Delta_2\}=\{2,2\}$}\label{s:22}

The $\{2,2\}$ geometries are identified by the \textbf{Folk-theorem} with the geometries of
Lagrangians SCFTs with gauge group $SU(2)\times SU(2)$ and hence are known.
However it is instructive
to recover their classification from the solution of the corresponding Deligne-Simpson problem
and Diophantine considerations. 
The perfect agreement of our Number-Theoretic analysis with the physical expectations
 yields a highly non-trivial check on our methodology and provides additional evidence for the $\mu$-rigidity conjecture.

The only \emph{isotrivial} geometry with dimensions $\{2,2\}$ is $\cn=4$ with the gauge group
$SU(2)\times SU(2)$. Therefore we may assume the geometry to be non-isotrivial. 
We have $\varkappa=2$, the conformal manifold has dimension 2,
 there are no enhanced divisors, and there must be 5 special points on $\mathscr{P}\simeq\mathbb{P}^1$. Moreover from the discussion in \S.\,\ref{reddd} we have that an indecomposable
 $\{2,2\}$ geometry should be irreducible.
Then $\mu$-rigidity yields

\begin{fact} A non-isotrivial $\{2,2\}$ geometry has 5 discriminant components of type $I_{n_j}$ ($n_j>0$,
$j=1,\dots,5$) and its Picard-Fuch equation is the 4-order Pochhammer ODE (i.e.\! graph {\rm(f)}) with fugacities\footnote{\ Recall from \textbf{Fact \ref{selrule}} that there is no irreducible lift of the $PSp(4,\Z)$ monodromy with
$\xi(\alpha_\star)=1$.}
\be\label{fuga}
\xi(\alpha_\star)=-1,\qquad \xi(\alpha_a)=1\ \ \text{for }a\neq\star.
\ee
\end{fact}

The solutions of the Pochhmmer ODE have a well-known explicit expression in terms of contour integrals (see e.g.\! \cite{pochh}) from which one reads the monodromy representation \cite{pochh}.
Consider the following six explicit $4\times 4$ matrices:
\be\label{hara}
\begin{aligned}
&\mu_1=\left(\begin{smallmatrix}1 & 0 & 0 & 0\\ 0 & 1& 0 &0\\
0 & 0 & 1 & 0\\ -2 & -2 & -2 & 1\end{smallmatrix}\right) &&\mu_2=
\left(\begin{smallmatrix}1 & 0 & 0 & 0\\ 0 & 1& 0 &0\\
-2 & -2 & 1 & 2\\ 0 & 0 & 0 & 1\end{smallmatrix}\right)
&&\mu_3=
\left(\begin{smallmatrix}1 & 0 & 0 & 0\\ -2 & 1& 2 &2\\
0 & 0 & 1 & 0\\ 0 & 0 & 0 & 1\end{smallmatrix}\right)\\
&\mu_4=
\left(\begin{smallmatrix}1 & 2 & 2 & 2\\ 0 & 1& 0 &0\\
0 & 0 & 1 & 0\\ 0 & 0 & 0 & 1\end{smallmatrix}\right),
&&\mu_5=
\left(\begin{smallmatrix}3 & 2 & 2 & 2\\ -2 & -1& -2 &-2\\
2 & 2 & 3 & 2\\ -2 & -2 & -2 & -1\end{smallmatrix}\right)
&&M=\left(\begin{smallmatrix}0 & 1 & 1 & 1\\ -1 & 0& 1 &1\\
-1 & -1 & 0 & 1\\ -1 & -1 & -1 & 0\end{smallmatrix}\right)
\end{aligned}
\ee
The first four may be written as ($i=1,2,3,4$)
\be
\mu_i=1+2\,v_ i\otimes v^t_i\, M
\ee
where $\{v_4,v_3,v_2,v_1\}$ is the standard basis\footnote{\ Note the inversion of the order of the basis elements: it is chosen to agree with the conventions in the classical literature on the Pochhammer ODE.} of $\Z^4$.
The $\mu_a$'s satisfy (now $a=1,2,\dots,5$)
\be
\begin{aligned}
&(\mu_a-\boldsymbol{1})^2=0 &\quad& \text{rank}(\mu_a-\boldsymbol{1})=1,&\quad&\mu_a-\boldsymbol{1}\equiv 0\bmod2\\
&\mu_1\,\mu_2\,\mu_3\,\mu_4\,\mu_5=-\boldsymbol{1} &&\det M=1 && \mu_a^t\,M\,\mu_a=M, 
\end{aligned}
\ee 
so that $\{\mu_1,\dots,\mu_5\}$ is (a lift to $Sp(4,\Z)$ of) a monodromy representation in $PSp(4,\Z)$ with five
local monodromies of type $I_2$ where $Sp(4,\Z)$ is seen as the integral matrix group which 
leaves invariant the \emph{principal} polarization $M$. In particular

\begin{corl} The Brauer class of the rigid representation associated to graph {\rm(f)} with fugacities \eqref{fuga}
is trivial.
\end{corl}

The change of basis which puts the polarization $M$ in 
the standard form is
\be
S=\left(\begin{smallmatrix}1 & 0 &-1 &1\\
0 & 0 & 1 &-1\\ 0 & 1 & 0 & 0\\ 0 & -1 &0 & 1\end{smallmatrix}\right),\qquad \Omega\overset{\rm def}{=} S^t M S=
\left(\begin{smallmatrix}0& 0 &1 &0\\
0 & 0 & 0 &1\\ -1 & 0 & 0 & 0\\ 0 & -1 &0 & 0\end{smallmatrix}\right).
\ee
In the standard symplectic basis the integral monodromies read
\be
\ell_a=S^{-1}\mu_a S\quad a=1,2,\dots,5.
\ee
In particular $\ell_4$ becomes
\be
\ell_4=\left(\begin{smallmatrix}1 & 0 &2 &0\\
0 & 1 & 0 &0\\ 0 & 0 & 1 & 0\\ 0 & 0 &0 & 1\end{smallmatrix}\right)
\ee
while the five local monodromies are cyclically permuted by the $\Z_5$ symmetry:
\be\label{uuuu1234}
\ell_{a+1\bmod 5}= R\, \ell_a\, R^{-1},
\ee
where
\be\label{Rmat}
R=\left(\begin{smallmatrix}0 &1 & 1& -1\\
1 & -1 & -1 & 2\\ -1 & 0 & 1 &0\\ 0 &-1 & 0 &1\end{smallmatrix}\right)\in Sp(4,\Z),\qquad R^5=-\boldsymbol{1},
\ee
which makes manifest that \emph{all five} local monodromies $\ell_a$ are of type $I_2$. We conclude
\begin{corl} There is a $\{2,2\}$ geometry (with principal polarization)
having 5 irreducible discriminant components of type $I_2$. All other (inequivalent) $\{2,2\}$ (principally polarized) geometries
arise from monodromy representations in $Sp(4,\Z)$ which are conjugate over $GL(4,\C)$
to the $\{\ell_1,\dots,\ell_5\}$ one, but not conjugate in $Sp(4,\Z)$.
\end{corl}

Consider the matrix $Q\in Sp(4,\R)$
{\renewcommand{\arraystretch}{1.4}\be\label{kkkki8z}
Q=\left(\begin{array}{c|c}\sqrt{2}\,\boldsymbol{1}_2 & \\\hline
 &\frac{1}{\sqrt{2}}\,\boldsymbol{1}_2 \end{array}\right)\in Sp(4,\R)
\ee}
which has the property that the five matrices $Q\ell_a Q^{-1}$ have integral entries.
One has
\be
\begin{aligned}
&\gcd_{ij}\big(Q\ell_1 Q^{-1}-\boldsymbol{1}\big)_{ij}=4\\
&\gcd_{ij}\big(Q\ell_a Q^{-1}-\boldsymbol{1}\big)_{ij}=1\quad\text{for }a=2,3,4,5
\end{aligned}
\ee
for instance, say,
\be
Q\ell_4Q^{-1}=\left(\begin{smallmatrix} 1 & 0 & 1 & 0\\ 0 & 1 & 0 & 0\\
0 & 0 & 1 & 0\\ 0 & 0 & 0 & 1\end{smallmatrix}\right)
\ee
We interpret the monodromy representation $\{\tilde\mu_a\equiv Q \ell_a Q^{-1}\}$ as a second  \emph{inequivalent} special geometry with $\{\Delta_1,\Delta_2\}=\{2,2\}$
having one $I_4$ and four $I_1$.\footnote{\ The matrix $QRQ^{-1}$ which permutes the 5 local monodromies
is now valued in $Sp(4,\tfrac{1}{2}\Z)$ so the five local monodromies, while conjugate in $Sp(4,\mathbb{Q})$ are not longer conjugate in the Siegel modular group $Sp(4,\Z)$.}

\begin{fact}\label{fact20} There are two inequivalent non-isotrivial, irreducible special geometries with dimensions $\{2,2\}$ with five discriminants of type $I_2$ and, respectively, four $I_1$ and one $I_4$. 
\end{fact}
%

According to the \textbf{Folk-theorem}, the above two special geometries should describe
two inequivalent non-isotrivial Lagrangian SCFTs with gauge group $SU(2)\times SU(2)$.
Indeed in the table of \cite{Martone} there are precisely two such models: $\#13$ with five $I_2$'s,
and $\#31$ with four $I_1$ and one $I_4$, in perfect agreement with our findings. These two $\cn=2$ models are Gaiotto class-$S[A_1]$ SCFTs \cite{gaiotto} as well as quiver
gauge theories with quivers
\be\label{iiiu7zz}
\begin{aligned}
&I_2^{\,5}\colon &&\xymatrix{\fbox{2}\ar@{-}[r]& *++[o][F-]{2}\ar@{-}[r]& *++[o][F-]{2}\ar@{-}[r]& \fbox{2}}
\\
\\
&I_1^{\,4},\,I_4\colon &&\xymatrix{*++[o][F-]{2}\ar@<0.5ex>@/^1pc/@{-}[rrrr]\ar@<-0.5ex>@/_1pc/@{-}[rrrr]&&&& *++[o][F-]{2}}\
\end{aligned}
\ee
\smallskip

\noindent
From the physical side we do not expect other solutions, since there are no other Lagrangian SCFT
with these properties. We shall give a math proof of this fact in the next subsection.

\begin{rem} As a byproduct of the analysis we get explicit expressions for the
 periods of all $\{2,2\}$ SCFTs in terms of Pochhammer transcendents
which have well-known integral representations and connection coefficients.
\end{rem}

\subsubsection{No other $\{2,2\}$ special geometries}\label{s:noother}

We present a direct math proof that indeed there
are no other integral symplectic monodromy representations conjugate in $GL(4,\C)$ to the (unique) rigid complex
representation given by the solution of the Deligne-Simpson problem associated to graph (f). 
The reader may prefer to skip this technical subsection. 

\begin{fact} (Up to isomorphism) there are no other irreducible, non-isotrival special geometries
with Coulomb dimensions $\{2,2\}$ besides the two listed in {\bf Fact \ref{fact20}}.
\end{fact}

\begin{proof} We work in the Haraoka basis where the local monodromies are given by the matrices $\mu_i$ and the symplectic form by $M$, cf.\! eq.\eqref{hara}. Let $L\in Sp(M,\C)$ be a complex $4\times4$ matrix such that
$L^t ML=M$ while the five matrices $\tilde\mu_i \equiv L \mu_i L^{-1}$ ($i=1,\dots,5$)
 have integral entries.
For $i=1,\dots,4$ we have
\be\label{ooooooqqq21}
\tilde\mu_i-\boldsymbol{1}=2\, L v_i\otimes v_i^t L^t M\quad \text{(not summed over $i$!)}
\ee 
where $\{v_4,v_3,v_2,v_1\}$ is the standard column-vector basis of $\Z^4$ (written in reverse order).
We extend eq.\eqref{ooooooqqq21} to $i=5$ by setting $v_5=(-1,1,-1,1)^t$.
It is easy to check that 
\be
v_i^t\mspace{0.5mu} M v_j=1\quad \text{for}\quad 1\leq i< j\leq 5.
\ee
If the first four matrices in the \textsc{lhs} of \eqref{ooooooqqq21} have to be integral, 
$L v_i$ must be a scalar multiple of a 4-vector with integral coefficients, so 
there are 4 complex numbers $\rho_a$ such that
{\renewcommand{\arraystretch}{1.5}\be
L=\begin{bmatrix}\rho_1 n_{1,1} & \rho_2 n_{1,2} & \rho_3 n_{1,3} & \rho_4 n_{1,4}\\
\rho_1 n_{2,1} & \rho_2 n_{2,2} & \rho_3 n_{2,3} & \rho_4 n_{2,4}\\
\rho_1 n_{3,1} & \rho_2 n_{3,2} & \rho_3 n_{3,3} & \rho_4 n_{3,4}\\
\rho_1 n_{4,1} & \rho_2 n_{4,2} & \rho_3 n_{4,3} & \rho_4 n_{4,4}\\
\end{bmatrix}\equiv N\,\mathrm{diag}(\rho_1,\dots,\rho_4)\in Sp(4,\C)
\ee}
\hskip-4pt where $n_{a,i}$ are integers, $\gcd\{n_{1,i},n_{2,i},n_{3,i},n_{4,i}\}=1$ for $i=1,\dots,4$,
and the four columns vectors $\boldsymbol{n}_i\equiv (n_{a,i})$ are linear independent over $\C$. $N$
is the integral matrix with entries $n_{a,i}$.
 One has
\be\label{kkkjqwert}
1=\det L= (\rho_1\rho_2\rho_3\rho_4) \det N\qquad \det N\in \Z.
\ee
Moreover the integrality of $\tilde\mu_i$ implies 
\be\label{kkk321}
2\rho_i^2\in \mathbb{Z},\quad\text{and}\quad (\rho_i\rho_j)^{-1}\in\mathbb{Z}\qquad \text{for }1\leq i<j\leq 4,
\ee
where the second condition follows from $(Lv_i)^t M(L v_j)=1$ for $0\leq i < j\leq 4$.
Requiring that $\tilde\mu_5$ is also integral extends  \eqref{kkk321} to $i,j=1,\dots,5$
where $\rho_5$ is the complex number (unique when it exists) such that
\be
\rho_5\, n_{i,5}\equiv (Lv_5)_i= -\rho_1\, n_{i,1}+\rho_2\, n_{i,2}-\rho_3\, n_{i,3}+\rho_4\, n_{i,4}
\ee
where $\{n_{1,5}, n_{2,5},n_{3,5},n_{4,5}\}$ are coprime integers.
Dividing out by $\rho_5$ we see that
$\rho_i/\rho_5\in \mathbb{Q}$, so $\rho_i =q_i \rho_5$ with $q_i\in\mathbb{Q}$ ($q_5=1$).
Since $2\rho_i^2= 2\rho_5^2 q_i^2$, from the first eq.\eqref{kkk321}
we conclude that there is a \emph{square-free}
integer $d$ such that 
\be
\rho_i=\sqrt{\frac{d}{2}}\;k_i,\qquad k_i\in\mathbb{N}.
\ee
Then 
\be\label{juqwertqt}
(\rho_i\rho_j)^{-1}= \frac{2}{d\, k_i\, k_j}\in \Z\qquad \text{for }1\leq i<j\leq 5.
\ee
From eq.\eqref{kkkjqwert} we get
\be
1= \frac{d^{\,2}}{4}\,k_1\,k_2\,k_3\,k_4\,\det N=\frac{\det N}{\frac{2}{d\,k_1\, k_2}\,\frac{2}{d\,k_3\, k_4}}
\ee
Now we have two possibilities:
\begin{itemize}
\item[$d=2$] then $k_i=1$ and $\rho_i=1$ for $i=1,\dots,5$ while $\det N=1$,
so that $N\equiv L\in Sp(4,\Z)$ and $2\rho_i^2=2$ for $i=1,2,3,4,5$.
In this solution we have five $I_2$'s. It is our first integral monodromy corresponding to
the upper quiver in \eqref{iiiu7zz};
\item[$d=1$] so that $2\rho_i^2=k_i^2$ is
a perfect square for all $i$. Eq.\eqref{juqwertqt} shows that at most one
$k_i$ ($i=1,\dots,5$) can be $2$ all others should be 1. By the $\Z_5$ symmetry \eqref{kkk321}
we may assume that $k_1=k_2=k_3=k_4=1$. Then one gets $k_5=2$ (cf.\! \textbf{Lemma \ref{finally!}} in appendix \ref{dddeeff}). Therefore we get a monodromy with one $I_4$ and four $I_1$'s which is the second special geometry in {\bf Fact \ref{fact20}}. 
\end{itemize}
\end{proof}

\subsection{Irreducible geometries with $\{2,\Delta\}$, $\Delta\neq 2$}

We already discussed these geometries when discussing the \textbf{Folk-theorem}, now
we look at them from the present $\mu$-rigidity viewpoint. The are 4 special points on $\mathscr{P}$
and the Picard-Fuchs ODE has a graph $\mathfrak{G}$ with 4
branches. Graph (e) was ruled out, so all $\mu$-rigid special geometries
should arise from graph (d) which corresponds to an ODE reduction of the Appell PDE aka
Okubo $II^*$ ODE \cite{mima}.

\begin{fact}\label{888uuu} Assume $\mu$-rigidity. 
{\bf(1)} There are no
 non-isotrivial, irreducible geometries with dimensions $\{2,\Delta\}$
and $\Delta\not\in\{2,3,4,6\}$. 
{\bf(2)} A non-isotrivial, irreducible geometry with dimensions $\{2,3\}$
should have two knotted discriminant of types $I_{n_1}$, $I_{n_2}$ (with $n_1$, $n_2$ satisfying the conditions in \emph{appendix \ref{appeApriori}}),
while the third discriminant is $u_2=0$ of type $I_{2n_3}$ {\rm($n_3>0$)}.
{\bf(3)}
A non-isotrivial, irreducible geometry with dimensions $\{2,2m\}$ ($m\geq2$)
should have two non-enhanced discriminant of type $I_{n_1}$, $I_{n_2}$
 (with $n_1, n_2$ as in \emph{appendix \ref{appeApriori}}),
while the third discriminant is $u_2=0$ of type $I_{n_3}^*$ ($n_3>0$).
\end{fact}

These results can also be inferred from the \textbf{Folk-theorem}.
The obvious solution to the conditions in appendix \ref{appeApriori} is $n_1=n_2\in\{1,2\}$.

\begin{proof}[Argument] {\bf(1)} $\Delta$ is a dimension in rank-1 so if the statement does not hold
$\Delta=m/(m-1)$ for $m\in\{3,4,6\}$. Then (say) for $m=4,6$
\be
(\Delta_1,\Delta_2)=\frac{2}{m-1}\left(m-1,\frac{m}{2}\right)=\lambda(d_1,d_2)
\ee
so $(d_1,d_2)$ are coprime $\neq1$ and both axes are enhanced. 
Since $(m-1)>2$ and $\delta<4$ both are regular. The $\varrho$-monodromy along the
$u_2=0$ axis then has eigenvalues
\be
\exp(\pm2\pi i/2),\qquad \exp(\pm2\pi i (1-\Delta)/2)=\exp(\pm 2\pi i/(2(m-1)) 
\ee
while the eigenvalues of $\mu=\pm\varrho^{1/(m-1)}$ are not algebraic numbers of degree $\leq2$. 

Statement  {\bf(2)} reduces to {\bf(3)} for $2m=6$
by the following trick. The family over $\mathring{\mathscr{P}}$
of a $\{2,3\}$ geometry has local monodromies $\mu_1$, $\mu_2$, $\mu_3$, $\mu_4$ with
with $\mu_1\mu_2\mu_3\mu_4=1$, $\mu_1$, $\mu_2$ (knotted) monodromies of type $I_1$,
while $(\mu_3)^3=1$ and $(\mu_4)^2=\varrho_4$ where $\varrho_4$ is the monodromy around the
axis $u_2=0$ of dimension $2\equiv d_1$. 
The local monodromies\footnote{\ The change of sign in $\mu_3$ is due to the different sign of the phases $\exp(\pm 2\pi i/\Delta_2)$ for $\Delta_2=3$ and $\Delta_2=6$.} $\{\mu_1,\mu_2,-\mu_3,-\mu_4\}$ can also be seen (as a lift) of the local monodromy for a $\{2,6\}$ geometry where now $d_1=1$ so $-\mu_4=\varrho_4$.
Thus
\be
\varrho_4\Big|_{\{2,3\}}=(\varrho_4)^2\Big|_{\{2,6\}}
\ee
and so {\bf(3)} $\Rightarrow$ {\bf (2)}.

{\bf (3)} the two $\delta=1$ branches in graph (d) correspond to $I_{n_1}$ ($n_1>0$)
discriminants. The $\delta=2$ branch then corresponds to the enhanced axis of dimension $2m$
whose local $\mu$-monodromy has eigenvalues $\exp(\pm 2\pi i/2m)$ both with multiplicity 2 (and trivial Jordan blocks).
The forth branch has then $\delta=3$, and since the $u_2=0$ axis is not enhanced, by eq.\eqref{fffufug} it
must be a discriminant of type $I^*_{n_2}$ ($n_2>0$) or the local monodromy must satisfy $\mu^\ell=\pm1$
for $\ell\in\{3,4,6\}$ (that is, a semisimple Kodaira type $\neq I_0^*$). 
The second possibility looks inconsistent for the following reason.
The period matrix $\tau_{ij}$ along the $u_2=0$ axis must have the form
$\mathrm{diag}(\zeta,\tau^\prime)$ with $\zeta^\ell=1$ at all points $\tau\in\mathscr{M}$
which is not consistent with the discussion in \S.\,\ref{Folk2}.

When the $u_2=0$ axis is of type $I_{n_2}^*$,
or $\Delta\in\{3,4,6\}$ fugacities are
\be\label{fugggga}
\begin{gathered}
\xymatrix{& e^{-2\pi i /\Delta}\ar@{-}[d]\\
1\ar@{-}[r] & \pm e^{\pi i/\Delta} \ar@{-}[r] &-1 \ar@{-}[r] & 1\\
& 1\ar@{-}[u]}
\end{gathered}
\ee
An element $\lambda$ of the root lattice with $\xi(\lambda)\in\R$ should be even at
the central node;
the only decomposition into a sum of roots with this property is 
\be
\beta_1+\beta_2=\begin{smallmatrix}&1\\ 1&2&1&1\end{smallmatrix} +\begin{smallmatrix}1\\ 2&1\\ 1\end{smallmatrix}
\ee
with $\xi(\beta_a)=-1$, so that the representation is irreducible hence rigid.

\end{proof}

\begin{rem} All the known models with dimensions $\{2,\Delta\neq2\}$ (a part for the subtle $\#69$
discussed in \S.\,\ref{Folk2}) obey \textbf{Fact \ref{888uuu}} with $n_1=n_2\in\{1,2\}$.
According to the \textbf{Folk theorem} the known models exhaust the geometries with these dimensions.
Indeed a non-systematic search produced no integral representation corresponding to
a semisimple discriminant along the $u_2=0$ axis.
\end{rem}
\begin{rem} 
One can get the possible values of the integers
$n_1$ and $n_3$ by constructing the
explicit monodromy representation as we did in the Pochhammer case.
The explicit construction of the monodromy also allows to check the triviality of the Brauer class (as we did in the previous case) and the positivity condition. Of course, both statements hold true since we know that these geometries exists (for suitable $\{n_1,n_3\}\in \mathbb{N}^2$) because we may construct the corresponding SCFT by Lagrangian methods. 
\end{rem}

We give some examples.

\subsubsection{Example: irreducible special geometries with $\{\Delta_1,\Delta_2\}=\{2,4\}$}

While there is a substantial literature on the explicit monodromy representation for the
Okubo $II^*$ ODE, see e.g. \cite{sasai,sasai2}, to the best of our knowledge all published results
apply only under the simplifying assumption that the local monodromies are semisimple, which is not the case of interest for the present applications. Then,
to construct explicitly the geometries, we need to proceed by learned guess-work.
The nice aspect of the story is that if we manage to guess one monodromy representation, all others
are conjugate to it by rigidity, so a single happy guess is good enough. 
We got the following one:
\begin{fact}\label{kjjjhaa} The following four explicit integral $4\times 4$ matrices
{\rm\be
\begin{tabular}{l@{\hskip0.8cm}ll}
\phantom{matr}matrix & minimal equation & branch of conjugacy class\\\hline
$M_1= \left(\begin{smallmatrix}
-1 & -2 & 0 & 0\\ 
0 & -1 & 0 & 0\\
0 & 0 & 1 & 0\\
0 & 0 & 0 &1
\end{smallmatrix}\right)$ & $(M_1+\boldsymbol{1})^2(M_1-\boldsymbol{1})=0$ &$\xymatrix{\ar@{-}[r]&\fbox{2}\ar@{-}[r]&\fbox{1}}$\\
$M_2= \left(\begin{smallmatrix}
1 & 0 & 0 & 0\\ 
-1 & 1 & 0 & 1\\
-1 & 0 & 1 & 1\\
0 & 0 & 0 &1
\end{smallmatrix}\right)$ & $(M_2-\boldsymbol{1})^2=0$ & $\xymatrix{\ar@{-}[r]&\fbox{1}}$\\
$M_3= \left(\begin{smallmatrix}
0 & 1 & 1 & 0\\ 
-1 & 2 &1 & 0\\
0 & 0 & 1 & 0\\
1 & -1 & -1 &1
\end{smallmatrix}\right)$ & $(M_3-\boldsymbol{1})^2=0$ & $\xymatrix{\ar@{-}[r]&\fbox{1}}$\\
$M_4= \left(\begin{smallmatrix}
0 & 1 & -1 & 2\\ 
0 & 0 & -1 & 1\\
-1 & 2 & 1 & -1\\
-1 & 1 & 1 &-1
\end{smallmatrix}\right)$ & $M_4^2+\boldsymbol{1}=0$ &$\xymatrix{\ar@{-}[r]&\fbox{2}}$\\
\end{tabular}
\ee}
satisfy
\be
M_1M_2M_3M_4=\boldsymbol{1}\qquad\text{and}\qquad M_a^t\,\Omega\, M_a= \Omega\ \ (a=1,\dots,4)
\ee
where
\be
\Omega= \left(\begin{smallmatrix}0 & 1 & 0 & 0\\ -1 & 0 & 0 & 0\\
0 & 0 &0 &1\\
0 & 0 &-1 &0\end{smallmatrix}\right)
\ee
In particular the Brauer class of the order-4 Okubo $II^*$ ODE
with the fugacities in eq.\eqref{fugggga} (with $\Delta=4$) is trivial. 
\end{fact} 

\begin{corl} There exists a (non-isotrivial, irreducible) special geometry with dimensions $\{2,4\}$
two knotted divisors of type $I_1$ while the axis of dimension $2$ is a
component of the discriminant of type $I^*_2$.
\end{corl}

In the list of \cite{Martone} there are two SCFT with precisely this geometry $\#10$ ($SO(5)$ with six $\mathbf{4}$'s) and $\#43$ ($SO(5)$ with three $\mathbf{5}$'s). There is another SCFT with
dimensions $\{2,4\}$ whose geometry is non-isotrivial irreducible: $\#48$ ($SO(5)$ with two $\mathbf{4}$ and two $\mathbf{5}$)
which has two knotted discriminants of type $I_2$ while the axis of dimension $2$
has type $I_1^*$. Again this means that we have a second integral-symplectic representation which is
conjugated to the one in \textbf{Fact \ref{kjjjhaa}} in $GL(4,\C)$ but not in $Sp(4,\Z)$.
It follows from \textbf{Fact } (or the \textbf{Folk-theorem}) that there are no other geometries of this kind.

Going through the same (elementary) systematic procedure for the search
of integral symplectic monodromies conjugate to the one 
we used in \S.\,\ref{s:noother} one gets

\begin{fact} Besides the one in {\bf Fact \ref{kjjjhaa}} there is precisely one other irreducible geometry
with dimensions $\{2,4\}$ obtained
by conjugating the previous one
\be
\tilde M_a= U\, M_a\, U^{-1}\quad a=1,\dots,4,\qquad U=\frac{1}{\sqrt{2}}\left(\begin{smallmatrix}1 & 0 & 0 &-1\\
0 & 2 & 0 &0 \\
0 & 2 & -2 &0\\
0 & -1 & 1 & -1\end{smallmatrix}\right).
\ee 
Explicitly
\be
\begin{aligned}
\tilde M_1&=\left(\begin{smallmatrix}-1&-1& 1& 2\\
0& -1 & 0 &0\\
0 & -2 & 1 &0\\
0 & 1 & 0 &1\end{smallmatrix}\right) &\qquad \tilde M_2&=\left(\begin{smallmatrix}1&0& 0& 0\\
-2& 1 & 0 &0\\
0 & 0 & 1 &0\\
0 & 0 & 0 &1\end{smallmatrix}\right)\\
\tilde M_3&=\left(\begin{smallmatrix}-1&2& 0& 2\\
-2& 3 & 0 &2\\
-2 & 2 & 1 &2\\
0 & 0 & 0 &1\end{smallmatrix}\right) & \tilde M_4&=\left(\begin{smallmatrix}1&-1& -1& -4\\
0& -1 & 0 &-2\\
2 & -4 & -1 &-6\\
0 & 1 & 0 &1\end{smallmatrix}\right)
\end{aligned}
\ee
which confirms the physical expectation that the second geometry has one discriminant of type $I^*_1$
and two of type $I_2$ as required for $SO(5)$ with two $\mathbf{4}$ and two $\mathbf{5}$.
\end{fact}

Again our findings agree with the \textbf{Folk-theorem}. 

\subsubsection{Example: irreducible special geometry with $\{\Delta_1,\Delta_2\}=\{2,3\}$}
The Diophantine problem looks harder. However we may restrict the search by the following considerations: the monodromy $\mu_1$ associated to the regular axis is known by
\textbf{Proposition \ref{pro3}} (see also the more detailed analysis in \cite{car1,car2}). So we can
fix two of the required matrices
\be\label{iiiiii7zz}
\Omega=\left(\begin{smallmatrix}0 & 0 &1 & 0\\
0 &0 &0 &1\\
-1 & 0 &0 &0\\
0 & -1 & 0 & 0\end{smallmatrix}\right)\qquad\mu_1=\left(\begin{smallmatrix}0 & 1 &0 & 0\\
-1 &-1 &0 &0\\
0 & 0 &-1 &1\\
0 & 0 & -1 & 0\end{smallmatrix}\right)
\ee
which satisfy
\be
\mu_1^t\,\Omega\,\mu_1=\Omega,\qquad \mu_1^3=\boldsymbol{1}.
\ee
Then $\mu_2$, $\mu_3$ must be symplectic with respect to the principal polarization $\Omega$
and $\mu_a-\boldsymbol{1}$ should have rank 1 for $a=2,3$ i.e.
\be
\mu_a=\boldsymbol{1}+v_a\otimes v_a^t\,\Omega.
\ee
for some integral 4-vectors $v_a$.
By some try and error we find $v_2=(1,0,-1,-1)^t$ and $v_3=(0,1,1,1)^t$
so that
\be
\mu_2=\left(\begin{smallmatrix}2 & 1 &1 & 0\\
0 &1 &0 &0\\
-1 & -1 &0 &0\\
-1& -1 & -1 & 1\end{smallmatrix}\right),\qquad
\mu_3=\left(\begin{smallmatrix}1 &0 &0 & 0\\
-1 &0 &0 &1\\
-1 & -1 &1 &1\\
-1& -1 & 0 & 2\end{smallmatrix}\right)
\ee
since $v_1$ and $v_2$ are integral and primitive, these are monodromies of type $I_1$.
Then the relation $\mu_1\mu_2\mu_3\mu_4=\boldsymbol{1}$ yields
\be
\mu_4= \left(\begin{smallmatrix}-1 & 0 &0 & 1\\
1 &1 &-1 &3\\
0 & 0 &-1 &1\\
0& 0 & -1 & 1\end{smallmatrix}\right)
\ee
and the monodromy around the discriminant component $u_2=0$ in the Coulomb branch is
\be
\varrho_4\equiv \mu_4^2=\left(\begin{smallmatrix}1 & 0 &0 & 0\\
0 &1 &0 &6\\
0 & 0 &1 &0\\
0& 0 & 0 & 1\end{smallmatrix}\right)
\ee 
We conclude
\begin{fact} There is a special geometry with dimensions $\{2,3\}$, two knotted discriminant components of type $I_1$, and a third component $u_2=0$ of type $I_6$. 
\end{fact}
Of course this is just the well-known geometry of $SU(3)$ with six fundamentals. 
Here we reconstructed it from the $\mu$-rigidity perspective as a further check of the general philosophy. A (non-systematic) search has produced no other solution as physically expected.

\subsection{Rigid, non-isotrivial, rank-2 special geometries}

Rank-2 non-isotrivial, primitive, \emph{rigid} geometries have three special points on $\mathscr{P}$
and hence Dynkin graphs with 3 branches. Moreover $\Delta_2$ cannot be
an integral multiple of $\Delta_1$ (since in this case one has $\varkappa=\Delta_1$
while $\Delta_1\neq1,2$, so the geometry is automatically isotrivial \cite{Cecotti:2021ouq}). Thus
two special points are enhanced (they may or may not be part of the discriminant)
while the third one arises from a non-enhanced discriminant component (which now should be knotted). 
Assuming $\mu$-rigidity, the graph selection rules of \S.\,\ref{s:selection} allow
 two classes of geometries corresponding, respectively,
to the Dynkin graphs
of the Kac-Moody algebras $\widehat{E}_7$ and $E_7$. 

\subsubsection{Rigid geometries with a knotted discriminant of type $I_n$}\label{s:knott}

The vast majority of known special geometries of rank-2 SCFT with $\mathscr{R}_2=0$
have a unique knotted discriminant which is of Kodaira type $I_1$. To shed light on some puzzles
we shall encounter later, we begin by considering the implications of this fact \emph{per se},
without making \emph{any} other assumption (however plausible) nor using any restriction arising from our previous considerations.
In this subsection we assume only that the $\mu$-monodromy is given by three quasi-unipotent matrices $\mu_0,\mu_1,\mu_\infty\in Sp(4,\Z)$,
satisfying $\mu_0\mu_1\mu_\infty=\pm\boldsymbol{1}$, while $(\mu_1-\boldsymbol{1})$ is nilpotent of rank $1$.
The double sign $\pm$ applies only when $\varkappa=2$. The presence of a knotted discriminant guarantees that the geometry is indecomposable, but we do not rule out the possibility that $\mu$ is reducible.

If the minimal polynomials of both $\mu_0$ and $\mu_\infty$ are of degree 4 we get 
an $\widehat{E}_7$ graph and the representation generated by $\mu_0$
and $\mu_\infty$ is irreducible precisely when their minimal polynomials $P_0(z),P_\infty(z)$ have no common zero
i.e.\! $\gcd(P_0(z),P_\infty(z))=1$ or equivalently iff their resultant $\mathsf{Res}(P_0(z),P_1(z))\neq0$.
When this is the case the representation is automatically irreducible and rigid.
This ``nice'' case will be discussed in great detail in \S.\,\ref{s:ehat} below.
We ask what happens in the two ``less nice'' cases:
\begin{itemize}
\item[(I)] the minimal polynomial of, say, $\mu_0$ has degree $\leq3$;
\item[(II)] $P_0(z),P_\infty(z)$ have degree 4 while the polynomial $\gcd(P_0(z),P_\infty(z))$ has positive degree.
\end{itemize}
In either case the monodromy representation cannot be irreducible and decomposes
(over $\C$) into smaller ones. Since there is a non-trivial unipotent monodromy associated to the knotted discriminant, the geometry cannot be isotrivial, and since 2 is not a dimension, it is rigid.
Therefore we end up in the scenario discussed in \S.\,\ref{SS:quasi-isotrivial}.

\begin{fact}\label{777xw} A  rigid rank-2 geometry with a single knotted discriminant of type $I_n$ ($n>0$)
is either irreducible of class $\widehat{E}_7$
 or it is reducible quasi-isotrivial with Coulomb dimensions $\{3,4\}$ resp.\! $\{4,6\}$ and
with the discriminant types listed in {\bf Fact \ref{uuuu63}}. The two non-trivial periods
can be expressed in terms of Gauss' hypergeometric functions.
\end{fact}

\subsubsection{Rigid non-isotrivial geometries of class $E_7$}

Class-$E_7$ contains the non-isotrivial, primitive, rigid and $\mu$-rigid geometries
 where the knotted discriminant is
not of type $I_n$.
The $E_7$ class
is quite small. Indeed these geometries are severely constrained. 
The point is that class-$E_7$ rigid monodromies behave quite differently from their
$\widehat{E}_7$ cousins which are hypergeometric monodromies.
 In the $\widehat{E}_7$ case the Brauer class
is always trivial by the Levelt theorem \cite{levelt,BHeck}, while ``most'' $E_7$-class rigid monodromies are known to have non-trivial
Brauer classes \cite{Gou2}, so they do not describe special geometries: this is one of the reasons why
class-$E_7$ special geometries are relatively rare. 
The computation of the Brauer class is difficult, so this subclass of special geometries
is the hardest one to classify and construct in explicit terms. In principle
 we can settle the question
with a finite amount of work since we do have
explicit integral representation of the solutions \cite{haraoka},
but we shall not attempt to do it here.    

The following statement is proven 
in appendix \ref{8888bbbv}:

\begin{lem}\label{ttttt432} In a (non-isotrivial) class-$E_7$ special
 geometry the short $\delta=2$ branch cannot correspond to  a
non-enhanced divisor.
\end{lem}

Thus the $\delta=2$ branch represents a  coordinate axis in the Coulomb branch $\mathscr{C}$ which should be enhanced but cannot be part of the discriminant because the $\delta=2$ branch is too short to describe an enhanced discriminant. Working mod 2 we get:
\begin{fact}\label{pppp9cxb} In a class-$E_7$ special geometry the
 Kodaira types $I_n$ and $I^*_{2k}$ ($k\in\mathbb{N}$) are \textbf{excluded} for the non-enhanced (i.e.\! knotted) discriminant associated to the $\delta=3$ branch.
\end{fact}
Now we are ready to show:

\begin{fact}\label{pppp9cx} A rigid non-isotrivial (irreducible) geometry of class $E_7$
whose two dimensions $\{\Delta_1,\Delta_2\}$ are both \textbf{non new},
must have Coulomb dimensions $\{4,6\}$.
\end{fact}

\begin{cave} There are plenty of rigid non-isotrivial irreducible special geometries with dimensions $\{4,6\}$ which \emph{are not} of class $E_7$ but rather of class $\widehat{E}_7$.
Note that $\{4,6\}$ is the same ``troublesome'' dimension pair which allows for quasi-isotrivial geometries.
\end{cave}

 \begin{rem} There is one \emph{known} class-$E_7$ SCFTs with dimensions $\{4,6\}$
 and a non-enhanced divisor of type $I^*_1$ (model $\#66$);
 in addition one of the two possible geometries for the model $\#22$ (the non-cover one)
  would be
of class-$E_7$ with dimensions $\{4,6\}$ and a non-enhanced divisor of type $II$. 
Both are consistent with the above \textbf{Facts}. 
 \end{rem}

\begin{proof}[Argument] We call $\Delta_1$ the dimension of the regular axis corresponding to the $\delta=2$ branch, and $\Delta_2$ the other one (beware: $\Delta_1$ may be actually larger than $\Delta_2$). $\Delta_1$ is an old dimension $\neq2$. Applying the dimension formulae of \S.\,\ref{jjuutreq} to the regular axis of dimension $\Delta_1$ we get
\be
\Delta_2= 2+k\,\Delta_1\quad k\in\Z\setminus \{0\}.
\ee
When $\Delta_2$ is a \emph{new} dimension the other axis is also regular, so we must have a 
non-enhanced discriminant which must be non-semisimple for the geometry not to be isotrivial.
The non-enhanced discriminant should be associated with the $\delta=3$ branch and hence be of type $I^*_n$ ($n>0$). In this case we can apply the dimension formulae of \S.\,\ref{jjuutreq} also to the axis of dimension $\Delta_2$ getting a list of allowed values for $\Delta_1$ which should agree with the original one.
One gets the table
{\renewcommand{\arraystretch}{1.2}\be
\begin{tabular}{ccc@{\hskip20pt}l}
sequence & $\begin{smallmatrix}\text{rank-2 dimensions}\\ \text{in the sequence}\end{smallmatrix}$ &
$\begin{smallmatrix}\text{old dimensions $\Delta_1$}\\ \text{consistent with $\Delta_2$}\end{smallmatrix}$ & notes\\\hline
$2+\tfrac{6}{5}\,\Z$ & 8 & 4, 6 & $\Delta_1=\tfrac{6}{5}$ ruled out\\\hline
$2+\tfrac{4}{3}\,\Z$ & 6 &  & $\{\tfrac{4}{3},6\}$ allowed\\\
& 10 & 4 & $\{\tfrac{4}{3},10\}$ ruled out\\\hline
$2+\tfrac{3}{2}\,\Z$ & 8 & 4, 6 & $\Delta_1=\tfrac{3}{2}$ ruled out\\\hline
$2+3\,\Z$ & 5 & 3, 4 & $\{3,5\}$ \emph{not} ruled out\\
& 8 & 4, 6 & $\{3,8\}$ ruled out\\\hline
$2+4\,\Z$ & 6 &  & $\{4,6\}$ allowed \\
& 10 & 4 & $\{4,10\}$ \emph{not} ruled out\\\hline
$2+6\,\Z$ & 8 & 6, 4 & $\{6,8\}$ \emph{not} ruled out\\
\end{tabular} 
\ee }
\hskip-4pt Under the assumptions in  \textbf{Fact \ref{pppp9cx}}, $\Delta_2$ is an old dimension and
there are only two allowed dimension pairs with these properties
\be
\left\{\frac{4}{3},6\right\}\quad\text{and}\quad \{4,6\}.
\ee
The first one has $d_2=9$ which is not acceptable in rank-2.
Thus we remain with $\{4,6\}$.
\end{proof}

\subsubsection{Class-$E_7$ geometries with new dimensions?}\label{s:nnnewddd}

The proof of \textbf{Fact \ref{pppp9cx}}
leaves open the possibility of non-isotrivial class $E_7$-geometries of 
with an \emph{irreducible} non-enhanced discriminant of type $I^*_{2k+1}$ and dimensions in the list\footnote{\ Recall that when a rank-2 geometry with one dimension equal 12 is automatically isotrivial.}
$\{3,5\}$, $\{4,10\}$, $\{6,8\}$. 
All known SCFTs with dimensions $\{3,5\}$, $\{4,10\}$ and $\{6,8\}$ have geometries of class $\widehat{E}_7$ and one wonders if any class-$E_7$ geometry with a \emph{new} dimension exists. 

However we can write an explicit integral, rigid, irreducible monodromy representation 
of class-$E_7$ which is appropriate to dimensions $\{3,5\}$ without making any new computation.
Just take as $\mu_1=-R$  where $R$ is the matrix defined in eq.\eqref{Rmat}
and as $\mu_2$ the second matrix in eq.\eqref{iiiiii7zz}. Now
\be\label{Exo1}
\mu_1^t\,\Omega\,\mu_1=\mu_2^t \,\Omega\,\mu_2=\Omega,\quad \mu_1^5=\mu_2^3=\boldsymbol{1}
\ee   
while
\be
\mu_3\overset{\rm def}{=} (\mu_1\mu_2)^{-1}=\left(\begin{smallmatrix}1 & 1 & 0 &-1\\
-1 & 0 & 1 &-1\\
0 &1 & 1 &-1\\
-1 & 1 & 1 &-2
\end{smallmatrix}\right)
\ee
which has a Jordan block of size 2 associated to the eigenvalue $-1$ and 2 trivial Jordan blocks
associated to the eigenvalue $+1$. Moreover
\be\label{Exo2}
\mu_3^2=\boldsymbol{1}+ 2\,w\otimes w^t\,\Omega\quad\text{where}\quad w^t=\big(0\ 1\ 0\ 1\big),
\ee
so that $\mu_3$ is an integral symplectic monodromy of Kodaira type $I_1^*$.
Thus a special geometry of dimensions $\{3,5\}$ with a single discriminant component of type $I_1^*$
potentially may exist. 

\subparagraph{Regularity?} Local regularity along the axes poses no problem for the putative
$\{3,5\}$ geometry of class-$E_7$.
We have $h=u_1^2/u_2$ and $z=u_2^5/u_2^3$. Along the first axis
\be
a^\|\sim h\cdot z^{1/3}=u_1^{1/3},\qquad a^\perp\sim h\cdot z^{2/3}=u_2\cdot u_1^{-4/3}
\ee
along the second axis
\be
a^\|\sim h\cdot z^{2/5}=u_2^{1/5},\qquad 
a^\perp\sim h\cdot z^{1/5}=u_1\cdot u_2^{-2/5}.
\ee
The only mechanism which may prevent this geometry to exist is a failure of the
connection formula i.e.\! the locally good solution around one axis is not the
analytic continuation of the one which is good at the other axis. It looks a very unlikely possibility, so here we have a promising candidate for a novel special geometry which was not previously known.

\subsection{Rigid special geometries of class $\widehat{E}_7$}\label{s:ehat}

The Fuchsian ODE with rigid monodromy and graph $\widehat{E}_7$
is Thomae's order-4 hypergeometric equation whose monodromy is
well understood thanks to Levelt \cite{levelt} and Beukers-Heckman \cite{BHeck}.
In particular we know that
the Brauer class is always trivial. We write $\mu_0$, and $\mu_\infty$
for the local $\mu$-monodromies associated to the two $\delta=4$ branches
and $\mu_1$ for the one associated to the non-enhanced divisor i.e.\! to the $\delta=1$ branch of $\mathfrak{G}$. 
We absorb the possible sign in the lift to $Sp(4,\Z)$
 in
$\mu_\infty$, so that the relation $\mu_0\mu_1\mu_\infty=\boldsymbol{1}$ holds.
%
%

We write $X_a$ ($a=0,\infty$) for, respectively, the matrices $\mu_0$ and $\mu_\infty^{-1}$.
The fact that the conjugacy classes of the matrices $X_0$ and $X_\infty$ are described by $\delta=4$  
branches is equivalent to the statement that their minimal polynomials have degree 4.
In this case, up to conjugacy, they may be set in the form
\be\label{llelele}
X_a=\left(\begin{smallmatrix}0 & 0 & 0 & c_{0,a}\\
1 & 0 & 0 & c_{1,a}\\
0 & 1 & 0 & c_{2,a}\\
0 & 0 & 1 & c_{3,a}\end{smallmatrix}\right)\quad a=0,\infty
\ee
where the $c_{j,a}$'s are the coefficients in the characteristic polynomial of $X_a$:
\be\label{charrr}
P_a(z)\equiv\det[z-X_a]=z^4-c_{3,a}\,z^3-c_{2,a}\,z^2-c_{1,a}\,z-c_{0,a}.
\ee
Since $X_a\in Sp(4,\Z)$ the coefficients are integral and
$c_{0,a}=-1$ and $c_{3,a}=c_{1,a}$. Further restrictions on the $c_{i,a}$
are obtained by the
condition that the \textsc{rhs} of \eqref{charrr} is a product of cyclotomic polynomials.
From now on we assume that the characteristic polynomials of $X_0$, $X_\infty$ satisfy these conditions.

Suppose both $X_0$, $X_\infty$ are in the form \eqref{llelele} up to a common conjugation. In this case
\be
\begin{gathered}
X_1\equiv X_0^{-1}X_\infty=\boldsymbol{1}_4+ V\otimes W^t,\qquad\text{with }\ 
V=\left(\begin{smallmatrix}A\\ B\\ A \\ 0
\end{smallmatrix}\right),\quad W=\left(\begin{smallmatrix}0\\ 0\\ 0\\ 1\end{smallmatrix}\right)\\
\text{where}\qquad A=c_{1,\infty}-c_{1,0},\qquad B=c_{2,\infty}-c_{2,0}
\end{gathered}
\ee
is automatically in the conjugacy class described by a $\delta=1$ with all its eigenvalues are equal $1$.  A $\mu$-representation with these local monodromy classes is automatically rigid, and the three matrices
$\mu_0\equiv X_0$, $\mu_1\equiv X_1$ and $\mu_\infty=X_\infty^{-1}$
should generate the \emph{unique} monodromy representation up to conjugacy in $GL(4,\C)$.
This statement is the Levelt theorem \cite{levelt}. Since in our case the coefficients of the characteristic polynomials are rational integers, the Levelt representation is automatically integral and the Brauer class is trivial. The Levelt representation is $Sp(4,\Z)$-valued and irreducible iff the two characteristic polynomials, $P_0(z)$ and $P_\infty(z)$,
satisfy the following conditions:
\begin{itemize}
\item[(a)] $(P_0(z),P_\infty(z))$ is a pair of degree 4 products of cyclotomic polynomials
$\prod_{n}\Phi_n(z)^{m_n}$ with the property $P_a(z)=z^4\, P_a(1/z)$ (for $a=0,\infty$);
\item[(b)] each polynomial has at most one double root (corresponding to at most one non-trivial Jordan block of size 2) which then must be $\pm1$;
\item[(c)] $P_0(z)$ and $P_\infty(z)$ have no common root by \textbf{Fact \ref{distinctr}}.
In other words, their resultant must be non-zero, $\mathsf{Res}(P_0(z),P_\infty(z))\neq0$;
\item[(d)] there should exist a pair of \emph{coprime} integers $(d_1,d_2)$ such that a pair $(X_0,X_\infty)$ of $4\times 4$ matrices
with minimal polynomials $(P_0(z),P_\infty(z))$ have the property that $X_0^{d_1}$, $X_\infty^{d_2}$ have two trivial Jordan blocks of eigenvalue $\pm1$. This condition eliminates, say, the troublesome polynomial
pair $(\Phi_1^2\Phi_6,\Phi_2^2\Phi_3)$;
\item[(e)] we can disregard $\Phi_{12}(z)$ since it is consistent only with isotrivial geometries
and two $\Z$-irreducible polynomials (i.e.\! two new dimensions)
are possible only if they are $\Phi_8$ and $\Phi_{10}$ or $\Phi_5$ (cf.\! \S.\,\ref{s:sample}). 
\end{itemize}

\begin{defn} The three integral matrices $X_0$, $X_1$ and $X_\infty$
will be called the \emph{Levelt monodromy} (of class-$\widehat{E}_7$)
associated to the pair of polynomials $(P_0(z),P_\infty(z))$ which satisfies the above restrictions.
\end{defn}

The Kodaira type of $X_1$ is $I_f$ where the integer $f$ is
\be
f\equiv\gcd(A, B).
\ee
In the list of pairs $(P_0(z),P_\infty(z))$ satisfying conditions (a),(b),(c) only $(\Phi_1^2\Phi_6,\Phi_2^2\Phi_3)$ 
has $f\not\in\{1,2\}$ in facts $f=6$. But this ``bad'' pair is ruled out by condition (d).

\subsubsection{Symplectic polarization}
With our restrictions on the characteristic polynomials $(P_0(z),P_\infty(z))$,
the Levelt matrices $X_a$ preserve the symplectic form
{\renewcommand{\arraystretch}{1.1}\be
\check{\Omega}=\begin{pmatrix}0 & A & -B & C\\
-A & 0 & A & -B\\
B &- A & 0 & A\\
-C  & B & - A  &0\end{pmatrix}\quad\text{where}\quad
C= A(1+c_{2,0})-B\, c_{1,0}
\ee}
with
\be
\mathsf{Paf}(\check{\Omega})=A^2-B^2+AC,\qquad
\det(\check{\Omega})\equiv \mathsf{Paf}(\check{\Omega})^2
=\mathsf{Res}\big(P_0(z),P_\infty(z)\big).
\ee
Note that $\check{\Omega}$ is $f$ times an integral matrix with the GCD
of its entries 1. In particular
\be
f^2\mid \mathsf{Paf}(\check{\Omega}),\quad\text{and}\quad \check{\Omega}\,V=-\mathsf{Paf}(\check\Omega)\,W.
\ee
If $\check{\Omega}$ is principal (more generally $\mathsf{Paf}(\check{\Omega})$ is square-free), $f$ is automatically $1$ and the knotted
discriminant has Kodaira type $I_1$. In general $\check{\Omega}$ is $\Z$-equivalent to a symplectic matrix in the normal form \cite{matmat}
\be
f\!\left(\begin{smallmatrix}0 & 1 & 0 &0\\
-1 & 0 & 0 &0\\
0& 0& 0 &g\\
0 & 0 &- g&0\end{smallmatrix}\right),\qquad g\in\mathbb{N},\quad \mathsf{Paf}(\check\Omega)=f^2 g.
\ee 
The values of $A$, $B$, $C$, $f$ and $g$ for each pair of polynomials
satisfying the three conditions (a),(b),(c) are listed in table \ref{pairpoly}.

\begin{table}
\begin{scriptsize}
$$
\begin{array}{ccccccc||ccccccc}\hline\hline
(P_0(z),P_\infty(z)) & A & B & C & f & g & \text{allowed?} & (P_0(z),P_\infty(z)) & A & B & C & f & g & \text{allowed?}\\\hline
(\Phi_1^2\Phi_3,\Phi_2^2\Phi_4) &-3 & -2& -1& 1& 8& \#9,\#25&
 (\Phi_1^2\Phi_3,\Phi_2^2\Phi_6) & -2 & 0 &-2 & 2 & 2 & \text{no: dim}\\ 
 (\Phi_1^2\Phi_3,\Phi_4\Phi_6) &0& -2&2&2& 1& \text{no: reg}
 & (\Phi_1^2\Phi_3,\Phi_5) & -2 & -1 &-1 & 1& 5 &\#15,\#42\\
  (\Phi_1^2\Phi_3,\Phi_8) &-1 & 0&-1&1&2 &\#44&
(\Phi_1^2\Phi_3,\Phi_{10}) & 0 & -1 &1 & 1 & 1 & \text{no: dim}\\
(\Phi_1^2\Phi_3,\Phi_{12}) & -1 & 1 &-2& 1 & 2 &\text{no: dim}&
(\Phi_1^2\Phi_4,\Phi_2^2\Phi_3)& -5 & -2 &9& 1& 24 &\#25,\#37?,\#40?\\
 (\Phi_1^2\Phi_4,\Phi_2^2\Phi_6) &-3 & 2 &-1&1 &8 &\#5&
(\Phi_1^2\Phi_4,\Phi_3\Phi_6) &-2&1&0&1&3&\#26\\
 (\Phi_1^2\Phi_4,\Phi_5) & -3 & 1 &1& 1 & 5 &\#6,\#17&
(\Phi_1^2\Phi_4,\Phi_8) &-2 &2&-2&2&1&\text{no: dim}\\
(\Phi_1^2\Phi_4,\Phi_{10}) & -1&1 &-1&1&1&\#50?&
(\Phi_1^2\Phi_4,\Phi_{12}) & -2&3 &-4&1&3& \text{no: dim}\\
(\Phi_1^2\Phi_6,\Phi_2^2\Phi_3)&-6&0&18&6&2&\text{no: dim}&
 (\Phi_1^2\Phi_6,\Phi_2^2\Phi_4) &-5 &9&2 &1&24&\#5,\#23,\#40?\\
(\Phi_1^2\Phi_6,\Phi_3\Phi_4) &-4 &2&6&2&3&\text{no: reg}&
 (\Phi_1^2\Phi_6,\Phi_5) & -4 & 3 &3&1 & 5 &\#29\\ 
 (\Phi_1^2\Phi_6,\Phi_8) &-3 & 4&-3&1&2& \#2,\#19&
(\Phi_1^2\Phi_6,\Phi_{10}) &-2&3&-3& 1 &1 & \text{no: dim}\\
(\Phi_1^2\Phi_6,\Phi_{12}) &-3&5&-6& 1 &2 &\text{no: dim}&
 (\Phi_2^2\Phi_3,\Phi_4\Phi_6) &4 & 2&-6&2&3&\text{no: reg}\\
  (\Phi_2^2\Phi_3,\Phi_5) & 2 & 3 &3& 1 &1&\#8,\#46&
   (\Phi_2^2\Phi_3,\Phi_8) &3 &4&3&1&2&\#2,\#19\\
(\Phi_2^2\Phi_3,\Phi_{10})& 4 &3 &-3&1 & 5 &\text{no: dim}&
(\Phi_2^2\Phi_3,\Phi_{12})& 3 &5 &6&1 & 2 &\text{no: dim}\\
 (\Phi_2^2\Phi_4,\Phi_3\Phi_6) &2 &1&0&1&3&\#41?&
 (\Phi_2^2\Phi_4,\Phi_5) & 1& 1&1&1 &1&\#35?\#50\\ 
 (\Phi_2^2\Phi_4,\Phi_8) &2&2&2&2&1&\text{no: dim}&
(\Phi_2^2\Phi_4,\Phi_{10}) &3 &1 &-1&1 &5&\#50\\
(\Phi_2^2\Phi_4,\Phi_{12}) &2 &3 &4&1 &3&\text{no: dim}&
 (\Phi_2^2\Phi_6,\Phi_3\Phi_4) &0&-2&-2&2&1&\text{no: reg}\\
 (\Phi_2^2\Phi_6,\Phi_5) &0&-1&-1& 1&1&\text{no: dim}&
 (\Phi_2^2\Phi_6,\Phi_8)&1&0&1&1&2& \#44\\
(\Phi_2^2\Phi_6,\Phi_{10}) &2 &-1&1&1&5&\text{no: dim}&
(\Phi_2^2\Phi_6,\Phi_{12}) &1 &1&2&1&2&\text{no: dim}\\
(\Phi_3\Phi_6,\Phi_5) &-1&0&0&1&1&\#20&
 (\Phi_3\Phi_6,\Phi_8) &0 &1 &0& 1&1 &\text{cf. \S.\ref{yuyuy}}\\ 
 (\Phi_3\Phi_6,\Phi_{10})&1&0&0&1&1&\text{no: dim}&
  (\Phi_3\Phi_6,\Phi_{12})&0&2&0&2&1&\text{no: dim}\\
(\Phi_3\Phi_4,\Phi_5) &0 &1 &1&1 &1&\#28,&
(\Phi_3\Phi_4,\Phi_8) &1&2&1&1&2&\text{no: dim}\\
 (\Phi_3\Phi_4,\Phi_{10}) & 2 &1 &-1& 1 &1&\#3,\#59?&
  (\Phi_3\Phi_4,\Phi_{12}) & 1 &3 &2& 1 &6&\text{no: dim}\\
 (\Phi_4\Phi_6,\Phi_5) &-2&1&1&1&1&\#55?&
  (\Phi_4\Phi_6,\Phi_8) &-1 &2 &-1& 1 &2&\text{no: reg}\\
  (\Phi_4\Phi_6,\Phi_{10}) &0&1&-1&1&1&\#3,&
   (\Phi_4\Phi_6,\Phi_{12}) &-1&3&-2&1&6&\text{no: dim}\\
(\Phi_5,\Phi_8) &1&1&1&1&1&\text{no: dim}&
(\Phi_5,\Phi_{10}) &2&0&0&2&1& \text{no: dim}\\
(\Phi_5,\Phi_{12}) &1&2&2&1&1&\text{no: dim}&
(\Phi_8,\Phi_{10}) &1&-1&1&1&1& \#21\\
(\Phi_8,\Phi_{12}) &0&1&0&1&1&\text{no: dim}&
(\Phi_{10},\Phi_{12}) &-1&2&-2&1&1&\text{no: dim}\\\hline\hline
\end{array}
$$
\end{scriptsize}
\caption{\label{pairpoly}Pairs of degree-4 symmetric products of cyclotomic polynomials which have
at most one non-simple root of multiplicity 2 and no common zero.}
\end{table}

When the invariant $g=1$, the Levelt monodromy takes values in $Sp(4,\Z)$ and
is a candidate for the $\mu$-monodromy associated to a special geometry. 
Otherwise one should look for other monodromies, conjugate to the Levelt one
in $GL(4,\C)$, which do take value in $Sp(4,\Z)$. Moreover, as in \S.\,\ref{s:22}, there may be several $GL(4,\C)$-conjugate monodromies all taking value in $Sp(4,\Z)$, but not conjugate in $Sp(4,\Z)$, which  may correspond to inequivalent special geometries. Before going to the general case, we present
a few simple examples with $f=g=1$.

\subsubsection{Example 1: the $\{8/7,10/7\}$ geometry}

We know from \S.\,\ref{s:sample} that a non-isotrivial geometry in rank-2
with two \emph{new} dimensions must have dimensions $\{8/7,10/7\}$
(this is the geometry of the SCFT $\# 21$ in the list \cite{Martone}). The axes are regular enhanced with 4 distinct
eigenvalues so the  geometry must be of class $\widehat{E}_7$.
The polynomial pair is\footnote{\  $(\Phi_8,\Phi_{5})$ yields a different lift of the same $PSp(4,\Z)$ monodromy.}
$(\Phi_8,\Phi_{10})$. We have
\be
c_{1,0}=c_{2,0}=0,\quad c_{1,\infty}=-c_{2\infty}=1.
\ee
and we get
\be
\Omega=\left(\begin{smallmatrix}0 & -1 & -1 & -1\\
1 & 0 & -1 & -1\\
1 & 1 & 0 &- 1\\
1 & 1 & 1 & 0\end{smallmatrix}\right)
\ee
which has $f=g=1$ so that the Levelt monodromy is valued in $Sp(4,\Z)$.
By a change of basis we put the polarization in the standard form
\be
\mu_0=\left(\begin{smallmatrix}0 & 1 &1&0\\ 1 & 1&1 &1\\
-2& 0 & -1 & -1\\ 1 &-1& 0 &1\end{smallmatrix}\right)\quad \mu_1=
\left(\begin{smallmatrix}1 & 0 &0&0\\ 0 & 1&0 &0\\
-1& 0 & 1 & 0\\ 0 &0& 0 &1\end{smallmatrix}\right)\quad
\mu_\infty=\left(\begin{smallmatrix}-1 & 0 &-1&-1\\ -1 & 1&0 &-1\\
1& -1 & -1 & 0\\ 0 &1& 1 &1\end{smallmatrix}\right)\quad
\Omega=\left(\begin{smallmatrix}0 & 0 &1&0\\ 0 & 0&0 &1\\
-1& 0 & 0 & 0\\ 0 &-1& 0 &0\end{smallmatrix}\right)
\ee
which satisfy
\be
\begin{aligned}
&\mu_0^4-\mu_0^3+\mu_0^2-\mu_0+\boldsymbol{1}=0, &\quad&(\mu_1-\boldsymbol{1})^2=0, &\quad&\mu_\infty^4=-\boldsymbol{1},\\
&\mu_0\mu_1\mu_\infty=\boldsymbol{1}, && \mu_a^t\,\Omega\, \mu_a=\Omega &&\text{for }a=0,1,\infty.
\end{aligned}
\ee
It is clear that the (knotted) discriminant is of type $I_1$.

\subparagraph{Local regularity along the axes.} As a further illustration we
 check local regularity (which must hold since we know that this geometry exists).
We have $\varkappa=2$ with $h^2=(u_2/u_1)^7$ and $z=u_1^5/u_2^4$.
Along the first axis ($u_1=0$) the exponents are $7/8,5/8$ so:
\begin{align}
&a^\|\sim\sqrt{h^2}\, z^{7/8}=u_1^{7/8}\equiv u_1^{1/\Delta_1}, &\quad&
a^\perp\sim \sqrt{h^2}\, z^{5/8}= u_2\cdot u_1^{-3/8}
\intertext{along the second axis ($u_1=0$) the exponents are $7/10,9/10$ and}
&a^\|\sim\sqrt{h^2}\, z^{7/10}=u_2^{7/10}\equiv u_2^{1/\Delta_2}, &&
a^\perp\sim \sqrt{h^2}\, z^{9/10}= u_1\cdot u_2^{-1/10}.
\end{align}

\subsubsection{Example 2: $\{3,5\}$ geometries with one $I^*_n$ discriminant}
This geometry corresponds to SCFTs $\#8$, $\#36$, $\#42$ and $\#46$.
The polynomials are $(\Phi_2^2\Phi_3,\Phi_5)$ which give $f=g=1$, so the Levelt
monodromy is in $Sp(4,\Z)$ and the knotted discriminant is of type $I_1$.
The other component of the discriminant is along the axis of dimension 3
and is associated with the $\mu$-monodromy
\be
X_0=\left(\begin{smallmatrix}0 & 0 &0 &-1\\
1 &0&0&-3\\
0& 1 & 0 &-4\\
0 &0 &1 &-3\end{smallmatrix}\right),\qquad (X_0+\boldsymbol{1})^2(X_0^2+X_0+\boldsymbol{1})=0.
\ee
The corresponding $\varrho$-monodromy in $\mathscr{C}$, $\varrho_0\equiv X_0^3$
has one Jordan block of size 2 associated to the eigenvalue $-1$ and two trivial Jordan blocks with eigenvalues $+1$. It satisfies
\be
\varrho_0^2\equiv(X_0^3)^2=\boldsymbol{1}+6\,v\otimes v^t\,\check{\Omega},\qquad v\equiv(1\ 2\ 2\ 1)^t,
\ee 
  which shows that this enhanced discriminant has Kodaira type $I_3^*$ in agreement with physical expectations.
  
  \subparagraph{Local regularity.} This geometry has $\varkappa=1$, $h=u_1^2/u_2$
  and $z=u_2^3/u_1^5$. Along the second axis the exponents are $2/5$ and $1/5$
  \be
  a^\|\sim h\cdot z^{2/5}=u_2^{1/5},\qquad a^\perp \sim h\cdot z^{1/5}=u_1\cdot u_2^{-2/5}
  \ee
  along the first axis ($u_2=0$) the exponents are $1/3$ and $1/2$ so that
  \be
  a^\|\sim h\cdot z^{1/3}=u_1^{1/5},\qquad a^\perp\sim h\cdot z^{1/2}\sim u_2^{1/2}\cdot u_1^{-1/2}
  \ee
  which is the expected behaviour for a discriminant of type $I^*_n$ ($n\geq0$).

  \subsubsection{Example 3: $\{4,10\}$ geometry with a semisimple enhanced discriminant}
  This is the geometry of the SCFT $\#3$. This geometry has $\varkappa=2$
  so we have a little ambiguity in the $Sp(4,\Z)$ lift. The pair of polynomials may be $(\Phi_4\Phi_6,\Phi_{10})$
  or $(\Phi_4\Phi_6,\Phi_5)$. Both pairs have
$f=g=1$ so that the Levelt monodromy is already in $Sp(4,\Z)$. For, say, the first pair
we put the polarization in the standard form
  \be
  \Omega\equiv\left(\begin{smallmatrix}0 & 0 & 1 &0\\
  0 & 0 &0 &1\\ -1 & 0 & 0 &0 \\ 0 & -1 & 0 & 0\end{smallmatrix}\right)=- S^t\,\check{\Omega}\,S,\quad \text{where}\quad S=\left(\begin{smallmatrix}1 & 0 & 0 &0\\
  0 & 1 &0 &1\\ 0 & 0 & 1 &-1 \\ 0 & 0 & 0 & 1\end{smallmatrix}\right)
  \ee
  In this canonical basis the monodromy around the
  axis of dimension 4 reads
  \be
  \varrho_0= - S^{-1}X_0^2\,S =\left(\begin{smallmatrix}0 & 0 & 1 &0\\
  0 & 1 &0 &0\\ -1 & -1 & 1 &0 \\ 0 & -1 & -1 & -1\end{smallmatrix}\right) 
  \ee
  which is semisimple with eigenvalues $(\zeta_6,\zeta_6^{-1},1,1)$
  i.e.\! a discriminant of type $II^*$ or $II$. The regularity of this example 
  was already discussed in \S.\ref{s:eeeexam}; as shown there 
regularity fixes the signs in the lift to $Sp(4,\Z)$
while the axis of dimension 4 must be a discriminant of type $II^*$, as physically expected. 
%
%

  \subsubsection{Example 4: the $\{5/4,3/2\}$ geometry}
 This is the only other non-isotrivial geometry in the list of examples \cite{Martone} with an irreducible
 discriminant. It corresponds to the polynomials $(\Phi_3\Phi_6,\Phi_5)$ which have
 $f=g=1$ so that the Levelt monodromy is in $Sp(4,\Z)$. This geometry has $\varkappa=1$,
 $h=u_2^4/u_1^4$, $z=u_1^6/u_2^5$. The exponents for the first axis are $4/5$ and $3/5$
 so that
 \be
 a^\|\sim h\cdot z^{4/5}=u_1^{4/5},\qquad a^\perp\sim h\cdot z^{3/5}=u_2\cdot u_1^{-2/5}
 \ee
 while around the second axis the exponents are $2/3$ and $5/6$ which gives
 \be
 a^\|\sim h\cdot z^{2/3}=u_2^{2/3},\qquad a^\perp\sim h\cdot z^{h/6}=u_1\cdot u_2^{-1/6}.
 \ee

\subsection{$\widehat{E}_7$ class: non-Levelt principal polarizations}

The Levelt integral monodromy is a solution to our inverse problem only when
$\check{\Omega}$ is an integral multiple of a principal integral matrix, that is, when
$g=1$. 
Otherwise we have to search for representations
conjugate to the Levelt one in $GL(4,\C)$ which take value in $Sp(4,\Z)$
 much as we did in \S.\,\ref{s:22} for the Pochhammer ODE monodromy.
As in that discussion, there may be more than one inequivalent 
$Sp(4,\Z)$-valued representations for a given rigid $GL(4,\C)$-monodromy.
Even when the Levelt monodromy is principal it may be not the only one
valued in $Sp(4,\Z)$ for the given polynomial pair.

\subsubsection{Example: a $\{6,8\}$ geometry}

Consider the following two pairs of polynomials 
$(\Phi_2^2\Phi_6,\Phi_8)$ and $(\Phi_1^2\Phi_3,\Phi_8)$ which have $f=1$ and $g=2$. The first one is naturally interpreted as describing
$\varkappa=2$ geometry with dimensions $\{6,8\}$ where the axis of
dimension 6 has type $I^*_{n>0}$. The monodromy of the second pair
is obtained from the first one 
by a sign flip $\mu_0\to -\mu_0$, $\mu_\infty\to -\mu_\infty$ so it gives a different lift of the same
$PSp(4,\Z)$ monodromy. We have
\be
-2\, \Omega\equiv \left(\begin{smallmatrix} 0 & -2 & 0 & 0\\
2 & 0 &0 &0\\
0& 0 & 0 &-2\\
0 & 0 & 2 & 0\end{smallmatrix}\right) =S^t \,\check{\Omega}\,S,\qquad\text{where}\quad
S= \left(\begin{smallmatrix} 2 & 1 & 1 & 1\\
0 & 1 &0 &-1\\
0& 0 & 1 &1\\
0 & 0 & 0 & 1\end{smallmatrix}\right)
\ee
so that conjugating the Levelt representation with $S^{-1}$ we get the three matrices
\be
\mu_0=\left(\begin{smallmatrix} -1 & -1 & -1 & 0\\
2 & 1 &2 &0\\
0& 1 & -1 &-1\\
0 & 0 & 1 & 0\end{smallmatrix}\right)\qquad
\mu_1=\left(\begin{smallmatrix} 1 & 0 & 0 & 0\\
0 & 1 &0 &0\\
0& 0 & 1 &1\\
0 & 0 & 0 & 1\end{smallmatrix}\right)\qquad
\mu_\infty=\left(\begin{smallmatrix} 1 & 1 & 0 & -1\\
-2 & -1 &0 &0\\
2& 1 & 1 &2\\
-2 & -1 & -1 & -1\end{smallmatrix}\right)
\ee
which satisfy
\be
\begin{gathered}
\mu_0\,\mu_1\,\mu_\infty=\boldsymbol{1},\quad \mu_a^t\,\Omega\,\mu_a=\Omega\ \ (a=0,1,\infty)\\
\mu_\infty^4+\boldsymbol{1}=(\mu_0^2+\boldsymbol{1})^2(\mu_0^2-\mu_0+\boldsymbol{1})=0
\end{gathered}
\ee
One has
\be
\mu_0^6=\boldsymbol{1}+4 v\otimes v^t\Omega\qquad v=\left(\begin{smallmatrix} 0\\ 1\\ -1\\ 1
\end{smallmatrix}\right)
\ee
which means that the axis of dimension 6 is a discriminant of type $I_2^*$
as expected for model $\#44$ of \cite{Martone}. There should be another 
$GL(4,\C)$ equivalent but not $Sp(4,\Z)$-equivalent solution which describes
the geometries of the models $\#2,\#33,\#39$ where the type is $I^*_6$. Presumably they correspond
to another $Sp(4,\Z)$ inequivalent integral realization of the same $GL(4,\C)$ monodromy.

\subsection{$\widehat{E}_7$ class: Kodaira type of knotted discriminants}

\begin{fact} All knotted discriminants of class $\widehat{E}_7$ geometries  are of Kodaira type $I_1$.
\end{fact}
This is confirmed from the known examples \cite{Martone}. 

\begin{proof}[Argument]
Any integral symplectic monodromy
which is conjugate in $GL(4,\C)$ to the Levelt one has the local monodromomy $\mu_1$ associated to the
$\delta=1$ branch of type $I_n$ where $n\mid f$. Indeed $\mu_1=\boldsymbol{1}\bmod n$
and hence
\be
\mu_\infty =(\mu_0)^{-1}\bmod n
\ee   
and $P_\infty(z)=P_0(z)\bmod n$, while $f$ is the largest integer such that
$P_\infty(z)=P_0(z)\bmod f$. Looking to the table \ref{pairpoly} we see that $f=1$
except for the pair $(\Phi_1^2\Phi_6,\Phi_2^2\Phi_3)$ which has $f=6$
and seven pairs which have $f=2$. 
The pair $(\Phi_1^2\Phi_6,\Phi_2^2\Phi_3)$ cannot correspond to
any regular special geometry since both axes would be non-semisimple, and all
consistent assignments of dimensions would lead to $3\mid\varkappa$ which implies
an isotrivial geometry which is inconsistent with the presence of 3 non-semisimple discriminants.
Equivalently this pair is ruled out by requirement (d).
Three of the pairs with $f=2$, namely
$(\Phi_1^2\Phi_3,\Phi_2^2\Phi_6)$, $(\Phi_1^2\Phi_4,\Phi_8)$,
$(\Phi_2^2\Phi_4,\Phi_8)$ should also be excluded because all consistent dimension
assignments lead to $\varkappa>2$ (requirement (d)). The four remaining pairs with $f=2$ are
\be\label{uuyuuuu7}
\begin{aligned}
&(\Phi_1^2\Phi_6,\Phi_3\Phi_4) &&(\Phi_2^2\Phi_6,\Phi_3\Phi_4)
&& (\Phi_1^2\Phi_3,\Phi_6\Phi_4)
 && (\Phi_2^2\Phi_3,\Phi_6\Phi_4).
\end{aligned}
\ee
In a would-be special geometry associated with one of these pairs
the first axis is a non-semisimple discriminant of dimension in the list $\{3/2,6/5,3,6\}$.
The second axis is a semisimple discriminant: \emph{a priori} the Albanese variety
may have an automorphim of degree 3 or 4, but the first possibility would lead to
$3\mid\varkappa$, hence to a contradiction. Let us focus, say, on the first pair.
The dimensions pairs with $1/\Delta_i$ of the correct order in $\mathbb{Q}/\Z$
are $\{\Delta_1,\Delta_2\}=\{6,4\}$, $\{6/5,4\}$, $\{6,4/3\}$, and $\{6/5,4/3\}$
but the fractional dimensions may be ruled out by the values of the degree $(d_1,d_2)$.
We remain with  $\{6,4\}$ where the second axis of type
$IV^*$, $II^*$ (or, less likely, $IV$, $II$). Then $\varkappa=2$, $h^2= u_1/u_2$, and $z=u_2^3/u_1^{2}\bmod1$.
Along the second axis the exponents are $\pm1/4$ and $\pm1/3$,
then as $u_1\to 0$
\be
\begin{aligned}
&a_\|\sim \sqrt{h^2}\, z^{1/4} = \frac{u_1^{1/2}}{u_2^{1/2}}\cdot\frac{u_2^{3/4}}{u_1^{1/2}}= u_2^{1/4}\\
\text{either}\quad &a_\perp\sim \sqrt{h^2}\, z^{1/3} = \frac{u_1^{1/2}}{u_2^{1/2}}\cdot\frac{u_2}{u_1^{2/3}}= u_2^{1/2}\,u_1^{-1/6}\\
\text{or}\quad& a_\perp\sim \sqrt{h^2}\, z^{-1/3} = \frac{u_1^{1/2}}{u_2^{1/2}}\cdot\frac{u_1^{2/3}}{u_2}= u_2^{-3/2}\,u_1^{7/6}
\end{aligned}
\ee
and neither expression of $a_\perp$ is consistent with a regular total space with a singular fiber of type $IV^*$ or $IV$ (nor $II^*$, $II$).
 The same argument applies to the
second pair which has the same putative dimensions. For the other two pair
one has dimensions $\{3,4\}$, $\varkappa=1$, $h=u_2/u_1$,
$z=u_1^4/u_2^3$ and as $u_1\to 0$
\be
a_\|\sim h\,z^{1/4}=u_2^{1/4},\qquad a_\perp\sim\begin{cases} h\,z^{1/6}=u_1^{-1/3} u_2^{1/2}\\
h\,z^{-1/6} = u_1^{-5/2} u_2^{3/2}
\end{cases}
\ee  
Neither expression of $a_\perp$ is consistent with a bona fide special geometry.
Therefore all consistent pairs have $f=1$ and a fortiori $n=1$.
\end{proof}

\subsection{The puzzling SCFT $\#45$}\label{puzzle:s}

The SCFT $\#45$ in the list \cite{Martone} is puzzling on the physical side for reasons explained in. that reference. Its most plausible
(or less implausible) properties are listed in \cite{Martone}: Coulomb dimensions $\{4,6\}$,
a knotted discriminant of type $I_1$, the axis of dimension 6 is a discriminant of type $I_2$,
and the one of dimension 4 is a semisimple discriminant of (possibly) type $II^*$.
This list of properties from physical considerations  look a bit puzzling from the geometrical side too.
Let us summarize the situation.

If we assume irreducibility, the presence of a knotted $I_1$ discriminant suggests
a class-$\widehat{E}_7$ geometry. This requires the characteristic polynomials of $\mu_0$
and $\mu_\infty$ to be coprime, hence in the list of table \ref{pairpoly}. From the dimensions,
the fact that $\mu_\infty$ has a non-trivial Jordan block while $\mu_0$ is semisimple
will suggest as candidate polynomials pairs (taking into account the various $Sp(4,\Z)$ lifts)
precisely the ones in eq.\eqref{uuyuuuu7} which are not consistent with regularity. Moreover the most natural pair $(\Phi_2^2\Phi_6,\Phi_3\Phi_4)$ has a Levelt monodromy in $Sp(4,\Z)$ which, however,
has a knotted discriminant of the wrong type: $I_2$ instead of $I_1$.

To solve the conundrum we use regularity \emph{backwards:} we assume the minimal polynomial
of $\mu_\infty$ to be $\Phi_2^2\Phi_6$, since this identification is not problematic, and ask for which
characteristic polynomial of $\mu_0$ we would have both regularity \emph{and}
the physically expected properties. In particular we already know that the polynomial $\Phi_0(z)$
should have the form $\Phi_4(z)\,Q(z)$ for some degree-2 \emph{cyclotomic} polynomial
$Q(z)$. The purported geometry has
dimensions $\{4,6\}$, so $\varkappa=2$, $h^2=u_2/u_1$
and $z= u_1^3/u_2^2$. Assuming the first axis ($u_2=0$) is a discriminant of
type $II^*$ we must have
\be
a^\|\sim \sqrt{h^2}\cdot z^{\alpha_\|}=u_1^{1/4},\qquad
a^\perp\sim \sqrt{h^2}\cdot z^{\alpha_\perp}= u_2^{1/6} 
\ee
for some exponents $\alpha_\|$ and $\alpha_\perp$ of $\mu_0$.
The first relation yields $\alpha_\|=1/4$ consistently with the factor $\Phi_4(z)$
of $P_0(z)$. The second relation yields $\alpha_\perp=1/6$ which implies
$Q(z)=\Phi_6$. We conclude
\be
P_0=\Phi_4\Phi_6,\qquad P_\infty=\Phi_2^2\Phi_6,\qquad \gcd(P_0,P_\infty)=\Phi_6,
\ee
which leads to a reducible geometry. We have studied them
 in \S.\,\ref{reddd}. Since this geometry is rigid, it should be one of the
 \emph{quasi-isotrivial} possibilities listed in \S.\,\ref{SS:quasi-isotrivial} (see also appendix \ref{A:quasi}). Indeed the geometric data of model $\#45$ exactly matches with \textbf{Fact \ref{uuuu63}(3)}:
 same dimensions, same Kodaira types of the three special divisors.
 The same conclusion will follows from \textbf{Fact \ref{777xw}}.

  However unlikely this situation
appears, it is not ruled out by our analysis since $\#45$ does have all 
the characteristics expected for a quasi-isotrivial geometry. While this is certainly not the last word about the geometry of model $\#45$, it gives a plausible geometric
scenario consistent with the physical intuition.
 
 \begin{scen} The special geometry of the SCFT $\#45$ is quasi-isotrivial.
 \end{scen}

\subsection{$\widehat{E}_7$ class: allowed polynomial pairs}\label{yuyuy}

A part for the case of the troublesome dimension pair $\{4,6\}$
(and perhaps its close cousin $\{3,4\}$), assuming $\mu$-rigidity
the classification of rigid (non-cover) non-isotrivial geometries
is reduced to: \textit{(i)} listing the allowed pairs of coprime polynomials $(P_0,P_\infty)$
\textit{(ii)} checking local regularity,
then \textit{(iii)} for each such pair solve the Diophantine problem of finding
all $S\in \C^\times\backslash GL(4,\C)/Sp(4,\Z)$ such that $S^t\check{\Omega}S$ is a multiple of 
an integral skew-symmetric matrix of determinant 1 while $\mu_a=S^{-1}X_a S$ ($a=0,1,\infty$)
are integral matrices (here $X_a$ are the Levelt matrices).

The results of part \textit{(i)} of the program are listed in the last column of table \ref{pairpoly}.
There are various obvious reasons to rule out a polynomial pair. The first one is the non-existence
of a consistent pair of dimensions $\{\Delta_1,\Delta_2\}$ potentially associated to the given pair.
If, say, $\Delta_2$ is a new dimension the list of allowed $\Delta_1$ is short and this
allows to eliminate all pairs which do not correspond to an allowed set of dimensions.
Then one eliminates all pairs which would correspond to dimensions with $\varkappa\not\in\{1,2\}$
since the corresponding geometries would be isotrivial which is in contradiction with the presence of a $I_1$ discriminant. For instance, $\Phi_{12}$ cannot appear since it would automatically imply that the geometry is isotrivial. The pairs which are eliminated for dimensional reasons have 
\textsf{no:dim} in the last column of the table. Then we can eliminate some pairs based on local regularity. We have already eliminated the four pairs in eq.\eqref{uuyuuuu7}.
On the other extremum there is the list of pairs which are obviously allowed, namely the ones associated to known SCFTs.
The corresponding SCFT are written in the table (up to a small ambiguity).

After this exercise, we remain with two pairs which are neither obviously ruled out nor obviously
allowed:
\be\label{juyqw1234}
(\Phi_3\Phi_6,\Phi_8)\qquad (\Phi_4\Phi_6,\Phi_8)
\ee
 In both cases the second axis is regular of dimension $\Delta_2=8/k$ where $\gcd(k,8)=1$.
 The first axis has dimension $\Delta_1$ (written in minimal form) $3/\ell$ or $6/\ell$ for the first pair and
 $4/\ell$ or $6/\ell$ for the second one. The only allowed dimension pairs with these properties
 are
 \be\label{juqw1235c}
 \{4/3,8/3\},\quad \{8/3,4\},\quad \{4,8\},\quad \{6/5,8/5\},\quad \{6,8\}.
 \ee
 All but the first pair are realized in the geometries of known SCFT\footnote{\ The dimension pairs are realized in the following SCFTs: $\{8/3,4\}$ in $\#58$, $\{4,8\}$ in $\#4,52,62,65$, $\{6/5,8/5\}$ in $\#19$,
 and $\{6,8\}$ in $\#2,33,39,44$.}. The first three pairs have $\varkappa=4$ or $4/3$
 and so are isotrivial. We claim that the last two sets of dimensions cannot be realized
 by monodromies with the polynomials in eq.\eqref{juyqw1234}. Note that in the known geometries with these dimensions the axis of dimension $\neq8$ is a discriminant of \emph{non-semisimple} type,
 while the polynomials \eqref{juyqw1234} are satisfied by semisimple matrices. 
 
 Let us justify the claim.
 The last two dimension pairs in \eqref{juqw1235c} have both $\varkappa=2$
 and $\{d_1,d_2\}=\{3,4\}$. Hence
 \be
 \varrho_0=\pm \mu_0^3\ \ \text{with eigenvalues }
 \begin{cases}\pm(1,1,-1,-1) &\text{for }(\Phi_3\Phi_6,\Phi_8)\\
 \pm(-1,-1,i,-i) &\text{for }(\Phi_4\Phi_6,\Phi_8)
 \end{cases}
 \ee
 so the discriminant along the first axis has Kodaira type $I_0^*$ and, respectively,
 $III$ or $III^*$.
 
 For both dimension pairs $z=u_1^4/u_2^3$, while for the first (resp.\! second)
dimension pair $h^2=(u_2/u_1)^5$ (resp.\! $h^2=u_2/u_1$). 
We look for the exponents $\alpha_\perp$
required to get a local solution along the first axis ($u_2=2$)
which matches with the Kodaira type above in the several cases:
\be
\begin{tabular}{lllc}
\text{dimensions} & \text{polynomials} & \text{regular period\quad} & \text{required }$\alpha_\perp$\\\hline
$\{6,8\}$ &$(\Phi_3\Phi_6,\Phi_8)$ & $\sqrt{h^2}\, z^{\alpha_\perp}= u_2^{1/2}\cdot u_1^{-1/2}$ & $0$\\
$\{6,8\}$ &
$(\Phi_4\Phi_6,\Phi_8)$ & $\sqrt{h^2}\, z^{\alpha_\perp}= 
\left\{\begin{smallmatrix}
u_2^{2/3}\cdot u_1^{a_1}\\
u_2^{1/3}\cdot u_1^{a_2}\end{smallmatrix}\right.$ & $\left\{\begin{smallmatrix}-1/18\\ 1/18
\end{smallmatrix}\right.$\\
$\{6/5,8/5\}$ &$(\Phi_3\Phi_6,\Phi_8)$ & $\sqrt{h^2}\, z^{\alpha_\perp}= u_2^{1/2}\cdot u_1^{b_1}$ & $2/3$\\
$\{6/5,8/5\}$ &
$(\Phi_4\Phi_6,\Phi_8)$ & $\sqrt{h^2}\, z^{\alpha_\perp}= 
\left\{\begin{smallmatrix}
u_2^{2/3}\cdot u_1^{c_1}\\
u_2^{1/3}\cdot u_1^{c_2}\end{smallmatrix}\right.$ & $\left\{\begin{smallmatrix}11/18\\ 13/18
\end{smallmatrix}\right.$\\
\end{tabular}
\ee
The only meaningful exponent is $2/3$ in the third row. Hence $(\Phi_4\Phi_6,\Phi_8)$
is ruled out, while the pair $(\Phi_3\Phi_6,\Phi_8)$ may correspond to a
special geometry with an enhanced discriminant of type $I_0^*$ \emph{provided}
it has the exotic-looking
dimension pair $\{6/5,8/5\}$.

\subsection{Summary of rank-2}

The main result of this section is that the list of \emph{known}
rank-2 special geometries ($\equiv$ the SW geometries of known rank-2 SCFTs)
is ``essentially'' complete in the sense that the only \emph{unknown} geometries whose existence we have not yet ruled out are:
\begin{itemize} 
\item geometries which are weakly equivalent to a known one,
i.e.\! they have a monodromy which is $GL(4,\C)$-equivalent but not $Sp(4,\Z)$-equivalent to
a known one;
\item class-$E_7$ geometries with new dimensions not ruled out in \S.\,\eqref{s:nnnewddd}.  A natural candidate follows
from the explicit monodromy in eqs.\eqref{Exo1}-\eqref{Exo2}.
\item other reducible geometries with dimensions $\{3,4\}$ or $\{4,6\}$
besides $\#45$. (This class deserves a more detailed study);
\item additional cover geometries (looks unlikely);
\item rigid geometries with a non rigid monodromy (which conjecturally do not exist). 
\end{itemize}

\section{The inverse problem in rank $r\geq3$}\label{s:r3}

At first sight the inverse problem may seen impossibly hard when $r>2$. 
Even the specification of the inverse data $(\{\Delta_i\},\{\cd_a\},\{\varrho_a\})$
may look bewildering because now the discriminant components $\cd_a$
are quasi-cones over weighted projective
hypersurfaces $\cs_a\subset\mathscr{P}$ which \emph{a priori} may be very complicated
whereas for $r=2$ they were mere points in $\mathbb{P}^1$.

Luckily the story is much better than it may appear. The two basic ingredients that
we used in $r=2$, the rigidity of the Abelian family over $\mathring{\mathscr{S}}\equiv\mathscr{P}\setminus\mathscr{S}$
and the special properties of rigid representations of the fundamental group $\pi_1(\mathring{\mathscr{S}})$ apply (the second one a bit conjecturally) to all quasi-projective bases
$\mathring{\mathscr{S}}$ of any dimension. In the rank-2 case we used the first property to
conclude that when the geometry is not
decomposable, or a cover, and 2 is not a Coulomb dimension (the only case which remains open assuming the 
\textbf{Folk-theorem})  the special divisor $\mathscr{S}$ must consist of precisely 3 points in $\mathbb{P}^1$. 
One has to extend that argument to higher rank; this leads to:

\begin{que} Which hypersurfaces complements 
$\mathbb{P}(q_1,\dots,q_r)\setminus \mathscr{S}$ in a weighted projective space
are the higher dimensional analogues of the sphere with three punctures?  
\end{que} 
If we can give an \emph{a priori} bound on how bad the singularities of the special divisor $\mathscr{S}$
may be, the set of such special divisors should be finite for each $r$-tuple of well-posed weights
$(q_1,\dots,q_r)$ which form themselves a finite set since only finitely many
dimensions $r$-tuple $\{\Delta_1,\dots,\Delta_r\}$ are allowed in a given rank $r$.

%
 
We start by making some simple considerations about the \textbf{Question}.

\subsection{The special divisor $\mathscr{S}$}\label{kkki9223b}

\subsubsection{Conditions on $\mathscr{S}$ from rigidity and ``stratification''}

We assume the Coulomb branch to be smooth and the geometry to be primary. Just as in rank-2 we have the
\begin{prin}[Rigidity]
The space $\mathscr{M}$ of deformations of the (reduced) divisor $\mathscr{S}\subset\mathscr{P}$ which preserve the group $\pi_1(\mathscr{P}\setminus  \mathscr{S})$ should have dimension equal to the multiplicity of $2$ as a
Coulomb dimension. In facts, $\mathscr{M}$ should coincide with the conformal manifold of the special geometry.
\end{prin} 

\begin{exe}\label{kkkkk1234} We are mainly interested in rigid special geometries, but
below we shall consider also examples coming from 
the opposite extremum, namely a dimension $r$-tuple of the form $\{2,2,\dots,2\}$.
In this instance we expect finitely-many $r$-parameter families $\{\mathscr{S}_m\}_{m\in\mathscr{M}}$
of special divisors\footnote{\ In this situation all components $\cs_a$ arise from the discriminant.}. 
In view of the Gaiotto construction \cite{gaiotto}, the \textbf{Folk-theorem} yields
an explicit description of them: there is one family per pair $(g,n)$ of non-negative integers
such that $3g-3+n=r$. The conformal manifold of the $(g,n)$-family is the moduli space
$\cm_{g,n}$ of Riemann surfaces of genus $g$ with $n$-punctures $p_i$,
while $\mathscr{P}$ is the linear system of the divisor $2K+\sum_i p_i$ which is a copy of
$\mathbb{P}^{r-1}$. Fix $m\in\cm_{g,n}$ and consider the corresponding marked curve $\Sigma_m$.
A point $z\in\mathscr{P}$ 
corresponds to a quadratic differential $\phi_{m,z}$ on $\Sigma_m$ with
at most simple poles at the punctures (modulo overall normalization).
Consider the curve $C_{m,z}$ in the (holomorphic) cotangent space $T^*\Sigma_m$
(with fiber coordinate $y$) of equation  $y^2=\phi_{m,z}$. $\mathscr{S}_m\subset\mathbb{P}^{r-1}$
is the locus of $z$'s with singular curve $C_{m,z}$. 
\end{exe}

Rigidity is not the only restriction on $\mathscr{S}$. We have already discussed in \S.\,\ref{s:finD} other conditions following from the physical intuition of ``stratification''
and in particular that each irreducible component of the discriminant
admits a polynomial parametrization. 
%
%
Some examples of divisor satisfying these constraints are in order.

\begin{exe}\label{999ex} Consider the rank-3 special geometries of dimensions $\{2,2,2\}$.
As described in \textbf{Example \ref{kkkkk1234}} we have three families of them
with $(g,n)=(0,6),(1,3),(2,0)$. We focus on the simpler $(0,6)$ family. We have a quadratic differential
with simple poles at (say) $z_1\equiv 0$, $z_2=1$, $z_3=-1$, $z_4$, $z_5$, $z_6$
($z_4,z_5,z_6$ being local coordinates in $\cm_{0,6}$)
\be\label{uuuuuyyyww}
\phi_{m,u}=\frac{u_1\, z^2+ u_2\, z+u_3}{z(z^2-1)(z-z_4)(z-z_5)(z-z_6)}\,dz^2
\ee
where $(u_1:u_2:u_3)$ are homogeneous coordinates on $\mathscr{P}$. The discriminant
is the union of the loci where a zero of the quadratic polynomial in the numerator
collides with a zero of the denominator and the locus where the two zeros of the
numerators coalesce
\be
\cd\colon\quad (u_2^2-4 u_1u_3)\prod_{a=1}^6 (u_1\, z_a^2+u_2\, z_a+u_3)=0.
\ee 
Let us check that this divisor $\cd$ has the properties expected on physical grounds. We start with ``rigidity''.
$\cd$ is the union of a smooth planar conic $C$ and six lines 
\be
L_a\equiv\{u_1\, z_a^2+u_2\, z_a+u_3=0\}\subset\C^3,\quad a=1,\dots,6
\ee 
each one of which is isomorphic to
a copy of $\mathbb{P}^1$. Modulo projective equivalence there is a unique smooth conic.
The six lines are all tangent to the conic $C$, and each line $L_a$ is uniquely determined by its
point of tangency $p_a\in C$. Thus $\mathscr{S}_m$ is determined by the six point
$\{p_a\}$. However the automorphism group $O(3,\C)$ of $C$ acts on the points $\{p_a\}$; since the action of $O(3,\C)$ on $C$ is 3-transitive\footnote{\ This is the complexified version of the well-known \emph{holography} statement that the Euclidean conformal group acts on the $d$-sphere $S^d$ in a 3-transitive fashion.}, we can fix 3 points, and the actual deformation space of $\mathscr{S}_m$ is 3-dimensional
as required by physics.
%

\begin{rem}\label{kkii1234b} A cheap curve $C\subset\mathbb{P}^2$ with the property that
the deformation space $\mathscr{M}$ preserving the topology of its complement has dimension $s$ is a smooth conic together with $s+3$ tangent lines.
\end{rem}

Let us check ``stratification'' on this example. Restricting on the discriminant line $L_a$, the numerator
has a zero at $z=z_a$ which cancels a zero in the denominator, so that $\phi_2$
 becomes the quadratic differential of the model
with $(g,n)=(0,5)$; hence the restriction on $L_a$ yields back the first $\{2,2\}$ special geometry in eq.\eqref{iiiu7zz}.
The restriction on the conic $C$ produces the non-smooth special geometry discussed in \S.\,\ref{uuiyyyw}, i.e.\! a $\Z_2$ quotient of the free geometry. Indeed
\be\label{uuuuuyyyww2}
\phi_{m,u}\big|_C=\frac{u_1(z-u_2/2u_1)^2}{\prod_{a=1}^6(z-z_a)}\,dz^2
\ee
and setting $y=\sqrt{u_1}(z-u_2/2u_1) \tilde y/\prod_a(z-z_a)$, we get the equation
\be
\tilde y^2= \prod_{a=1}^6 (z-z_a),\qquad \lambda=\sqrt{\phi_2}= (\sqrt{u_1}\,z-\sqrt{u_3})\frac{dz}{\tilde y}.
\ee 
 \end{exe}

\begin{exe} Consider the rank-4 SCFT with dimensions $\{2,2,2,2\}$ associated to the 
sphere with 7 punctures. Eq.\eqref{uuuuuyyyww} gets replaced by
\be\label{pppo999992}
\phi_{m,u}=\frac{u_1\ z^3+u_2\,z^2+u_3\, z+u_4}{\prod_{a=1}^7(z-z_a)}\,dz^2,\qquad z_0,z_1,z_2=0,+1,-1
\ee
with discriminant $\mathscr{S}\subset\mathbb{P}^3$ of equation
\be
\Big(-27\, u_1^2 u_4^2+18\, u_1u_2 u_3 u_4-4\, u_1
   u_3^3-4\, u_2^3 u_4+u_2^2 u_3^2\Big)\prod_{a=1}^7\big(u_1 z_a^3+u_2z_a^2+u_3z_a+u_4\big)
   \ee
 the union of 7 planes $L_a$ and a \emph{singular} quartic $Q_4$, with the polynomial parametrization
   \be
   (u_1:u_2:u_3:u_4)= (t_1^3:(2t_2+t_3)t_1^2:(2t_2t_3+t_2^2)t_1:t_2^2t_3).
   \ee
$Q_4$ and the $a$-th plane are tangent along the the image of the line
   $t_2+z_at_1=0$. 
   If we restrict to a plane $L_a$, a zero in the numerator of \eqref{pppo999992}
   cancels a zero in the denominator and  
we get back the quadratic differential in one rank less, eq.\eqref{uuuuuyyyww}, as expected from
stratification. 
We may look at the discriminant restricted on a line $L_a$, say to $u_4=0$ we get a reducible curve: one component is the correct curve $u_2^2-4 u_1u_3=0$ in one rank less and
 the other is a ``spurious'' extra double line $u_3^2=0$ tangent to the conic.
This reflects the fact that we are (wrongly!) enforcing that the curve depends essentially on 4 parameters, and so it reduction on $u_4=0$ must contain 7 tangent lines, whereas the actual discriminant depends only on 3 conformal moduli and hence should contain only 6 tangent line. When both the original geometry and the reduced geometry along a
discriminant component are rigid we do not expect this mismatch to show up.  We expect
that for a rigid geometry the restriction of the special locus to a discriminant component should be a legitimate special locus in one less rank. We call this the \emph{hereditary} principle.
\end{exe}

\begin{rem}\label{ppaazzq} \emph{A priori} the hereditary principle applies only to the components of $\mathscr{S}$
which arise from the discriminant, not necessarily to the ones associated with \emph{regular} R-enhanced divisors. However the only conditions that may not hold are the algebraic properties
of the fundamental group of the complement which now may be Abelian etc. corresponding to (say)
an isotrivial geometry along the R-enhanced divisor. For instance, in $r=3$ the normalization of a
R-enhanced divisor may be $\mathbb{P}^1$ with just \emph{two} punctures not three as
in the ``legitimate'' rank-2 case.
\end{rem}

\begin{exe} More generally, for the rank-$r$ theory with all dimensions equal 2 whose Gaiotto surface is
a sphere with $r+3$ punctures $\mathscr{S}$ is the union of $r+3$ hyperplanes 
\be
H_a=\left\{\sum_{s=1}^r u_s\, z_a^{r-s}=0\right\}\subset\mathbb{P}^{\,r-1},\qquad a=1,\dots,r+3
\ee
plus a singular
hypersurface $Y_{2r-4}$ of degree $2r-4$ given by the discriminant of the degree $r-1$
polynomial $P(z)=\sum_{s=1}^r u_s\, z^{r-s}=0$ which admits the rational parametrization
\be
(u_1:u_2:\dots:u_r)=(t_0^{r-1},t_0^{r-2}e_1,t_0^{r-2} e_2,\dots, t_0 e_{r-2},e_{r-1}),\quad (t_0,t_1,\dots,t_{r-2})\in\mathbb{P}^1
\ee
where $e_k$ is the $k$-th elementary polynomials in the $(r-1)$ variables $(t_1,t_1,t_2,\dots,t_{r-2})$
with $t_1$ repeated twice. The validity of the hereditary principle is obvious.
\end{exe}

\subsubsection{Conditions on $\mathscr{S}$ from $\pi_1(\mathscr{P}\setminus\mathscr{S})$}\label{uuuuqwrt}

The above constraints from rigidity and ``stratification''
are not the only conditions on the divisor $\mathscr{S}\subset\mathscr{P}$.  The requirement that the geometry is indecomposable and non-isotrivial
(nor quasi-isotrivial) sets additional conditions on $\mathscr{S}$.
Indeed the complement $\mathscr{P}\setminus\mathscr{S}$
cannot be a product of lower dimensional quasi-projective space
and its fundamental group
 $\pi_1(\mathscr{P}\setminus\mathscr{S})$ cannot be finite, Abelian,
 solvable, or a non-trivial product, etc.
To see what this implies we review the fundamental  
group of the complement of the hypersurface $\mathscr{S}$ \cite{dimca}.

\subsection{The fundamental group of a hypersurface complement}\label{poi8871}

For special geometry in rank-$r$ we are interested in the fundamental group $\pi_1(\mathscr{P}\setminus\mathscr{S})$
of a weighted projective hypersurface $\mathscr{S}$ in a weighted projective
space $\mathscr{P}=\mathbb{P}(q_1,\dots,q_r)$ which we assume well-formed.

The first observation is that we can always reduce ourselves to hypersurface
complements in the \emph{ordinary} projective space\footnote{\ Below we shall use the symbols $u_1,\dots,u_r$ and $x_1,\dots,x_r$ for, respectively, the homogeneous 
coordinates in the weighted and ordinary projective space.} $\mathbb{P}^{r-1}$.
In our situation the hypersurface $\mathscr{S}$ can be assumed to have form
\be
\mathscr{S}=\{u_1=0\}\cup \{u_2=0\}\cup \cdots\cup \{u_r=0\}\cup \{h(u_1,\dots,u_r)=0\}
\ee 
where $\{u_i=0\}$ are (typically) the R-enhanced 
divisors and $\mathscr{K}\equiv \{h=0\}$ is a divisor (typically reducible) which does not lie
in any linear sub-space spanned by the coordinate axes. 
We use the notation $\ddot{\mathscr{P}}\equiv\mathscr{P}\setminus (u_1u_2\cdots u_r)$,
$\ddot{\mathbb{P}}^{\mspace{2mu}r-1}\equiv \mathbb{P}^{\mspace{2mu}r-1}\setminus(x_1x_2\cdots x_r)$,
so that $\mathscr{P}\setminus\mathscr{S}\equiv \ddot{\mathscr{P}}\setminus \mathscr{K}$.
By construction we have an \emph{unbranched} finite cover
\be
\varpi\colon\ddot{\mathbb{P}}^{\mspace{2mu}r-1}\to \ddot{\mathscr{P}},\qquad\quad
\varpi\colon x_i\mapsto u_i=x_i^{q_i},
\ee
with deck group $\Z_{q_1}\times\Z_{q_2}\times\cdots\times\Z_{q_r}$
so  we have an exact sequence of groups
\be
1\to \prod_a \Z_{q_a}\to \pi_1(\ddot{\mathscr{P}}\setminus \mathscr{K})\to 
\pi_1(\ddot{\mathbb{P}}^{\mspace{2mu}r-1}\setminus \varpi^{-1}(\mathscr{K})))\to 1
\ee
so that, \emph{modulo a finite group}, we can replace $\pi_1(\mathscr{P}\setminus\mathscr{S})$
by the fundamental group of the complement in $\mathbb{P}^{\mspace{2mu}r-1}$
of the hypersurface $\mathscr{W}$ of equation
\be
\mathscr{W}\colon\quad x_1\,x_2\,\cdots x_r\,h(x_1^{d_1},x_2^{d_2},\dots,x_r^{d_r})=0. 
\ee
Note that some arguments, e.g.\! the structure theorem of VHS, are insensitive to finite groups.
\medskip

\subsubsection{Zariski and Kempen-Zariski theorems}\label{s:zariskithm}
The fundamental group $\pi_1(\mathbb{P}^{\mspace{2mu}r-1}\setminus\mathscr{W})$ for $r\geq3$ is described
by two theorems due to Zariski and Kempen-Zariski respectively (see e.g.\! \cite{dimca} or \cite{haraoka}).
The first theorem states that we can effectively reduce the problem to two-dimensions (i.e.\! $r=3$).

\begin{thm}[Zariski]
$\mathscr{W}$ is a hypersurface in $\mathbb{P}^{\mspace{2mu}r-1}$
with complement $Y\equiv\mathbb{P}^{\mspace{2mu}r-1}\setminus \mathscr{W}$ and $E_k\subset\mathbb{P}^{\mspace{2mu}r-1}$
a generic $k$-dimensional linear subspace. The morphism induced by inclusion
\be
\iota_k\colon\pi_1(E_k\cap Y)\to \pi_1(Y)
\ee
is surjective for $k = 1$ and an isomorphism for $k \geq 2$.
In particular:
\begin{itemize}
\item[\rm (a)] The computation of $\pi_1(Y)$ is reduced to the case $r=3$
i.e.\! to the complement of  a planar curve $E_2\setminus E_2\cap \mathscr{W}$;
\item[\rm (b)] The group $\pi_1(Y)$
has a natural set of $d = \deg(\mathscr{W})\equiv \#\{E_1\cap \mathscr{W}\}$ generators $\gamma_j\equiv \iota_1(\ell_j)$ satisfying the relation $\gamma_1\cdots\gamma_d = 1$.
\end{itemize}
\end{thm}

In view of (b), to get a full description of the fundamental group $\pi_1(Y)$ we have
to complete its presentation by adding the relations between the generators
$\gamma_j$ which generate the kernel of the homomorphism $\iota_1$.
This is precisely what the Kempen-Zariski theorem does, see e.g.\! \cite{dimca,haraoka} for details. 
E.g.\! if $\mathscr{W}$ is a smooth planar curve of degree $d$, $\pi_1(Y)=\Z/d\Z$.
 The Kempen-Zariski relations take a particular simple form when the planar curve $\mathscr{W}\subset\mathbb{P}^2$ is the complexification of
a real arrangement of lines. 
By (b) we have one generator $\gamma_i$ per line (satisfying $\gamma_1\cdots\gamma_d=1$);
The Randell theorem \cite{randell} states that we have
one additional relation per intersection point between these lines. At a point where the $\ell$
lines $\gamma_{j_s}$ ($s=1,\dots,\ell$) meet (ordered in cyclic order) we have the `rotation' relations
\be
\gamma_{i_1}\gamma_{i_2}\cdots\gamma_{i_{s-1}}\gamma_{i_s}=\gamma_{i_2}\gamma_{i_3}\cdots\gamma_{i_s}\gamma_{i_1}=
\gamma_{i_3}\gamma_{i_4}\cdots\gamma_{i_1}\gamma_{i_2}=\cdots=
\gamma_{i_s}\gamma_{i_1}\cdots\gamma_{i_{s-2}}\gamma_{i_{s-1}}
\ee 
In particular if we have only pairwise intersections the fundamental group is commutative.
This is a particular case of the more general result

\begin{pro} If $X\subset\mathbb{P}^2$ is a curve with only nodal singularities
$\pi_1(\mathbb{P}^2\setminus X)$ is Abelian.\footnote{\ For line arrangements the converse is also true.}
\end{pro} 
%
%

\subsubsection{Rational covers}\label{kkkiqwert9} 
We have already noticed that for many purposes we are interested in the monodromy group only up to commensurability, i.e.\! up to finite groups. For $r\geq3$
there is a cheap way of getting new good special divisors $\mathscr{S}$ out of old ones such that
their complements have commensurable fundamental groups.
For easy of exposition, we assume that $\mathscr{P}$ is the ordinary projective space $\mathbb{P}^{\mspace{2mu}r-1}$. Suppose we have a rational map
$\psi\colon \mathbb{P}^{\mspace{2mu}r-1}\dashrightarrow\mathbb{P}^{\mspace{2mu}r-1}$
of degree $d$ and a good special divisor $\mathscr{S}$ with the property that
the hypersurface $\mathscr{S}^\prime\equiv\overline{\psi^{-1}(\mathscr{S})}$ contains all the points of indeterminacy of
$\psi$ as well as the branching locus of $\psi$. Then the restriction of the map between the complements
\be
\psi|\colon \mathbb{P}^{\mspace{2mu}r-1}\setminus \mathscr{S}^\prime\to
 \mathbb{P}^{\mspace{2mu}r-1}\setminus \mathscr{S}
\ee
is an unbranched cover with a deck group $F\subset\mathfrak{S}_d$ (hence finite). 
One has the exact sequence of groups
\be
1\to F\to \pi_1(\mathbb{P}^{\mspace{2mu}r-1}\setminus \mathscr{S})\to
\pi_1(\mathbb{P}^{\mspace{2mu}r-1}\setminus \mathscr{S}^\prime)\to 1
\ee
so that the fundamental groups of the two complements differ by a finite group.
When $d=1$, i.e.\! $\psi$ is a Cremona transformation, the complements are
isomorphic and in particular have the same fundamental groups (evident from the above exact sequence since $F=1$ in this case). The special divisor
$\mathscr{S}$, $\mathscr{S}^\prime$ are different but essentially equivalent
from the viewpoint of the monodromy analysis. For examples, see below.

More generally we may exploit the theory of Cremona transformations for weighted projective
space. See \cite{wcremona} for examples when $r=3$ with applications 
to the computation of fundamental groups of complements of curves in weighted
projective planes.

\subsection{First ``economical'' examples}

We focus on the basic case $r=3$, i.e.\! $\mathscr{P}$
of complex dimension 2, and assume the special geometry to be irreducible,
rigid (2 is not a Coulomb dimension), and non-isotrivial.  Under these conditions
in rank-2 the complement $\mathscr{P}\setminus\mathscr{S}$ was the sphere less 3 points;
 we look for complements
 $\mathscr{P}\setminus\mathscr{S}$
which are
\textit{``the higher dimensional analogues of the thrice-punctured sphere''} in the sense
that satisfy the ``same'' rigidity and non-triviality conditions.
The curve $\mathscr{S}$ should posses
 rather restrictive properties, in particular it should have no equisingular deformation\footnote{\
A deformation of a curve is \emph{equisingular} if it preserves the Milnor number of each of
its singularity, see \cite{dimca}.} since the complements of equisingular curves are diffeomorphic \cite{dimca}
so have the same fundamental group. 
To give the flavor of a planar curve satisfying all the relevant constraints, we present a simple
example which may actually arise from a fully-fledged rank-3 special geometry. Later we shall discuss a deeper interpretation of the same example.

 \medskip
 
To keep the example as simple as possible, we consider the special case where $\{\Delta_1,\Delta_2,\Delta_3\}$ are such:
  \textit{(i)} 1 and 2 are not dimensions, \textit{(ii)} the geometry is not automatically isotrivial ($\varkappa\in\{1,2\}$),
 and \textit{(iii)} 
  $\mathscr{P}\simeq \mathbb{P}^2$. 
The typical example is $\{\Delta_1,\Delta_2,\Delta_3\}=\{4,6,12\}$
where the isomorphism $\mathscr{P}\simeq \mathbb{P}^2$ is given in
 homogeneous coordinates $(x_1:x_2:x_3)=(u_1^6:u_2^4:u_3^2)$.
 The reduced $\mathscr{S}$ has the form 
\be
\mathscr{S}=\big(x_1\,x_2\,x_3\, h(x_1,x_2,x_3)\big)
\ee
 for some square-free homogeneous polynomial
$h\equiv h(x_1,x_2,x_3)$ of degree $d$ which is not divisible by any $x_i$.
$\mathscr{S}$ consists of three lines in general position, $\{x_a=0\}$,
 plus the planar curve $h=0$ which is
reduced but typically reducible.
$h$ should satisfy a number of conditions. First, $\mathscr{S}$ must have no
 equisingular deformations, that is, all modifications of the polynomial $h$ which preserve the Milnor numbers of the singularities of $\mathscr{S}$ may be undone by rescalings of the variables $x_i$ (`rigidity').
 If the curve $h$ is smooth in general position the rigidity condition is satisfied only when $d=1$.
 But then the curve $\mathscr{S}$ has only simple nodes, $\pi_1(\mathscr{P}\setminus\mathscr{S})$
 is Abelian, so the period map is constant, and the geometry (if it exists) is isotrivial.
Next the line $h=0$ may be in a special position, i.e.\! it may pass through the intersection point
 of two coordinate lines, say through $(x_2=0)\cap (x_3=0)$, so $h=\alpha\, x_2+\beta\, x_3$.
 But in this case 
\be
\pi_1(\mathscr{C}\setminus (x_1x_2 x_3 (\alpha x_2+ \beta x_3)))=\pi_1\!\big(\C\setminus (x_1)\big)\times \pi_1\!\big(\C^2\setminus (x_2 x_3 (\alpha x_2+ \beta x_3))\big)
\ee
  and the special geometry decomposes in a rank-1 special geometry with $\Delta=4$
  and a rank-2 geometry with $\{\Delta_1,\Delta_2\}=\{6,12\}$ (again isotrivial).
  We conclude that the minimal degree of $h$ is 2, in which case the reduced curve $h=0$ is either two distinct lines or
  a smooth conic. We consider the two possibilities in turn.
  
  \paragraph{$h=0$ are two distinct lines.} Neither line can be in general position because of rigidity;
  then either both lines pass through a crossing point $p_{ij}\equiv(x_i=0)\cap (x_j=0)$ of two coordinate lines, or
  they meet on a coordinate line (say on $x_3=0$). They cannot pass through the same crossing point $p_{ij}$ because
  this will lead to a decomposable geometry. Modulo permutation of the coordinates,
  we remain with the two possibilities
  \be
  \begin{cases}\alpha x_1+\beta x_3=0\\
  \gamma x_2+\delta x_3=0
  \end{cases}\quad\text{or}\quad 
    \begin{cases}\alpha x_1+\beta x_2+\gamma x_3=0\\
\alpha x_1+\beta x_2+\gamma^\prime x_3=0
  \end{cases}
  \ee
 In the second case the parameter $\gamma/\gamma^\prime$ cannot be absorbed by a rescaling,
 so this configuration is excluded.
We remain with the first possibility. 
We absorb the parameters by a suitable rescaling of the coordinates,
and focus on the affine patch $x_3\neq0$ of $\mathbb{P}^2$;
$x_3=0$ is then the line at infinity $L_\infty$, and we write $\mathscr{S}=X\cup L_\infty$.
 Now $\mathbb{P}^2\setminus \mathscr{S}=\C^2\setminus X$ where  $X\subset \C^2$ is the \emph{complexification} of the arrangement in $\R^2$ of the 4 real lines drawn solid in the figure 
\be\label{ju76z}
\begin{gathered}
\includegraphics[width=0.30\textwidth]{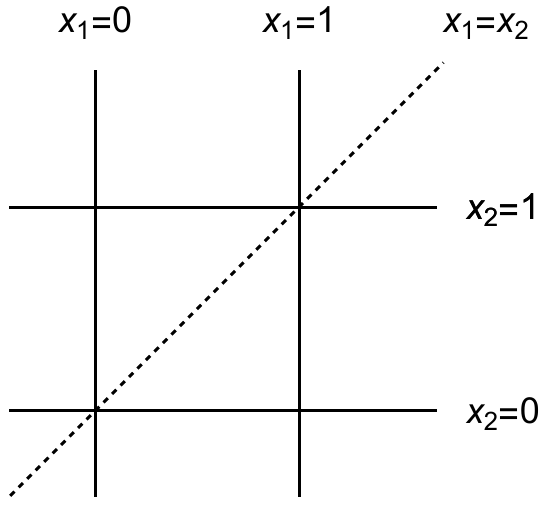}
\end{gathered}
\ee
Adding the dashed line $L_5\equiv\{x_1=x_2\}$ to $X$ will not spoil rigidity nor any other required property. We claim that we should add $L_5$ in order to satisfy
the ``natural'' conditions of $\pi_1(\mathscr{P}\setminus\mathscr{S})$ listed in
\S.\,\ref{uuuuqwrt}. Indeed in eq.\eqref{mosteco} we will show that the fundamental group
of the complement in $\C^2$ of the 4 solid lines is a non-trivial product. 
Thus 

\begin{exe}
Our first \textit{``economic'' higher dimensional analogue of the sphere less 3 points}
is $\mathbb{P}^2$ less the six lines
\be
\{x_1=0\},\ \{x_2=0\},\ \{x_3=0\},\ \{x_1=x_3\}, \ \{x_2=x_3\},\ \{x_1=x_2\}.
\ee   
i.e.\! the complement in
$\C^2$ of the five lines in the figure \eqref{ju76z}. We check the validity of the ``hereditary principle''.
Each special line contains precisely 3 special points, namely the two intersections with other lines
visible in figure \eqref{ju76z} plus the intersection with the line at infinity which also has 3 special
points associated to each set of parallel lines in \eqref{ju76z}. We shall consider the reduction of the natural Picard-Fuchs equation with singularity locus \eqref{ju76z} to one discriminant line (in the sense of \S.\,\ref{s:PFstrat}), in \S.\,\,\eqref{s:AAAA} below. 
\end{exe}

 \paragraph{$h$ is a smooth conic.} Let us consider the case where $h=0$ is a conic.
 \textbf{Remark \ref{kkii1234b}} gives us a first good configuration:
 $h$ is the (unique) smooth conic in $\mathbb{P}^2$ which is tangent to all 3 coordinate lines $\{x_a=0\}$ i.e.\!
 \be\label{juiqw12uu}
 h\colon\quad x_1^2+x_2^2+x_3^2-2(x_1x_2+x_2x_3+x_3x_1)=0.
 \ee
 Each one of the 4 components of $\mathscr{S}$ is a copy of $\mathbb{P}^1$ and the (reduced)
 restriction of $\mathscr{S}$ to it consists of 3 distinct points in perfect accordance
 with the hereditary principle.

 Another rigid configuration is a conic which goes through all 3 intersection points $p_{ij}$ i.e.\!
  \be
 h\colon\quad x_1x_2+x_2x_3+x_3x_1=0
 \ee
 The three lines $x_i+x_j=0$ ($i<j$) are tangent to $h=0$ at the crossing point $p_{ij}$;
 adding some of them would not spoil rigidity. If we add exactly \emph{two} of them, say the tangents to $h=0$ in $p_{12}$
 and $p_{13}$,
 all components have 3 special points, but for the line $x_1=0$ which has only 2 such points and then must be a regular R-enhanced discriminant (cf.\! \textbf{Remark \ref{ppaazzq}}).
If   we draw all 3 tangents we get a special divisor that, while rigid, is not
hereditary: some components have 4 special points. This configuration should be disregarded.
Thus to construct a new ``economic'' example we add just two tangents and draw
them \emph{dashed}
 meaning that the monodromy along them is allowed to be trivial (cf.\! the dashed line in  \eqref{ju76z}).
 If we set the local monodromy along 0,1, or dashed lines to 1 we get
 a configuration with (respectively) 1, 2 or 3 lines $\{x_a=0\}$ which are regular R-enhanced.
 Note that this is consistent with the fact that only coordinate planes can be regular special divisors. 
 
 Finally we have the less symmetric situations
where the conic is tangent (say) to the line $x_3=0$ in the crossing point $(0:1:0)$
  and also passes through another crossing point, say $(0:0:1)$. While this curve
   is rigid,
  the conic component has only two special points: since it is not a coordinate line
  it cannot be a regular component and so, having only 2 special points, violates
  hereditary. This configuration should also be disregarded.

\begin{exe}
Our two next ``economic'' examples of \emph{higher dimensional analogues}
of  the 3-punctured sphere $\mathbb{P}^1\setminus\{0,1,\infty\}$ are:
{\bf(1)} the complement in
$\C^2$ of a conic with three tangent lines and {\bf(2)} the complement of three lines
plus the conic
passing through their crossing points and two out of the three lines tangent to the
conic which pass through a crossing point of the original lines.
\end{exe}
 
  \medskip
  
  \paragraph{Higher degree $d\geq3$.}
  We can increase the degree $d$
  of $h$ above 2
  and rigidity will imply stronger and stronger non-genericity conditions and wilder singularities
  making the configuration of curves more and more
  ``unlikely'' and ``ugly''. Since our purpose is merely to present nice economical examples,
  we shall not pursue them here. Many other examples (possibly all) can be produced using the cover construction outlined in \S.\,\ref{kkkiqwert9} and so are essentially equivalent to the one described before. It is conceivable
  that all consistent special divisors arise in this way from the previous ``economic'' divisors (see below for some preliminary evidence).

 \paragraph{Summary of $r=3$ examples.} In conclusion, for our illustrative purposes we focus on the three ``cheap''
  examples of \emph{higher dimensional analogues
of  the thrice-punctured sphere $\mathbb{P}^1\setminus\{0,1,\infty\}$}
given by the complement $\mathbb{P}^2\setminus \mathscr{S}_{(a)}$ ($a=1,2,3$)
where
\begin{itemize}
\item[(1)] 
$\mathscr{S}_{(1)}\colon x_1\,x_2\,x_3\,(x_3-x_1)\,(x_2-x_3)\,(x_1-x_2)=0$
\item[(2)] $\mathscr{S}_{(2)}\colon x_1\,x_2\,x_3\,(x_1+x_2)\,(x_2+x_3)\,(x_1 x_2+x_2x_3+x_1x_3)=0$
\item[(3)] $\mathscr{S}_{(3)}\colon x_1\,x_2\,x_3\,(x_1^2+x_2^2+x_3^2-2x_1x_2-2x_2x_3-2 x_3x_1)=0$.
\end{itemize}

\begin{rem} The coordinate lines $x_i=0$ may correspond to R-enhanced loci which are or not part of the discriminant. On the contrary, all components of $h=0$ should be part of the discriminant.
Hence a \emph{necessary condition} for the coordinate axis $\{u_i=0\}\cap\{u_j=0\}$
not to lay in the discriminant is that $h$ does not cross the point $p_{ij}\in \mathbb{P}^{2}$. 
Therefore,
when the three axes are all regular (and the degree of $h\leq2$), the only possibilities are (3) or
(2) with the two ``dashed'' lines $x_1+x_2=0$, $x_2+x_3=0$ deleted.
In rank-2 there are only 2 non-isotrivial geometries with all axes regular ($\#20$ and $\#21$)
and they both correspond to Argyres-Douglas SCFTs. 
Non-isotrivial geometries with all axes regular and an irreducible discriminant are quite rare.  
\end{rem}

\paragraph{Rational covers.} In the case of the thrice-punctured sphere we had to study
the representations of a single group $\pi_1(\mathbb{P}^1\setminus\{0,1,\infty\})\simeq F_2$
(the free group in two generators). In the present situation it may seem that we have
to study three different groups already at the level of our ``economic'' examples.
However it is not so \cite{app1}

\begin{fact} {\bf (1)} The complements $\mathbb{P}^2\setminus \mathscr{S}_{(1)}$ and
 $\mathbb{P}^2\setminus \mathscr{S}_{(2)}$
are biholomorphic so, in particular, have isomorphic fundamental groups
$\pi_1(\mathbb{P}^2\setminus \mathscr{S}_{(1)})\simeq\pi_1(\mathbb{P}^2\setminus \mathscr{S}_{(2)})$.\\ 
{\bf (2)} We have an unbranched double cover $\mathbb{P}^2\setminus \mathscr{S}_{(1)}\to
\mathbb{P}^2\setminus \mathscr{S}_{(3)}$ with deck group $\Z_2$ generated by $x_1\leftrightarrow x_2$, so that
\be
1\to\Z_2\to \pi_1(\mathbb{P}^2\setminus \mathscr{S}_{(3)})\to \pi_1(\mathbb{P}^2\setminus \mathscr{S}_{(1)})\to 1
\ee
so the fundamental groups agree modulo a finite group.
\end{fact} 

\begin{proof}{\bf(1)} The Cremona transformation
\be
c\colon\mathbb{P}^2\dashrightarrow\mathbb{P}^2,\qquad (x_1,x_2,x_3)\mapsto (y_1,y_2,y_3)\equiv\big(x_1x_2+x_2x_3,-x_1x_3, x_1x_2\big)
\ee
(projectively equivalent to a standard quadratic transform \cite{classicalAG}) has the property that
the \emph{reduced} divisor $(c^*\mathscr{S}_{(1)})_\text{red}$ is equal to $\mathscr{S}_{(2)}$.
Indeed
\begin{multline}
c^*\big(y_1y_2 y_3 (y_3-y_1)(y_2-y_3)(y_1-y_2)\big)=\\
= - x_1^3 x_2^3 x_3^2(x_1+x_3)(x_2+x_3)(x_1x_2+x_1x_3+x_2x_3),
\end{multline}
while the points of indeterminacy of $c$ ($(1:0:0)$, $(0:1:0)$, $(0:0:1)$)
and of $c^{-1}$ ($(1:0:0)$, $(0:1:0)$, $(1:0:1)$) belong to the excised curves so that $c$
is an isomorphism when restricted to the complements.
\textbf{(2)} Same story, except that the rational map 
\be
\begin{aligned}
&r\colon\big(x_1,x_2,x_3\big)\mapsto
\big(y_1,y_2,y_3)\equiv\big(x_1x_2,(x_3-x_1)(x_3-x_2),x_3^2\big)
\\
&r^*\!\left(y_1y_2y_3\Big(\sum_a y_a^2-2\sum_{a<b}y_ay_b\Big)\right)=x_1x_2x_3^4(x_1-x_2)^2(x_1-x_3)(x_2-x_3)
\end{aligned}
\ee
has now degree 2. The Jacobian of the map is $2(x_1-x_2)x_3^2$,
so the map is an unbranched cover when restricted to the complement. 
\end{proof}

\paragraph{Fundamental groups.} It is convenient to think of the ``good'' locus $\mathring{\mathscr{P}}$
as the complement of the arrangement of lines $\mathscr{S}_{(1)}$
which allows us to write a simple presentation of its fundamental group using the
Randell theorem \cite{randell}.
We number the lines in \eqref{ju76z} (i.e.\! the component of $\mathscr{S}_{(1)}$ distinct of the line at infinity $L_\infty=\{x_3=0\}$) as
\be
L_1=\{x_1=0\},\ L_2=\{x_2=x_3\},\ L_3=\{x_1=x_2\},\ L_4=\{x_1=x_3\},\ L_5=\{x_2=0\}
\ee
and write $\mu_a$ for the monodromy of a loop going once around $L_a$ in the positive direction.
The Randell theorem then gives \cite{randell,haraoka}\footnote{\ Compare the line arrangement in figure \eqref{ju76z} with the one in the figure \textbf{13.8} of \cite{haraoka}.}
{\renewcommand{\arraystretch}{1.05}
\be\label{hhhhhy4}
\pi_1(\mathbb{P}^2\setminus\mathscr{S}_{(1)})=
\pi_1(\C\setminus \cup_{a=1}^5 L_a)=\left\langle\mu_1,\mu_2,\mu_3,\mu_4,\mu_5\;\left|\;
\begin{matrix}\mu_1\mu_2=\mu_2\mu_1,\ \ \mu_4\mu_5=\mu_5\mu_4,\\
\mu_5\mu_3\mu_1=\mu_3\mu_1\mu_5=\mu_1\mu_5\mu_3,\\
\mu_4\mu_3\mu_2=\mu_3\mu_2\mu_4=\mu_2\mu_4\mu_3\end{matrix}\right.\right\rangle
\ee
}
and we have
\be
\mu_\infty=(\mu_1\mu_2\mu_3\mu_4\mu_5)^{-1}.
\ee
When the dashed line $L_3$ has trivial local monodromy, $\mu_3=1$, (the ``\emph{most} economical'' example)
all relations become commutations and we get\footnote{\ As before, $F_2$ is the free group in two generators (isomorphic to $\Gamma(2)$ as abstract groups).}
\be\label{mosteco}
\pi_1(\mathbb{P}^2\setminus\mathscr{S}_\text{m.e.})\simeq F_2\times F_2
\ee
which is a special instance of \textbf{Theorem 4.11} of \cite{dimca2} (in turn a special case of
\cite{prodthm}). 

\begin{rem} As an abstract group $\pi_1(\mathbb{P}^2\setminus \mathscr{S}_{(3)})$ is the Artin braid group of the triangle graph \cite{artin1,artin2}
\be
\begin{gathered}
\xymatrix{&\bullet\ar@{-}[dr]^4\ar@{-}[dl]_4\\
\bullet\ar@{-}[rr]_2 &&\bullet}
\end{gathered}
\ee
For a van Kampen-like presentation of $\pi_1(\mathbb{P}^2\setminus \mathscr{S}_{(3)})$, see instead \cite{wcremona}. 
\end{rem} 

\subsubsection{Comparison with Appell's hypergeometric PDEs}\label{s:AAAA}

Above we have followed our physical intuition to get ``economic'' examples of
 complements of planar curves which may be seen, at least informally, as
higher dimensional counterparts of the sphere with three punctures. It turns out
that the 3 varieties we
got realize the desired analogy in a much deeper sense than expected \emph{a priori}.
Indeed, as stressed by Riemann in 1851, the most important property of the thrice-punctured sphere is that 
the degree-2 representations of its fundamental group are the monodromies
of  the order-2 hypergeometric ODE, and these representations are rigid. 
Appell \cite{appp} extended Riemann ideas to
higher dimensions introducing four hypergeometric functions $F_i(x,y)$ ($i=1,\dots,4$) depending on two variables
$x$, $y$ which satisfy certain systems of linear PDEs with regular singularities
(a Picard-Fuchs system) which are, in a sense, the 2-variable generalization of Gauss' hypegeometric ODE. One studies the monodromy representation of this Picard-Fuchs systems
and finds that they are \emph{rigid} (under some natural condition) \cite{haraoka,app1,app2,app3,app4}.
The monodromy of the Appell function is \emph{precisely} a (rigid) representation of the
fundamental group $\pi_1(\mathscr{P}\setminus\mathscr{S}_{(1)})$ of the complement of
the arrangement of lines in figure \eqref{ju76z} which then plays the same role in the theory of rigid monodromy
of Picard-Fuchs PDEs on $\mathbb{P}^2$ than the thrice-punctured sphere  plays classically
in the theory of rigid monodromies of Picard-Fuchs ODEs on $\mathbb{P}^1$!

Important examples of rigid monodromy representations of the group \eqref{hhhhhy4}
are explicitly know \cite{app1,app2,app3,app4}: they are the monodromies of the four Appell's
functions $F_1(x,y)$, $F_2(x,y)$, $F_3(x,y)$, $F_4(x,y)$ which satisfy \emph{pairs}
of second order PDEs \cite{appp,erde}. 
In facts, the divisors $\mathscr{S}_{(1)}$, $\mathscr{S}_{(2)}$ and $\mathscr{S}_{(3)}$ are the loci of the regular singularities
of the Picard-Fuchs PDEs satisfied, respectively, by $F_1$, $F_3$, and $F_4$ (the divisor for $F_2$ is obtained from the one for $F_1$ by a linear redefinition of the coordinates \cite{app1}).
Their monodromy representations have the form
\be
\varrho\colon \pi_1(\mathbb{P}^2\setminus\mathscr{S}_{(1)})\to GL(n,\C)\quad
\text{where}\ \ n=\begin{cases}3 & \text{for }F_1\\
4 & \text{for }F_2,F_3,F_4
\end{cases}
\ee

The theory of the Appell PDEs may be used to illustrate the ``stratification''
at the level of the Picard-Fuchs PDEs (cf.\! \S.\,\ref{s:PFstrat}). The restrictions
of the Appell PDE system for $F_1$ to each one of the irreducible components
of the special divisor $\mathscr{S}_{(1)}$ is studied in \textbf{Example 12.2}
of \cite{haraoka}. We quote the result: \emph{the restriction of the PDE
 to each line in the special locus
yields back the Gauss hypergeometric ODE.}  

\subsection{Appell's monodromies?} 
Appell's monodromy representation cannot be directly identified with the monodromy of
a family of Abelian varieties of dimension 3 since the latter has dimension (degree) 6,
while in the Appell case we have 3-dimensional representations
$V_3$ and 4-dimensional ones $V_4$. Indeed, the Appell functions are the higher
dimensional analogue of the Gauss hypergeometric functions ${}_2F_1(z)$
whereas already for $r=2$ the SW periods were expressed in terms of the generalized
hypergeometric functions ${}_4F_3(z)$ which solve Thomae's order-4
linear ODE. These equations have the same singular set, but different order.
For our applications we need the higher dimensional analogue of the Thomae
hypergeometric -- with a dimension-$6$ monodromy representation --
 not Appell's analog of the Gauss equation (which is good enough as long as
 understanding the special divisor $\mathscr{S}$ goes).

However we may still think of a ``functorial'' construction where we get a dimension 6 monodromy representation $V_6$ of $\pi_1(\mathbb{P}^2\setminus \mathscr{S}_{(1)})$
starting from Apell's monodromies $V_3$ and $V_4$ of dimension (respectively) 3 and 4. 
Applying 
 the standard operations of linear algebra to $V_3$, $V_4$
 one constructs ``functorially'' various 6-dimensional representations such as $V_3\oplus V_3$ or $V_3\odot V_3$ or $V_4\wedge V_4$.
More generally, we may consider non-trivial factorizations of the form
\be
\mu\colon \pi_1(\mathbb{P}^2\setminus\mathscr{S}_{(1)})\to G(\Z)\xrightarrow{\ \;\iota\ } Sp(6,\Z)
\ee
where we assume $G(\Z)$ to be the smallest group with this property.
The Mumford-Tate group $G$ of the Abelian family $\mathscr{A}\to\mathring{\mathscr{P}}$ \cite{griffithsMT,griffithsMT2,periods} is the connected component of the
$\mathbb{Q}$-algebraic group which is   
 the Zariski $\mathbb{Q}$-closure $\overline{G(\Z)}^{\,\mathbb{Q}}$ of the discrete subgroup $G(\Z)\subset Sp(6,\mathbb{Q})$. The real Lie group $G(\R)$ of real-valued
points of $G$ is a Hermitian group\footnote{\ A Hermitian group is a real Lie group $G$ such that the symmetric space $G/K$ ($K\subset G$ a maximal compact subgroup) is K\"ahler. All Hermitian groups are of Hodge type (the converse is false).}. $\iota$ extends to
an inclusion of real Lie groups $G(\R)\xrightarrow{\;\iota\;} Sp(6,\R)$ given by a representation
which is symplectic in the ``motivic'' sense \cite{Deligne?,milneShi}.
The allowed pairs of algebraic groups and symplectic representations (for any rank $r$) may be read for instance in the table on page 10 of
\cite{milneShi}. If $G(\R)\subset Sp(6,\R)$ is a non-trivial subgroup, it
can be only one of the  semi-simple real Hermitian maximal subgroups
 \be
 SL(2,\R),\quad SL(2,\R)\times SL(2,\R),\quad SL(2,\R)\times Sp(4,\R),\quad
SU(2,1)
\ee
or a subgroup thereof. The first two have no real dimension \textbf{6} representation in the table \cite{milneShi}.
For the third one the only allowed representation is $\mathbf{2}\oplus\mathbf{4}$ which leads
to a reducible representation. As we learned from the
reducible examples in rank-2, this may still produce a non-trivial geometry
(for instance a quasi-isotrivial geometry with a generic fiber which is isogenic to the
product of a fixed elliptic curve with an Abelian surfaces whose 4 periods
satisfy an Appell's PDE) but here we are mainly interested in $\mu$-irreducible geometries and leave apart this possibility for future work. 

We remain with $G(\R)=SU(2,1)$ and a real representation $V_6$ such that
$V_6\otimes\C\simeq \boldsymbol{3}\oplus \boldsymbol{\overline{3}}$.
The embedding of $SU(2,1)$ into $Sp(6,\R)$ is as follows: identify $Sp(6,\R)$ with the group of real 
$6\times 6$ matrices $S$ satisfying
{\renewcommand{\arraystretch}{1.2}\be
S^t\,\tilde\Omega\, S=\tilde\Omega\quad\text{where }\ \tilde\Omega=\begin{pmatrix}0 & \eta\\
-\eta &0\end{pmatrix}\ \ \text{and } \eta=\mathrm{diag}(+1,+1,-1)
\ee
and $SU(2,1)$ with the group of $3\times 3$ complex matrices $U$ such that $U^\dagger \eta\, U=\eta$.
Then 
\be\label{uuu666u}
SU(2,1)\hookrightarrow Sp(6,\R),\quad U\mapsto \begin{pmatrix}\mathrm{Re}\,U & \mathrm{Im}\,U\\ -\mathrm{Im}\,U & \mathrm{Re}\,U\end{pmatrix}.
\ee}
\hskip-4pt We have the constraint that the monodromy group should contain a non-trivial unipotent element $S_0$
 i.e.\! such that $S_0-\boldsymbol{1}$ is nilpotent of rank 1, in facts 
 \be
 S_0=\boldsymbol{1}+v\otimes v^t\,\tilde\Omega\qquad v\neq0.
 \ee
 But $S_0$ is not of the form in the right side of \eqref{uuu666u},
 so $S_0\not\in SU(2,1)$. We conclude
 \begin{fact} A non-isotrivial, $\mu$-irreducible rank-$3$ special geometry
 has the standard Mumford-Tate group $Sp_6$.
 \end{fact} 
 Thus we cannot construct a rank-3 $\mu$-irreducible special geometry by a simply-minded
 use of Appell's transcendents. 
 To get a candidate monodromy for a putative $\mu$-irreducible (and $\mu$-rigid)
 special geometry based on our 
 ``economic'' analogue of the 3-punctured sphere we should look for ``new'' transcendents
 which generalize to multiple variables the Thomae ones. 
 In the next subsection we present evidence that nice transcendents of the appropriate kind exist.
 
 \subsection{An ``economic'' $r=4$ example}\label{r4eco}
 
%
%
 
 We look for would-be for rank-4 analogues of the 3-punctured sphere.
 For simplicity we take $\mathscr{P}$ to be $\mathbb{P}^3$: the physically relevant case of
 a weighted projective space $\mathbb{P}(d_1,\dots,d_4)$ can be reduced to this one
 by a finite cover as before.

  To construct a natural analogue in $\mathbb{P}^3$ of the planar configuration
 of a conic with 3 tangent lines we start from the
 Cayley cubic $X_3\subset\mathbb{P}^3$ which is the unique (modulo $PGL(4,\C)$
 projective equivalence) singular normal cubic hypersurface with 4 singularities
 which are of type $A_1$ (simple nodes) \cite{classicalAG,klein}. Choosing the coordinates
 so that the 4 nodes are at
 \be
 (1:0:0:0),\quad (0:1:0:0),\quad (0:0:1:0),\quad (0:0:0:1),
 \ee
 the equation is
 \be
 x_1x_2x_3+x_2x_3x_4+x_3x_4x_1+x_4x_1x_2=0.
 \ee
 Since the Cayley cubic is singular, it has only 9 lines and 11 tritangent planes
 instead of 27 lines and 45 tritangents as a generic cubic. The lines
 are of two types \cite{klein}:
 {\renewcommand{\arraystretch}{1.4}
 \be
 \begin{array}{lll}\hline
\text{type} &\text{description} & \text{equation}\\\hline
 \textbf{A} &\text{6 lines through 2 nodes\quad} &  x_i=x_j=0,\quad 1\leq i < j\leq 4\\
 \textbf{B} & \text{3 other lines} &  x_i+x_j=x_k+x_l=0,\quad {i,j,k,l}=\{1,2,3,4\}\\\hline
 \end{array}
 \ee
 The 11 tritangent planes (i.e.\! planes tangent in three lines, counted with multiplicity) are:
 \be
 \begin{array}{lll}\hline
$\#$ & \text{equation} & \text{it contains}\\\hline
4 &x_i=0,\quad i=1,\dots,4 & \text{3 distinct \textbf{A} lines}\\
6 & x_i+x_j=0,\quad 1\leq i<j\leq 4 & \text{1 \textbf{A} line (multiplicity 2) \& 1 \textbf{B} line} \\
1&x_1+x_2+x_3+x_4=0 & \text{3 lines of type \textbf{B}} \\\hline
 \end{array}
 \ee}
 
 \begin{propo} As an example of rank 4 analogue to the 3-puncture sphere we
 propose the complement in $\mathbb{P}^3$ of the divisor $\mathscr{S}_8$ whose 8 irreducible components are
 \begin{itemize}
 \item the Cayley cubic $X_3$;
 \item the 6 tritangent planes containing double lines through two nodes;
 \item the tritangent plane containing the 3 lines of type \textbf{B}, 
 \end{itemize}
 that is
 \begin{multline}
 (x_1+x_2) (x_1+x_3) (x_1+x_4) (x_2+x_3) (x_2+x_4) (x_3+x_4)\times\\
 \times (x_1+x_2+x_3+x_4)(x_1x_2x_3+x_2x_3x_4+x_3x_4x_1+x_4x_1x_2)=0
 \end{multline}
 \end{propo}
 Clearly this divisor is rigid. Let us check hereditarity. The restriction to 
 a tritangent containing two lines, say
 $x_1+x_2=0$, followed by the change of coordinates $(x_1,x_3,x_4)=(y_3,2y_1-y_3,-2y_2+y_3)$
 gives 
 \be
 -128\,y_1y_2y_3^2(y_1-y_3)(y_2-y_3)(y_1-y_2)^2=0
 \ee
 which, after reduction, is exactly the rank-3 line configuration in figure \eqref{ju76z}
 Clearly adding further planes would spoil this property.
 The restriction to the last tritangent plane yields, after reduction, a plane with 3 lines in general position,
 which is consistent provided that the restriction of the rank-4 geometry to this plane
 is a product of rank-1 geometries (actually just isotrivial after taking into account the finite cover). 
 
 On the other hand the birational interpretation (as a singular del Pezzo surface) of the Cayley cubic
sees it as the image of the linear system of cubics passing through the six points which are
 the pairwise intersections of 4 lines in general position in $\mathbb{P}^2$ \cite{klein}. This yields the rational parametrization of $X_3\subset\mathbb{P}^3$
 \be
 \begin{split}
 &(x_1:x_2:x_3:x_4)=\\
 &=(y_2^2y_3+y^2_3y_2+y_1y_2y_3:
 y_3^2y_1+y^2_1y_3+y_1y_2y_3:y_1^2y_2+y^2_2y_1+y_1y_2y_3:
 -y_1y_2y_3)
 \end{split}
 \ee 
 so that $X_4$ is an admissible Coulomb branch in the sense of \textbf{Definition \ref{uuu51234}}.
 
 \subsubsection{Relation to recent work on Fuchsian systems}
 
 The divisor $\mathscr{S}_8$ has recently appeared in the math literature
 as the singular locus of a Fuchsian system
 of rank 8 in $\mathbb{P}^3$, see ref. \cite{Fuch8}. In other words
 that paper gives a monodromy representation
 \be
 \mu\colon \pi_1(\mathbb{P}^3\setminus \mathscr{S}_8)\to GL(8,\C)
 \ee
 which when restricted to a suitable irreducible component of $\mathscr{S}_8$ yields
 a \emph{rank 6} representation
 \be
 \phi\colon\pi_1(\mathbb{P}^2\setminus\mathscr{S}_{(1)})\to GL(6,\C).
 \ee
 This seems exactly the gadged we need: a dimension 6 representation of the rank-3
 group $\pi_1(\mathbb{P}^2\setminus\mathscr{S}_{(1)})$ which arises from the
 restriction of a dimension 8 representation for the rank-4 group $\pi_1(\mathbb{P}^3\setminus \mathscr{S}_8)$ as expected from ``stratification'' of a rank-4 special geometry. Indeed this is almost right but, just as in rank-2,
 the math literature mostly focuses on representations with semisimple
 local monodromies, and so we cannot simply borrow
 their results and get a solution to the inverse problem of
 special geometry. This remark applies \emph{a fortiori}
 in higher rank.
However the existence of natural Fuchsian equations and monodromies
of roughly the correct type, with the proper ``stratification'' properties
yields us confidence that we are on the right track.

\subsection{Lauricella singular divisors and all that}

From the above examples we infer that a plausible class of special
divisors $\mathscr{S}\subset \mathbb{P}^{\,r-1}$ with all the required
properties (the ``higher dimensional analogues of the 3-punctured sphere'')
are given by the loci of (regular) singularities of flat Gauss-Manin
connections which generalizes the hypergeometric equation.
Lauricella \cite{lauricella} gave a list of fours hypergeometric functions
$F_A(x_i)$, $F_B(x_i)$, $F_C(x_i)$, $F_D(x_i)$ depending on $n$ variables (and suitable parameters)
for all $n\in\mathbb{N}$. The singular loci of the Fuchsian PDEs satisfied by
the Lauricella functions of $n$ variables yield examples of rigid special divisors $\mathscr{S}$
with the required properties. The singular loci of the PDEs satisfied by the Lauricella functions
are studied e.g.\! in \cite{sing1,sing2}. 

The singular locus of the PDE for $F_C$ with $m$ variables, when restricted to the affine
chart $\C^m\equiv\mathbb{P}^m\setminus H_\infty$  is the zero set
of the reducible polynomial \cite{sing1}
\be
\prod_{i=1}^m x_i \prod_{\varepsilon_i=\pm1}(1+\varepsilon_1\sqrt{x_1}+\cdots+\varepsilon_m\sqrt{x_m})
\ee
which for $m=2$ yields back the three lines $x_a=0$ plus the conic \eqref{juiqw12uu}
i.e.\! the special set $\mathscr{S}_{(3)}$.

For the Lauricella function $F_A$ of $m$ variables one gets the singular locus \cite{sing2}
is the zero set of the polynomial \cite{sing2}
\be
\prod_{i=1}^m x_i \prod_{\{i_1,\dots,i_r\}\subset \{1,\dots,m\}}\left(1-\sum_{p=1}^rx_{i_p}\right)
\ee
so for $m=2$ we get $x_1x_2(1-x_1)(1-x_1-x_2)=0$ which with the change of variables
$x_2\leftrightarrow 1-x_2$ gives $x_1x_2(1-x_1)(1-x_2)(x_2-x_1)=0$ i.e.\!
the special set $\mathscr{S}_{(1)}$ (reinserting the line at infinity). Lauricella's function $F_B$
gives an equivalent singular locus.

Finally Lauricella function $F_D$ has a singular locus
\be
\prod_{i=1}^m x_k(1-x_k) \prod_{1\leq i<j\leq m}(x_i-x_j)=0
\ee
which for $m=2$ gives back $\mathscr{S}_{(1)}$.

These Lauricella singular loci (together with the hyperplane at infinity) give
the basic examples of admissible special divisors for \emph{smooth}
special geometry in rank $r=m+1$ when $\mathscr{P}\simeq \mathbb{P}^{r-1}$.
In addition we have the generalizations of the special divisor in \S.\,\ref{r4eco} based on the
Klein cubic for $m=3$, and finite covers thereof.
The situation is similar when $\mathscr{P}$ is a general weighted projective space,
and indeed the difference is a finite group cf.\! \S.\,\ref{poi8871}.

\subsection{Conclusion}

The examples and the discussion above strongly suggest that the classification
(possibly up to finite covers) of all complements $\mathscr{P}\setminus\mathscr{S}$
with the required properties is feasible for all $r$, and certainly for $r=3$.
When we have a list of putative special divisors, we can state the inverse program
in quite explicit terms. The monodromy representation is expected to be rigid,
and then in principle one can proceed with the recursion in the rank of the inverse problem which however leads to
an enormous ``vast program''.

\section*{Acknowledgments}
I have benefited of discussions with Michele Del Zotto, Mario Martone, Robert Moscrop, and Artan Sheshmani.
I thank Chris Peters, Carlos Simpson, Alexandru Dimca, Yoshishige Haraoka, and V.P. Kostov
for clarifications on their respective works.

\appendix

\section{Sketch of Lemma \ref{le:zero}}\label{app:zero}

\begin{lem*} The connected component $(\mathscr{X}_0^\text{\rm reg})^0$
 of the smooth locus $\mathscr{X}_0^\text{\rm reg}\subset\mathscr{X}$ of the fiber over the origin $0$
 has the structure
 \be\label{zerostruc}
 (\mathscr{X}_0^\text{\rm reg})^0= A_0\times \C^{r-\ell},\quad\ell\equiv\dim A_0=\dim\mathscr{R}_{\Delta=1},
 \ee
 where $A_0$ is a polarized Abelian variety,
 while the (general) smooth fibers have the form $\mathscr{X}_u=A_0\times\mathscr{Y}_u$ with
 $\mathscr{Y}_u$ a polarized Abelian variety of dimension $r-\ell$. 
 \end{lem*}
 
 We know from the discussion in \S.\,\ref{s:sing} that $(\mathscr{X}_0^\text{\rm reg})^0$ is a connected, commutative, algebraic group over $\C$ of dimension $r$. Hence by the Chevalley-Barsotti theorem it is a fibration over an Abelian variety $A_0$ with fiber a product of copies of $\C$ and copies of $\C^\times$. $\exp(z\,\ce)$ acts by automorphisms of the commutative group $(\mathscr{X}_0^\text{\rm reg})^0$ for all $z\in\C$.
 
 Let $L$ be the (commutative) Lie algebra of the group $(\mathscr{X}_0^\text{\rm reg})^0$
 isomorphic to $T_0(\mathscr{X}_0^\text{\rm reg})^0$ which in turn is isomorphic to 
 \be
 L\simeq (\mathscr{R}_+/\mathscr{R}_+^2)[1]=L_0\oplus (\bigoplus_{\Delta_i>1} L_{\Delta_i-1})
 \ee
 The action of $\exp(z\, \ce)$ on the group induces an action on the Lie algebra which
 acts on each summand as 
 \be\label{howaact}
 \exp(z\,\ce)\Big|_{L_d}= e^{z\,d}.
 \ee
 We can reconstruct the group from its Lie algebra modulo a discrete subgroup $D$
 \be\label{pooomi4}
 (\mathscr{X}_0^\text{\rm reg})^0=\Big(\C^{\ell_0}\times \prod_{d>0}\C^{\ell_d}\Big)\Big/D,
 \qquad \ell_d=\dim L_d. 
 \ee
 The subgroup $D$ should be preserved by the automorphism $\exp(z\,\ce)$ for all $z$. In view of 
 \eqref{howaact}, a point in $D$ has the form $p\times \{0\}\times\cdots\times \{0\}$
 with $p\in \C^{\ell_0}$ so that 
  \be\label{FACTrr}
 (\mathscr{X}_0^\text{\rm reg})^0=(\C^{\ell_0}/D)\times \C^{r-\ell_0}.
 \ee
 We have to show that the group $\C^{\ell_0}/D$ is an Abelian variety (i.e.\! that $D$ is a maximal lattice). For simplicity of notation, we
 assume that there is just one Coulomb coordinate $u_1$ of dimension $1$,
 while $\Delta_i>1$ for $i>1$. The general case is an analogous.
 
Let $u\in\mathscr{C}$ be a generic point. 
 In a neighborhood of $u$ preserved by $\exp(z\ce)$   we write $\Omega= dz^i\wedge du_i$ .
 The Hamiltonian vector field generated by $u_i$ is $\partial/\partial z^i$ which acts on
 the fiber coordinates by translation $z^j\to z^j+\alpha\,{\delta^j}_i$. The kernel of the exponential map
 $\mathsf{exp}\colon\mathsf{Lie}(\mathscr{X}_u)\to\mathscr{X}_u$ is then a lattice in
 $\C^r$ which we write as equivalence relations for the coordinates of the fiber $\mathscr{X}_u$
 \begin{equation}\label{hhhhyq}
 z^i\sim z^i+A^{ij}(u)\,m_j+B^{ij}(u)\,n_j,\qquad m_j,n_j\in\Z,
 \end{equation}
 where the matrix $A(u)$ is invertible while the matrix $A(u)^{-1}B(u)$ is symmetric
 with positive-definite imaginary part. Now
 \begin{equation}
 A^{ij}(e^{z\ce}u)= e^{z(1-\Delta_i)}\,A^{ij}(u),\qquad  B^{ij}(e^{z\ce}u)= e^{z(1-\Delta_i)}\,B^{ij}(u)
 \end{equation}
 which shows that as $z\to -\infty$ (that is, as $e^{z\ce}u\to 0$), the fiber coordinates
 $w_i$ with $i>1$ become valued in $\C$, while
 \begin{equation}
 z^1\sim z^1+ A^{1j}(u)\,m_j+B^{1j}(u)\,n_j.
 \end{equation}
Comparing with eq.\eqref{pooomi4}, we get $D\equiv \{A^{1j}(u)m_j+B^{1,j}n_j\}\subset\C$.
$D\subset \C$ is discrete subgroup, hence a lattice of rank $\leq2$.
If the rank is $<2$, we get from \eqref{hhhhyq} that $\mathscr{X}_u$ was not
compact, contrary to the assumption.  Returning to the general case,
we see that the factor $A_0\equiv\C^{\ell_0}/D$ in {FACTrr}
 $(\mathscr{X}_0^\text{\rm reg})^0$ is an Abelian variety of dimension equal to the
 multiplicity of 1 as a Coulomb dimension, while the general fiber
 has the form $\mathscr{X}_u=A_0\times \mathscr{Y}_u$ for some Abelian variety 
 $\mathscr{Y}_u$. The fiber is then a product of the fixed Abelian variety $A_0$
 with the $u$-dependent variety  $\mathscr{Y}_u$. Since the Lie algebra of $A_0$
 is given by the Hamiltonian fields of the $u_i$'s with $\Delta_i=1$,
 the geometry is a global product of a trivial geometry with base $\C^{\ell_0}$
 and fixed fiber $A_0$ and a rank $(r-\ell_0)$ interacting geometry.

\section{Proof of Lemma \ref{ttthelemma}}\label{defpprr}

\begin{lem*} A monodromy representation associated to the graph {\rm(e)}
with local monodromies in the allowed lists is irreducible only when
$\varkappa=2$ and the three semisimple monodromies have eigenvalues
\be
(1,1,-1,-1),\quad (i,i,-i,-i),\quad (\pm \zeta_3,\pm \zeta_3,\pm\zeta_3^{-1},\pm\zeta_3^{-1}),
\ee
where $\zeta_3$ is a primitive 3-rd root of 1.
\end{lem*}

\begin{proof}
We consider the graph (e)
\begin{tiny}$$
\begin{gathered}
\xymatrix{& \text{\ovalbox{2}}\ar@{-}[d]\\
\text{\ovalbox{2}}\ar@{-}[r] & \text{\ovalbox{4}}\ar@{-}[r]& \text{\ovalbox{2}}\\
& \text{\ovalbox{1}}\ar@{-}[u]}
\end{gathered}
$$\end{tiny}
The possible fugacities are
\be
\begin{aligned}
&\begin{smallmatrix}&& \zeta_2^{-2}\\
\zeta_1^{-2} && \pm \zeta_1\zeta_2\zeta_3 && \zeta_3^{-2}\\
&& 1\end{smallmatrix}&\qquad &
\begin{smallmatrix}&& -1\\
\zeta_1^{-2} && \pm \zeta_1 && -1\\
&& 1\end{smallmatrix}\\
\\
&\hskip10pt\begin{smallmatrix}&& -1\\
-1 & & \pm 1 & & -1\\
&& 1\end{smallmatrix}
&\qquad 
&
\begin{smallmatrix}&& \zeta_2^{-2}\\
\zeta_1^{-2} & & \pm \zeta_1\zeta_2 & & -1\\
&& 1\end{smallmatrix}
\end{aligned}
\ee 
where $\zeta_a$ stands for a root of unity of degree 2 (i.e.\! a primitive 3-rd, 4-th, or 6-th root of unity).
In the first and second cases, for either sign, we have  the $R_+$ decomposition
\be
\begin{smallmatrix}&1\\
1&2&1\\ & 0\end{smallmatrix}\ +\ \begin{smallmatrix}&1\\
1&2&1\\ & 0\end{smallmatrix}\ +\ \begin{smallmatrix}&0\\
0&0&0\\ & 1\end{smallmatrix} 
\ee
We consider the third case. If the central fugacity is $+1$, the central simple root
$\alpha_\star$ and $(\alpha-\alpha_\star)$ are in $R_+$.
We take the opposite sign $-1$. Then we have the decomposition in $R_+$
\be
\begin{smallmatrix}&1\\
1&1&1\\ & 0\end{smallmatrix}\ +\ \begin{smallmatrix}&1\\
0&1&0\\ & 0\end{smallmatrix}\ +\ \begin{smallmatrix}&0\\
1&1&0\\ & 0\end{smallmatrix}\ +\ \begin{smallmatrix}&0\\
0&1&1\\ & 0\end{smallmatrix}\ +\ \begin{smallmatrix}&0\\
0&0&0\\ & 1\end{smallmatrix}
\ee
It remains the forth case. Suppose $\pm\zeta_1\zeta_2=1$. Both $\alpha_\star$
and $(\alpha-\alpha_\star)$ are in $R_+$. 
Suppose $\pm\zeta_1\zeta_2=-1$. Then we have the $R_+$ decomposition
\be
\begin{smallmatrix}&2\\
2&3&1\\ & 1\end{smallmatrix}\ +\ \begin{smallmatrix}&0\\
0&1&1\\ & 0\end{smallmatrix}
\ee
Rearranging  the order of the eigenvalues, we get the same conclusion whenever $\pm\zeta_1\zeta_2^{-1}=\pm1$.
Therefore the situation with graph (e) when the representation may be irreducible is
the forth case with the minus sign (possible only if $\varkappa=2$)
and $\zeta_1=e^{\pi i/2}$ and $\zeta_2=\pm e^{2\pi i/3}$. 
\end{proof}

\section{Deferred proof from \S.\,\ref{s:noother}}\label{dddeeff}

\begin{lem}\label{finally!} Let $\Lambda_n$ be $\Z^{2n}$ equipped with the standard principal integral 
symplectic form $\langle-,-\rangle$. We are given four elements $n_i\in \Lambda_2$ ($i=1,\dots,4$)
such that:
\begin{itemize}
\item[\rm(1)] the $n_i$ are linear independent over $\mathbb{Q}$;
\item[\rm(2)] each $n_i=(n_{i,a})$ is primitive i.e.\! $\gcd\nolimits_a(n_{i,a})=1$ for all $i$;
\item[\rm(3)] for $1\leq i<j\leq4$ we have $\langle n_i,n_j\rangle=2$.
\end{itemize} 
Define $n_5\equiv -n_1+n_2-n_3+n_4$. We have  $n_5\equiv0\bmod2$.
\end{lem}
\begin{proof}
We write $e_i$ for the standard basis elements of $\Lambda_2$ with 
\be
\langle e_1,e_2\rangle=-\langle e_2,e_1\rangle=\langle e_3,e_4\rangle=-\langle e_4,e_3\rangle=1
\ee
 and all other zero. 
Since $n_1$ is primitive, by a $Sp(4,\Z)$ rotation we may assume $n_1=e_1$.
Then
\be
n_i = 2e_2+ a_i e_1 +z_i\quad i=2,3,4,\quad a_i\in\Z,\quad z_i\in \Z e_3+\Z e_4\equiv \Lambda_1. 
\ee
We have ($i,j=2,3,4$)
\be\label{poiuqaa}
\langle n_i,n_j\rangle =2(a_i-a_j)+\langle z_i, z_j\rangle =\begin{cases}2 & i<j\\
-2 & i<j\\
0 & i=j\\
\end{cases}
\ee  
Let $\bar z_i$ be the reduction mod 2 of $z_i$. The $\bar z_i$'s belong to the 
principal symplectic $\mathbb{F}_2$-space $\mathbb{F}_2^2$.
Eq.\,\eqref{poiuqaa} says that the span of the $\bar z_i$'s is either trivial or a Lagrangian $\mathbb{F}_2$-subspace. Let us show that it is not trivial. Suppose that $\bar z_i=0$ for $i=2,3,4$.
By (2) $a_i$ must be odd for $i=2,3,4$ and so the \textsc{lhs} of \eqref{poiuqaa} is
0 $\bmod4$ for all $i,j$, giving a contradiction. We conclude that the $\mathbb{F}_2$-span of $\bar z_2,\bar z_3,\bar z_4$ is
one-dimensional. Suppose (say) that $\bar z_3\neq 0$.  
By a $Sp(2,\Z)$ rotation we may assume $z_3=(1+2m)e_3$ for some $m\in\Z$.
Now we have
\be
\begin{aligned}
n_1&=e_1\\
n_2&= 2e_2+ a_2 e_1+ b_2 e_3 +2 c_2 e_4\\
n_3&= 2 e_2+ a_3 e_1 +(1+2m)e_3\\
n_4 &= 2 e_2+a_4 e_1+ b_4 e_3 + 2 c_4 e_4
\end{aligned}
\ee
where $a_2,b_2$ and, respectively, $a_4,b_4$ cannot be both even. We have
\be\label{tttty667}
\begin{aligned}
1&=\tfrac{1}{2}\langle n_2,n_3\rangle= a_2-a_3-(1+2m)c_2\\
1&=\tfrac{1}{2}\langle n_3,n_4\rangle=a_3-a_4+(1+2m)c_4\\
1&=\tfrac{1}{2}\langle n_2,n_4\rangle=a_2-a_4+b_2c_4-b_4 c_2.
\end{aligned}
\ee
From now on all equalities are meant to be $\bmod\;2$. We have
\be
\begin{aligned}
0&\equiv(a_2-a_4)-(c_2-c_4)\\
1&\equiv (a_2-a_4)(1-b_4)+(b_2-b_4)c_2\mod2
\end{aligned}
\ee
$b_2$, $b_4$ cannot be both odd. If they are both even, $a_2$, $a_4$ should be both odd, and the equation is still impossible. So one of the $b_2$, $b_4$ is even and the other odd. Say $b_2\equiv0$ and $b_4\equiv1$. Then $a_2\equiv1$, $c_2\equiv1$, and $c_4\equiv a_4$. Eqs.\eqref{tttty667} give (mod 2)
\be
a_3\equiv 1,\qquad a_4\equiv1.
\ee
Now (mod 2)
\be
\begin{aligned}
n_1&\equiv e_1&\qquad n_2&\equiv e_1\\
n_3&\equiv e_1+e_3 & n_4&\equiv e_1+e_3
\end{aligned}
\ee
and
\be
n_5\equiv0\bmod2.\qedhere
\ee
\end{proof}

\section{\emph{A priori} restrictions on the allowed Kodaira types}\label{appeApriori}

Three out of four allowed Dynkin graphs have a $\delta=1$ branch: (a), (d), and (f)
with the fugacity in eq.\eqref{fffufug}. We know that $\delta=1$ corresponds to a non-enhanced
divisor component of type $I_n$ with $n>0$. We wish to give an \emph{a priori}
restriction on the value of $n$.

\begin{fact}\label{jt5554x} In a non-isotrivial irreducible rank-2 special geometry with non-enhanced
divisors of type $I_{n_1},\dots, I_{n_t}$ we have:
\begin{itemize}
 \item[\bf(1)] $p\nmid n_k$ for all odd primes $p$ and $k=1,\dots,t$
\item[\bf(2)] $\gcd(n_k)\in\{1,2\}$
\item[\bf(3)] in graph {\rm(d)} if $\gcd(n_1,n_2)=2$ then the dimensions are $\{2,4\}$
\item[\bf(4)] $n_k/n_h\in \mathbb{Q}^2$
\end{itemize}
\end{fact}
\begin{proof}[Argument] \textbf{(1)}
Let us suppose $p\mid n$. The Crawley-Boevey solution of the Deligne-Simpson problem
holds on every algebraically closed field $\mathbb{K}$. We take for $\mathbb{K}$ the algebraic closure of $\mathbb{F}_p$. The reduction mod $p$ of the local monodromy of type $I_n$ is the identity matrix
and the graph loses the corresponding branch. As a $\mathbb{K}$-representation the reduction mod $p$
of the original monodromy representation $V_\mathbb{K}$ is reducible, since the graph amputated of the $\delta=1$ branch is not a root. Suppose $V_\mathbb{K}$ is a direct sum of irreducible representations. Each summand must be
a root in $(\Sigma_+)_p$ which now is defined using the fugacities (i.e.\! the eigenvalues) of the
local monodromies mod $p$ or more precisely by the roots of their minimal polynomials mod $p$. 
Minkowski theorem (and its generalizations \cite{serreM}) implies that if $\ell$ is the minimal
positive integer such that the integral-valued matrix $M$ satisfies
$(\mu^\ell-\boldsymbol{1})^{k+1}$ for some $k\in\mathbb{N}$, the reduction of $M$ mod an odd prime would satisfy a similar condition with the same minimal $\ell$ (but possibly a different $k$).
Hence, apart for a few special cases, the fugacities will remain roots of unity (in the appropriate
algebrically closed field) of the same order after the reduction mod $p$.
An exception is, for instance, graph (a) where the two long branches represent monodromies with minimal
polynomials $\Phi_1(z)^2\Phi_6(z)$ and respectively $\Phi_2(z)^2\Phi_3(z)$, without common roots,
but which have the same order $\ell=6=2\cdot 3$: in this case we may have $3\mid n$.
However these special cases do not correspond to special geometries since
no dimension pair $\{\Delta_1,\Delta_2\}$ is consistent with these minimal polynomials (as we show in \S...).
For the cases which do are allowed we cannot find a decomposition of the  root with the $\delta=1$ branch omitted into elements of $(\Sigma_+)_p$ (otherwise the original representation would also be reducible) i.e.\!
the monodromy representation with $p\mid n$ is impossible. 

\textbf{(2)(3)} For graph (f) this follows from the direct analysis of \S.\,\ref{s:noother}.
Graphs (a) and (d) have precisely two branches with $\delta\neq1$ with local monodromies
$\mu_0$, $\mu_\infty$. Let
$\ell=\gcd(n_k)$ which is a power of 2 by {\bf(1)}.  One has
\be
\det[z-\mu_\infty]-\det[z-\mu_0]=0\bmod \ell,
\ee
For graph (a) this relation can be satisfied only for $\ell=1,2$ (with the exception mentioned above),
see also \S.\,\ref{yuyuy}.
For graph (d) we have
$\det[z-\mu_0]=\Phi_m(z)^2$, $\det[z-\mu_0]=(z^2-1)^2$ which are congruent modulo $\ell>1$
only when $m=4$ (i.e. $\{\Delta_1,\Delta_4\}=\{2,4\}$) in which case
\be
\Phi_4(z)-(z^2-1)^2=4 z^2.
\ee
However iff $\mu_\infty=\mu_0\bmod4$, the image of $\mu_\infty$
in $GL(4,\Z/4\Z)$ would be semisimple with minimal polynomial $z^2+1=0\bmod 4$
which is inconsistent with the fact that $1$ is an eigenvalue. We remain with $\gcd(n_1,n_2)=1,2$. 

\textbf{(4)} By symmetry of the graphs under interchange of any two $\delta=1$ branches
and rigidity
there exists $U\in Sp(4,\C)$ such that
\be
M_{\pi(i)}=U M_i U^{-1} 
\ee
where $\pi$ is the permutation of the indices $\{1,2,\dots,s\}$ labelling the branches
which corresponds to the interchange of the given two $\delta=1$ branches. Since the $M_i$'s
are integral matrices, $U$ is defined over $\mathbb{Q}$, so (by irreducibility) $U\in Sp(4,\mathbb{Q})$.
Under conjugacy in $Sp(4,\mathbb{Q})$ the type $n$ of a $I_n$ conjugacy class will change as
\be
n\leadsto n^\prime q^2\quad \text{with}\quad q\in\mathbb{Q}.
\ee
It follows that if we have a number of non-enhanced discriminants of type $I_{n_i}$ ($i=1,\dots,s$)
we must have
\be\label{juqw12}
\frac{n_i}{n_j}\in \mathbb{Q}^{\mspace{1mu}2}\quad \text{for all }\ i,j.\qedhere
\ee
\end{proof}

\begin{rem}
The authors of \cite{M4} proposed eq.\eqref{juqw12} on the physical basis of the Dirac quantization of electromagnetic charge.
\end{rem}

\section{Proof of Lemma \ref{ttttt432}}\label{8888bbbv}

\begin{lem*} In a non-isotrivial geometry of class $E_7$ the short $\delta=2$
branch cannot stand for a non-enhanced divisor. 
\end{lem*}

\begin{proof} We write $\mu_2$, $\mu_3$ and $\mu_4$ for the local monodromies
associated to the branches with $\delta=2$, $3$ and $4$, respectively.

If the $\delta=2$ branch is non-enhanced, $\mu_2$ is of $I^*_0$ type i.e.\!
semisimple with two eigenvalues $+1$ and two $-1$. The fugacity of the peripheral node
is $-1$ while the contribution of this branch to the fugacity of the central node is a factor
$\pm1$ depending on the choice of order of the two eigenvalues. In particular, we can always
fix the overall sign of the central fugacity according to convenience.

Suppose $\mu_3$ is non-semisimple. Then it should have eigenvalues $(+1,+1,-1,-1)$
and a non-trivial Jordan block for one of the two eigenvalues. The fugacities 
along the branch (starting from the peripheral node) are $(1)$--$(-1)$--$\star$. 
The fugacities are
\be
\xymatrix{&& (-1)\ar@{-}[d]\\
(1)\ar@{-}[r]&(-1)\ar@{-}[r]& (\pm \xi)\ar@{-}[r]&(\lambda/\xi)\ar@{-}[r]&(1/\lambda\xi)\ar@{-}[r]&(\xi/\lambda)}
\ee
for some complex numbers $\xi$, $\lambda$ with $\lambda\neq\xi^{\pm1}$; when $\lambda\neq\pm1$
$\mu_4$ is semisimple, otherwise has a non-trivial Jordan block. Consider the decomposition
\be
\begin{smallmatrix}&& 1\\
1 & 1 & 2 & 2 & 1 & 1\end{smallmatrix}\ +\ \begin{smallmatrix}&& 1\\
0 & 1 & 2 & 1 & 1 & 0\end{smallmatrix}
\ee
Both elements are positive roots in $R_+$ (cf.\! \textsc{Planche VI} in \cite{bourba}), so the monodromy representation is not irreducible
and should be descarded. We conclude that $\mu_3$ should be semisimple.

Now assume $\mu_4$ is non-semisimple. It should have eigenvalues $(\pm 1,\pm1, \xi,\xi^{-1})$.
By a clever ordering of the eigenvalues we get the fugacities
\be
\xymatrix{&& (-1)\ar@{-}[d]\\
(\eta^2)\ar@{-}[r]&(\eta)\ar@{-}[r]& (- 1)\ar@{-}[r]&(1)\ar@{-}[r]&(\xi)\ar@{-}[r]&(\xi^{-2})}
\ee
and we have the decomposition
\be
\begin{smallmatrix}&& 1\\
1 & 2 & 3 & 3 & 2 & 1\end{smallmatrix}\ +\ \begin{smallmatrix}&& 1\\
0 & 0 & 1 & 0 &0 & 0\end{smallmatrix}
\ee
which are positive $E_7$ roots in $R_+$ (cf.\! \textsc{Planche VI} in \cite{bourba}).
We conclude that $\mu_4$ should also be semisimple. Hence all three monodromies
$\mu_2$, $\mu_3$ and $\mu_4$ are semisimple and the geometry is isotrivial contrary to the assumption. 
\end{proof}

\section{Quasi-isotrivial geometries}\label{A:quasi}

We consider the regularity conditions along the axes of the putative quasi-isotrivial
geometries discussed in \S.\,\ref{SS:quasi-isotrivial}. There are two cases: dimensions $\{3,4\}$
and dimensions $\{4,6\}$. In this appendix $m$ is an integer attached to the Kodaira type
 which is equal 1 for $I_n$, $2$ for $I^*_n$, 3 for $IV^*$, 4 for $III^*$ and 6 for $II^*$.
 For types $IV$, $III$ and $II$ replace in the following fromulae
  $1/m\mapsto (m-1)/m$ with, respectively, $m=3,4,$ and $6$.
 
 We write $\alpha_1$, $\alpha_2$ for the ``exponents'' of a local $\mu$-monodromy $\mu$
 around a special point corresponding to an axis. This means that that eigenvalues
 of $\mu$ are $\{e^{2\pi i\alpha_1}, e^{2\pi i\alpha_2}, e^{-2\pi i\alpha_1},e^{-2\pi i\alpha_2}\}$.

\subparagraph{Dimensions $\{3,4\}$.}
Along the first axis ($u_2=0$)  we must have
\be
\frac{u_2}{u_1} \left(\frac{u_1^4}{u_2^3}\right)^{\!\alpha_1}=u_1^{1/3},\qquad
\frac{u_2}{u_1} \left(\frac{u_1^4}{u_2^3}\right)^{\!\alpha_2}
=\begin{cases}
u_2^{1/m} u_1^{x} &m=1,2,3,4,6\\
u_2^{(m-1)/m} u_1^{x}&m=3,4,6
\end{cases}
\ee
i.e.\! 
\be
\{\alpha_1,\alpha_2\}=\begin{cases}
 \{1/3,(m-1)/3m\} &m=1,2,3,4,6\\
 \{\alpha_1,\alpha_2\}= \{1/3,1/3m\}&m=3,4,6
 \end{cases}
\ee
$\exp(2\pi i\alpha_2)$ should be an algebraic number of degree 2,
and we have
the list of possible exponents
\be
 \{1/3,1/6\}_{m=2},\quad  \{1/3,1/4\}_{m=4}.
\ee

On the second axis: 
\be
\frac{u_2}{u_1} \left(\frac{u_1^4}{u_2^3}\right)^{\!\alpha_1}=u_2^{1/4},\qquad
\frac{u_2}{u_1} \left(\frac{u_1^4}{u_2^3}\right)^{\!\alpha_2}
=\begin{cases}
u_1^{1/m} u_2^{x}\\
u_1^{(m-1)/m} u_2^{x}
\end{cases}
\ee
i.e.\! 
\be
\{\alpha_1,\alpha_2\}= \begin{cases}\{1/4,(m+1)/4m\} &m=1,2,3,4,6\\
 \{1/4,(2m-1)/4m\} &m=3,4,6
\end{cases}
\ee
and $m=1,3$ and we have 
\be
\{1/4,1/2\}_{m=1},\quad \{1/4,1/3\}_{m=3}
\ee
Since the monodromies around the two exceptional fibers of $\ce_2$ have the same exponents,
a pair in this list is allowed only if it has a non-trivial overlap with the list for the first axis.
If both exponents are equal the geometry becomes isotrivial, so we remain with two
pairs of couples
\be
\Big(\{1/3,1/6\}_{m=2}\ \&\ \{1/4,1/3\}_{m=3}\Big)\quad\text{and}\quad
\Big(\{1/3,1/4\}_{m=4}\ \& \ \{1/4,1/2\}_{m=1}\Big)
\ee
which IFF they exist would correspond to geometries with dimensions $\{3,4\}$
and respectively
\begin{itemize}
\item the axis of dimension 
4 is of type $IV^*$ and
the axis of dimension 3 has type $I^*_0$ (it cannot be of type $I^*_n$ with $n>0$
since the $\mu$-monodromy, hence  the $\varrho$-monodromy is semisimple);
\item the axis of dimension 3 has type $III^*$ and the axis of dimension 4 has type
$I_n$ with $n>0$ (since $n=0$ will produce
a ``dummy'' monodromy summand).
\end{itemize}

\subparagraph{Dimensions $\{4,6\}$.}
Along the first axis ($u_2=0$)  we must have
\be
\frac{u_2^{1/2}}{u_1^{1/2}} \left(\frac{u_1^3}{u_2^2}\right)^{\!\alpha_1}=u_1^{1/4},\qquad
\frac{u_2^{1/2}}{u_1^{1/2}} \left(\frac{u_1^3}{u_2^2}\right)^{\!\alpha_2}
=\begin{cases}
u_2^{1/m} u_1^{x} &m=1,2,3,4,6\\
u_2^{(m-1)/m} u_1^{x}&m=3,4,6
\end{cases}
\ee
i.e.\! 
\be
\{\alpha_1,\alpha_2\}=\begin{cases}
 \{1/4,(m-2)/4m\} &m=1,2,3,4,6\\
\{1/4,(2-m)/4m\}&m=3,4,6
 \end{cases}
\ee
and 
the list of possible exponents is
\be
\{1/4,-1/4\}_{m=1},\  \{1/4,0\}_{m=2},\ \{1/4,1/6\}_{m=6},\ \{1/4,-1/6\}_{m=6}
\ee
but only the third one seems to be plausible for integral dimensions.
On the second axis: 
\be
\frac{u_2^{1/2}}{u_1^{1/2}} \left(\frac{u_1^3}{u_2^2}\right)^{\!\alpha_1}=u_2^{1/6},\qquad
\frac{u_2^{1/2}}{u_1^{1/2}} \left(\frac{u_1^3}{u_2^2}\right)^{\!\alpha_2}
=\begin{cases}
u_1^{1/m} u_2^{x} &m=1,2,3,4,6\\
u_1^{(m-1)/m} u_2^{x}&m=3,4,6
\end{cases}
\ee
i.e.\! 
\be
\{\alpha_1,\alpha_2\}= \begin{cases}\{1/6,(m+2)/6m\} &m=1,2,3,4,6\\
\{1/6,(3m-2)/6m\} &m=3,4,6
\end{cases}
\ee
and we have 
\be
\{1/6,1/2\}_{m=1},\quad \{1/6,1/3\}_{m=2},\quad \{1/6,1/4\}_{m=4},
\ee
We disregard the last possibility because it has two exponents in common
with the monodromy around the other axis. We remain with two possibilities
\begin{itemize}
\item dimension $\{4,6\}$ where the axis of dimension 4 has Kodaira type $II^*$
and the axis of dimension 6 has type $I_n$ with $n>0$ (since $n=0$ will produce
a ``dummy'' monodromy summand);
\item  dimension $\{4,6\}$ where the axis of dimension 4 has Kodaira type $II^*$
and the axis of dimension 6 has type $I_0^*$. Indeed the $\mu$-monodromy
(hence the $\varrho$-monodromy) is semisimple around the second axis.
\end{itemize}

\begin{small}
  
\end{small}
%
%
%
%
%
%
%
%
%
%
%
%
%
%
%
%
%
%
%

\end{document}